\documentclass[laurea,oneside,12pt,english]{unatesi}
\usepackage{enumitem}
\usepackage{amsfonts}
\usepackage{amssymb}
\usepackage{graphicx}
\usepackage{array}
\usepackage{appendix}
\usepackage{mathrsfs}
\usepackage{hyperref}
\usepackage{amsmath}
\usepackage{amsthm}
\usepackage{cancel}
\usepackage{inputenc}
\usepackage[]{natbib}
\usepackage{midpage}
\numberwithin{equation}{section}
\theoremstyle{plain} 
\newtheorem{thm}{Theorem}[section] 
\newtheorem{cor}{Corollary}[section]
\newtheorem{lem}{Lemma}[section]
\newtheorem{prop}{Proposition}[section]
\theoremstyle{definition}
\newtheorem{defn}{Definition}[section]
\theoremstyle{remark}
\newtheorem{oss}{Remark}[section]
\addto\captionsitalian{
  }
  \newcommand{\abs}[1]{\lvert#1\rvert}
\university{Universit\`a degli Studi di Napoli ``Federico II''}
\faculty{Scuola Politecnica e delle Scienze di Base\\Dipartimento di Fisica ``Ettore Pancini''}
\logo{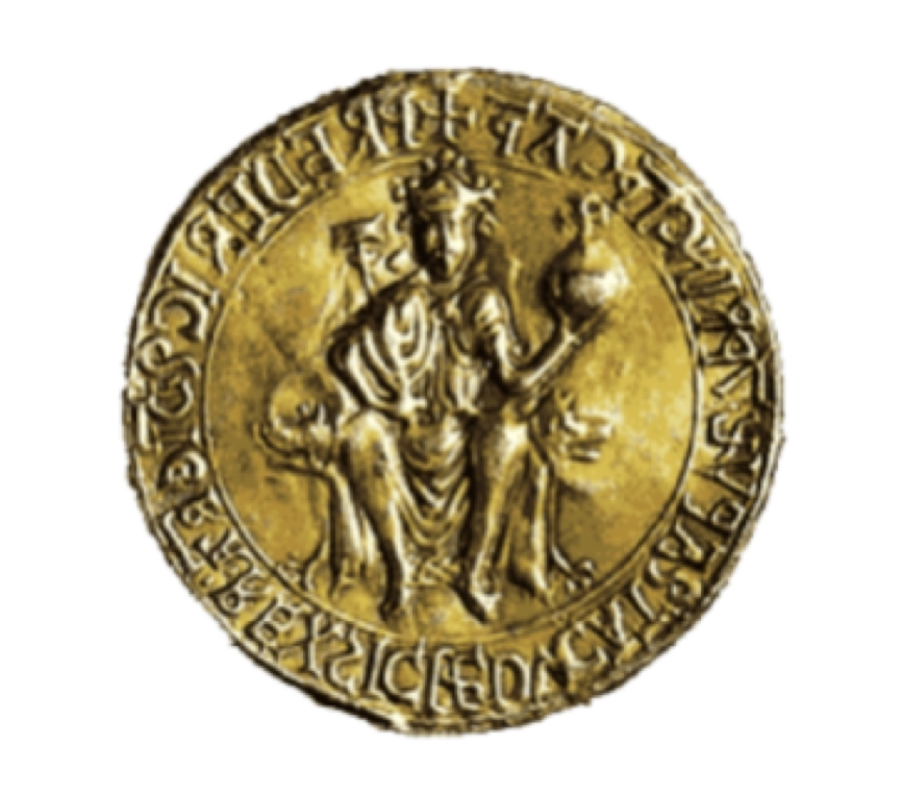}
\course{Fisica}
\degreeyear{2016/2017}

\title{Asymptotic Structure and \\ Bondi-Metzner-Sachs group in General Relativity}
\chair{Prof. Giampiero Esposito}
\numberofmembers{1}
\author{Francesco Alessio}
\matrnum{N94000312}
\begin{document}
\maketitle
\frontmatter
\tableofcontents
\mainmatter
\section*{Introduction}
In this work the asymptotic structure of space-time and the main properties of the Bondi-Metzner-Sachs (BMS) group, which is the asymptotic symmetry group of asymptotically flat space-times, are analysed. Every chapter, except the fourth, begins with a brief summary of the topics that will be dealt through it and an introduction to the main concepts. The work can be divided into three principal parts. \\The first part includes the first two chapters and is devoted to the development of the mathematical tools that will be used throughout all of the work. In particular we will introduce the notion of space-time and will review the main features of what is referred to as its causal structure and the spinor formalism, which is fundamental in the understanding of the asymptotic properties. \\In the second part, which includes the third, fourth and fifth chapters, the topological and geometrical properties of null infinity, $\mathscr{I}$, and the behaviour of the fields in its neighbourhood will be studied. Particular attention will be paid to the peeling property. \\The last part is completely dedicated to the BMS group. We will solve the asymptotic Killing equations and find the generators of the group, discuss its group structure and Lie algebra and eventually try to obtain the Poincar\'e group as its normal subgroup.\\
The work ends with a brief conclusion in which are reviewed the main modern applications of the BMS group.
\chapter{Causal Structure}
\label{chap:1}
\begin{center}
\begin{large}
\textbf{Abstract}
\end{large}
\end{center}
In this first chapter we analyse what is referred to as the causal structure of the space-time. In particular we will start by giving, in the first section, the definition of a space-time, that will be used throughout all of the work. In the other sections the notions of orientability and causality, i.e. chronological and causal past and future sets and their topological properties, will be considered. Having defined these concepts, from section \ref{sect:1.6} we will start to explore the meaning of \lq causality-violating\rq\hspace{0.1mm} space-time and to take into account the restrictions to impose on a space-time for it to be \lq physical\rq. The last section is devoted to the global hyperbolicity and the existence of Cauchy surfaces and will be very important to discuss the asymptotic properties, which are subjects of the last chapters. \\
The main bibliography for this chapter is furnished by the beautiful works of Penrose, Hawking and Geroch. 
\section{Introduction}
\label{sect:1.1}
The mathematical model we shall use for the description of space-time, i.e. the collection of all events, is a four-dimensional manifold $\mathscr{M}$ (see Appendix \ref{A} for definitions). In fact a manifold corresponds naturally to our intuitive ideas of the continuity of space and time. So far this continuity is thought to be valid for distances greater than a certain cut-off of about $10^{-33}$ cm (the Planck length) and actually has been established for distances down to $10^{-15}$ cm by experiments on pion scattering. For the description of phenomena that occur at distances lesser than this cutoff our model for space-time could become inappropriate and other different structures may emerge, due to quantum effects. It is worth remarking that the first physicist who introduced the Planck scale value was the Soviet theoretical physicist Matvei Petrovich Bronstein in his work \lq Quantization of Gravitational Waves\rq\hspace{0.1mm} of 1936 in which he analysed the problem of the measurability of the gravitational field. He calculated the ``\textit{absolute minimum for the indeterminacy}" in the weak-field framework and formulated the following conclusion: \\
``\textit{The elimination of the logical inconsistencies connected with this requires a radical reconstruction of the theory, and in particular, the rejection of a Riemannian geometry dealing, as we have seen here, with values unobservable in principle, and perhaps also the rejection of our ordinary concepts of space and time, replacing them by some much deeper and non-evident concepts.}"
\\ 
\citep{Bron}\\
In such a way, the quantum limits of General Relativity were revealed for the first time.\\
Before investigating the causal structure of space-time, which explores the causal relationships between the events,
we will start by asking the question  ``What is the underlying manifold of our universe?". To answer we need to make some  physical and reasonable assumptions.\\
The first consideration is that no \lq edges\rq \hspace{1.5mm}of the universe have ever been observed. The edges can be mathematically represented by boundaries and hence we assume $\mathscr{M}$ to have no such boundaries. Furthermore we take $\mathscr{M}$ to be a connected Hausdorff manifold.  In fact we don't have knowledge of any disconnected components and moreover there could not be any communication between separated connected components of our universe. The Hausdorff condition says that any pair of points can be separated by disjoint neighbourhoods. Thus violating Hausdorff condition would imply a violation of concept of \lq distinct events\rq.\\
We know that General Relativity requires more than merely a manifold: there must be a metric tensor field $g$ defined over it that possesses a Lorentz signature. The following theorem is remarkable.
\begin{thm}{\citep{Ger68}}$\\ $
Let $\mathscr{M}$ be a connected, Hausdorff 4-manifold with a $C^{\infty}$ Lorentzian metric tensor $g$. Then the topology of $\mathscr{M}$ has a countable basis.
\end{thm}
Thus we may infer that $\mathscr{M}$ is paracompact, according to \ref{thm:cba}. This property, physically, prevents a manifold from \lq being too large\rq . Paracompatness has a number of important consequences. It can be shown that paracompactness implies that all connected components of $\mathscr{M}$ can be covered by a countable family of charts and that there exists a partition of unity which allows us to define a Riemannian metric over $\mathscr{M}$  as discussed in Appendix A (see \eqref{cor:sco}, \ref{thm:par} and \eqref{cor:met}).\\
The order of differentiability, $r$, of the metric must be sufficient for the field equation to be defined. Those equations, involving the metric tensor components $g_{ab}$, can be defined in a distributional sense if $g_{ab}$ and its inverse $g^{ab}$ are continuous and have locally square integrable generalized first derivatives with respect to the coordinate system. But this condition is not sufficient, since it guarantees neither the existence nor the uniqueness of geodesics, for which a $C^2$ metric is required.  In the remainder we will simply assume the metric to be $C^{\infty}$ because probably the order of differentiability of the metric is not physically relevant. In fact, since one can never measure the metric exactly, but only with some margin of error, one could never determine that there is an actual discontinuity in its derivatives of any order. Thus we are led to this definition of space-time:
\begin{defn} $\\ $
\label{defn:st}
A \textit{space-time} $(\mathscr{M},g)$ is a real, four-dimensional connected $C^{\infty}$ Hausdorff manifold without boundary with a globally defined $C^{\infty}$ tensor field $g$ of type $(0,2)$, which is non-degenerate and Lorentzian. By \textit{Lorentzian} is meant that for any $p\in\mathscr{M}$ there is a basis in $T_p$ (the tangent space to $\mathscr{M}$ at $p$) relative to which $g_p$ is represented by the matrix $\mathrm{diag}(1,-1,-1,-1)$.
\end{defn}
\begin{oss} $\\ $
Two space-times $(\mathscr{M},g)$ and $(\mathscr{M'},g')$ will be taken to be equivalent if there is a diffeomorphism $\theta:\mathscr{M}\rightarrow\mathscr{M'}$ which carries the metric $g$ into the metric $g'$, i.e. $\theta_*g=g'$. So it would be more correct to define the space-time to be the equivalence class of $(\mathscr{M},g)$, two space-times being equivalent if their metrics are linked by a diffeomorphism. However we will work with just one representative member of the above mentioned equivalence class.
\end{oss}
\section{Orientability}
\label{sect:1.2}
The presence of the metric tensor enables us to give the following
\begin{defn}$\\ $
Let $(\mathscr{M},g)$ be a space-time, with $p\in\mathscr{M}$. Then any tangent vector $X_p\in T_{p}$ is said to be: \textit{timelike}, \textit{spacelike} or \textit{null} according as $g(X_p,X_p)=g_{ab}X^a_pX^b_p$ (summation over the repeated indices, according to Einstein's convention) is positive, negative or zero.
\end{defn}
The \textit{null cone} at $p$ is the set of null vectors in $T_p$. The null cone in $T_p$ disconnects  the timelike vectors into two separate components, the \textit{future-directed} one and the \textit{past-directed} one. Similarly, the set of all triads of unit, mutually orthogonal, spacelike vectors at $p$ can be divided into two classes, which could be designated the \textit{left-handed} and \textit{right-handed} triads.\\
Physically the designation of future- and past-directed timelike vectors corresponds to a choice of a direction for the arrow of time, while the designation of left- and right-handed triads to a choice of spatial parity. Those choices can be made at each point of $\mathscr{M}$. We may ask whether or not such designations can be made globally over the entire $\mathscr{M}$. Then we are led to the following
\begin{defn}$\\ $
\label{defn:to}
A space-time $(\mathscr{M},g)$ is said to be \textit{time-orientable} if a designation of which timelike vectors are to be future-directed and which past-directed can be made at each of its point, where this designation is continuous from point to point over the entire manifold $\mathscr{M}$.
\end{defn}
Each of these two designations is called a \textit{time-orientation}. A similar definition holds for \textit{space-orientability} and \textit{space-orientation}, involving the triads mentioned above.\\
A space-time is clearly time-orientable if there exists a nowhere vanishing timelike vector field, i.e. one can choose at each point one of the two oppositely directed unit timelike vectors along a given direction, this choice being continuous from point to point. Conversely, since a space-time is paracompact, there exists a smooth Riemannian metric $h_{ab}$ defined on $\mathscr{M}$. Thus at $p\in\mathscr{M}$ there will be a unique future-directed timelike vector $t^a$ that can be chosen to be the unit eigenvector with positive eigenvalue $\lambda$, of $g$ with respect to $k$, i.e. $(g_{ab}-\lambda h_{ab})t^b=0$, $h_{ab}t^at^b=1$. Thus we have obtained the following
\begin{prop}$\\ $
A space-time $(\mathscr{M},g)$ is time-orientable if and only if there exists a smooth non-vanishing timelike vector field $t^a$ on $\mathscr{M}$.
\end{prop}
The above definitions of orientability are not in general the most useful. Fortunately there is a  much simpler characterization, that is given below.\\
Consider a space-time $(\mathscr{M},g)$, and fix a point $p\in\mathscr{M}$. Consider then a closed curve $\gamma\in\mathscr{M}$ beginning and ending at $p$. We can fix a time-orientation at $p$ and carry this choice continuously about $\gamma$ from point to point, until we revert to $p$, and thus have a final orientation which will be either the same or the opposite as with which we began, according to the fact that $\gamma$ is time-preserving or time-reversing. Then by definition of time-orientability we have that $(\mathscr{M},g)$ is time-orientable if any closed curve $\gamma$ through each of its point $p$ is time-preserving. We can easily show the converse. Consider again a point $p\in\mathscr{M}$ and choose there a time-orientation. Then choose a time-orientation at any other point $q\in\mathscr{M}$ by carrying the choice made at $p$ continuously along some curve joining $p$ and $q$. The resulting time-orientation at $q$ will be unambiguous: in fact, given another curve joining $p$ and $q$, it can be combined with the original to obtain a closed curve through $p$, that has to be, by hypothesis, time-preserving. Thus, repeating this procedure for every point of $\mathscr{M}$ we can obtain in each of them a definite time-orientation. We can conclude that
\begin{prop} $\\ $
A space-time $(\mathscr{M},g)$ is time-orientable if and only if every closed curve through each of its points $p$ is time-preserving.
\end{prop}
Again, the same argument can be carried out for space-orientability.\\
Thus, in order to decide whether  or not a space-time is time orientable one has only to \lq test\rq\hspace{0.1mm} time-orientation about all the closed curves through a given point $p$. However the number of curves to test can be reduced considerably, as will be shown. We call two closed curves, $\gamma$ and $\gamma'$ through $p$ \textit{homotopic} if $\gamma$ can be continuously deformed into $\gamma'$, i.e. if there exists a 1-parameter family of curves, $\gamma_{\lambda}$, with the parameter $\lambda$ varying in $[0,1]$, such that $\gamma_0=\gamma$ and $\gamma_1=\gamma'$. For example any two closed curves through the origin in the plane are homotopic, while a closed curve on the annulus that wraps around the hole is not homotopic to one which does not. Since a continuous deformation cannot result in the discontinuous change from time-preserving to time-reversing, two closed homotopic curves are both time-preserving or time-reversing. Then we just need to evaluate the time-orientation of curves belonging to different homotopy equivalence classes. 
\begin{defn} $\\ $
A manifold $\mathscr{M}$ is \textit{simply connected} if any two closed curves through any $p\in\mathscr{M}$ are homotopic.
\end{defn}
Then we have
\begin{prop}$\\ $
Every space-time $(\mathscr{M},g)$ based on a simply connected manifold is both time- and space-orientable.
\end{prop}
We illustrate now what are the appropriate orientation properties of a \lq physically\rq\hspace{0.1mm} realistic model of our universe. Consider a space-time $(\mathscr{M},g)$ and fix a point $p\in\mathscr{M}$. Each point of $\mathscr{M}$ defines several points of a new space-time, $(\mathscr{\tilde{M}},\tilde{g})$, as follows. Consider pairs $(\gamma,q)$, where $q$ is a point of $\mathscr{M}$ and $\gamma$ a generic curve in $\mathscr{M}$ from $p$ to $q$. Two of such pairs $(\gamma,q)$ and $(\gamma',q')$ are called equivalent if $q=q'$ and $\gamma$ and $\gamma'$ are homotopic, and we say $(\gamma',q')\sim(\gamma,q)$. We define $\mathscr{\tilde{M}}$ to be the set of the just defined equivalence classes, i.e. $\mathscr{\tilde{M}}\equiv\{(\gamma',q') :(\gamma',q')\sim(\gamma,q)\,\hspace{2mm}\forall\gamma,q\in\mathscr{M}\}$. Thus, a point of $\mathscr{\tilde{M}}$ is just a point of $\mathscr{M}$ and a curve from $p$ to that point, up to a continuous deformation of the curve. We build the metric $\tilde{g}_{ab}$ defining the distance between two points $\tilde{q}=(\gamma,q)$ and $\tilde{q}'=(\gamma',q')$ of $\mathscr{\tilde{M}}$ just as the distance between $q$ and $q'$ in $\mathscr{M}$. The resulting space-time $(\mathscr{\tilde{M}},\tilde{g})$ is called \textit{universal covering space-time} of $(\mathscr{M},g)$. The two space-times are locally indistinguishable, but globally they are not. In fact if we suppose $(\mathscr{M},g)$ to be simply connected then $(\gamma',q')\sim(\gamma,q)$ only provided that $q=q'$. In this case the universal covering will be identical to the original space-time. To each point of $\mathscr{M}$ there corresponds just one point of $\mathscr{\tilde{M}}$. But if we take $(\mathscr{M},g)$ to be, for example, the two dimensional annulus (which is not simply connected), as shown in figure \eqref{fig:1.1}, then for the point $q$ in $\mathscr{M}$ the curve $\gamma$ can be continuously deformed into $\gamma'$, $(\gamma,q)\sim(\gamma',q)$, and hence the two pairs define the same point on $\mathscr{\tilde{M}}$, but it cannot be deformed continuously into $\gamma''$ which winds one time around the hole, $(\gamma,q)$ is not equivalent to $(\gamma'',q)$ and the two pairs define different points in $\mathscr{\tilde{M}}$. More generally a curve which reaches $q$ from $p$ after winding around the hole $n$ times will be deformable to a curve which winds around the hole the same number of times. Each point $q$ of $\mathscr{M}$, therefore, will give rise to an infinite number of points of $\mathscr{\tilde{M}}$, one for each value of the integer $n$. 
\begin{figure}[h]
\begin{center}
\includegraphics[scale=0.5]{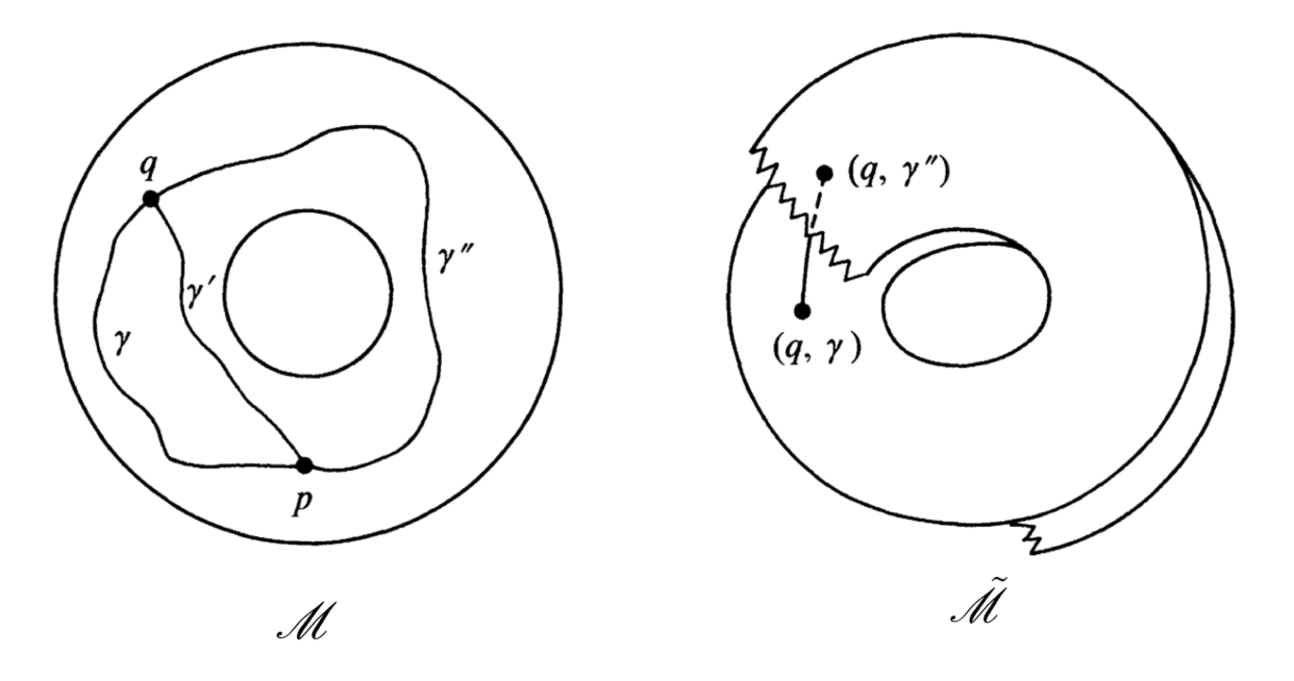}
\caption{The universal covering space $\mathscr{\tilde M}$ of the two-dimensional annulus. Each point of $\mathscr{M}$ defines an infinite number of points of $\mathscr{\tilde{M}}$.}
\label{fig:1.1}
\end{center}
\end{figure}
\\As shown in figure \ref{fig:1.1} we have just unwrapped the annulus.
Universal covering space-times of any space-time are always simply connected and therefore always time- and space-orientable.
 Their importance is related to the fact that they are physically indistinguishable form the original space-time because the only effect of taking the universal covering space-times is to produce possibly several copies of each local region in the original one, leaving the \lq local physics\rq\hspace{0.1mm} unaffected. The conclusion is that no physical possibilities would be lost by demanding that one space-time should be time- or space-orientable. We can eventually state the following
 \begin{prop}
 If a space-time $(\mathscr{M},g)$ is not simply connected (and hence not time-orientable), there always exists a simply connected (and hence time-orientable) space-time $(\mathscr{\tilde{M}},\tilde{g})$, which is its universal covering.
 \end{prop}
\section{The Exponential Map}
\label{sect:1.3}
To proceed further we shall need some simple properties of the \textit{exponential map}. Here and in the remainder we will assume the presence of a $C^{\infty}$, torsion-free connection $\nabla$ on the space-time $(\mathscr{M},g)$. For any $p\in\mathscr{M}$, the exponential map is a smooth ($C^{\infty}$) map from some open subset of the tangent space $T_p$, into $\mathscr{M}$
\begin{equation*}
\mathrm{exp}_p: V\in T_p\longrightarrow \mathrm{exp}_p(V)=q\in\mathscr{M}
\end{equation*}
such that the affinely parametrized geodesic with tangent vector $V$ at $p$ and parameter value $0$ at $p$ acquires the parameter value $1$ at $q$. This map is not defined for all $V\in T_p$, since a geodesic $\gamma(t)$ may not be defined for all $t$. If $t$ takes all values the geodesic is said to be a \textit{complete} geodesic. The manifold $\mathscr{M}$ is said to be \textit{geodesically complete} if all geodesics on $\mathscr{M}$ are complete, i.e. $\mathrm{exp}_p$ maps the whole $T_p$ into $\mathscr{M}$ for every $p\in\mathscr{M}$. That means that every affinely parametrized geodesic in $\mathscr{M}$ extends to arbitrarily large parameter values. But whether $\mathscr{M}$ is complete or not, it may well be that several different elements of $T_p$ are mapped to the same point of $\mathscr{M}$, as shown in Figure \ref{fig:1.3}, or that the map is badly behaved for certain elements of $T_p$ (because its jacobian vanishes). 
\begin{figure}[h]
\begin{center}
\includegraphics[scale=0.39]{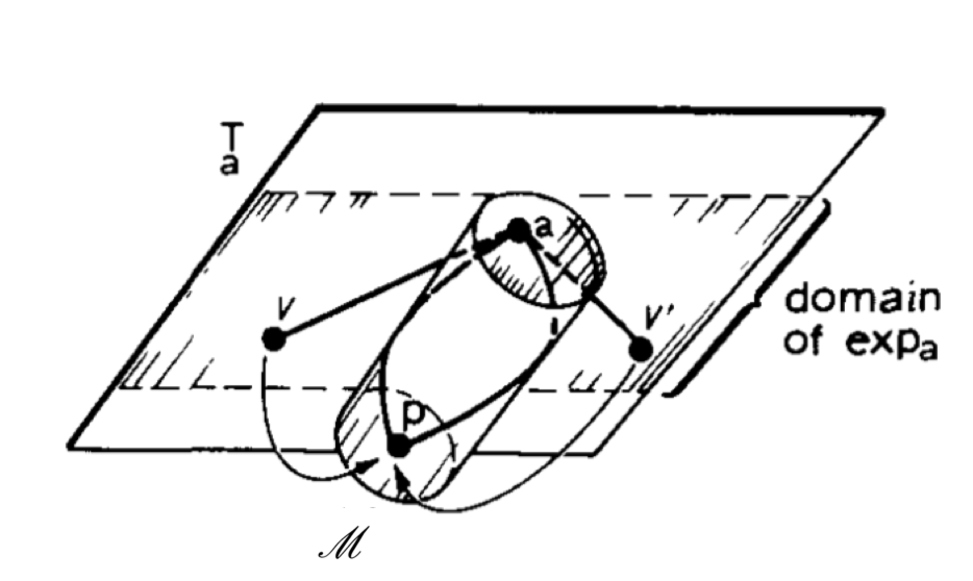}
\hspace{10mm}
\includegraphics[scale=0.39]{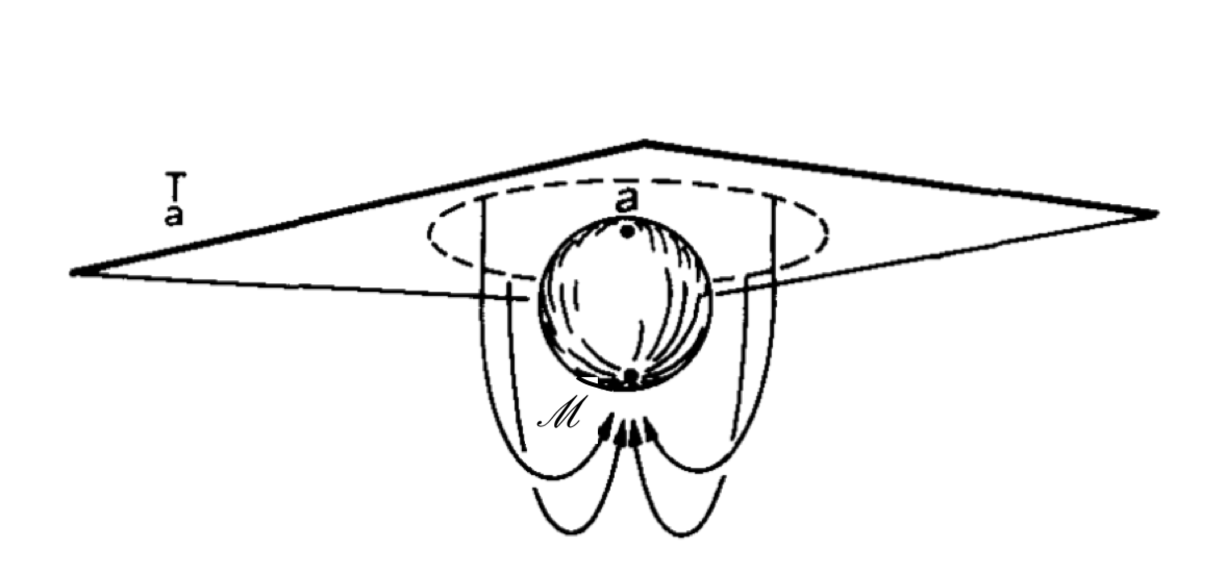}
\caption{On the left $\mathscr{M}$ is a 2-dimensional cylinder. Here $\mathrm{exp}_a$ maps a strip in $T_a$ onto $\mathscr{M}$, wrapping it around $\mathscr{M}$ infinitely many times so that the point $p\in\mathscr{M}$ is the image of infinitely many points in $T_a$, in particular of $v$ and $v'$. On the right $\mathscr{M}$ is a 2-sphere and all the circles of $k\pi$ radius and center at the origin of $T_a$ are mapped to a single point of $\mathscr{M}$.}
\label{fig:1.3}
\end{center}
\end{figure}
\\We require for the present that, for each $p\in\mathscr{M}$ there is some open neighbourhood $\mathscr{U}_0$ of the origin in $T_p$ and an open neighbourhood $\mathscr{U}_p$ of $p$ in $\mathscr{M}$ such that $\mathrm{exp}_p$ is a $C^{\infty}$ diffeomorphism from $\mathscr{U}_0$ to $\mathscr{U}_p$. Such a neighbourhood $\mathscr{U}_p$ is called \textit{normal  neighbourhood} of $p$. Furthermore, one can choose $\mathscr{U}_p$ to be \textit{convex}, i.e. to be such that any point $q$ in $\mathscr{U}_p$ can be joined to any other point $r$ in $\mathscr{U}_p$ by a unique geodesic starting at $q$ and totally contained in $\mathscr{U}_p$. Within a normal neighbourhood one can define coordinates $(x^1,...,x^4)$ by choosing any point $q\in \mathscr{U}_p$, choosing a basis $\{\mathrm{\textbf{E}}_a\}$ of $T_q$, and defining the coordinate of the point $r\in \mathscr{U}_p$ by the relation $r=\mathrm{exp}_q(\mathrm{\textbf{E}}_ax^a)$. In this way one assigns to $r$ the coordinates, with respect to the basis $\{\mathrm{\textbf{E}}_a\}$, of the point $\mathrm{exp}^{-1}_q(r)$ in $T_q$. Then $\left.(\partial/\partial x^i)\right|_{q}=\mathrm{\textbf{E}}_i$ and $\left.\Gamma^i_{jk}\right|_q=0$. Such coordinates will be called \textit{normal coordinates} based on $q$. 
 \section{Chronology and Causality}
 \label{sect:1.4}
Let $(\mathscr{M},g)$ be a space-time with fixed time-orientation and $p$ and $q$ any two points of $\mathscr{M}$. 
\begin{defn}$\\ $
\label{defn:prec}
The point $p$ \textit{chronologically precedes} $q$, $p\ll q$, if there exists a future-directed timelike curve (i.e. whose tangent vector is timelike future-directed) with past endpoint $p$ and future endpoint $q$.  
\end{defn}
The \lq precedes\rq\hspace{0.1mm} relation is the central one of what is called the causal structure of space-time. Physically it means that the events represented by $p$ and $q$ are causally related, in the sense that a signal can be sent from $p$ to be received later from $q$.\\
We have immediately the following
\begin{thm}$\\ $
\label{thm:prec}
If $p\ll q$ and $q\ll r$ then $p\ll r$. 
\end{thm} 
\begin{proof}It is sufficient to draw the timelike curves $\gamma_{pq}$ and $\gamma_{qr}$, that must exist by hypothesis and joining them at $q$. Their union, $\gamma_{pr}$ is a union of timelike curves, and hence timelike, which connects $p$ and $r$. Thus $p\ll r$.
\end{proof} 
\begin{defn}$\\ $
\begin{itemize}
\item The set $I^-(p)=\{q\in\mathscr{M}:q\ll p\}$ is called the \textit{chronological past} of $p$;
\item The set $I^+(p)=\{q\in\mathscr{M}:p\ll q\}$ is called the \textit{chronological future} of $p$;
\item Given a subset $S\subset\mathscr{M}$ the set $I^-[S]=\bigcup_{p\in S}I^-(p)$ is called the \textit{chronological past} of $S$;
\item Given a subset $S\subset\mathscr{M}$ the set $I^+[S]=\bigcup_{p\in S}I^+(p)$ is called the \textit{chronological future} of $S$.
\end{itemize}
\end{defn}
Since one can always perform a sufficiently small deformation of a timelike curve while preserving the timelike nature of the curve, it follows that for all $q\in I^{\pm}(p)$ there exists an open neighbourhood $O$ of $q$ such that $O\subset I^{\pm}(p)$. Thus
\begin{prop}$\\ $
\label{thm:opensets}
$I^{\pm}(p)$ is an open subset of $\mathscr{M}$, for every $p\in\mathscr{M}$. 
\end{prop}
The same property holds for $I^{\pm}[S]$, being the union of open sets. \\
As an example, in Minkowski space-time with the usual coordinates $(t,x,y,z)$, if $p=(0,0,0,0)$ then $I^-(p)=\{(t,x,y,z):t<-(x^2+y^2+z^2)^{1/2}
\}$ and $I^+(p)=\{(t,x,y,z):t>(x^2+y^2+z^2)^{1/2}
\}$ are just the interiors of the past and future light-cones of $p$.
\begin{defn}$\\ $
\label{defn:cprec}
The point $p$ \textit{causally precedes} $q$, $p\prec q$, if there exists a future-directed causal curve (i.e. whose tangent vector is timelike or null future-directed) with past endpoint $p$ and future endpoint $q$.  
\end{defn}
\begin{oss}$\\ $
Note that \cite{PenTDT} defines the chronologically and causally precedes relations using not arbitrary timelike and null curves, but geodesics which are easier to handle mathematically.
\end{oss}
We have similar definitions for the \textit{causal past} of $p$, $J^-(p)$ and for the \textit{causal future} of $p$, $J^+(p)$. In the remainder, to show the topological properties of the above defined sets, we will only use the future ones, being clear that they are valid for the past ones too. \\
In the above example of Minkowski space-time we have $J^-(p)=\{(t,x,y,z):t\leq-(x^2+y^2+z^2)^{1/2}
\}$ and $J^+(p)=\{(t,x,y,z):t\geq(x^2+y^2+z^2)^{1/2}
\}$, these sets being closed, and furthermore we have that the boundaries $\dot{I}^{\pm}(p)$ of $I^{\pm}(p)$ are generated by the null geodesics starting from $p$. However neither of this last properties is valid in general. We can immediately give an example of space-time in which $J^{\pm}(p)$ is not closed and where $ \dot{I}^{\pm}(p)$ are not generated by null geodesics starting from $p$. Let $(\mathscr{M},g)$ be a space-time and consider a closed subset $C$ of $\mathscr{M}$. Then $\mathscr{M}-C$ with the metric induced by $g$ is, if connected, itself a space-time (that $C$ be closed was necessary to ensure that $\mathscr{M}-C$ even be a manifold). This simple argument allows us to buil a new space-time from a given one, by, for example, removing a point of it. If now we consider Minkowski space-time with one point removed, as shown in Figure \ref{fig:1.4}, then $J^+(p)$ is not closed since the null geodesic beyond the removed point, which extends from $p$, is not part of $J^+(p)$ whereas it is part of $\dot{J}^+(p)$. 
\begin{figure}[h]
\begin{center}
\includegraphics[scale=0.4]{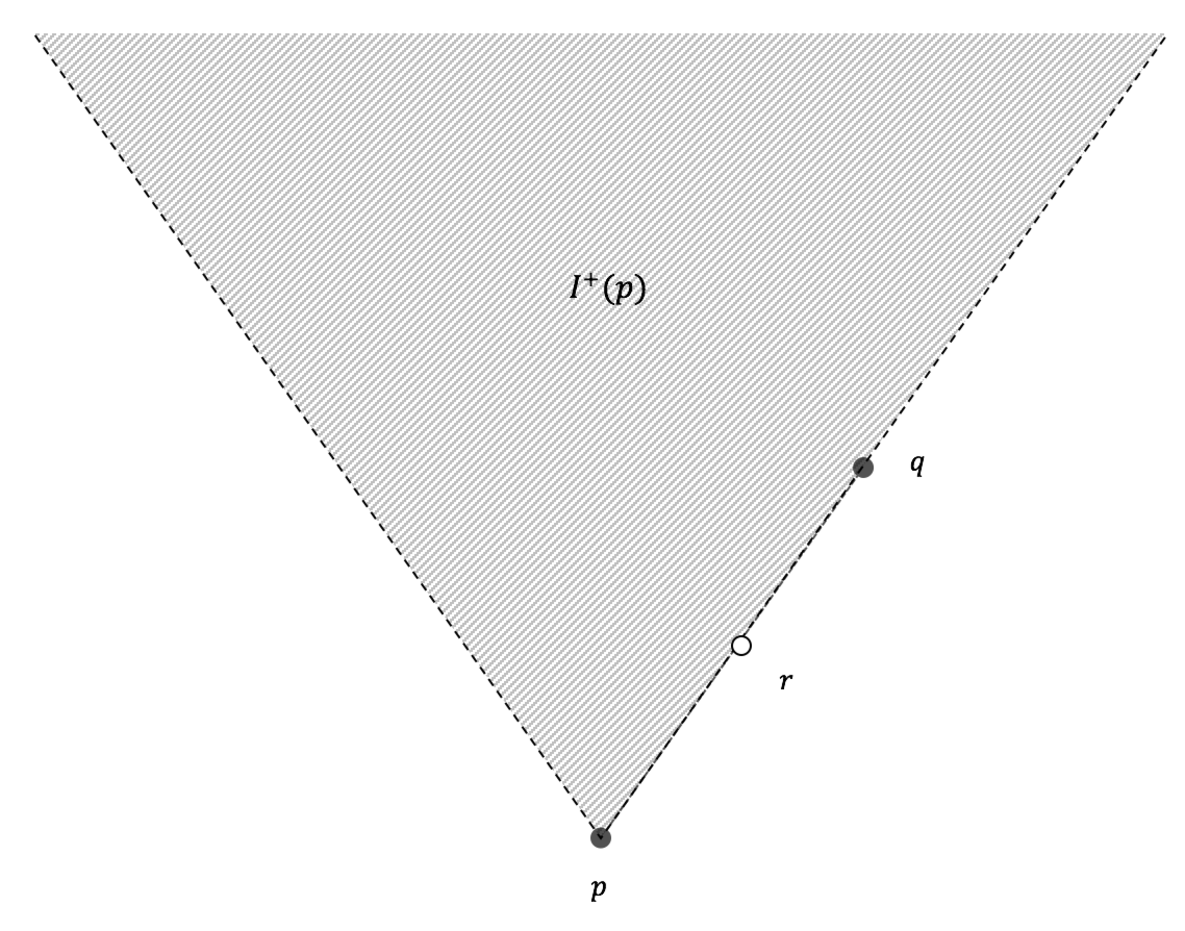}
\caption{Minkowski space-time with a point $r$ removed. In this space-time no causal curve connects $p$ and $q$, hence $q$ is not in $J^+(p)$. However $q\in\overline{J^+}(p)$. Thus $J^+(p)$ is not closed.}
\label{fig:1.4}
\end{center}
\end{figure}
\\Anyway the two properties mentioned above remain valid \textit{locally}, as stated by the following 
\begin{thm}$\\ $
\label{thm:nge}
Let $(\mathscr{M},g)$ be an arbitrary space-time and let $p\in\mathscr{M}$. Then there exists a convex normal neighbourhood  $\mathscr{U}$ of $p$\\
Furthermore, for any such $\mathscr{U}$, $\left.I^{+}(p)\right|_{\mathscr{U}}$, i.e. the chronological future of $p$ in the space-time $(\mathscr{U},g)$, consists of all points reached by future-directed timelike geodesics starting from $p$ and contained within $\mathscr{U}$, and has its boundary $\left.\dot{I}^{+}(p)\right|_{\mathscr{U}}$ generated by future-directed null geodesics in $\mathscr{U}$ starting from $p$.
\end{thm}
The proof of the first proposition can be found in \cite[pg.~32]{Hicks} while the second in \cite[pg.~103]{HawEll}.\\
Let $q\in J^+(p)$ and $\gamma$ be a causal curve beginning at $p$ and ending at $q$. Since $\gamma$ is a compact subset of $\mathscr{M}$ (being the continuous image of a closed interval) it can be covered by a finite number of convex normal neighbourhoods $\mathscr{U}_i$, $i=1,..n$ as shown in Figure \ref{fig:1.5}.
\begin{figure}[h]
\begin{center}
\includegraphics[scale=0.4]{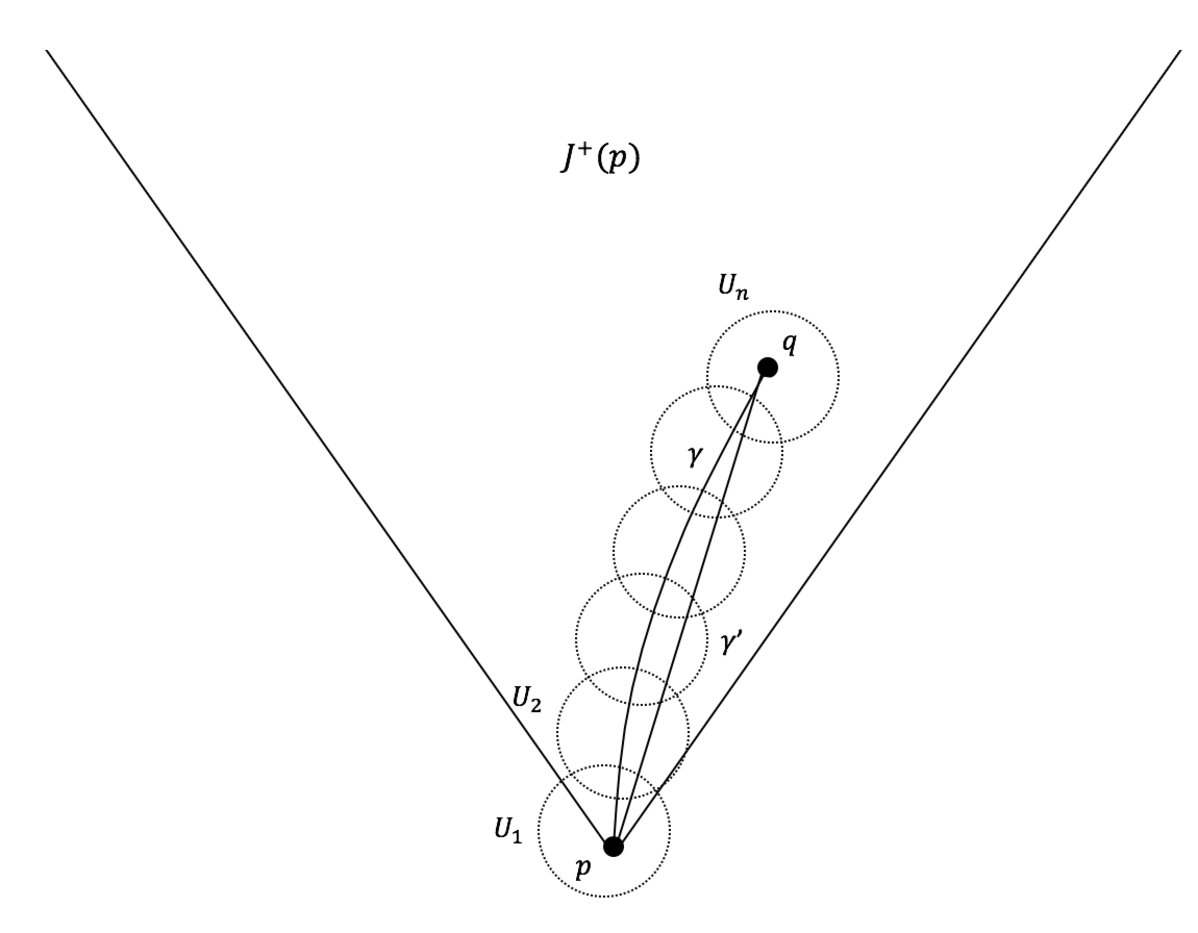}
\caption{A causal curve $\gamma$ from $p$ to $q$ and its deformation $\gamma'$ into a timelike geodesic.}
\label{fig:1.5}
\end{center}
\end{figure}
\\ If $\gamma$ failed to be a null geodesic in any such neighbourhood, then, using Theorem \ref{thm:nge} we can always deform $\gamma$ into a timelike geodesic in that neighbourhood and then extend this deformation to the other neighbourhoods to obtain a timelike curve from $p$ to $q$, call it $\gamma'$. Thus we have the following
\begin{cor} $\\ $
\label{cor:ngeo}
If $q\in J^{+}(p)-I^{+}(p)$, then any causal curve connecting $p$ to $q$ must be a null geodesic.
\end{cor}
Sometimes the set $E^{+}(p)=J^{+}(p)-I^{+}(p)$ is called the \textit{future horismos} of $p$.\\
Since for any set $S\subset\mathscr{M}$ it can be shown that $J^{+}[S]\subset \overline{I^{+}}[S]$ and clearly $I^{+}[S]\subset J^{+}[S]$, it follows immediately that \begin{equation*}
\overline{J^{+}}[S]=\overline{I^{+}}[S].
\end{equation*}
Similarly we have $I^{+}[S]=\mathrm{int}\left[J^{+}[S]\right]$ and hence
\begin{equation*}
\dot{J}^{+}[S]=\dot{I}^{+}[S].
\end{equation*}
\section{Pasts, Futures and Achronal Boundaries}
\label{sect:1.5}
From \ref{thm:nge} we saw that the boundary of $I^+(p)$ or $J^+(p)$ is formed, at least locally, by the future-directed null geodesics starting from $p$. To derive the properties of more general boundaries we introduce the concepts of achronal and future sets.
\begin{defn}$\\ $
\label{defn:ach}
A set $S\subset\mathscr{M}$ is said to be an \textit{achronal set} if $I^+[S]\cap S=\emptyset$, i.e. if no two points of $S$ are chronologically related. 
\end{defn}
Note that a set can be locally spacelike without being achronal, as shown in Figure \ref{fig:1.6}. Examples of achronal sets are the future light cone in Minkowski space-time, $t=(x^2+y^2+z^2)^{1/2}$, the null hyperplane $t=z$ and the spacelike plane $t=0$.
\begin{figure}[h]
\begin{center}
\includegraphics[scale=0.4]{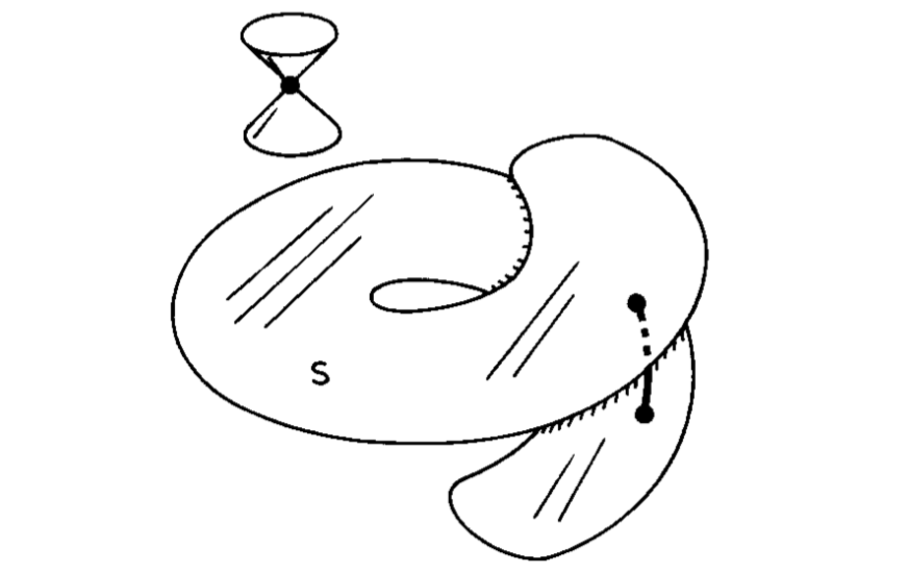}
\caption{A spacelike hypersurface in Minkowski space that is not achronal, since there are two points joined by a timelike curve, i.e. chronologically related.}
\label{fig:1.6}
\end{center}
\end{figure}
\begin{defn}$ \\	$
A set $F\subset\mathscr{M}$ is said to be \textit{future set} if $F=I^+[S]$ for some set $S\subset\mathscr{M}$.
\end{defn}
By proposition \ref{thm:opensets} a future set $F$ is always open.
\begin{defn}$\\ $
If $F$ is a future set, its boundary $\dot{F}$ is called \textit{achronal boundary}, i.e. $\dot{F}=\dot{I}^+[S]$.
\end{defn}
The next theorem asserts that the boundary of the future of a set, the above defined achronal boundary,  even if does not need to be smooth, always forms a \lq well behaved\rq , 3-dimensional, achronal surface.
\begin{thm}$\\ $
\label{thm:embman}
Let $(\mathscr{M},g)$ be a space-time and let $F$ be a future set for $S\subset\mathscr{M}$, $F=I^+[S]$. Then the achronal boundary $\dot{F}=\dot{I}^+[S]$ is an achronal, 3-dimensional, embedded $C^0$ topological submanifold of $\mathscr{M}$.
\end{thm}
\begin{proof}
Let $q\in\dot{F}$. If $p\in I^+(q)$, then $q\in I^-(p)$ and since $I^-(p)$ is open, an open neighbourhood $O$ of $q$ is contained in $I^-(p)$ (see Figure \ref{fig:1.7}). We have $O\cap F\neq\emptyset$, $q$ being on $\dot{F}$. Thus we have $p\in I^+(q)\subset I^+[O\cap F]\subset F$. In particular $I^+(q)\subset F$. By following the same argument we have $I^-(q)\subset\mathscr{M}-F$. If $\dot{F}$ failed to be achronal we could find two points in it, say $q$ and $r$, such that $r\in I^+(q)$, and hence $r\in F$, by the previous result. However, this is impossible since $F$ is open and there is no point lying both in $F$ and $\dot{F}$. Thus $\dot{F}$ is achronal. 
\begin{figure}[h]
\begin{center}
\includegraphics[scale=0.4]{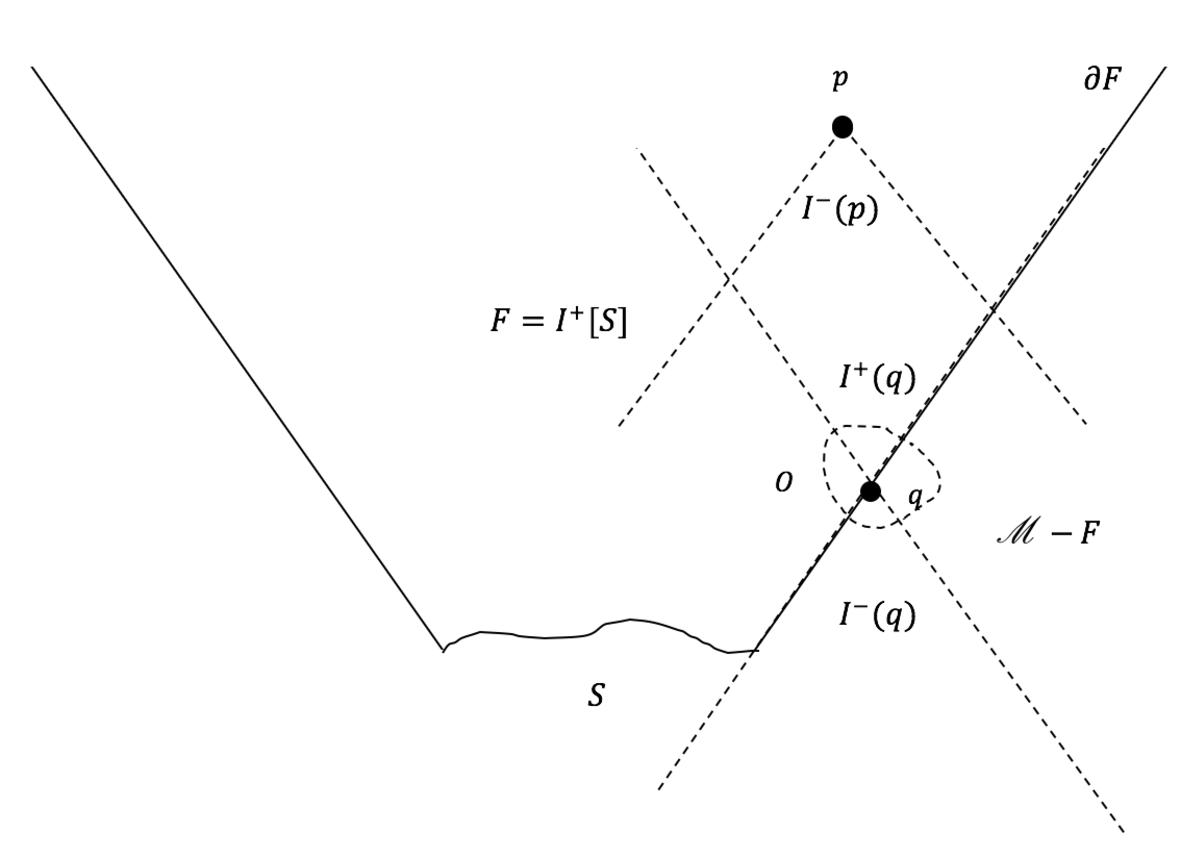}
\caption{A space-time diagram showing a set $S$ and the boundary of its future, $F$.}
\label{fig:1.7}
\end{center}
\end{figure}
\\To obtain the manifold structure of $\dot{F}$ we introduce normal coordinates $(x^0,x^1,x^2,x^3)$ in a neighbourhood $\mathscr{U}_{\alpha}$ of $q$ such that $\partial/\partial x^0$ is timelike in $\mathscr{U}_{\alpha}$ and that the integral curves of $\partial/\partial x^0$, $\{x^i=\mathrm{const},i=1,2,3\}$ enter $I^+(q)\subset F$ and  $I^-(q)\subset \mathscr{M}-\overline{F}$. But this implies that each such curve intersects $\dot{F}$, and since $\dot{F}$ is achronal, it must intersect it at precisely one point (otherwise we would obtain two or more points joined by a timelike curve). 
\begin{figure}[h]
\begin{center}
\includegraphics[scale=0.4]{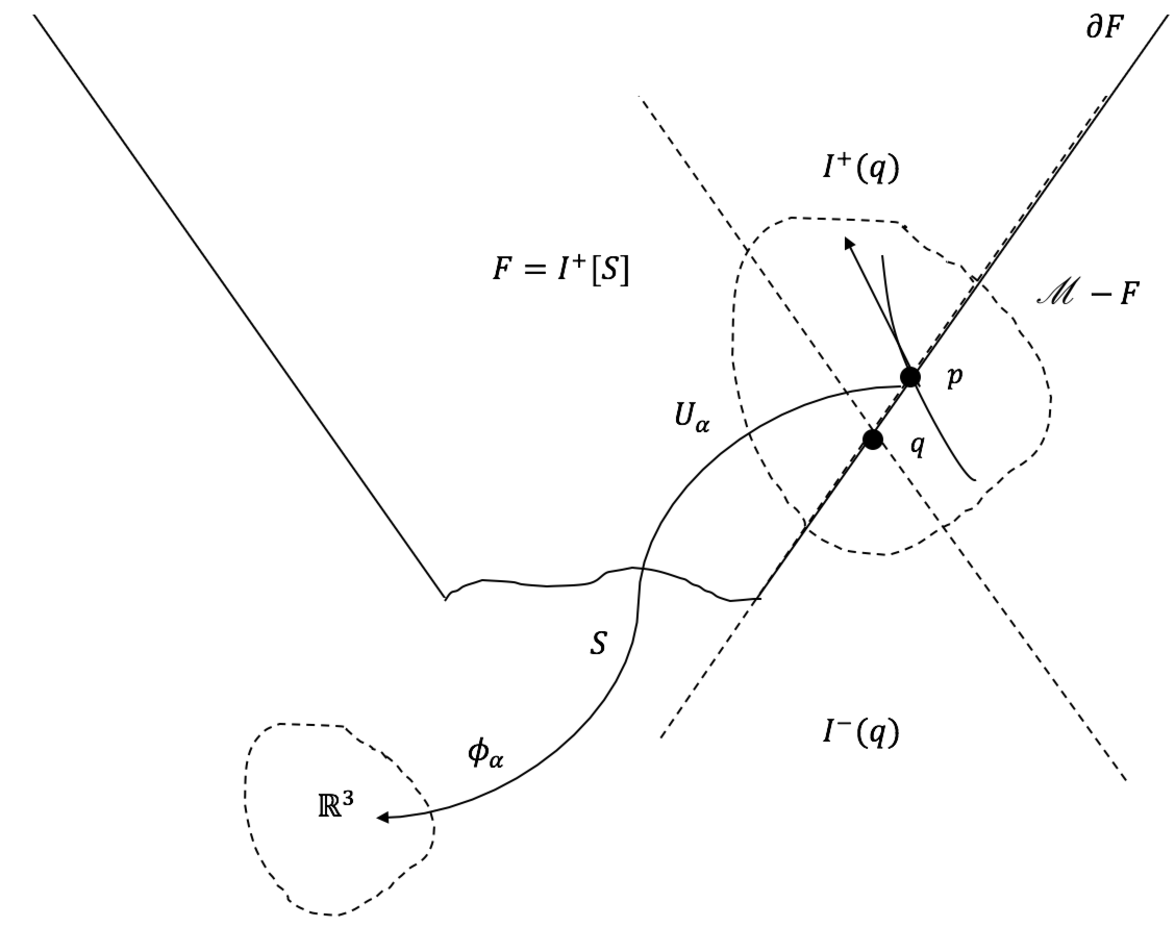}
\caption{This picture shows the timelike vector and its integral curve that intersect $\dot{F}$ in just one point $p$, defining a coordinate system.}
\label{fig:1.8}
\end{center}
\end{figure}
\\Thus in each such neighbourhood, we get a one-to-one association of points of $\dot{F}$ with the coordinates $(x^1,x^2,x^3)$ characterizing the integral curve of $\partial/\partial x^0$, i.e. $\phi_{\alpha}:\dot{F}\cap \mathscr{U}_{\alpha}\rightarrow \mathbb{R}^3$ defined by $\phi_{\alpha}(p)=x^i(p)$ ($i=1,2,3$) for $p\in\dot{F}\cap \mathscr{U}_{\alpha}$. Furthermore the value of $x^0$ at the intersection point must be a $C^0$ function of the coordinates $(x^1,x^2,x^3)$ and thus the map $\phi_{\alpha}$ is a homomorphism. Since this construction can be repeated for all $q\in\dot{F}$ we obtain a collection $\{\dot{F}\cap \mathscr{U}_{\alpha},\phi_{\alpha}\}$ that is a $C^0$ atlas for $\dot{F}$, which makes it an embedded topological manifold.
\end{proof}
For the purpose of what follows we need to introduce several definitions that will play an important role. First, it will be convenient to extend the definition of timelike and causal curve from piecewise differentiable to continuous, it being essential in taking limits. 
\begin{defn}$\\ $
A continuous curve $\gamma:I\rightarrow\mathscr{M}$, where $I$ is an interval of $\mathbb{R}$, is \textit{future-directed causal} if for every $t\in I$, there is a neighbourhood $G$ of $t$ in $I$ and a convex normal neighbourhood $\mathscr{U}$ of $\gamma (t)$ in $\mathscr{M}$ such that for any $t_1\in G$, $\gamma(t_1)\in J^-[\gamma (t)]$ if $t_1<t$ and $\gamma(t_1)\in J^+[\gamma (t)]$ if $t_1>t$.
\end{defn}
\begin{figure}[h]
\begin{center}
\includegraphics[scale=0.5]{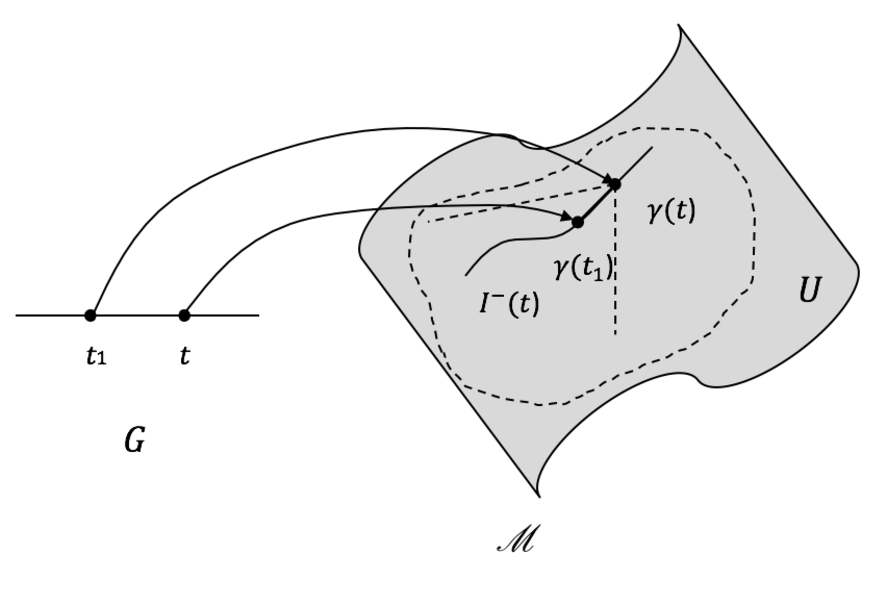}
\caption{A continuous timelike curve in the case in which $t_1<t$.}
\label{fig:1.9}
\end{center}
\end{figure}
The same definition holds for timelike curve, with $I^{\pm}[\gamma (t)]$ replacing $J^{\pm}[\gamma (t)]$. The sense of the definition is that, for a continuous curve, locally, pairs of points on the curve can be joined by a differentiable timelike or causal curve. Note that the timelike or causal nature of the curve is left unchanged by a continuous, one-to-one, reparametrization, and hence two curves which differ by such a reparametrization will be considered equivalent.\\
Next we need the notion of extendibility of a curve, and hence we give before the definition of endpoint of a non-spacelike curve.
\begin{defn}$\\ $
A point $p\in\mathscr{M}$ will be said to be \textit{future endpoint} of a future-directed causal curve $\gamma:I\rightarrow\mathscr{M}$ if for every neighbourhood $O$ of $p$ there is a $t\in I$ such that $\gamma (t_1)\in O$ for every $t_1\geq t$.
\end{defn}
Note that the endpoint need not lie on the curve, i.e. there need not exist a value of $t$ such that $p=\gamma (t)$. This allows us to give the following
\begin{defn}$\\ $
A causal curve is \textit{future-inextendible} if it has no future endpoint.
\end{defn}
A similar definition holds for \textit{past-inextendibility}.\\
We give now the definition of convergence of causal curves.
\begin{defn}$\\ $
\label{defn:conv}
Let $\{\lambda_n\}$ be an infinite sequence of causal curves. 
\begin{itemize}
\item A point $p$ will be said to be a \textit{convergence point} of $\{\lambda_n\}$ if, given any open neighbourhood $O$ of p, there exists an $N$ such that $\lambda_n\cap O\neq\emptyset$ for all $n>N$.\\
\item A curve $\lambda$ will be said to be a \textit{convergence curve} of  $\{\lambda_n\}$ if each $p\in\lambda$ is a convergence point.\\
\item A point $p$ will be said to be a \textit{limit point} of $\{\lambda_n\}$ if every open neighbourhood of $p$ intersects infinitely many $\lambda_n$.\\
\item A curve $\lambda$ will be said to be the \textit{limit curve} of $\{\lambda_n\}$ if there exists a subsequence $\{\lambda'_n\}$ for which $\lambda$ is a convergence curve.
\end{itemize}
\end{defn}
The previous definitions allow us to state and prove the following
\begin{thm}$\\ $ 
\label{thm:caus}
Let $S$ be an open set and $\{\lambda_n\}$ be an infinite sequence of causal curves which are future-inextendible with limit point $p$. Then through $p$ there is a causal curve $\lambda$ which is future-inextendible and which is a limit curve of $\{\lambda_n\}$. 
\end{thm}
\begin{proof}
Let $\mathscr{U}_1$ be a convex normal coordinate neighbourhood about $p$ and let $\mathscr{B}(p,b)$ be the open ball of coordinate radius $b>0$ with center $p$. Let $\{\lambda(1,0)_n\}$ be a subsequence of $\{\lambda_n\}\cap \mathscr{U}_1$ which converges to $p$. Since the sphere $\dot{\mathscr{B}}(p,b)$ is compact it will contain the limit point of the $\{\lambda(1,0)_n\}$, the latter being a subsequence. Any such limit point must lie either in $J^-(p)$ or $J^+(p)$ because of the causal nature of the curves. Choose\begin{equation}
 x_{11}\in J^{+}(p)\cap\dot{\mathscr{B}}(p,b)
 \end{equation}
 to be one of these limit points, and choose $\{\lambda(1,1)_n\}$ to be a subsequence of $\{\lambda(1,0)_n\}$ which converges to $x_{11}$. We can continue inductively, defining\begin{equation*}
x_{ij}\in J^{+}\cap \dot{\mathscr{B}}\left(p,\frac{i}{j}b\right)
\end{equation*}
as a limit point of the subsequence $\{\lambda(i,j-1)_n\}$ for $i\geq j\geq 1$, $\{\lambda(i-1,i-1)_n\}$ for $i\geq 0$ and $j=0$, and defining $\{\lambda(i,j)_n\}$ as a subsequence of the above subsequence which converges to $x_{ij}$.  
\begin{figure}[h]
\begin{center}
\includegraphics[scale=0.5]{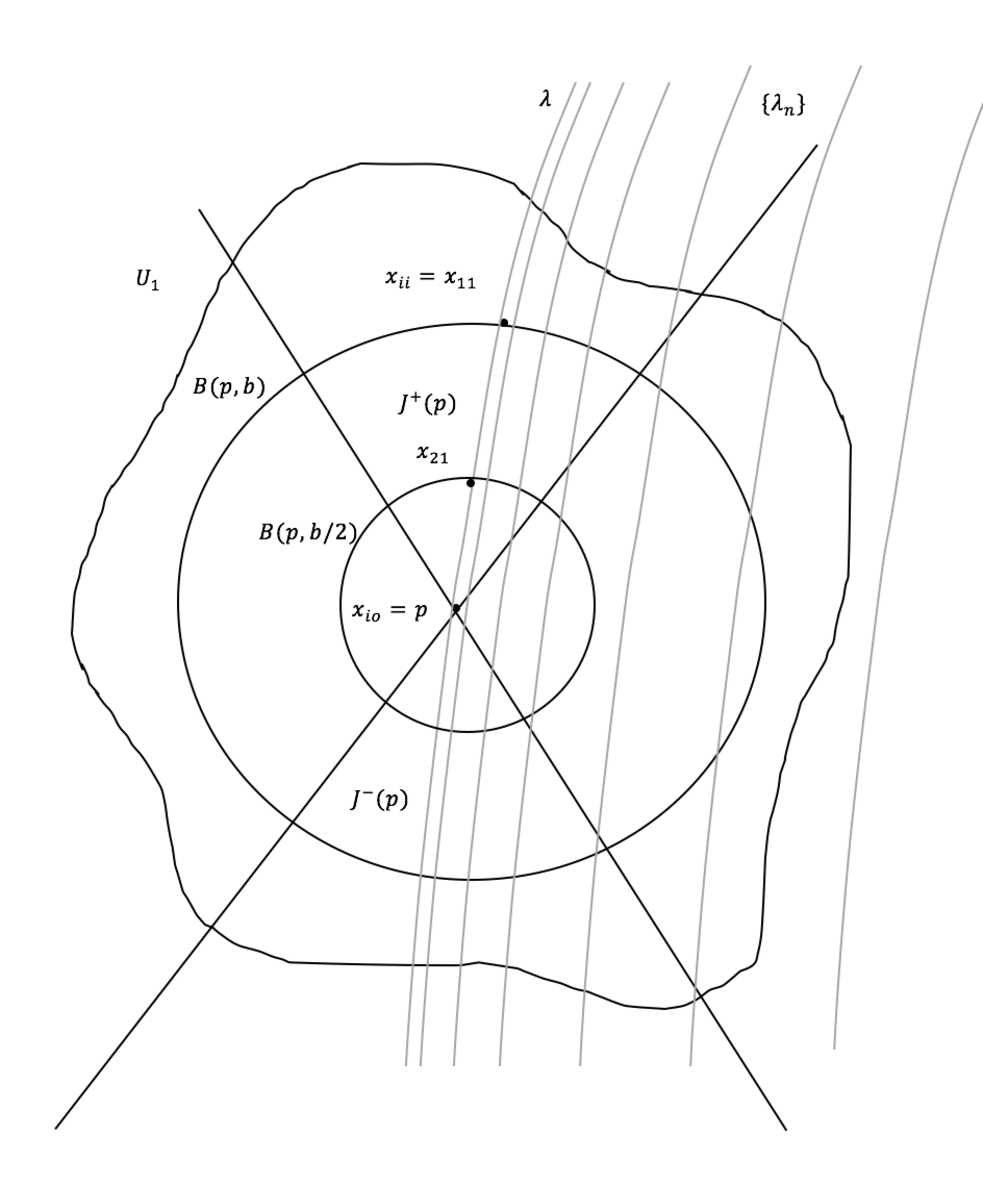}
\caption{The causal curve $\lambda$ through $p$ of a family of causal curves $\{\lambda_n\}$ for which $p$ is a limit point}
\label{fig:1.10}
\end{center}
\end{figure}
For example the point $x_{21}$, situated on $\mathscr{B}\left(p,\frac{1}{2}b\right)$ is a limit point of  $\{\lambda(2,0)_n\}$ and a convergence point of  $\{\lambda(2,1)_n\}$. We are just constructing, in turn, all the coordinate spheres whose radii are rational multiples, between $0$ and $1$ of $b$ and continuing to extract limit points lying on these spheres and subsequences converging to these points. Since any two of the $x_{ij}$ will have a causal separation, the closure of the union of all the $x_{ij}$ will give a causal curve $\lambda$ from $p=x_{i0}$ to $x_{11}=x_{ii}$. To show that $\lambda$ is a limit curve of $\{\lambda_n\}$ we have to construct a subsequence $\{\lambda'_n\}$ of the $\{\lambda_n\}$ such that for each $q\in\lambda$, $\{\lambda'_n\}$ converges to $q$. We choose $\{\lambda'_m\}$ to be a member $\{\lambda(m,m)_n\}$ which intersects each of the balls $\mathscr{B}(x_{mj},m^{-1}b)$ for $0\leq j \leq m$. Since each curve of the family $\{\lambda'_m\}$ above defined intersects the balls constructed with centers the various points of $\lambda$, we can say that $\{\lambda'_m\}$ converges to $\lambda$ and thus that $\lambda$ is a limit curve for $\{\lambda_n\}$. We can repeat this construction by letting $\mathscr{U}_2$ be a convex neighbourhood about $x_{11}$ and using as starting sequence $\lambda'_n$. In this way one can extend $\lambda$ indefinitely and thus $\lambda$ is future-inextendible. 
\end{proof}
As application of the previous statement, we prove now a fundamental theorem characterizing the nature of achronal boundaries.
\begin{thm}$\\ $
Let $C$ be a closed subset of the space-time manifold $\mathscr{M}$ and let $F$ be its chronological future, i.e. $I^+[C]=F$. Then every point $p\in\dot{F}$, the achronal boundary, with $p\notin C$ (i.e. $p\in\dot{F}-C$) lies on a null geodesic $\lambda$ which lies entirely in $\dot{F}$ and either is past-inextendible or has a past endpoint on $C$. 
\end{thm}
\begin{proof}
Choose a sequence $\{q_n\}$ of points in $F$ which converges to $p\in\dot{F}$. For each $q_n$ we can consider $\lambda_n$, a past-directed timelike curve connecting $q_n$ to a point in $C$. Consider the space-time manifold $\mathscr{M}-C$ (here we use the assumption that $C$ is closed, for otherwise $\mathscr{M}-C$ would not define a manifold). On $\mathscr{M}-C$, each $\lambda_n$ is obviously a past-inextendible causal limit curve and hence $p$ is a limit point of the sequence $\{\lambda_n\}$. 
\begin{figure}[h]
\begin{center}
\includegraphics[scale=0.5]{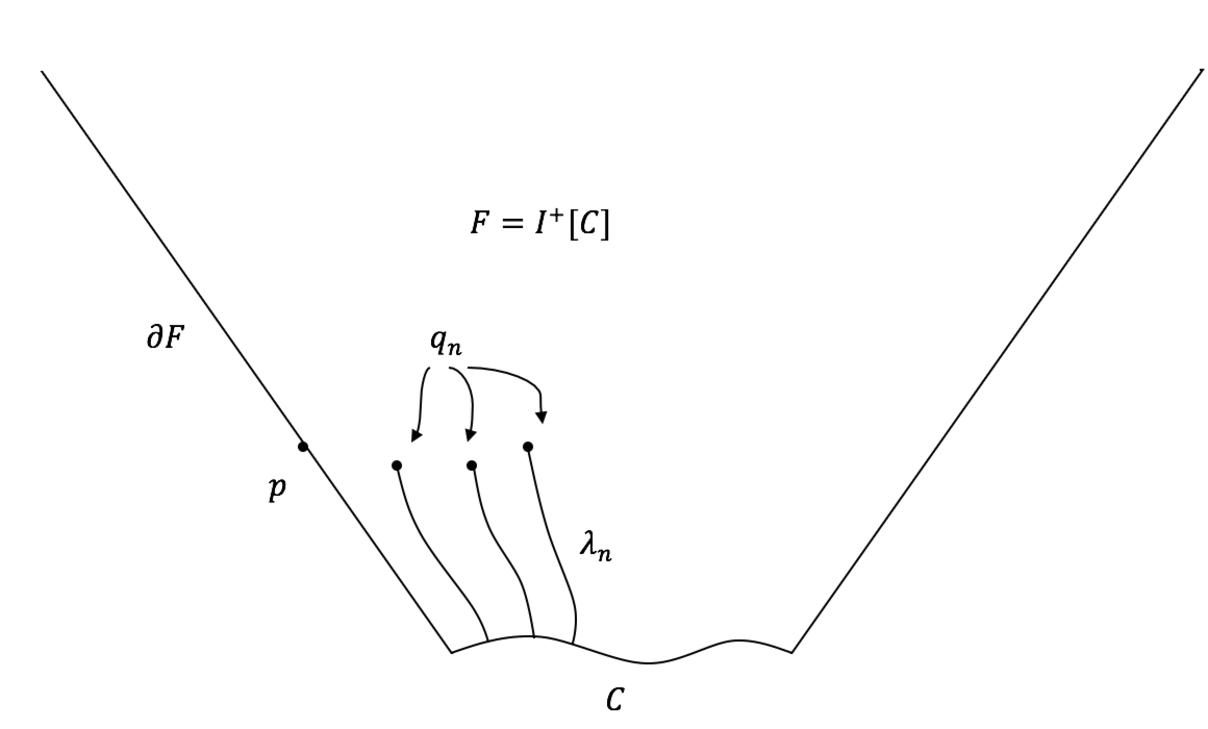}
\caption{A space-time diagaram showing a sequence of points in $F$ convering to $p\in\dot{F}$.}
\label{fig:1.11}
\end{center}
\end{figure}
\\Then, using theorem \ref{thm:caus}, there exists a past inextendible causal limit curve $\lambda$ passing through $p$, whose points are limit points of $\{\lambda_n\}$ in $F$. Hence $\lambda\subset\overline{F}[C]$. But if $\lambda$ were in $F$, then by corollary \eqref{cor:ngeo} we would have $p\in F=I^+[C]$, since $p$ could be connected to $C$ by a causal curve which is not null geodesic. This contradicts the fact that $p\in\dot{F}$. Thus $\lambda\in\dot{F}$. Furthermore, since $\dot{F}$ is achronal (it is an achronal boundary), using corollary  \eqref{cor:ngeo}, we obtain that $\lambda$ is a null geodesic. Since $\lambda$ is past-inextendible in $\mathscr{M}-C$, in $\mathscr{M}$ it must either remain past-inextendible or have past endpoint on $C$.
\end{proof}
An example where $\lambda$ is past-inextendible is provided by point $q$ in Figure \ref{fig:1.4}.
\section{Global Causality Conditions}
\label{sect:1.6}
In this section we will investigate the concept of a \lq globally causally well behaved\rq\hspace{0.1mm} space-time. In fact, according to theorem \ref{thm:nge}, space-times in General Relativity \textit{locally} have the same qualitative causal structure as in Special Relativity, but \textit{globally} very significant differences may occur.\\
The postulate of \textit{local causality} \citep[see][pg.~60]{HawEll} asserts that the equations governing the matter fields must be such that if $\mathscr{U}$ is a convex normal neighbourhood and $p$ and $q$ are points in $\mathscr{U}$, then a signal can be sent in $\mathscr{U}$ between $p$ and $q$ if and only if $p$ and $q$ can be joined by a causal curve lying entirely in $\mathscr{U}$. Obviously whether the signal can be sent from $p$ to $q$ or from $q$ to $p$  will depend on the direction of time in $\mathscr{U}$ and hence it is a problem regarding the orientability, already discussed in section \ref{sect:1.2}.\\
It is this postulate which sets the metric $g$ apart from the other fields and gives it its distinctive geometrical character. In fact, observation of local causality allows one to measure the metric up to a conformal factor, using the experimental fact that nothing travels faster than light, which is a consequence of the particular equations of electromagnetism.\\
However globally, as remarked, nothing ensures us that the space-time may be not causality-violating. But we may now wonder what we actually mean by causality violations. The most obvious manifestation of such violation would be the existence, on a large scale, of closed timelike or causal curves, i.e., with the notation of \eqref{defn:prec} and \eqref{defn:cprec}, that an event $p$ would satisfy $p\ll p$ or two events $p$ and $q$ would satisfy $p\prec q$, $q\prec p$ with $p\neq q$. In fact the existence in a space-time of such curves, would seem to lead to the possibility of logical paradoxes. An example can be that one could travel with a rocketship round a closed timelike curve and, arriving back before one's departure, one could prevent oneself from setting out. Hence, if we are assuming that there is a simple notion of \lq free will\rq , i.e. the ability to choose how to act, one could have no difficulty in altering and influencing his own past. We might argue that individuals with this abilities violate our most basic conceptions of how the world operates, and so it is entirely proper to impose, as an additional condition for physically acceptable space-times that they possess no such causality violations. We note here, that the mere Einstein field equations do not put any restriction on the causality behaviour of the space-time and hence those restrictions have to be imposed \lq artificially\rq\hspace{0.1mm}. A concrete example of this is the \textit{anti-de Sitter (AdS) space-time} , the space of constant curvature $R<0$. It has the topology of $S^1\times\mathbb{R}^3$ and can be represented as the hyperboloid
\begin{equation*}
-u^2-v^2+x^2+y^2+z^2=-1
\end{equation*} in the flat five-dimensional space $\mathbb{R}^5$ with metric
\begin{equation*}
g=-du\otimes du-dv\otimes dv+dx\otimes dx+dy\otimes dy+dz\otimes dz.
\end{equation*}
It can be shown that there exist closed timelike curves in this space \citep[see]{ADS}. However AdS space-time is not simply connected, and if one unwraps the circle $S^1$ one obtains the universal covering space of anti-de Sitter space which does not contain any closed timelike curves, which has the topology of $\mathbb{R}^4$. By \lq anti-de Sitter  space\rq\hspace{0.1mm} one usually means its universal covering.\\
But even if it is generally believed and is customary to dismiss space-time with closed causal curves, retaining them \lq physically unrealistic\rq , it is often convenient to study space-times possessing causality violations because an unrealistic model, in physics, may as well have an important, but indirect and not immediately tangible, physical value.\\
Another simple example of (flat) space-time with topology $S^1\times\mathbb{R}^3$ which possesses closed timelike curves is obtained by identifying the $t=0$ and $t=10$ hyperplanes of Minkowski space-time, as shown in Figure \ref{fig:1.12}. 
\begin{figure}[h]
\begin{center}
\includegraphics[scale=0.5]{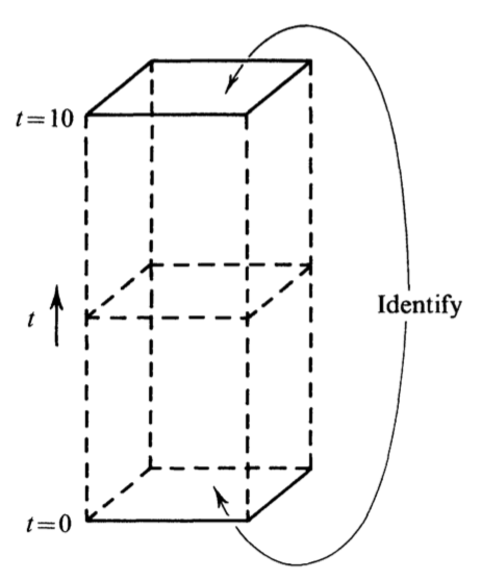}
\caption{Tube of Minkowski space-time with top and bottom identified.}
\label{fig:1.12}
\end{center}
\end{figure}
\\In this space-time the integral curves of the vector $\partial/\partial t$ will be closed space-time curves and it is not difficult to see that for all $p\in\mathscr{M}$ we have $I^+(p)=I^-(p)=\mathscr{M}$. However there are other examples of space-times with closed causal curves, which are not obtained making topological identifications in an \lq artificial\rq\hspace{0.1mm} way, but opportunely twisting the light cones, as in Figure \ref{fig:1.13}.
\begin{figure}[h]
\begin{center}
\includegraphics[scale=0.5]{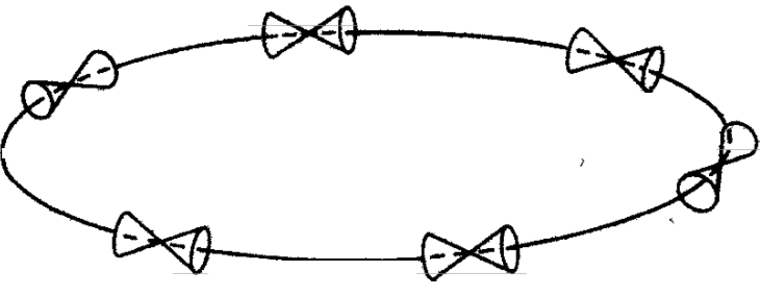}
\caption{A space-time where the light cones \lq tip over\rq\hspace{0.1mm} sufficiently to permit the existence of closed timelike curves.}
\label{fig:1.13}
\end{center}
\end{figure}
\\From the previous arguments, following \cite{HawEll}, we give the following
\begin{defn}$\\ $
\begin{itemize}
\item A spacetime $(\mathscr{M},g)$ is said to satisfy the \textit{chronology condition} if it does not contain closed timelike curves;\\
\item The set of points at which the chronology condition does not hold, i.e. those points through which pass closed timelike curves, is called \textit{chronology violating set}.
\end{itemize}
\end{defn}
The following theorems hold
\begin{thm}$\\ $
The \textit{chronology violating} set of $\mathscr{M}$ is the disjoint union of sets of the form $I^+(q)\cap I^-(q)$, $q\in\mathscr{M}$. 
\end{thm}
\begin{thm}$\\ $
If $\mathscr{M}$ is compact, the chronology violating set of $\mathscr{M}$ is non-empty.
\end{thm}
The proofs can be found in \cite[pg.~189-190]{HawEll}.
From this last result it would seem reasonable to assume that a space-time should not be compact, in agreement with the arguments carried out in section \ref{sect:1.1}. Similarly we can define the \textit{causality condition} and hence the \textit{causality violating set} and it turns out that it is formed by the disjoint union of sets of the form $J^+(q)\cap J^-(q)$, $q\in\mathscr{M}$. As we will see the chronology and causality conditions are the \lq largest\rq\hspace{0.1mm} restrictions one can impose a space-time.\\
Also there are other possible types of causality violations, weaker than the existence of closed causal curves. In fact it would seem reasonable to exclude situations in which there are causal curves who return arbitrarily close to their point of origin or which pass arbitrarily close to other causal curve, because an arbitrary small perturbation of the metric in space-times like these would produce causality violation. As an example we can consider the space-time in Figure \ref{fig:1.14} in which there exist causal curves which come arbitrarily close to intersecting themselves, although none of them actually do. In fact here the light cones on the cylinder tip over until one null direction is horizontal, and then tip back up.
\begin{defn}$\\ $
A space-time is $(\mathscr{M},g)$ is \textit{future-distinguishing} at $p\in\mathscr{M}$ if $I^+(p)\neq I^+(q)$ for each $q\in\mathscr{M}$, with $q\neq p$. If a space-time is future distinguishing at every $p\in\mathscr{M}$ it is said to satisfy the \textit{future-distinguishing condition}.
\end{defn}
A similar definition holds for the concept of \textit{past-distinction}. Clearly if a space-time contains closed causal curves, it cannot be either past- or future-distinguishing. In fact, if a space-time would contain a closed causal curve, each pair $(p,q)$, with $p\neq q$, of points on that closed curve would be such that $I^+(p)=I^+(q)$. Hence we have the simple 
\begin{prop}$\\ $
If a space-time time $(\mathscr{M},g)$ is past- and future-distinguishing at $p$, then $(\mathscr{M},g)$ is causal at $p$.
\end{prop}
\begin{defn}$\\ $
\label{defn:stcaus}
A space-time $(\mathscr{M},g)$ is said to be \textit{strongly causal} at $p\in\mathscr{M}$ if every neighbourhood $O$ of $p$ contains a neighbourhood $O'$ of $p$ which is not intersected more than once by any causal curve. If a space-time  is strongly causal at every $p\in\mathscr{M}$ it is said to satisfy the \textit{strong causality condition}.
\end{defn}
\begin{oss}$\\ $
\label{oss:stcaus}
By defining an open set $O$ to be \textit{causally convex} if and only if for every $p,q\in O$, $p\ll r\ll q$ implies $r\in O$, an equivalent definition of the strong causality may be the following: $(\mathscr{M},g)$ is strongly causal at $p\in\mathscr{M}$ if and only if $p$ has arbitrarily small causally convex neighbourhoods. Here \lq arbitrarily small\rq\hspace{0.1mm} means that such a neighbourhood $O$ of $p$ can be found inside any open set containing $p$. Thus this definition is equivalent to the previous one.
\end{oss}
\begin{figure}[h]
\begin{center}
\includegraphics[scale=0.3]{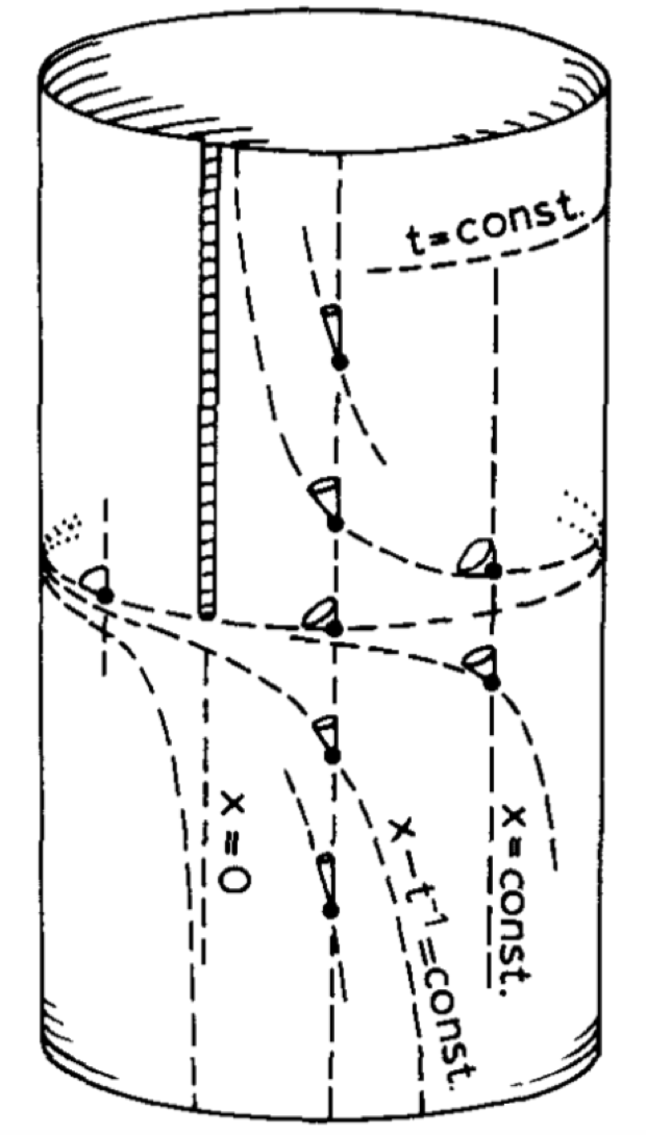}
\caption{Consider the metric form $ds^2=dtdx+t^2dx^2$, with $\partial/\partial t$ future pointing, on the strip $x\leq1$ of the $(t,x)$ plane. If we identify $(t,-1)$ with $(t,1)$ for each $t$ we obtain a space-time with a closed causal curve, the null geodesic $t=0$. Removing the point $(0,0)$ leaves a space-time with no closed causal curves, but which is neither future- nor past-distinguishing, since for any two different points $p$ and $q$ on the strip $t=0$ we have $I^+(p)=I^+(q)$. Hence the space-time is not strongly causal. If we remove the future-endless null geodesic $x=0$, $t\geq 0$ the new space-time obtained is not past distinguishing but future distinguishing.}
\label{fig:1.14}
\end{center}
\end{figure}
Hence, roughly speaking, if a space-time is not strongly causal at $p$, near $p$ there exist causal curves which come arbitrarily close to intersecting themselves.\\
Suppose that the future-distinguishing condition does not hold, i.e. there exist $q$ and $p$ such that $I^+(p)=I^+(q)$ with $q\neq p$. Choose $O_p$ and $O_q$ to be two disjoint open sets around $p$ and $q$ and choose $x\in I^+(p)\cap O_p$, then $q\ll x$. Choose $y$ in $O_q$ with $q\ll y \ll x.$ Then $p\ll y$ and hence, by \ref{thm:prec}, $p\ll x$, i.e. there is a timelike curve from $p$ to $x$ via $y\notin O_p$. Hence there is a causal curve intersecting more then once a neighbourhood of $p$ and, since this holds for arbitrary small $O_p$, the strong causality condition does not hold in $p$. We obtained the following
\begin{prop}$\\ $
If a space-time $(\mathscr{M},g)$ is strongly causal at $p$, then $(\mathscr{M},g)$ is future-distinguishing at $p$. 
\end{prop}
The following definition is very important, since, as we will see, it allows us to regard the causal structure of a space-time as a fundamental structure from which the topology of the space-time manifold can be derived, under appropriate hypothesis. 
\begin{defn}$\\ $
A \textit{local causality neighbourhood} is a causally convex open set $\mathscr{U}$ with compact closure.
\end{defn}
It can be shown that, in virtue of the previous definition, the following theorem holds:
\begin{thm}$\\ $
\label{thm:lcn}
A space-time $(\mathscr{M},g)$ is strongly causal at $p$ if and only if $p$ is contained in some local causality neighbourhood.
\end{thm}
The proof to this theorem can be found in \cite[pg.~30]{PenTDT}. 
Hence we have
\begin{prop}$\\ $
\label{prop:lcn}
Let $(\mathscr{M},g)$ be a space-time. Let $S\subset\mathscr{M}$ and suppose that strong causality holds at every point of $\mathscr{M}$. Then $\mathscr{M}$ can be covered by a locally finite (countable) system of local causality neighbourhoods. If $S$ is compact, then a finite number of such neighbourhoods will suffice.
\end{prop}
This proposition follows from \ref{thm:lcn} and from the definition of paracompactness. \\
We can now construct a collection of subsets of a space-time manifold $\mathscr{M}$ in the following way. Let $O$ be an open subset of $\mathscr{M}$ and let $p$,$q\in O$. The we write that $p\ll_{O}q$ if a timelike curve lying in $O$ exists from $p$ to $q$ and $p\prec_{O}q$ if a causal curve in $O$ exists from $p$ to $q$. We define
\begin{equation*}
\left<p,q\right>_O=\{\left.r\right| p\ll_{O}r\ll_{O}q\}
\end{equation*}
and 
\begin{equation*}
\left<p,q\right>=\left<p,q\right>_{\mathscr{M}}
\end{equation*}
so that $\left<p,q\right>=I^+(p)\cap I^-(q)$. Obviously the sets $\left<p,q\right>$ and $\left<p,q\right>_O$ are open. It can be shown that \citep[see][sec.~4]{PenTDT}:
\begin{itemize}
\item Any point $r\in\mathscr{M}$ is contained in some set $\left<p,q\right>$;
\item If $x$,$p$,$q$,$r$,$s\in\mathscr{M}$ are such that $x\in\left<p,q\right>\cap\left<r,s\right>$, then there exist $u$,$v\in\mathscr{M}$ such that $x\in\left<u,v\right>\subset\left<p,q\right>\cap\left<r,s\right>$.
\end{itemize}
Hence we can put a topology on $\mathscr{M}$, called the \textit{Alexandrov topology}. The base for such a topology is constituted by the sets of the form $\left<p,q\right>$, i.e. a set is defined to be an open set in the Alexandrov topology if it is a union of sets of the form $\left<p,q\right>$. An important question may be whether or not the Alexandrov topology agrees with the manifold topology. In fact, generally, it turns out the Alexandrov topology is \lq coarser\rq\hspace{0.1mm} \citep{Haw76} than the manifold topology. The next theorem gives the complete condition that the two topologies should agree.
\begin{thm}$\\ $
The following three restrictions on a space-time $(\mathscr{M},g)$ are equivalent:
\begin{enumerate}
\item $(\mathscr{M},g)$ is strongly causal;
\item The Alexandrov topology agrees with the manifold topology;
\item The Alexandrov topology is Hausdorff.
\end{enumerate}
\end{thm}
Proofs and further details can be found in \cite{PenKro}.\\
This means essentially that, under the assumption of strong causality condition, one can determine the topological structure of the space-time just by observation of causal relationships. Then, in a certain way, we can say that causal structure is more fundamental than other structures. 
\\As we have seen, various degrees of causality restriction on a space-time are possible, e.g. in order of decreasing restrictiveness: strong causality, future- and past-distinction, causality condition, chronology condition (is worth noting that, as shown by \cite{Car71} or by \cite{Beem} there are a number of inequivalent conditions, each more restrictive than strong causality on which we are not focusing). Each of these is \lq reasonable\rq\hspace{0.1mm} from the physical point of view since if any of one is violated it is possible to slightly modify the metric in order to obtain closed causal trips, and thus causality violation. However, it is possible to construct examples (see Figure \ref{fig:1.15}) where strong causality is still satisfied, but a modification of the metric tensor in an arbitrarily small neighbourhood of two or more points produces closed causal curves. 
\begin{figure}[h]
\begin{center}
\includegraphics[scale=0.4]{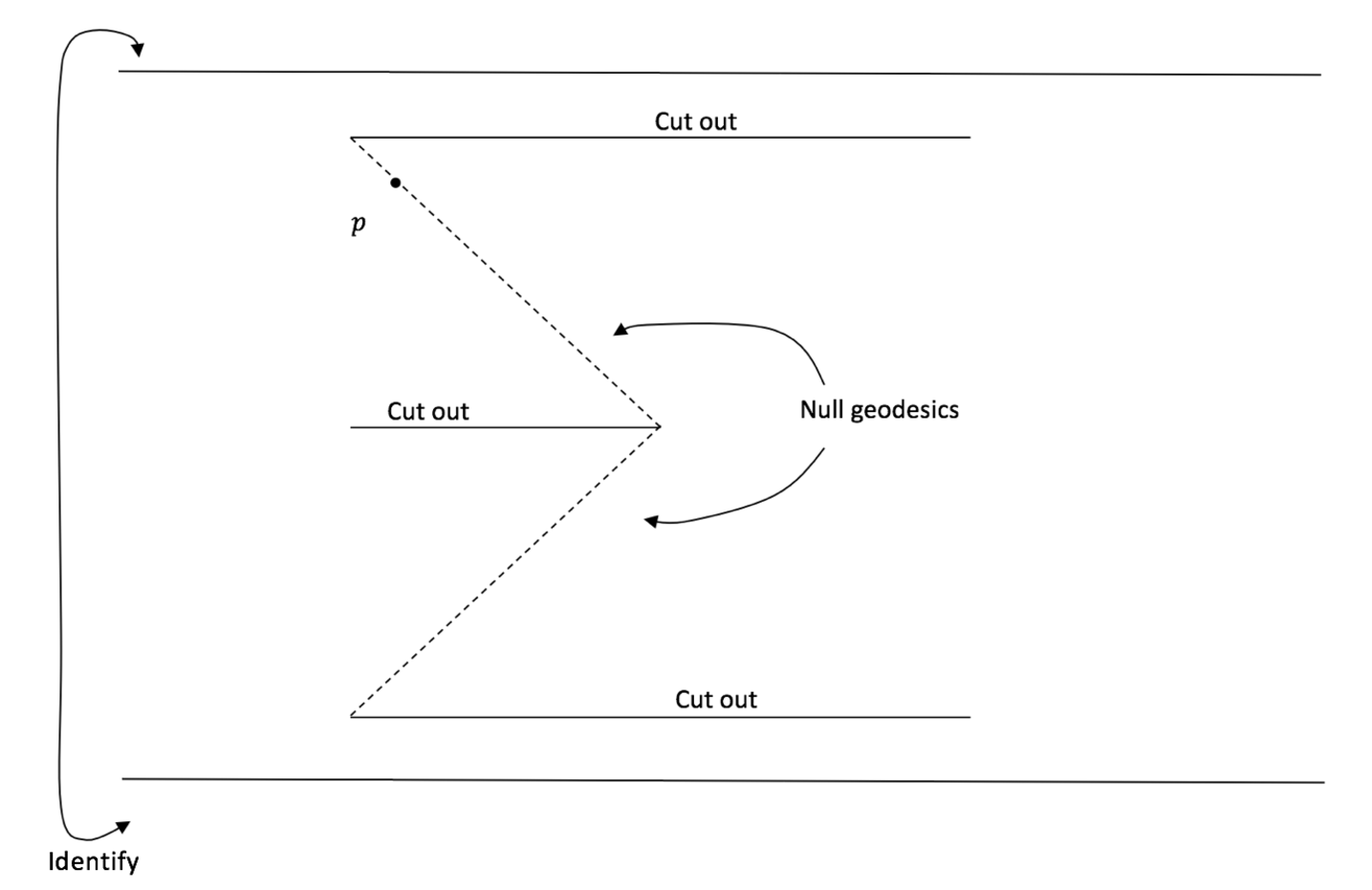}
\caption{A space-time satisfying the strong causality condition, but in which a slight variation of the metric would permit there to be closed timelike curves through $p$. The strips have been removed from the cylinder and the light cone is at $45^{\circ}$.}
\label{fig:1.15}
\end{center}
\end{figure}
\\Again, it would seem inappropriate to regard such space-times as having satisfactory causal behaviour. A motivation for this statement is that General Relativity is presumably the classical limit of a quantum theory of space-time, as remarked in the introduction to this chapter, and hence the metric tensor, according to the Uncertainty Principle, does not have an exact value at every point. Thus in order to be physically significant, a space-time must have some kind of stability, that has to be a property of \lq nearby\rq\hspace{0.1mm} space-times. To give a precise mathematical meaning to \lq nearby\rq\hspace{0.1mm} we have to define a topology on the set of all space-times, all non-compact four-dimensional manifolds and all Lorentz metrics on them. Here we do not consider the problem of uniting under the same topological space manifolds with different topologies, and focus only on putting a topology on the set of all $C^{\infty}$ Lorentzian metrics. There are various way in which this can be done, whether one defines \lq nearby\rq\hspace{0.1mm} metrics to be nearby just in its values ($C^0$ topology) or also in its derivatives up to the $k$th order ($C^k$ topology) and whether one requires it to be nearby everywhere (open topology) or only on compact sets (compact open topology). \	\
Let $\mathscr{M}$ be a space-time manifold. Let $\mathscr{D}(\mathscr{M})$ be the collection of all $C^{\infty}$, symmetric, $(0,2)$ rank tensors on $\mathscr{M}$. The set of Lorentzian metrics is a subset of $\mathscr{D}(\mathscr{M})$ and so will inherit the topology from $\mathscr{D}(\mathscr{M})$. Let $h_{ab}$ be any positive-definite metric on $\mathscr{M}$ (that exists in virtue of paracompatness of space-times) with associate covariant derivative $\nabla_a$, $C$ any closed subset of $\mathscr{M}$ and $k$ any non-negative integer. We define a distance function on pairs of elements $t_{ab},t'_{ab}\in\mathscr{D}(\mathscr{M})$ as follows:
\begin{equation}
\label{eqn:134}
\rho(t_{ab},t'_{ab})=\sup_{C}\sum_{n=0}^{k}2^{-n}\frac{\left| t-t'\right|_n}{1+\left| t-t'\right|_n}
\end{equation}
where 
\begin{equation*}
\left| t-t'\right|_n=\{\left[\nabla_{a_1}...\nabla_{a_n}\left(t_{rs}-t'_{rs}\right)\right]\left[\nabla_{b_1}...\nabla_{b_n}\left(t_{uv}-t'_{uv}\right)\right]h^{a_1b_1}...h^{sv}\}^{1/2},
\end{equation*}
This definition can be found in \cite{GerSing}. The sense of the above defined distance between elements of $\mathscr{D}(\mathscr{M})$ is that two elements are \lq close\rq\hspace{0.1mm} if their values and first $k$ derivatives are close, respect to the metric $h_{ab}$ on the set $C$. The complicated form of \eqref{eqn:134} is necessary to ensure that the least upper bound exists even if $\left| g-g'\right|_n$ may be unbounded. Thus, fixing the integer $k$, for any arbitrary choice of $h_{ab}$ and $C$, it remains defined a topology on $\mathscr{D}(\mathscr{M})$: a neighbourhood of the metric $g_{ab}\in\mathscr{D}(\mathscr{M})$ consists of all $g'_{ab}\in\mathscr{D}(\mathscr{M})$ such that $\rho(g_{ab},g'_{ab})<\epsilon$, for $\epsilon>0$. Since we are not interested in such arbitrary choices, we fix the pair $(h_{ab},C)$ and this choice defines a distance function \eqref{eqn:134} and a family of open sets on $\mathscr{D}(\mathscr{M})$. The aggregate of all finite intersection and arbitrary unions of all open sets of the above defined family defines a topology on $\mathscr{D}(\mathscr{M})$. \\
For our purpose we are only interested in the \textit{$C^0$ open topology} obtained with the choice $k=0$ and requiring $C=\mathscr{M}$. It is worth noting that the construction we used for the $C^0$ open topology has been made following the work of \cite{GerSing} (in which can be found more features and examples), but an equivalent formulation can be made in terms of the bundle of metrics over a manifold \cite[see][]{HawEll}. We are now ready to give the \lq maximally restrictive\rq\hspace{0.1mm} and \lq right\rq\hspace{0.1mm} causality condition for a space-time, that is due to \cite{Haw69}. 
\begin{defn}$\\ $
\label{defn:stc}
A space-time $(\mathscr{M},g)$ is said to be \textit{stably causal} if $g$ has an open neighbourhood in the $C^0$ open topology such that there are no closed causal curves in any metric belonging to the neighbourhood.
\end{defn}
In other words a space-time is stably causal if it cannot be made to contain closed causal curves by arbitrarily small perturbations of the metric and hence by arbitrarily small expansions of the light cones.\\
Other authors give the definition of stable causality in a different, but equivalent, way \citep[see][]{Wald,GerHor}. This definition lies on the idea that if we define a new metric $\tilde{g}_{ab}$ at $p\in\mathscr{M}$ as 
\begin{equation}
\label{eqn:stca}
\tilde{g}_{ab}=g_{ab}-t_at_b,
\end{equation}
with $t^a$ timelike vector at $p$, then $\tilde{g}_{ab}$ is also Lorentzian and its light cone is strictly larger than that of $g_{ab}$ (see Figure \ref{fig:1.16}).
\begin{figure}[h]
\begin{center}
\includegraphics[scale=0.4]{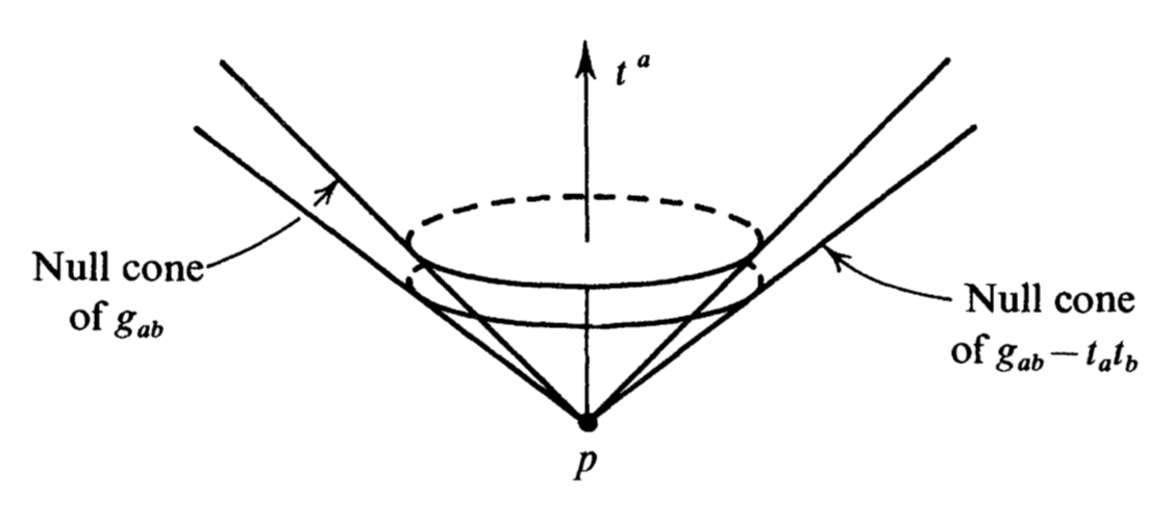}
\caption{The null cones of the metric $\tilde{g}_{ab}$ are more \lq opened out\rq\hspace{0.1mm} than those of $g_{ab}$.}
\label{fig:1.16}
\end{center}
\end{figure}
\\If we \lq open out\rq\hspace{0.1mm} the light cone at each point, if the space-time is stably causal, then we expect it to still not contain closed causal curves.  Hence they define a space-time $(	\mathscr{M},g)$ to be stably causal if there exists a continuous non-vanishing timelike vector field $t^a$ such that the space-time $(	\mathscr{M},\tilde{g})$, with $\tilde{g}$ given by \eqref{eqn:stca}, possesses no closed causal curves. The main consequence of stable causality is given by the following characterization.
\begin{thm}$\\ $
\label{thm:tf}
A space-time $(\mathscr{M},g)$ is stably causal if and only if there exists a differentiable function $f$ on $\mathscr{M}$ such that $\nabla^af$ is a future-directed timelike vector field.
\end{thm}
\begin{proof}
Suppose that there exists $f$ such that $t^a=\nabla^af$ is future-directed timelike. If we consider an arbitrary future-directed causal curve $\gamma$ with tangent vector $v^a$, we have $g_{ab}v^a\nabla^bf>0$ and thus $v(f)>0$. As consequence there can be no closed causal curves in $(\mathscr{M},g)$ since $f$ cannot return to its initial value because it is monotonically decreasing along the curve $\gamma$.\\
Now let $t^a=\nabla^a f$ and set $\tilde{g}_{ab}$ as in \eqref{eqn:stca}. It is easy to see that the inverse of $\tilde{g}_{ab}$ is given by 
\begin{equation*}
\tilde{g}^{ab}=g^{ab}+t^at^b/(1-t^ct_c).
\end{equation*}
Then we obtain 
\begin{equation*}
\tilde{g}^{ab}\nabla_af\nabla_bf=t^at_a+\left(t_at^a\right)^2/(1-t^ct_c)=t^at_a/(1-t^ct_c)>0.
\end{equation*}
Hence $\tilde{g}^{ab}\nabla_bf$ is a timelike vector in the metric $\tilde{g}_{ab}$. By repeating the previous argument it follows that the space-time $(\mathscr{M},\tilde{g})$ contains no closed causal curves. Thus $(\mathscr{M},g)$ is stably causal.
The converse is more complicated to show and its rigorous proof can be found in \cite{HawEll}.
\end{proof}
\begin{oss}$\\ $
The existence of function $f$ can be thought as an assignment of a sort of \lq cosmic time\rq\hspace{0.1mm} on the space-time, in the sense that it increases along every future-directed causal curve. In virtue of this property $f$ is called \textit{time-function}. 
\end{oss}
Hence, as we have just shown, it is always possible to introduce a meaningful concept of time in space-times that satisfy the stable causality condition. Although this result does perhaps gives confidence that the stable causality is the \lq right\rq\hspace{0.1mm} condition, it has to be pointed out that the time-function has a little direct physical significance. In particular, the spacelike surfaces given by $f=\mathrm{const}$ may be thought of as surfaces of simultaneity in space-time, though they are not unique. \\Anyway, there is an important corollary of theorem \ref{thm:tf}.
\begin{cor}$\\ $
If a space-time $(\mathscr{M},g)$ is stably causal, then $(\mathscr{M},g)$ is strongly causal.
\end{cor}
\begin{proof}
Let $f$ be a time-function on $(\mathscr{M},g)$. Given any $p\in\mathscr{M}$ and any open neighbourhood $O$ of $p$, we can choose an open neighbourhood $O'\subset O$ of $p$ such that the limiting value of $f$ along every future-directed causal curve leaving $O'$ is greater than the limiting value of $f$ on every future-directed causal curve entering $O$. Thus, since $f$ increases along every future-directed causal curve, no causal curve can enter $O'$ twice.
\end{proof}
This corollary puts the stable causality as the maximally restrictive causality condition which is acceptable on physical ground. 
\section{Domains of Dependence and Global Hyperbolicity\hspace{1cm}}
\label{sect:1.7}
All prerelativistic theories of space-time were governed by the concept of instantaneous action-at-distance. This means that to predict events at future points in space-time one has to know the state of the \textit{entire} universe at a certain time and assume some reasonable boundary condition at infinity. However, for relativity theory, we introduced in section \ref{sect:1.7} the postulate of local relativity that asserts, essentially, that locally two points $p$ and $q$ can be causally related if and only if there exists some causal curve joining them. Hence, so far, we have only discussed about whether or not an event $p$ can influence an event $q$ by means of a signal, i.e. what is called the \textit{domain of influence} \citep{GerHor,Ger71} of a point $p$. We ask now a slightly different, but related, question, i.e. whether information given on a certain set $S$ of space-time will determine the physical situation in some other region. It is immediately clear that the mathematics appropriate to answer this question will not be a relation between single points of space-time, but between regions.
\begin{defn}$\\ $
\label{defn:dd}
Let $(\mathscr{M},g)$ be a space-time and let $S\subset\mathscr{M}$ be an achronal set. Then
\begin{center}
$D^+[S]=\{\left.p\in\mathscr{M}\right|$ every past-inextendible causal curve through $p$ intersects $S\}$
\end{center}
is called the \textit{future domain of dependence of $S$} or the \textit{future Cauchy development of $S$} .
\end{defn}
The \textit{past domain of dependence of $S$} or the \textit{past Cauchy development of $S$} is defined similarly by interchanging the roles of past an future and is denoted by $D^-[S]$. Finally the \textit{total domain of dependence of $S$} or the \textit{total Cauchy development of $S$} is $D[S]=D^+[S]\cup D^-[S]$.
\begin{oss}$\\ $
This definition agrees with ones given in \cite{Pen67}, \cite{PenTDT}, \cite{Ger70b} and \cite{GerHor}, but note that \cite{Wald} and \cite{HawEll}  replace \lq timelike\rq\hspace{0.1mm} with \lq casual\rq\hspace{0.1mm}. If we denote the latter set by $\tilde{D}^+[S]$ it is easy to show that $\overline{\tilde{D}^+}[S]=D^+[S]$ \citep[see][pg.~202]{HawEll}. Hence, the only effect of such a change would be to eliminate certain boundary points from $D^+[S]$.
\end{oss}
\begin{oss}$\\ $
For simplicity one may also normally restrict attention to the case when $S$ is closed. This is due to the fact that if we knew data on an open set, that on its closure would follow by assuming, reasonably, continuity of the data. The \lq achronal\rq\hspace{0.1mm} property requested for $S$ is due to the fact that it does not appear to be generally useful to define $D^+[S]$ when $S$ is not achronal. Clearly we have $S\subset D^+[S]$.
\end{oss}
The set $D^+[S]$ is of interest because, if nothing can travel faster than light, then any signal sent to $p\in D^+[S]$ must have \lq registered\rq\hspace{0.1mm} on $S$. Thus, if we are given appropriate information about initial conditions on $S$, we should be able to predict what happens at any $p\in D^+[S]$. If a point $p\in I^+[S]$ but $p\notin D^+[S]$, then it should be possible to send a signal to $p$ without influencing $S$ and a knowledge of conditions on $S$ should not suffice to determine conditions at $p$. The physical meaning of $D[S]$ is that, roughly speaking, it represents the complete region of space-time throughout which the physical situation would be expected to be determined, given suitable data, i.e. information,  on an achronal $S$. All the above arguments are valid assuming that the local physic laws are of a suitable \lq deterministic\rq\hspace{0.1mm} and \lq causal\rq\hspace{0.1mm} nature. Some examples illustrating domains of dependence are given in Figure \ref{fig:1.17}.
\begin{figure}[h]
\begin{center}
\includegraphics[scale=0.35]{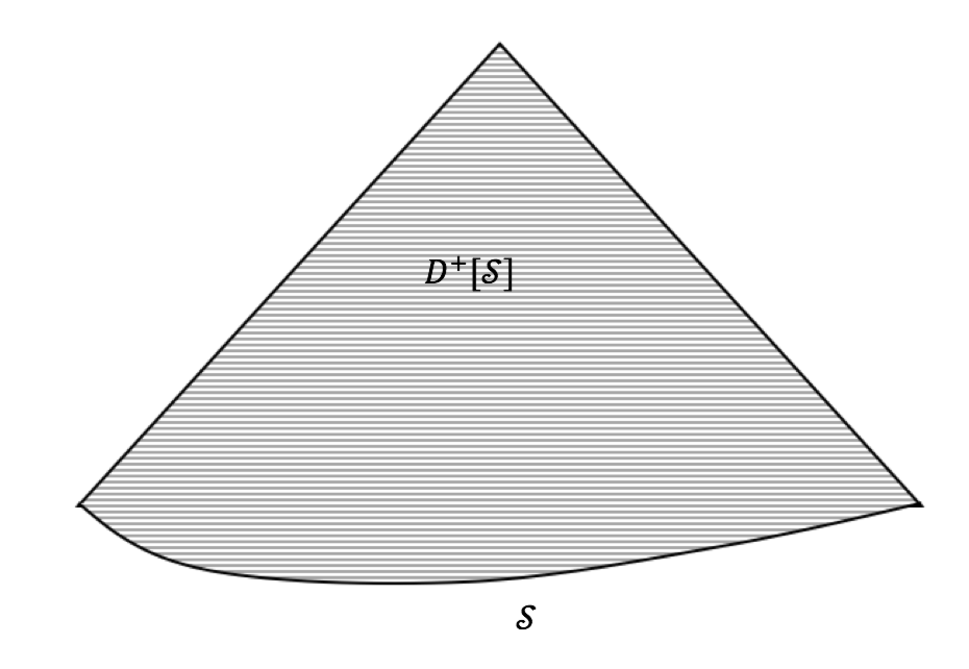}
\hspace{5mm}
\includegraphics[scale=0.35]{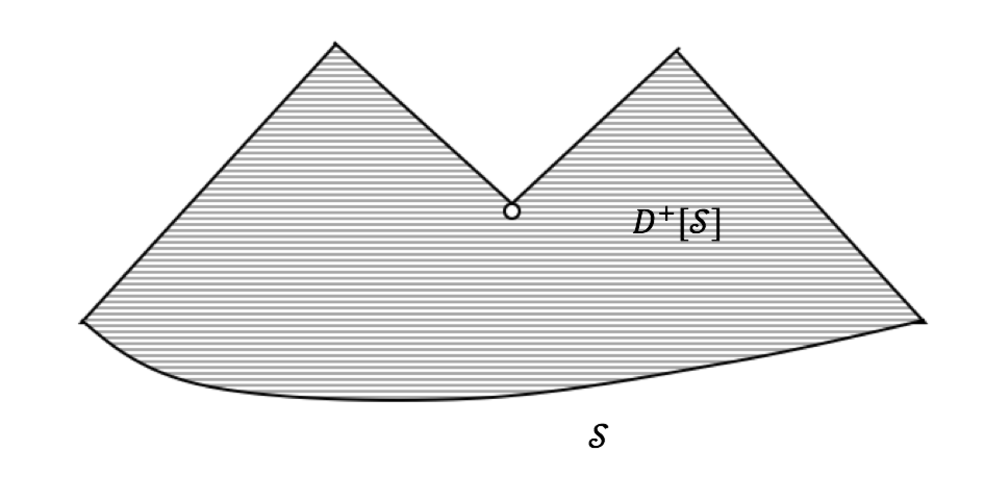}
\hspace{5mm}
\includegraphics[scale=0.35]{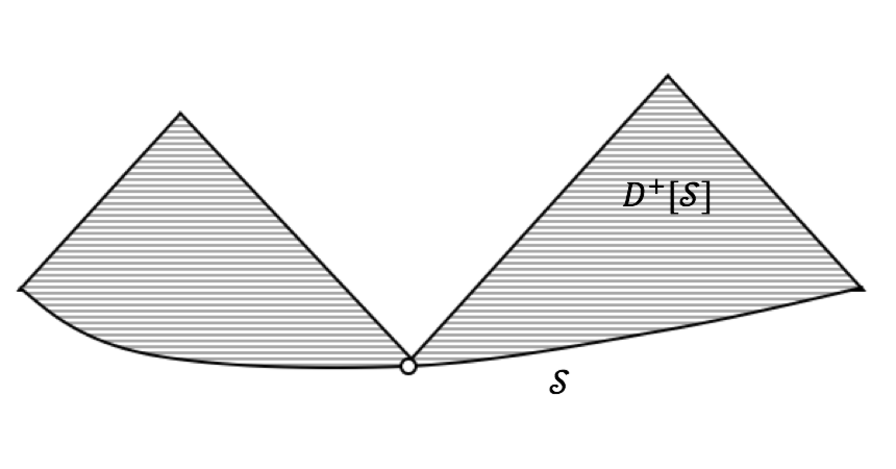}
\caption{On the left, the future domain of dependence of an achronal set $S$. On the right, the effect on $D^+[S]$ of removing a point from the manifold $\mathscr{M}$. Below, the effect on $D^+[S]$ of removing a point from the achronal set $S$.}
\label{fig:1.17}
\end{center}
\end{figure}
\begin{thm}
\label{thm:closed}
Let $S\subset\mathscr{M}$ be an achronal set. Then $\overline{D^+}[S]=D^+[S]\cup \overline{S}$. In particular, if $S$ is closed, so is $D^+[S]$. 
\end{thm}
\begin{proof}
Since $S\subset D^+[S]$, we have $D^+[S]\cup\overline{S}\subset\overline{D^+}[S]$. If we let $\{p_i\}$ be a sequence of point in $D^+[S]$ with accumulation point $p$, to prove the theorem we have to show that $p\in D^+[S]$ or in $\overline{S}$. Suppose that $p\notin\overline{S}$, so that there is some neighbourhood $O$ of $p$ which does not intersect $S$. Let $\gamma$ be any  past-inextendible timelike curve from $p$. Since the sequence $\{p_i\}$ accumulate at $p$, there is some $p_j$ and some timelike curve $\gamma'$ from $p_j$ into the past such that $\gamma'$ joins $\gamma$ in $O$ and thereafter coincides with $\gamma$. But $p_j\in D^+[S]$ and so $\gamma'$ intersects $S$. The intersection cannot occur in $O$, and so must take place after $\gamma$ and $\gamma'$ coincide. Consequently, $\gamma$ intersects $S$ and so $p\in D^+[S]$.
\end{proof}
\begin{thm}$\\ $
\label{thm:U}
Let $p$ be a point of $D^+[S]$. Then $I^-(p)\cap I^+[S]$ is contained in $D^+[S]$. 
\end{thm}
\begin{proof}
Let $q\in I^-(p)\cap I^+[S]$, and let $\gamma$ be a past-inextendible timelike curve from $q$. To prove the theorem we must show that $\gamma$ intersects $S$. Since $q\in I^-(p)$, $\gamma$ can surely be extended to the future to $p$. This new curve, call it $\gamma'$, since it passes through $p\in D^+[S]$ by hypothesis, must intersect $S$. But $q\in I^+[S]$, and so, since $S$ is achronal, $\gamma'$ must intersect $S$ at a point to the past of $q$. Thus $\gamma$ intersects $S$.
\end{proof}
From the previous theorem, by taking the union over all $p\in D^+[S]$ we obtain the following
\begin{cor}$\\ $
\label{cor:int}
$$\mathrm{int}[D^+[S]]=I^-[D^+[S]]\cap I^+[S]$$
\end{cor}
\begin{defn}$\\ $
The \textit{future Cauchy horizon} of an achronal set $S$ is defined as
\begin{equation*}
H^+[S]=D^+[S]-I^-[D^+[S]].
\end{equation*}
\end{defn}
Similarly we have an analogous definition for the \textit{past Cauchy horizon}, $H^-[S]$, and for the \textit{total Cauchy horizon}, $H[S]=H^+[S]\cup H^-[S]$. \\
\begin{figure}[h]
\begin{center}
\includegraphics[scale=0.35]{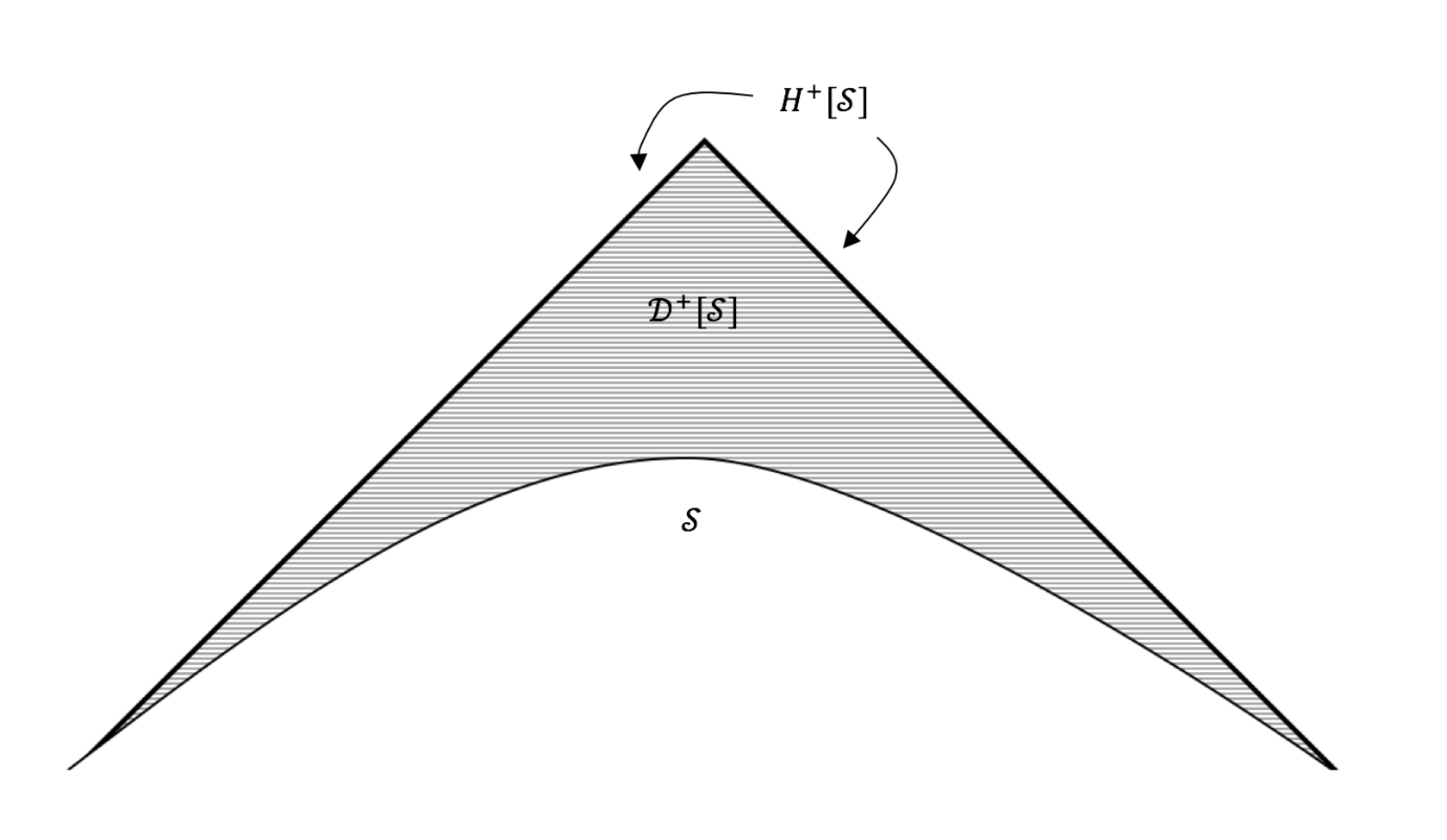}
\caption{If $S$ is the spacelike hypersurface $t=-(x^2+y^2+z^2+1)^{1/2}$ in Minkowski space, then $D^+[S]=\{\left.t\right| (x^2+y^2+z^2)^{1/2}\leq-t\leq (x^2+y^2+z^2+1)^{1/2}\}$ and $H^+[S]=\{\left.t\right| t=-(x^2+y^2+z^2)^{1/2}\}$. We also have $D^-[S]=\{\left.t\right| t\leq -(x^2+y^2+z^2+1)^{1/2}\}$ and $H^-[S]=\emptyset$.}
\label{fig:1.18}
\end{center}
\end{figure}
\begin{thm}$\\ $
Let $S\subset\mathscr{M}$ be a closed and achronal set. Then $H^+[S]$ is closed and achronal.
\end{thm}
The proof can be found in
{Wald}, pg. 203.
The future Cauchy Horizon may be described as the future boundary of $D^+[S]$ and it marks the limit of the region that can be predicted from knowledge of data on $S$.
\begin{thm}$\\ $
\label{thm:S}
Let $p\in D^+[S]-H^+[S]$. Then every past-inextendible causal curve from $p$ intersects $S$.
\end{thm}
\begin{proof}
Let $\gamma(t)$ be a past-inextendible causal curve from $p=\gamma(0)$, with $t\in[0,+\infty)$. Since $\mathscr{M}$ is a space-time, we can always introduce a Riemannian metric over $\mathscr{M}$ and we denote the associated distance by $d$. Since $p\in D^+[S]-H^+[S]$ we can consider a point $q\in D^+[S]\cap I^+(p)$ such that $d(p,q)<1$. We have that $\gamma(1)\in I^-(q)$ and hence we may find a past-directed timelike curve $\gamma'(t)$ for $t\in[0,1]$ from $q$ that satisfies the following requirements:
\begin{equation}
\label{eqn:req}
\gamma(t)\in I^-[\gamma'(t)],\hspace{1.5cm}d(\gamma(t),\gamma'(t))<(1+t)^{-1},\hspace{1.5cm}t\in[0,1].
\end{equation}
Since $\gamma(2)\in I^-(\gamma'(1))$, in virtue of the above hypothesis, we may extend $\gamma'$ so that this extension is timelike and keeps on satisfying the first two requirements of \eqref{eqn:req}, but for $t\in[0,2]$. Continuing in this way we can construct a timelike curve $\gamma'$ always subjected to the first two conditions of \eqref{eqn:req}, whose parameter $t\in[0,+\infty)$ and without past endpoint, otherwise for the second of \eqref{eqn:req} $\gamma$ would have such an endpoint too. Since $\gamma'$ is a timelike curve from $q\in D^+[S]$ it follows that it must intersect $S$. For the first of \eqref{eqn:req} there must be some points of $\gamma$ in the past of  $S$. Let $r$ be the first of those points, at which $\gamma$ leaves $D^+[S]$.
\begin{figure}[h]
\begin{center}
\includegraphics[scale=0.3]{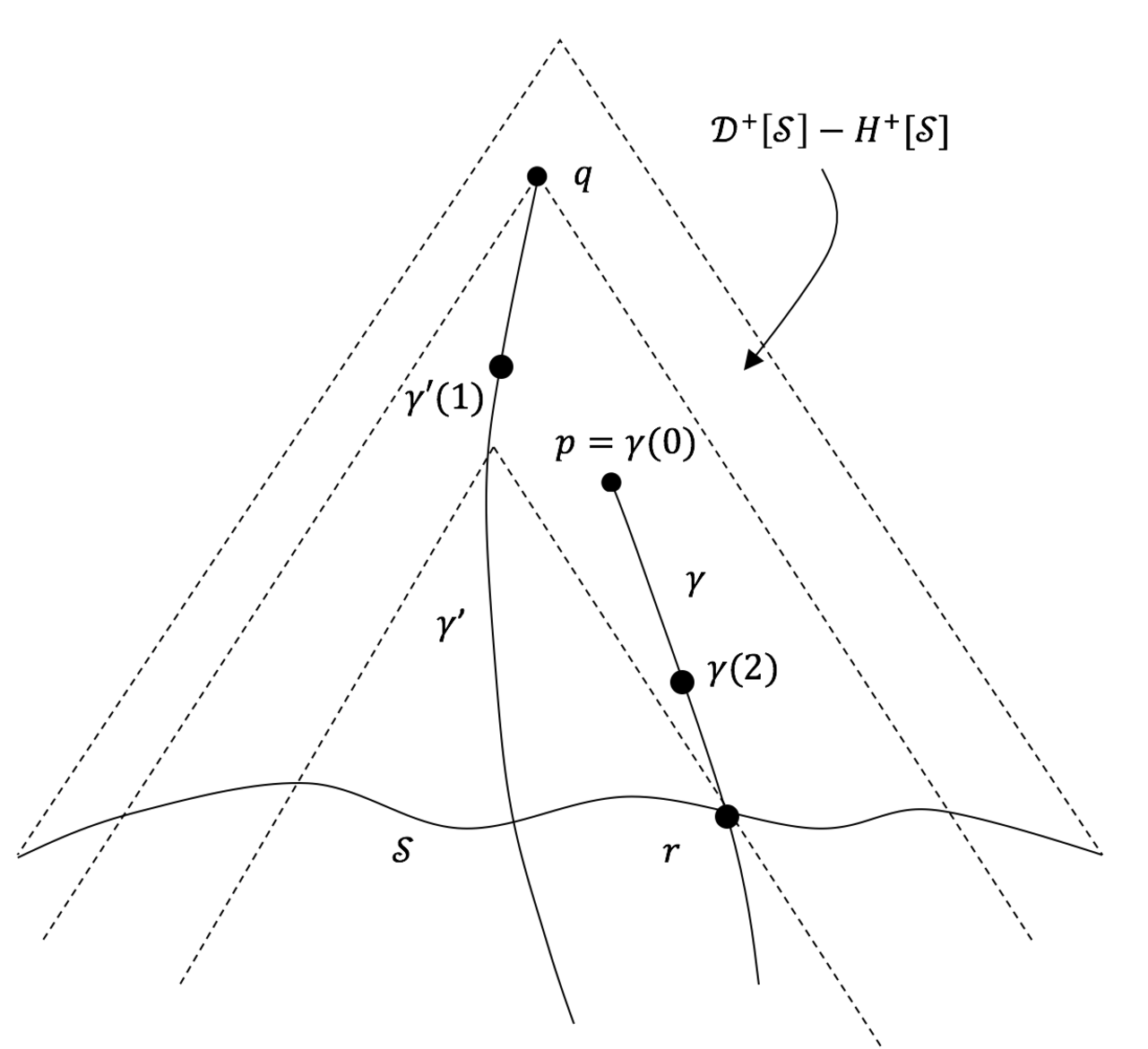}
\caption{A space-time diagram showing the construction used in theorem \ref{thm:S}.}
\label{fig:1.21}
\end{center}
\end{figure}
Since $\overline{D^+}[S]=D^+[S]\cup\overline{S}$ (theorem \ref{thm:S}), $r$ is either in $\overline{S}$ or in $D^+[S]$. The point $r$ cannot lie in $I^+[S]$, otherwise in virtue of theorem \ref{thm:U} $r\in I^-(q)\cap I^+[S]\subset\mathrm{int}[D^+[S]]$ and it would not be leaving $D^+[S]$. Then $r$ must be in $\overline{S}$. Finally $r$ must be in $S$ itself, for otherwise a timelike curve from $q$ and passing through $r$ would violate the fact that $q\in D^+[S]$. 
\end{proof}
The future Cauchy horizon will intersect $S$ if $S$ is null or if $S$ has an \lq edge\rq , has shown in Figure \ref{fig:1.22} 
\begin{figure}[h]
\begin{center}
\includegraphics[scale=0.281]{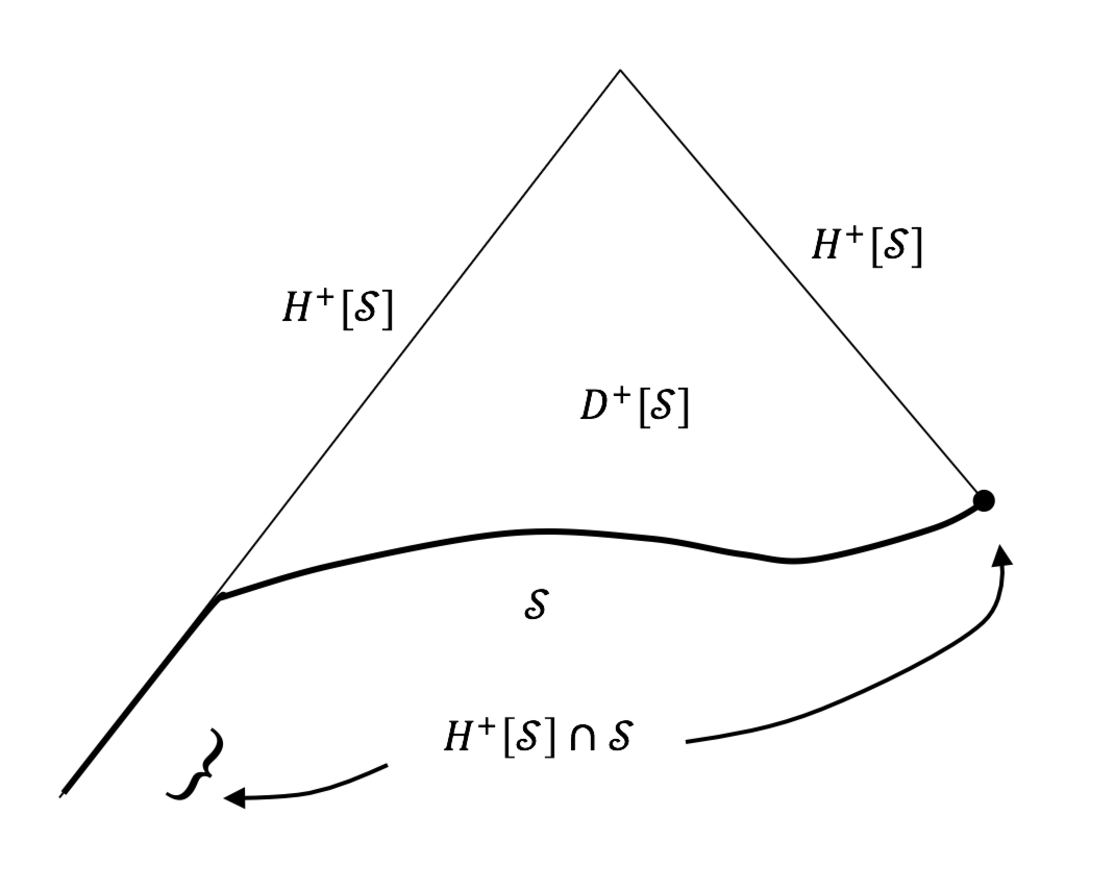}
\caption{The intersection between $S$ and $H^+[S]$.}
\label{fig:1.22}
\end{center}
\end{figure}
\\To make this precise we give the following. 
\begin{defn}$\\ $ 
Let $S$ be an achronal set. The \textit{edge of $S$}, $\mathrm{edge}[S]$, is defined as the set of points $p\in\overline{S}$ such that for every open neighbourhood $O$ of $p$ contains $q\in I^+(p)$ and $r\in I^-(p)$ and a timelike curve $\lambda$ from $r$ to $q$ which does not intersect $S$.
\begin{figure}[h]
\begin{center}
\includegraphics[scale=0.35]{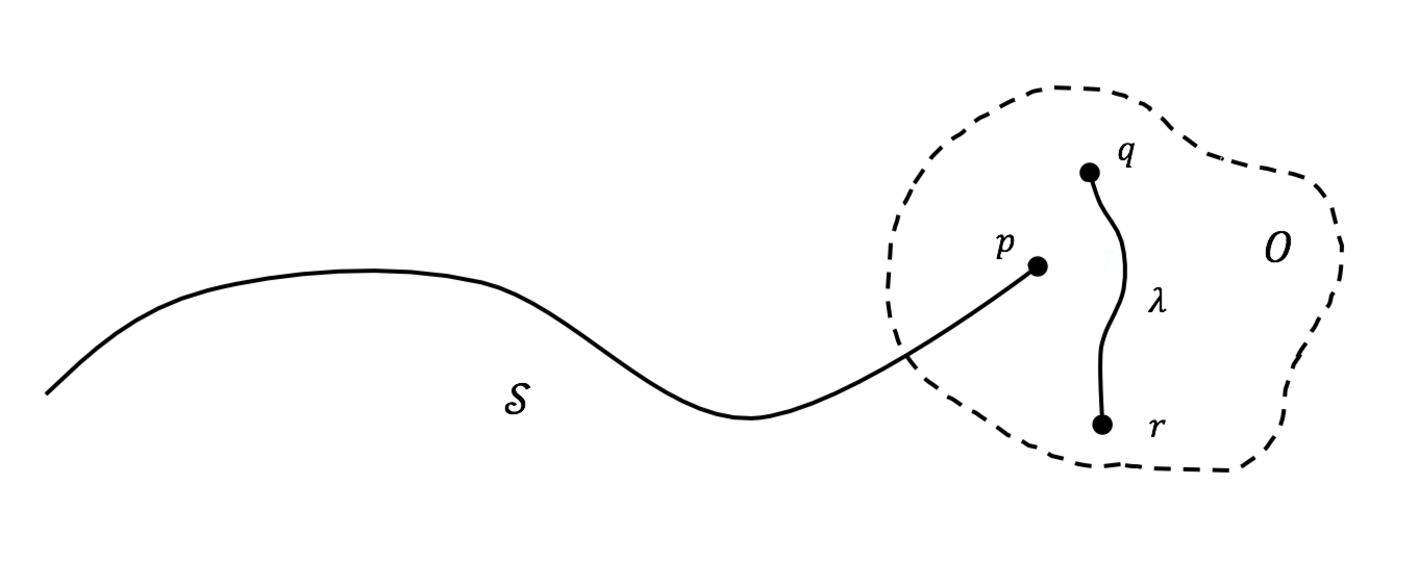}
\caption{The edge of a closed and schronal set $S$.}
\label{fig:1.19}
\end{center}
\end{figure}
\end{defn}
The intuitive meaning of $\mathrm{edge}[S]$ is illustrated in Figure \ref{fig:1.20}. Clearly we have $\overline{S}-S\subset\mathrm{edge}[S]$. We can think $\mathrm{edge}[S]$, roughly speaking, as the set of limit points of $S$ not in $S$, together with the set of points in whose vicinity $S$ fails to be a topological 3-manifold, i.e. those points of $S$ at which $S$ is not locally homeomorphic to $\mathbb{R}^3$.
\begin{figure}[h]
\begin{center}
\includegraphics[scale=0.35]{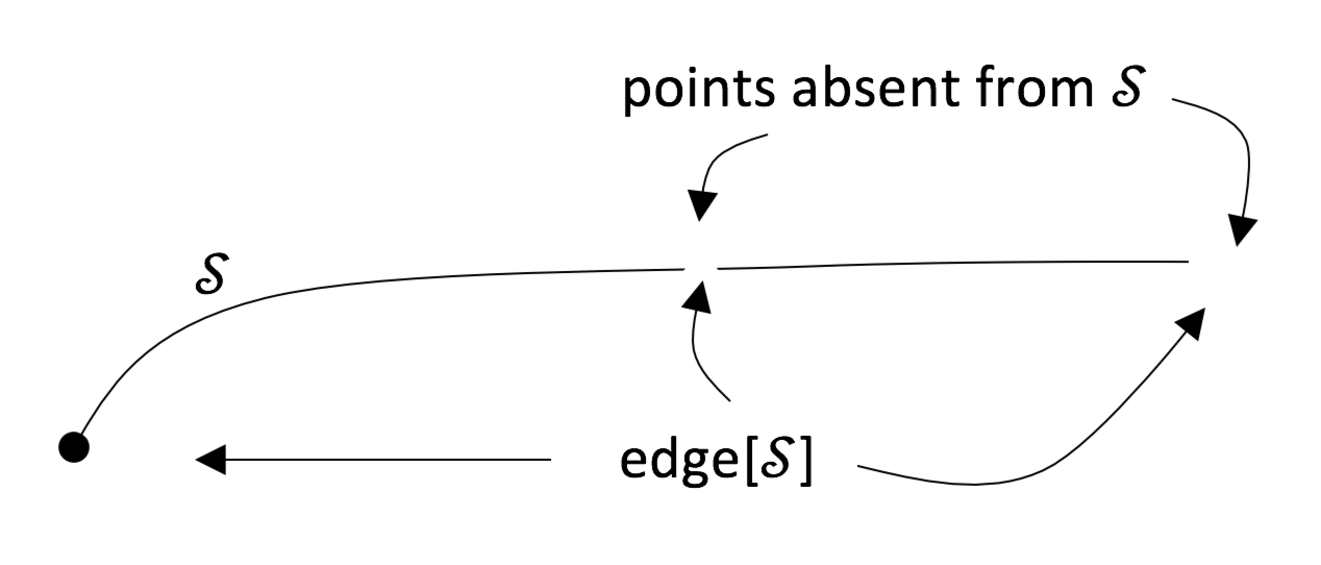}
\caption{An example of an achronal set $S$ and its edge. Note that here $S$ is not closed.}
\label{fig:1.20}
\end{center}
\end{figure}
\\The next theorem, whose proof is similar to that of theorem \ref{thm:embman} makes this statement more precise.
\begin{thm}$\\ $
\label{thm:edge}
If $S$ is an achronal set with $\mathrm{edge}[S]=\emptyset$, then $S$ is a three-dimensional, embedded $C^0$ topological submanifold.
\end{thm}
Also it is easy to show \citep{HawEll} that if $S$ achronal $$\mathrm{edge}[S]=\mathrm{edge}[H^+[S]].$$
In the example of Figure \ref{fig:1.18} we have seen that the Cauchy Horizon is a null hypersurface. Even if this property does not hold in general, it turns out that $H^+[S]$, always contains null geodesics through its points not included in $\mathrm{edge}[S]$. This important property of the Cauchy horizon has been used in theorems about singularities.
\begin{thm}$\\ $
\label{thm:HNG}
Let $S\subset\mathscr{M}$ be achronal and $p\in H^+[S]-\mathrm{edge}[S]$. Then there exists a segment $\Gamma$ of a past-directed null geodesic from $p$ which remains entirely in $H^+[S]$, and which either has no past endpoint or else has a past endpoint on $\mathrm{edge}[S]$.
\end{thm}
\begin{proof}
Choose, in the chronological future of $p$, $I^+(p)$, a sequence of points $\{p_i\}$ which converges to $p$. Since $p\in H^+[S]$, none of the $p_i$ lies in $D^+[S]$. Therefore, through each $p_i$ we may draw a past-directed timelike curve $\gamma_i$, without endpoint, such that $\gamma_i$ does not intersect $S$ and, in particular, does not enter $D^+[S]$. Since $p$ is a limit point of the sequence of curves $\{\gamma_i\}$, in virtue of theorem \ref{thm:caus} there exists a limit curve of that sequence through $p$. In particular we denote by $\Gamma'\subset I^+[S]$ a partial limit segment of a past-directed null geodesic from $p$, such that, given any point $q\in\Gamma'$ and any neighbourhoods $O_q$ and $O_{\Gamma'}$ of $q$ and $\Gamma'$,  respectively, an infinite number of $\gamma_i$ remain in $O_{\Gamma'}$, at least until they reach $O_q$. Let $q\in\Gamma'$. Then each past-directed timelike curve from $q$ enters $I^-(p)\cap I^+[S]\subset D^+[S]$ (see theorem \ref{thm:U}), and so intersects $S$. Furthermore, no point of $I^+(q)$ is in $D^+[S]$, for otherwise at least one $\gamma_i$ would enter in $D^+[S]$. Thus $q\in H^+[S]$ and each partial limit segment is in $H^+[S]$.
\begin{figure}[h]
\begin{center}
\includegraphics[scale=0.35]{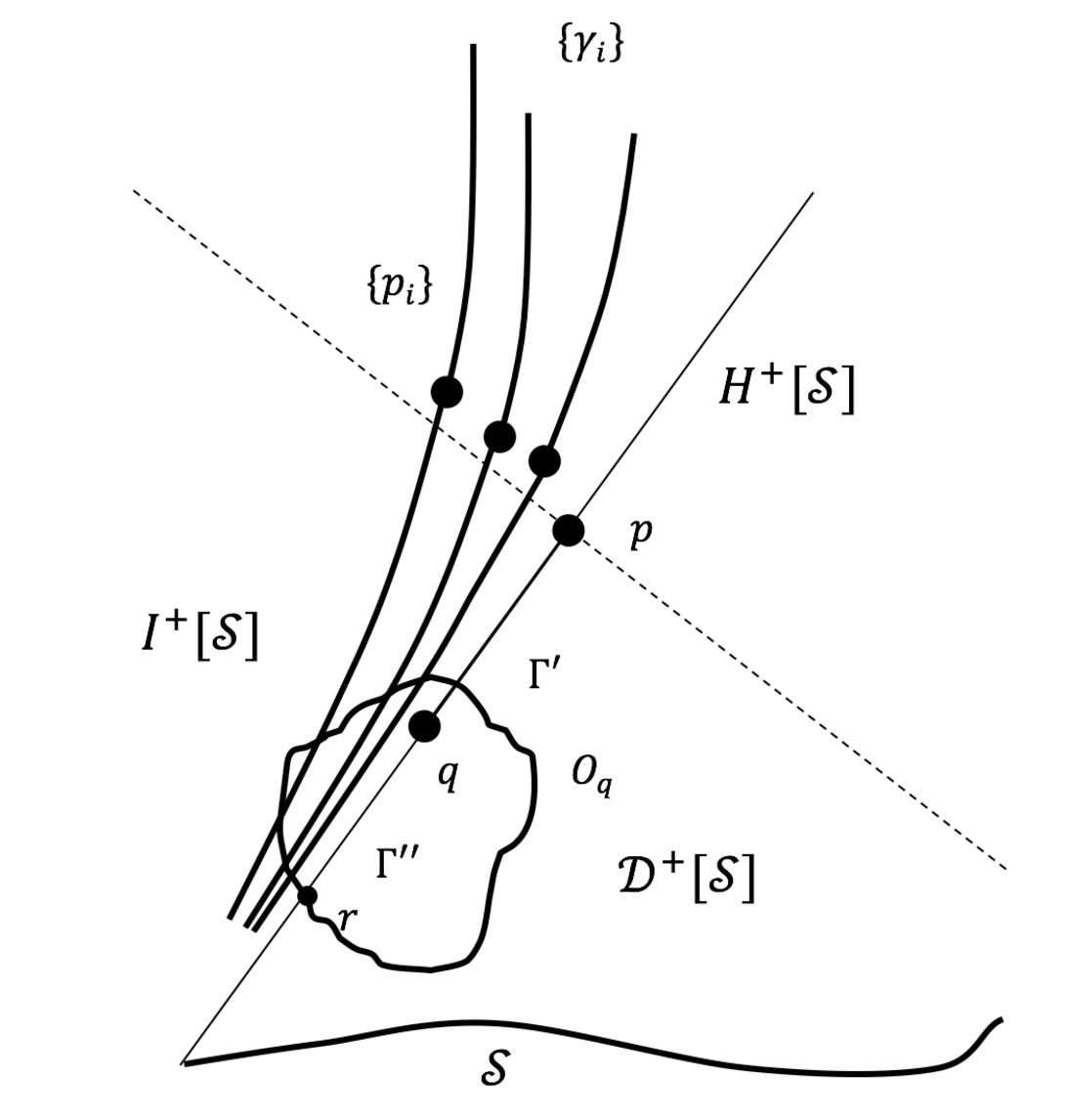}
\caption{The scheme used in the proof of theorem \ref{thm:HNG}.}
\label{fig:1.23}
\end{center}
\end{figure}
\\Let $\Gamma'$ be a partial limit and suppose that $\Gamma'$ has a past endpoint $q\in I^+[S]$ (see Figure \ref{fig:1.22}). Choose a small compact neighbourhood $O_q$ of $q$. The sequence of points at which the $\gamma_i$ first leave $O_q$ must have, by compactness, an accumulation point $r$. Since the $\gamma_i$ cannot enter $D^+[S]$, it follows that $r\notin I^-(q)$. Moreover the $\gamma_i$ pass arbitrarily close to $q$ and are timelike curves, thus $I^+(q)\subset I^+(r)$. Therefore there is a segment $\Gamma''\subset O_q$ of a null geodesic joining $q$ to $r$. The tangent vectors to $\Gamma'$ and $\Gamma''$ must agree at $q$, for otherwise $\Gamma''$,  and hence the $\gamma_i$, would enter $I^-(p)\cap I^+[S]\subset D^+[S]$. Therefore $\Gamma'\cup\Gamma''$ is a partial limit of the $\gamma_i$. We have shown that each partial limit which has a past endpoint in $I^+[S]$ may be extended, as a partial limit, beyond that endpoint. It follows immediately that there exists a partial limit $\Gamma$ which either has no past endpoint or else has a past endpoint $s\notin I^+[S]$. In the latter case, since $\Gamma\subset D^+[S]$, we must have $s\in \overline{D^+}[S]=D^+[S]\cup\overline{S}$ (see theorem \ref{thm:closed}) and, since $s\notin D^+[S]$, it follows that $s\in\overline{S}$. Furthermore $s$ cannot be in $\mathrm{int}[S]$, for otherwise the partial limit $\Gamma$ from $p$ to $s$ would be timelike. Hence $s\in\overline{S}-\mathrm{int}[S]\subset \mathrm{edge}[S]$. 
\end{proof}
We have as consequence, using theorem \ref{thm:edge} the following
\begin{cor}$\\ $
If $\mathrm{edge}[S]=\emptyset$, then $H^+[S]$ is an achronal, three-dimensional, embedded $C^0$ topological manifold which is generated by null geodesic segments which have no past endpoint.
\end{cor}
We turn now attention to Cauchy surfaces in a certain space-time $(\mathscr{M},g)$.
\begin{defn}$\\ $
An achronal set $S$ for which $D[S]=\mathscr{M}$ is said to be a \textit{Cauchy surface}.
\end{defn}
It follows immediately that for any Cauchy surface $S$, we must have $\mathrm{edge}[S]=\emptyset$. Hence by theorem \ref{thm:edge} every Cauchy surface is a three-dimensional, embedded $C^0$ topological submanifold of $\mathscr{M}$.\\
The set $S$ in the example in Figure \ref{fig:1.18} clearly $S$ is not a Cauchy surface. However it is easy to see that in Minkowski space-time there are Cauchy surfaces. If we consider, for example, the plane $S$ given by $t=0$  it is clear that $D[S]=\mathscr{M}$, and thus $S$ is a Cauchy surface. The same surface is also a Cauchy surface in the (extended) Schwarzschild solution, but this is not true in Reissner-N\"ordstrom solution (with which we are not dealing). Hence it is clear that to say that $S$ is a Cauhy surface for $\mathscr{M}$ is a statement about both $S$ and the whole space-time $(\mathscr{M},g)$ in which it is embedded.\\Intuitively, that $S$ be a Cauchy surface for $\mathscr{M}$ means that initial data on $S$ determines the entire evolution of $\mathscr{M}$, past and future. Thus, in a certain sense,  one could think of a space-time with a Cauchy surface as being \lq predictive\rq. Conversely, in space-times with no Cauchy surfaces we have a breakdown of predictability in the sense that a complete knowledge of conditions at a single \lq instant of time\rq\hspace{0.1mm} can never suffice to determine the entire history of the universe. Hence there are some good reasons for believing that all physically realistic space-times must admit a Cauchy surface. However one could not know the initial data on a certain surface $S$ unless one was to the future of every point in the surface, which would be impossible in most cases. In fact, in general, it is not possible to tell, by examining only a neighbourhood of $S$, whether or not $S$ will be a Cauchy surface because a space-time $\mathscr{M}$ in which it appears, during the early stages of evolution, that $S$ will be a Cauchy surface may, at some much later time, develop so as to have no Cauchy surface (see the Reissner-Nordstr\"om example in \cite{Ger71}, pg. 94). Furthermore there are a number of known exact solutions of the Einstein equations which do not admit such surfaces (AdS, Taub-NUT, Reissner-Nordstr\"om, etc.) and one must take into account that, often, there could be extra information coming in from infinity or from the singularity which would upset any predictions made simply on the basis of data on $S$. Thus in General Relativity one's ability to predict the future is limited both by the difficulty of knowing data on the whole of a spacelike surface and by the possibility that even if one did it would still be insufficient. It follows that the assumption of the existence of a Cauchy surface seems to be a rather strong condition to impose and there does not seem to be any physically compelling reason for believing that the universe should admit a Cauchy surface (this will be even more justified later). However, it is worth noting that requiring the existence of a Cauchy surface is a very useful tool if we want to study the Cauchy problem in General Relativity. \\
Before briefly discussing the global hyperbolicity we give some important results about Cauchy surfaces.
\begin{thm}$\\ $
\label{thm:cs}
Let $S$ be an achronal set and $\mathscr{M}$ be connected. Then $S$ is a Cauchy surface for $\mathscr{M}$ if and only if $H[S]=\emptyset$.
\end{thm}
\begin{proof}
If $S$ is a Cauchy surface, then $D[S]=\mathscr{M}$ by definition, and hence $H[S]=\emptyset$. Conversely if $H^+[S]=\emptyset$ it is in particular closed and hence $D^+[S]$ is closed too. Let $p\in D^+[S]$. Since $p\notin H^+[S]$, there is a point $r$ to the future of $p$ in $D^+[S]$. Since $p\notin H^-[S]$, there is a point $s$ to the past of $p$ in $D^-[S]$. By theorem \ref{thm:U} the open neighbourhood $I^-(r)\cap I^+(s)$ of $p$ is in $D[S]$. Since $D[S]$ is both open and closed and $\mathscr{M}$ is connected $D[S]=\mathscr{M}$.
\end{proof}
Hence the mere presence of a non-empty Cauchy horizon means that an achronal surface cannot be a Cauchy surface. 
\begin{thm}$\\ $
\label{thm:ccvS}
Let $S$ be a Cauchy surface and let $\gamma$ be an inextendible causal curve. Then $\gamma$ intersects $S$, $I^+[S]$ and $I^-[S]$.
\end{thm}
\begin{thm}$\\ $
\label{thm:ngcs}
Let $S$ be a closed, achronal set. Then $S$ is a Cauchy surface if and only if every inextendible null geodesic in $\mathscr{M}$ intersects $S$ and enters $I^+[S]$ and $I^-[S]$, i.e. intersects and then re-emerge from $S$.
\end{thm}
The proof can be found in \cite{Ger70}. Often theorems \ref{thm:HNG} and \ref{thm:cs} can be used together to establish the existence of a Cauchy surface. In fact, if there is some condition that ensure us that the geodesics required in \ref{thm:HNG} do not exist for a certain $S$, then it follows that $H[S]$ is empty and hence that $S$ is a Cauchy surface. Theorems \ref{thm:ccvS} and \ref{thm:ngcs} just provide, following these arguments, some characterization of a Cauchy surface in terms of the behaviour of causal curves and null geodesics. Note that the \lq only if\rq\hspace{0.1mm} part of theorem \ref{thm:ngcs} is just a special case of theorem \ref{thm:ccvS}.\\
To discuss about global hyperbolicity there are various definitions and it can be shown they are completely equivalent one to each other. In particular one can choose to follow the definitions of \cite{Wald}, \cite{HawEll}, whose approach is  adopted here, or \cite{Ger70} and \cite{Ler}. However it must be pointed out the original idea was introduced by \cite{Ler}.
\begin{defn}{\cite{HawEll}}$\\ $
\label{defn:ghbH}
A set $N$ is \textit{globally hyperbolic} if:
\begin{itemize}
\item Strong causality holds in $N$;
\item $\forall p$, $q\in N$ the set $J^+(p)\cap J^-(q)$ is compact and contained in $N$ .
\end{itemize}
\end{defn} 
Clearly to obtain the definition of a hyperbolic space-time $(\mathscr{M},g)$ it suffices to replace $N$ by $\mathscr{M}$ in the previous definition.\\
This can be thought of as saying that $J^+(p)\cap J^-(q)$ does not contain any points on the \lq edge\rq\hspace{0.1mm} of space-time, i.e. at infinity or at a singularity. The reason for the name \lq global hyperbolicity\rq\hspace{0.1mm} is that the wave equation for a $\delta$-function source at a point $p$ located inside a globally hyperbolic set $N$ has a unique solution which vanishes outside $N-J^+(p)$. We will see how the concept of global hyperbolicity is strictly related to that of existence of a Cauchy surface. 
\begin{defn}$\\ $
\label{defn:causim}
A set $N$ is said to be \textit{causally simple} if for every compact set $H$ contained in $N$, $J^+[H]\cap N$ and $J^-[H]\cap N$ are closed in $N$. 
\end{defn}
We have seen in section \ref{sect:1.4} that the sets $J^+(p)$ and $J^-(p)$ are not always closed. However if we suppose the space-time to be hyperbolic they are closed, how is stated by the next theorem.
\begin{thm}$\\ $
Let $(\mathscr{M},g)$ be a globally hyperbolic space-time. Then it is causally simple, i.e. for a compact set $H\subset\mathscr{M}$ the sets $J^{\pm}[H]$ are closed.
\end{thm}
We prove this theorem in the simple case in which $H$ consists of a single point $p$. The complete proof can be found in \cite{HawEll}, pg. 207.
\begin{proof}
Choose $p\in\mathscr{M}$ and suppose $J^+(p)$ is not closed. It follows that we can find a point $r\in\overline{J^+}(p)$ with $r\notin J^+(p)$. Choose $q\in I^+(r)$. Then we would have $r\in \overline{J^+(p)\cap J^-(q)}$ but $r\notin J^+(p)\cap J^-(q)$, which is a contradiction since $J^+(p)\cap J^+(q)$ is compact, and hence closed by hypothesis of global hyperbolicity (see theorem \ref{thm:cc1}). Hence $J^+(p)$ is closed. 
\end{proof}
It can be shown to be valid the following
\begin{cor}$\\ $
Let $(\mathscr{M},g)$ be a globally hyperbolic space-time. If $H_1$ and $H_2$ are compact sets in $\mathscr{M}$ then $J^+[H_1]\cap J^-[H_2]$ is compact.
\end{cor}
Hence in definition \ref{defn:ghbH} of global hyperbolicity the points $p$ and $q$ can be replaced by compact sets $H_1$ and $H_2$.\\
In order to define global hyperbolicity according to \cite{Ler} it is necessary to introduce a topology on certain collection of curves in the space-time $(\mathscr{M},g)$. For points $p$,$q\in\mathscr{M}$ such that strong causality holds in $J^+(p)\cap J^-(q)$ we define $C(p,q)$ to be the space of all causal curves from $p$ to $q$. A point in $C(p,q)$ is a causal curve from $p$ to $q$, up to a reparametrization. In this way two curves $\gamma(t)$ and $\lambda(u)$ will be considered equivalent and to represent the same point in $C(p,q)$ if there exists a continuous monotonic function $f(u)$ such that $\gamma(f(u))=\lambda(u)$. We are interested in defining a topology $\mathscr{T}$ on $C(p,q)$ to make $(C(p,q),\mathscr{T})$ into a topological space. We say that a neighbourhood $W$ of a point $\gamma$ in $C(p,q)$ consists of all the curves in $C(p,q)$ whose points in $\mathscr{M}$ lie in a neighbourhood $O$ of the points of $\gamma$ in $\mathscr{M}$. Let $O\subset\mathscr{M}$ be open, and define $W(O)\subset C(p,q)$ by \begin{equation}
\label{eqn:tplg}
W(O)=\{\left.\gamma\in C(p,q)\right|\gamma\subset O\}.
\end{equation}
We define the topology by calling a subset, $W$, \textit{open} if it can be expressed as 
\begin{equation*}
W=\bigcup W(O),
\end{equation*}
where each $W(O)$ has the form \eqref{eqn:tplg}.
\begin{figure}[h]
\begin{center}
\includegraphics[scale=0.4]{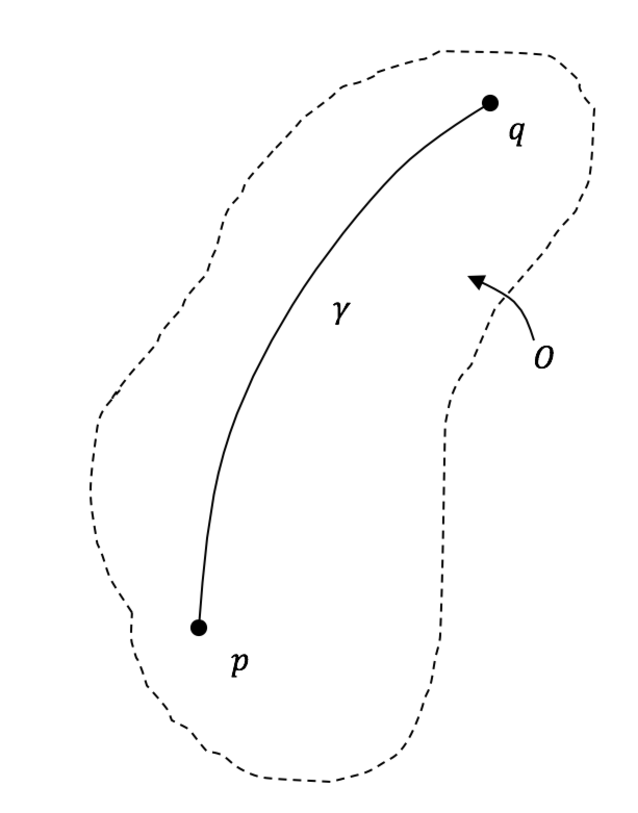}
\caption{A neighbourhood $O$ of the points of $\gamma$ in $\mathscr{M}$. A neighbourhood of $\gamma$ in $C(p,q)$ consists of all causal curves from $p$ to $q$ whose points lie in $\mathscr{M}$.}
\label{fig:1.24}
\end{center}
\end{figure}
We are saying that the set of all curves in $C(p,q)$ which lie in $O$, whilst $O$ ranges over all open sets in $\mathscr{M}$, defines a basis for the topology $\mathscr{T}$ on $C(p,q)$. Since we are requiring strong causality to hold, there exist no closed causal curves on $\mathscr{M}$. It is easy to see that this property implies that the topology $\mathscr{T}$ is Hausdorff. Furthermore, in absence of closed causal curves it can be shown that the topological space $(C(p,q),\mathscr{T})$ has a countable basis and hence is second countable \citep{Ger70}. It is worth noting that in the original work by Leray, the topology $\mathscr{T}$ is introduced in an arbitrary space-time in which there are no causality restrictions. Furthermore the notion of convergence defined by $\mathscr{T}$ is the following: $\gamma_n\rightarrow\gamma$ if for every set $O\subset\mathscr{M}$ with $\gamma\in O$, there exists a $N$ such that $\lambda_n\subset O$ for all $n>N$. This definition of convergence, in absence of closed causal curves, coincides with definition \ref{defn:conv}.
\begin{thm}$\\ $
\label{thm:equiv1}
Let strong causality hold on an open set $N$ such that $$N=J^-[N]\cap J^+[N].$$ Then $N$ is globally hyperbolic if and only if $C(p,q)$ is compact for all $p$,$q\in N$. 
\end{thm}
\begin{proof}
Suppose first that $C(p,q)$ is compact. Let $\{r_n\}$ be a sequence of points in $J^-(q)\cap J^+(p)$ and let $\{\gamma_n\}$ be a sequence of curves in $C(p,q)$ through the corresponding $r_n$. Since $C(p,q)$ is compact there exists a subsequence $\{\gamma'_n\}\subset\{\gamma_n\}$ which converges to a curve $\gamma$ in $C(p,q)$ in the topology $\mathscr{T}$. Since $\gamma$, regarded ad a subset of $\mathscr{M}$, is compact we can find an open neighbourhood $O$ of $\gamma$ with compact closure $\overline{O}$ (Cover $\gamma$ with open sets with compact closure , use compactness of $\gamma$ to extract a finite subcover, and take the union). Call $\{r'_n\}\subset\{r_n\}$ the subsequence of points through which the curves $\{\gamma'_n\}$ pass. Then, by the notion of convergence, there exists a $N$ such that $\gamma'_n\subset O$ for all $n>N$ and, since $r'_n\in\gamma'_n$ and $O\subset\overline{O}$, the sequence $\{r'_n\}$ converges to a point $r\in \overline{O}$, by compactness of $\overline{O}$. The point $r$ must lie on $\gamma$, for otherwise we would contradict the fact that $\gamma$ is the limit curve of $\{\gamma'_n\}$. Thus every infinite sequence in $J^-(q)\cap J^+(p)$ has a subsequence converging to a point in $J^-(q)\cap J^+(p)$ and hence, by theorem \ref{thm:BW},  $J^-(q)\cap J^+(p)$ is compact, i.e. $N$ is globally hyperbolic. \\
Conversely, suppose $J^-(q)\cap J^+(p)$ is compact, i.e. $N$ is globally hyperbolic. Suppose $\{\gamma_n\}$ be an infinite sequence of causal curves from $p$ to $q$, so that $\gamma_n\in C(p,q)$. In the space-time whose manifold is $\mathscr{M}-q$ the sequence $\{\gamma_n\}$ is a sequence of future-inextendible causal curves. Hence, using theorem \ref{thm:caus}, in $\mathscr{M}-q$ there will be a future-directed causal curve $\gamma$ from $p$ which is inextendible and such that there is a subsequence $\{\gamma'_n\}\subset\{\gamma_n\}$ which converges to $r$, for every $r\in\gamma$, i.e. $\gamma$ is a limit curve of $\{\gamma_n\}$.
The compact set $J^-(q)\cap J^+(p)$ can be covered by a finite number of local causality neighbourhood $U_i$. By strong causality any future-inextendible causal curve which intersects one of these neighbourhoods must leave it and not re-enter it. Hence no future-inextendible causal curve can be \textit{imprisoned} \citep{HawEll} in $J^-(q)\cap J^+(p)$. Thus the curve $\gamma$ in $\mathscr{M}$ must have a future endpoint at $q$, because it cannot be imprisoned in the compact set $J^-(q)\cap J^+(p)$ and it cannot leave the set except at $q$. Let $U$ be any neighbourhood of $\gamma$ in $\mathscr{M}$ and let $r_i$ $(1\leq i\leq k)$ be a finite set of points on $\gamma$ such that $r_1=p$ and $r_k=q$ and each $r_i$ has a neighbourhood $O_i$ with $J^-[O_{i+1}]\cap J^+[O_i]\subset U$. For sufficiently large $n$, $\lambda'_n$ will be contained in $U$ and thus the sequence $\{\lambda'_n\}$ converges to $\lambda$ in the topology of $C(p,q)$ and so $C(p,q)$ is compact.
\end{proof}
The above theorem proves the equivalence between definition \ref{defn:ghbH} and the following, which is due to Leray. 
\begin{defn}{\cite{Ler}}$\\ $
\label{defn:ghbL}
An open set $N$ is \textit{globally hyperbolic} if:
\begin{itemize}
\item Strong causality holds in $N$;
\item $C(p,q)$ is compact for every $p$,$q\in N$. 
\end{itemize}
\end{defn}
The next theorem relates the existence of Cauchy surfaces to the absence of causal curves and will play a fundamental role in the establishment of the above claimed equivalence between the different definitions of global hyperbolicity.
\begin{thm}$\\ $
\label{thm:GHSC}
Let $(\mathscr{M},g)$ be a space-time which admits a Cauchy surface $S$. Then $(\mathscr{M},g)$ is strongly causal.
\end{thm}
\begin{proof}
Clearly we have $\mathscr{M}=D^+[S]\cup D^-[S]\cup S$. Suppose strong causality were violated at $p\in I^+[S]$. From definition \ref{defn:stcaus} it follows that we could find a convex normal neighbourhood $U$ of $p$ contained in $I^+[S]$ and a nested family of open sets $O_n\subset U$ which converges to $p$ such that for each $n$ there exists a future-directed causal curve $\gamma_n$ which begins in $O_n$, leaves $U$, and ends in $O_n$. Using theorem \ref{thm:caus}, there exists a limit causal curve $\gamma$ through $p$. This curve must be either inextendible or closed through $p$, in which case it could be made inextendible by going \lq around and around\rq. Since none of the $\gamma_n$ can enter in $I^-[S]$, for otherwise $S$ would not be achronal, $\gamma$ also cannot enter $I^-[S]$. However, this contradicts theorem \ref{thm:ccvS} and hence strong causality cannot be violated in $I^+[S]$. Similarly one can repeat the above arguments for  $I^-[S]$. In the case $p\in S$ we can choose the family $\{O_n\}$ so that any future-directed causal curve starting in $O_n$ leaves $O_n$ in $I^+[S]$. Thus the limit curve $\gamma$ could not enter $I^-[S]$, which is again in contrast with theorem \ref{thm:ccvS}.
\end{proof}
The following theorem, whose proof is due to \cite{Ger70}, gives a fundamental characterization of globally hyperbolic space-times in terms of existence of Cauchy surfaces.
\begin{thm}{\cite{Ger70}}$\\ $
\label{thm:equiv2}
A space-time $(\mathscr{M},g)$ is globally hyperbolic, according to definition \ref{defn:ghbL}, if and only if it has a Cauchy surface.
\end{thm}
\begin{oss}$\\ $
Note that theorem \ref{thm:GHSC} plays a fundamental role in the \lq if\rq\hspace{0.1mm} part.
\end{oss}
The above result allows us to give the last definition of global hyperbolicity, which is due to Wald.
\begin{defn}{\cite{Wald}}$\\ $
\label{defn:ghbW}
A space-time $(\mathscr{M},g)$ is \textit{globally hyperbolic} if it possesses a Cauchy surface.
\end{defn}
Together with theorem \ref{thm:equiv1}, theorem \ref{thm:equiv2} shows the complete above mentioned equivalence between all three definition (\ref{defn:ghbH}, \ref{defn:ghbL} and \ref{defn:ghbW}) of global hyperbolicity.\\
The last result we are going to discuss, which is due to \cite{Ger70}, greatly strenghtens theorem \ref{thm:GHSC} and provides an important topological property of globally hyperbolic space-times.
\begin{thm}{\cite{Ger70}}$\\ $
Let $(\mathscr{M},g)$ be a globally hyperbolic space-time. Then $(\mathscr{M},g)$ is stably causal. Furthermore, a global time function, $f$, can be chosen such that each Cauchy surface of constant $f$ is a Cauchy surface. Thus $\mathscr{M}$ can be foliated by Cauchy surface and its topology is $\mathbb{R}\times S$, where $S$ denotes any Cauchy surface.
\end{thm}
All the results we have obtained confirm the fact that the existence of a Cauchy surface in a space-time, i.e. the property of global hyperbolicity, is a very strong condition. In particular there cannot be any kind of causal anomalies due to the presence of the stable causality condition. We can say that sufficiently small variations in the metric do not destroy global hyperbolicity. 
\begin{figure}[h]
\begin{center}
\includegraphics[scale=0.48]{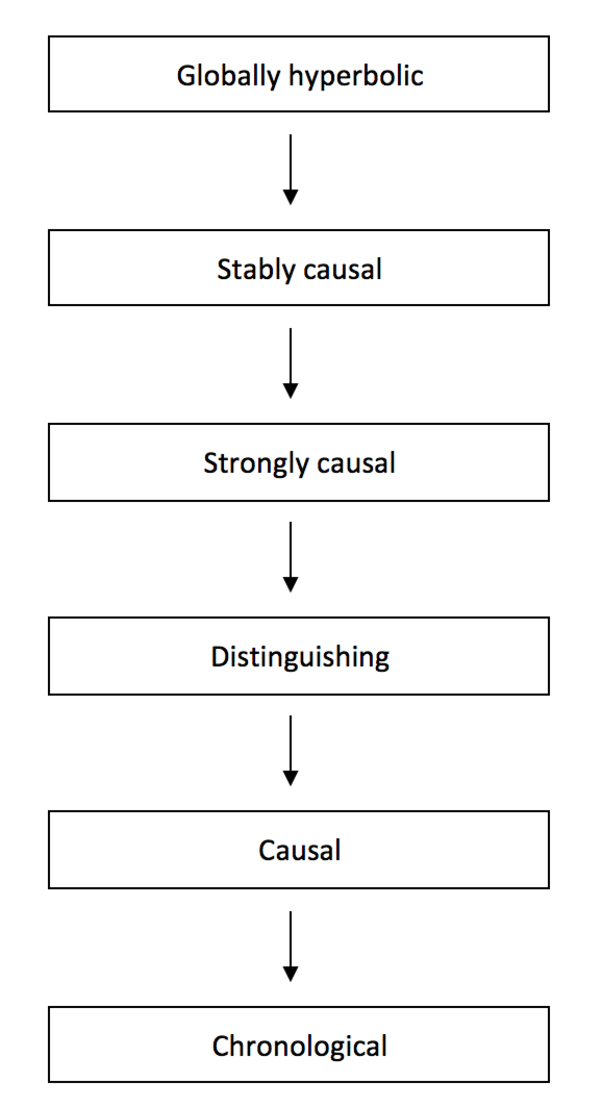}
\caption{A diagram illustrating the strengths of the causality conditions. Global hyperbolicity is the most restrictive causality assumption.}
\end{center}
\end{figure}
\\Furthermore there is a severe restriction on the topology. If we fix a timelike vector field on the space-time and consider two Cauchy surfaces of constant $f$, $S$ and $S'$, we can define a mapping from $S$ to $S'$ which sends each point $p$ of $S$ to that point of $S'$ reached by the integral curve of our vector field passing through $p$ (there must existc such a point, since $S'$ is a Cauchy surface and it must be unique, by achronality of $S'$). This mapping is smooth and its inverse exists (reversing the role of $S$ and $S'$) and so we have produced a diffeomorphism from $S$ to $S'$. Hence all the Cauchy surfaces of constant $f$ are topologically identical, i.e. diffeomorphic. Thus the global structure is very \lq dull\rq\hspace{0.1mm} and \lq tame\rq. 
All the theorems proven above for globally asymptotic space-time can always be applied to any region of the form $\mathrm{int}[D[S]]$, for any closed achronal set $S$.\\
It is also worth remarking that global hyperbolicity plays a key role in proving singularity theorems. In fact, if $p$ and $q$ are points lying in a hyperbolic set $N$ with $q\in J^+(p)$, then it can be shown that there exists a causal geodesic from $p$ to $q$ whose length is greater than or equal to that of any other causal curve form $p$ to $q$. The proof of this result can be found in \cite{Ave63} and in \cite{Seif67}. 
\chapter{Spinor Approach to General Relativity}
\label{chap:2}
\begin{center}
\begin{large}
\textbf{Abstract}
\end{large}
\end{center}
In this chapter we will deal with the spinor formalism, which was firstly developed by \cite{Pen60}. Even if the reader may find these sections more mathematical than physical, a very large use of this method will be done in the remainder of the work. In particular, the asymptotic properties of the space-time will be discussed by making use of the spinor approach, which makes them easier to develop. Furthermore it will be shown that this method is more than just a mere mathematical instrument equivalent to the tensors. In fact the spinor structure of a space-time emerges as deeper and more basic even than its pesudo-Riemannian structure. Moreover the range of applications of the spinor formalism is quite large and there is no possibility of even trying to give a reasonable discussion of them. A brief list of some more familiar or important examples of topics where it has been applied is:
\begin{itemize}
\item Exact solutions;
\item Gravitational radiation;
\item Numerical computations;
\item Black Hole Physics.
\end{itemize}
\section{Introduction}
\label{sect:2.1}
A way to deal with the theory of space-time, different from the usual one based on tensor calculus, is given by the spinor formalism. In fact the spinor structure of space-time gives, in a certain way, a deeper description of it as we will see that pseudo-Riemannian structure naturally emerges as its consequence.\\
The formalism is essentially based on the $(2\rightarrow 1)$ homomorphism between the group $\mathrm{SL}(2,\mathbb{C})$ of unimodular ($2\times 2$) complex matrices and the connected component of Lorentz group $\mathscr{L}$ (that is usually denoted by $L^{\uparrow}_{+}$). The easier way to express this isomorphism is to associate to a 4-vector $u^a=(u^0,u^1,u^2,u^3)$ a hermitian matrix $A$ such that 
\begin{equation}
\label{eqn:11}
A=\left(\begin{matrix}u^{00} & u^{01} \\ u^{10} & u^{11}\end{matrix}\right)=\left(\begin{matrix}u^0+u^{3}& u^1-iu^2\\u^1+iu^2 & u^0-u^3\end{matrix}\right)=u^{a}\sigma_{a},
\end{equation}
where we introduced $\sigma_{a}$ as
\begin{equation*}
\sigma_{a}=(\mathbb{I},\vec{\sigma}),
\end{equation*}
the matrices \begin{equation}
\label{eqn:pauli}
\sigma_1=\left(\begin{matrix}0 & 1 \\ 1 & 0\end{matrix}\right),\hspace{0.5cm}\sigma_2=\left(\begin{matrix}0& -i\\i & 0\end{matrix}\right),\hspace{0,5cm} \sigma_3=\left(\begin{matrix}1 & 0\\ 0& -1\end{matrix}\right),
\end{equation}
being the Pauli matrices. The components of the vector $u^a$ can be obtained as 
\begin{equation}
\label{eqn:8}
u_{a}=\frac{1}{2}\mathrm{tr}(\sigma_aA)
\end{equation}
since 
\begin{equation*}
\frac{1}{2}\mathrm{tr}(\sigma_aA)=\frac{1}{2}\mathrm{tr}(\sigma_a\sigma_bu^b)=\frac{1}{2}\mathrm{tr}\left(([\sigma_a,\sigma_b]/2+\{\sigma_a,\sigma_b \}/2\right)u^b)=
\end{equation*}
\begin{equation*}
=\frac{1}{2}\mathrm{tr}\left(i\epsilon_{abc}\sigma_cu^b+\eta_{ab}\mathbb{I}u^b\right)=u_a,
\end{equation*}
where the last equality is due to the traceless property of Pauli matrices. 
If now we consider, given any matrix $Q\in \mathrm{SL}(2,\mathbb{C})$, i.e. a $(2\times 2)$ complex matrix with $\mathrm{det}[Q]=1$ the product 
\begin{equation}
\label{eqn:9}
A'=QAQ^{\dagger},
\end{equation}
preserves both the determinant, $\mathrm{det}[A]=\mathrm{det}[A']$, i.e. the form
\begin{equation*}
\eta_{ab}u^au^b=\left(u^0\right)^2-\left(u^1\right)^2-\left(u^2\right)^2-\left(u^3\right)^2,
\end{equation*}
\begin{equation*}
\eta_{ab}=\textrm{diag}(1,-1,-1,-1),
\end{equation*} 
that expresses the pseudo-norm of $u^a$, and the hermicity. 
\\The argument can be carried out more generally. 
\begin{defn} $\\ $
\label{defn:action}
Let $G$ be a Lie group and $\mathscr{M}$ be a manifold. Define the \textit{action} of $G$ on $\mathscr{M}$ as a differentiable map $\sigma:G\times M\rightarrow M$ which satisfies the conditions
\begin{enumerate}
\item $\sigma(e,p)=p$\hspace{1.5cm} for any $p\in\mathscr{M}$
\item $\sigma(g_1,\sigma(g_2,p))=\sigma(g_1g_2,p)$
\end{enumerate}
 where $e$ is the identity of $G$, $g_1$ and $g_2$ are elements of $G$ and $g_1g_2$ is the product operation in $G$. 
 \end{defn}
The operation \eqref{eqn:9} can be now regarded to be an action $\sigma(A,u)$  of $\mathrm{SL}(2,\mathbb{C})$ on the space-time point of coordinates $u^a$ which possesses two important properties. We thus obtain a linear transformation of $u^a$ which preserves both reality and pseudo-norm, i.e. a Lorentz transformation
\begin{equation*}
u^a\rightarrow \Lambda^{a}{}_{b}u^b,
\end{equation*}
with $\Lambda^{a}{}_{b}\in \mathscr{L}$. Indeed if we obtain $u'^a$ from \eqref{eqn:9} using \eqref{eqn:8}, those components will be related to the old ones by a Lorentz transformation determined by $Q$. In this way we have built a homomorphism $\psi$: 
\begin{equation}
\label{eqn:10}
\psi:\pm Q\in \mathrm{SL}(2,\mathbb{C})\longrightarrow \psi(\pm Q)= \Lambda\in\mathscr{L}.
\end{equation}
It is easy to show \citep{Oblak16} that $\mathscr{L}\cong\mathrm{SL}(2,\mathbb{C})/\mathbb{Z}_2$. In other words $\mathrm{SL}(2,\mathbb{C})$ is the double cover of the connected component of the Lorentz group in four dimensions, and it is also its universal cover.\\
Usually we consider the components of a vector in a pseudo-orthonormal frame as an ordered array. However we have just shown that a completely equivalent way to order them and to perform transformations is to regard them as elements of a matrix \eqref{eqn:11}. Looking at a vector represented with a matrix allows us to interpret it as not the most elementary \lq vectorial\rq\hspace{0.1mm} object in space-time, but as a sort of \lq divalent\rq\hspace{0.1mm} quantity, composition of two monovalent quantities, called spin vectors.
\section{Spinor Algebra}
\label{sect:2.2}
\begin{defn} $\\ $
A \textit{spin space} $S$ is a complex 2-dimensional vector space equipped with a symplectic form, $\epsilon$, i.e. a bilinear skew-symmetric form. The elements of $S$ are called \textit{spin vectors} or \textit{spinors}.
\end{defn}
The presence of the symplectic form $\epsilon$ allows us to introduce in $S$ a skew-symmetric scalar product.
\begin{defn}$\\ $
The bilinear map defined as
\begin{equation*}
 (\xi,\eta)\in S\times S\longrightarrow[\xi,\eta]=-[\eta,\xi]\in \mathbb{C}.
\end{equation*}  
is called \textit{skew-symmetric scalar product}.
\end{defn}
With such a scalar product every spin vector is self-orthogonal. If $\eta$ is orthogonal to $\xi$  and not proportional to it the two spin vectors constitute a basis for $S$. Hence we give the following:
\begin{defn}$\\ $
Two spin vectors $(o,\iota)$ are said to form a \textit{normalized spin basis} if they satisfy $$[o,\iota]=1.$$
\end{defn}
\begin{oss}$\\ $
It is also possible to work with a non-normalized spin basis, as done in \cite{Penrin1}. In the remainder the indices $A,B,...$ will take values $0$ and $1$ and are referred to the components of spin vectors in a certain basis.
\end{oss}
Thus any spin vector $\xi\in S$ admits the representation 
\begin{equation*}
\xi=\xi^0o+\xi^1\iota,
\end{equation*}
and its components in the basis $(o,\iota)$ are denoted by $\xi^A$. 
\\Obviously we have\begin{equation*}
o^A=(1,0),\hspace{2cm}\iota^A=(0,1).
\end{equation*}
Since $S$ is a vector space, it admits a dual, denoted by  $S^*$. Through the skew-symmetric scalar product it is possible to define a natural isomorphism between $S$ and $S^*$:
\begin{equation}
\label{eqn:12}
\xi\in S\longrightarrow[\xi,\hspace{1.6mm}]\in S^*,
\end{equation}
that is a linear map
\begin{equation*}
\eta\in S\longrightarrow [\xi,\eta]\in\mathbb{C}.
\end{equation*}
The symplectic form can be identified with an element of $S^*\otimes S^*$,  $\epsilon_{AB}=-\epsilon_{BA}$, such that 
\begin{equation*}
[\xi,\eta]=\epsilon_{AB}\xi^A\eta^B.
\end{equation*}
The condition for $(o,\iota)$ to be a spin basis becomes
\begin{equation*}
\epsilon_{AB}o^Ao^B=\epsilon_{AB}\iota^A\iota^B=0,\hspace{1cm}\epsilon_{AB}o^A\iota^B=1.
\end{equation*}
In the frame $(o,\iota)$ we have
\begin{equation*}
\epsilon_{AB}=\left(\begin{matrix}0 & 1 \\ -1& 0\end{matrix}\right).
\end{equation*}
Since $\epsilon_{AB}$ is non-singular, there exists the inverse $(\epsilon^{-1})^{AB}$ which can be identified with an element of $S\otimes S$. By convention we denote
\begin{equation*}
\epsilon^{AB}=-(\epsilon^{-1})^{AB}=\left(\begin{matrix}0 & 1 \\ -1& 0\end{matrix}\right).
\end{equation*}
It is easy to show that
\begin{equation}
\label{eqn:13}
\epsilon^{AB}=o^A\iota^B-\iota^Ao^B.
\end{equation}
In this way the natural isomorphism \eqref{eqn:12} can be read in the following way: given a spin vector $\xi$ of components $\xi^B$ its dual can be identified with $\xi_B=\epsilon_{AB}\xi^A=\xi^A\epsilon_{AB}$. From this last property it follows that $\xi^B=\epsilon^{BC}\xi_{C}$.
Simple relations to show are the following:
\begin{align*}
&\epsilon_C ^{\enskip A}=\delta_C^{\enskip A}=-\epsilon^A_{\enskip C},\\
&\xi_A^{\enskip A}=-\xi^A_{\enskip A},\\
&\epsilon^{AB}\epsilon_{AB}=\epsilon_A^{\enskip A}=-\epsilon^A_{\enskip A}=2.
\end{align*}
We can always apply the usual symmetrization $(\quad)$ and antisymmetrization $[\quad]$ operations to a multivalent spinor $\tau_{...AB...}$. Since $S$ is a 2-dimensional space, for any multivalent spinor $\tau_{...AB...}$, we have $\tau_{...[ABC]...}=0$ because at least two of the bracketed indices must be equal.
As consequence we have the Jacobi identity
\begin{equation*}
\epsilon_{A[B}\epsilon_{CD]}=0=\epsilon_{AB}\epsilon_{CD}+\epsilon_{AC}\epsilon_{DB}+\epsilon_{AD}\epsilon_{BC}.
\end{equation*}
\begin{thm} $\\ $
\label{thm:asy}
Let $\tau_{...AB...}$ be a multivalent spinor. Then
 \begin{equation}
 \label{eqn:19}
 \tau_{...[AB]...}=\frac{1}{2}\epsilon_{AB}\tau_{...C}^{\quad C}{}_{...}
\end{equation} 
\end{thm}
The proof of this theorem is quite easy and can be found in \cite{Penrin1} or in \cite{Stew}.\\
For any multivalent spinor $\tau_{...AB...}$ the following decomposition holds
 \begin{equation}
 \label{eqn:14}
\tau_{...AB...}=\tau_{...(AB)...}+\tau_{...[AB]...}
\end{equation}
By theorem \ref{thm:asy} we get
\begin{equation}
\label{eqn:15}
\tau_{...AB...}=\tau_{...(AB)...}+\frac{1}{2}\epsilon_{AB}\tau_{...C}^{\quad C}{}_{...}
\end{equation}
\section{Spinors and Vectors\hspace{1cm}}
\label{sect:2.3}
\begin{defn}
We define the \textit{conjugation operation} as the following map of spin vectors from $S$ to a new spin space $S'$
\begin{equation*}
\alpha^A+c\beta^A\in S\longrightarrow\overline{\alpha^A+c\beta^A}=\bar{\alpha}^{A'}+\bar{c}\bar{\beta}^{A'}\in S'.
\end{equation*}
\end{defn}
\begin{oss}$\\ $
Some authors denote this operation as anti-isomorphism because its action on a complex number $c$ is to map it into its complex conjugate $\bar{c}$. For example, if we consider 
\begin{equation*}
\xi^A=\left(\begin{matrix}
a \\ b 
\end{matrix}\right),
\end{equation*}
then
\begin{equation*}
\overline{\xi^A}=\bar{\xi}^{A'}=\left(\begin{matrix}
\bar{a}\\ \bar{b}
\end{matrix}\right).
\end{equation*}
\end{oss}
Having introduced $S$ and $S'$ we are now ready to build tensorial quantities. Define a \textit{hermitian spinor} $\tau$ as one for which $\bar{\tau}=\tau$. Of course for this to make sense $\tau$ must have as many primed indices as unprimed ones, and their relative positions must be the same. For example, take an element of the tensor product $S\otimes S'$, $\tau^{AA'}$. Let $(o,\iota)$ and $(\bar{o},\bar{\iota})$ be the spin bases respectively for $S$ and $S'$. Then there exist scalars $\xi$, $\eta$, $\zeta$ and $\sigma$ such that
\begin{equation*}
\tau^{AA'}=\xi o^{A}\bar{o}^{A'}+\eta\iota^{A}\bar{\iota}^{A'}+\zeta o^{A}\bar{\iota}^{A'}+\sigma\iota^{A}\bar{o}^{A'}.
\end{equation*}
The hermiticity condition is equivalent to the statement that $\xi$ and $\eta$ are real and that $\zeta$ and $\sigma$ are complex conjugates. Thus the set of hermitian spin vectors $\tau^{AA'}$ forms a real vector space of dimension 4. This is the reason for which there exists an isomorphism between $T_pM$, the tangent space at a point in a 4-manifold and $S\otimes S'$:
\begin{equation*}
T_pM\cong S\otimes S'.
\end{equation*}
Similarly the set of hermitian spinors $\tau_{AA'}$ forms the dual of the vector space described above, isomorphic to $T^*_pM$. We can express those isomorphisms using the Infeld-van der Waerden symbols, $\sigma^a_{\enskip AA'}$ and $\sigma_a^{\enskip AA'}$. Note that the index $a$ is vectorial, and runs from $0$ to $3$. The correspondence $(A,A')\rightarrow a$ between spin vectors and vectors thus reads as
\begin{subequations}
\begin{align}
\label{eqn:16}
&v^{AA'}\longrightarrow v^a\equiv \sigma^a_{\enskip AA'}v^{AA'},\\
\label{eqn:17}
&v^a\longrightarrow v^{AA'}\equiv v^a\sigma_a^{\enskip AA'}.
\end{align}
\end{subequations}
We note now that, since $S$ and $S'$ are different vector spaces, we do not need to distinguish between $S\otimes S'$ and $S'\otimes S$. It follows that primed and unprimed indices can be interchanged, i.e. we have
\begin{equation*}
\tau_{AA'}=\tau_{A'A}.
\end{equation*}
Often in the remainder the Infeld-van der Waerden symbols will not appear again, their use being implicit.\\
Every spin basis defines a tetrad of vectors $(l,n,m,\bar{m})$ as
\begin{align}
&l^a=o^A\bar{o}^{A'},\hspace{0.5cm}n^a=\iota^A\bar{\iota}^{A'},\hspace{0.5cm}m^a=o^A\bar{\iota}^{A'},\hspace{0.5cm}\bar{m}^a=\iota^A\bar{o}^{A'}, \nonumber \\
\label{eqn:18}
&l_a=o_A\bar{o}_{A'},\hspace{0.5cm}n_a=\iota_A\bar{\iota}_{A'},\hspace{0.5cm}m_a=o_A\bar{\iota}_{A'},\hspace{0.5cm}\bar{m}_a=\iota_A\bar{o}_{A'}.
\end{align}
It is easy to show that $l^al_a=n^an_a=m^am_a=\bar{m}^a\bar{m}_a=0$ while $l_an^a=-m^a\bar{m}_a=1$ and the other mixed products vanish.
\begin{defn} $\\ $
The tetrad of vectors \eqref{eqn:18} is called  a \textit{Newman-Penrose (N-P) null tetrad}.
\end{defn}
We introduce the hermitian spinors 
\begin{equation}
\label{eqn:54}
g_{ABA'B'}=\epsilon_{AB}\epsilon_{A'B'},\hspace{1cm}g^{ABA'B'}=\epsilon^{AB}\epsilon^{A'B'},
\end{equation}
with tensor equivalent $g_{ab}$ and $g^{ab}$, which are obviously symmetric. It is easy to show, using \eqref{eqn:18}, that 
\begin{equation}
\label{eqn:69}
g_{ab}=2l_{(a}n_{b)}-2m_{(a}\bar{m}_{b)},\hspace{1cm}g^{ab}=2l^{(a}n^{b)}-2m^{(a}\bar{m}^{b)},
\end{equation}
and that $g_{ab}$ possesses all the properties of a metric tensor,
\begin{equation*}
l^a=g^{ab}l_b,\hspace{1cm}l_a=g_{ab}l^b,\hspace{1cm}etc.
\end{equation*}
i.e. can be used to raise or lower vectorial indices, and it satisfies
\begin{equation*}
g_{ab}g^{bc}=\delta^c_a,\hspace{1cm}g_{ab}g^{ab}=4.
\end{equation*}
\\
Furthermore, defining a tetrad of vectors
\begin{align}
&e_{\hat{0}}=\frac{(l+n)}{\sqrt{2}}=\frac{o^A\bar{o}^{A'}+\iota^A\bar{\iota}^{A'}}{\sqrt{2}},\hspace{1cm}e_{\hat{1}}=\frac{(m+\bar{m})}{\sqrt{2}}=\frac{o^A\iota^{A'}+\iota^{A}o^{A'}}{\sqrt{2}}, \nonumber \\
\label{eqn:35}
&e_{\hat{2}}=\frac{i(m-\bar{m})}{\sqrt{2}}=\frac{i(o^A\bar{\iota}^{A'}-\iota^{A}\bar{o}^{A'})}{\sqrt{2}},\hspace{1cm}e_{\hat{3}}=\frac{(l-n)}{\sqrt{2}}=\frac{o^A\bar{o}^{A'}-\iota^{A}\bar{\iota}^{A'}}{\sqrt{2}},
\end{align}
we have in such a base
\begin{equation}
\label{eqn:71}
g_{\hat{a}\hat{b}}=\eta_{\hat{a}\hat{b}}=\textrm{diag}(1,-1,-1,-1).
\end{equation}
The existence of a spinor structure fixes the signature of space-time to be Minkowskian.
In this case a suitable choice for the Infeld-van der Waerden symbols could be
\begin{equation*}
\sigma_{\hat{a}}^{\enskip AA'}=\frac{1}{\sqrt{2}}\sigma_{\hat{a}},\hspace{1cm}a=0,1,2,3
\end{equation*}
where $\sigma_{\hat{a}}$ are the usual Pauli matrices \eqref{eqn:pauli}.
\begin{figure}[h]
\begin{center}
\includegraphics[scale=0.45]{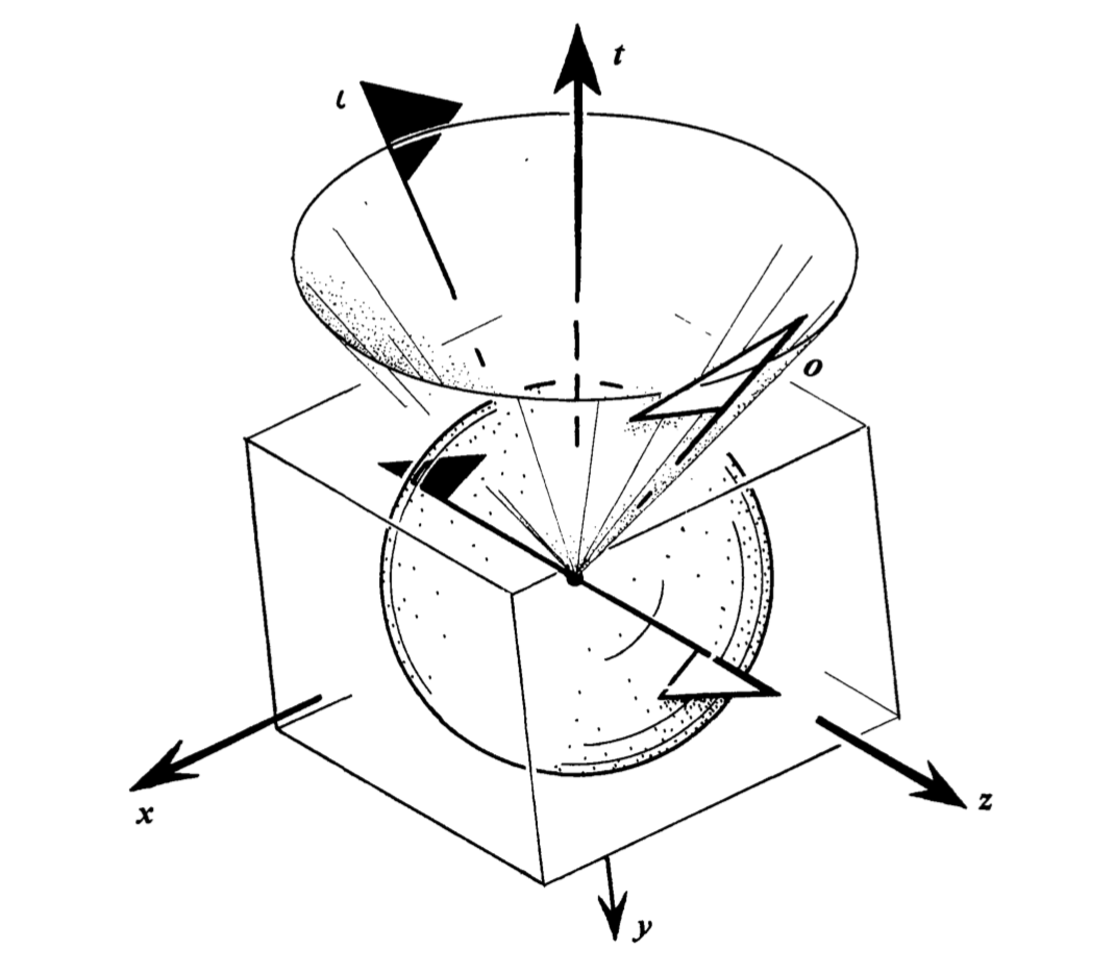}
\caption{The standard relation between a spin-frame $(o^A,\iota^A)$ and the tetrad of equation \eqref{eqn:35}, with $t=e_{\hat{0}}$, $x=e_{\hat{1}}$, $y=e_{\hat{2}}$ and $z=e_{\hat{3}}$.}
\label{fig:2.6}
\end{center}
\end{figure}
\begin{defn} $\\ $
The tetrad of vectors introduced in \eqref{eqn:35} is called a \textit{Minkowski tetrad}.
\end{defn}
We may ask now how to describe a curved space-time, that is of main interest in General Relativity. To do that the most efficient way is to use the tetrad formalism. \\
We introduce a tetrad vector $e_a^{\enskip\hat{c}}$ and a dual basis of covector $e^a_{\enskip\hat{c}}$, i.e. $e_a^{\enskip\hat{c}}e^b{}_{\hat{c}}=\delta_a{}^b$. Note that the hatted indices label the vectors, while unhatted ones label the components with respect to some arbitrarily chosen basis. For any generic metric tensor $g$, we have
\begin{equation*}
g=g_{ab}dx^a\otimes dx^b=e_a^{\enskip\hat{c}}e_b^{\enskip\hat{d}}\eta_{\hat{c}\hat{d}}dx^a\otimes dx^b=e^{\hat{c}}\otimes e^{\hat{d}}\eta_{\hat{c}\hat{d}},
\end{equation*}
where $e^{\hat{a}}=e^{\hat{a}}{}_{b}dx^b$ is the tetrad 1-form.
We denote the inverse of a tetrad as $e^a_{\enskip\hat{c}}$ so that
\begin{equation*}
e^a_{\enskip\hat{c}}e_a^{\enskip\hat{d}}=\delta_{\hat{c}}^{\enskip \hat{d}}.
\end{equation*}
We may now try to include tetrads in the definitions of the isomorphism \eqref{eqn:16}, \eqref{eqn:17} in the following way
\begin{align*}
&v^{AA'}\longrightarrow v^a=e^a_{\enskip \hat{c}}v^{\hat{c}}= e^a_{\enskip \hat{c}}\sigma^{\hat{c}}_{\enskip AA'}v^{AA'}\equiv  e^{a}_{\enskip AA'}v^{AA'},\\
&v^a\longrightarrow v^{AA'}= v^{\hat{c}}\sigma_{\hat{c}}^{\enskip AA'}= v^{a}e_{a}^{\enskip \hat{c}}\sigma_{\hat{c}}^{\enskip AA'}\equiv v^a e_{a}^{\enskip AA'},
\end{align*}
where $e_{a}^{\enskip\hat{c}}\sigma_{\hat{c}}^{\enskip AA'}= e_{a}^{\enskip AA'}$ is called the \textit{soldering form}, and is, by construction, a spinor-valued one form which encodes the relevant informations about the metric. \\
\section{Null Flags and Spinor Structure on $\mathscr{M}$}
\label{sect:2.4}
We now proceed to the space-time interpretation of spin vectors. As we have already seen in \eqref{eqn:18} every univalent spinor $\kappa^A$ defines a real null vector $k^a=\kappa^A\bar{\kappa}^{A'}$. This is a special case of a more general theorem.
\begin{thm} $\\ $
 Every non-vanishing real null vector $k^a$ can be written in one or other of the forms
\begin{equation}
\label{eqn:34}
k^a=\pm \kappa^A\bar{\kappa}^{A'}.
\end{equation}
\end{thm}
\begin{proof}[Proof:] To prove that $k^a$ defined in \eqref{eqn:34} is null is sufficient to note that 
\begin{equation*}
\kappa^A\kappa_A=(\kappa^0o^A+\kappa^1\iota^A)(\kappa^0o_A+\kappa^1\iota_A)=(\kappa^0)^2o^Ao_A+(\kappa^1)^2\iota^A\iota_A+(\kappa^0\kappa^1)(o^A\iota_A+\iota^Ao_A)=0.
\end{equation*}
Suppose conversely that $k^a=\omega^{AA'}$ is real and null. The nullity condition is
\begin{equation*}
\epsilon_{AB}\epsilon_{A'B'}\omega^{AA'}\omega^{BB'}=0
\end{equation*}
which says that the $2\times2$ matrix $\omega^{AA'}$ has vanishing determinant, so that the rows or the columns are linearly dependent. This means that there exist univalent spin vectors $\kappa$, $\lambda$, such that
\begin{equation*}
\omega^{AA'}=\kappa^A\bar{\lambda}^{A'}.
\end{equation*}
The reality condition is 
\begin{equation*}
\kappa^{A}\bar{\lambda}^{A'}=\lambda^A\bar{\kappa}^{A'}.
\end{equation*}
Multiplying by $\kappa_A$ implies that $\kappa_A\lambda^A=0$ so that $\lambda$ must be proportional to $\kappa$. Rescaling $\kappa$ we get \eqref{eqn:34}.
\end{proof}
As just shown every univalent spin vector $\kappa^A$ defines a null vector $k^a$, but given a real $\theta$, $e^{i\theta}\kappa^A$ defines the same null vector, so that in $\kappa$ there is some additional phase information. Complete now $\kappa$ to a spin basis  $(\kappa,\mu)$ and take into account the vectors
\begin{equation*}
s^a=\frac{1}{\sqrt{2}}(\kappa^A\bar{\mu}^{A'}+\mu^{A}\bar{\kappa}^{A'}),
\end{equation*}
\begin{equation*}
t^a=\frac{i}{\sqrt{2}}(\kappa^A\bar{\mu}^{A'}-\mu^{A}\bar{\kappa}^{A'}),
\end{equation*}
that are both spacelike and  orthogonal to $k^a$. Together those vectors span a spacelike 2-surface orthogonal to $k$. If now we perform the phase change $\kappa\rightarrow e^{i\theta}\kappa$ we get $\bar{\kappa}\rightarrow e^{-i\theta}\kappa$. Since $(\kappa,\mu)$ form a spin basis, i.e. $\kappa_A\mu^A=1$, it follows that $\mu\rightarrow e^{-i\theta}\mu$. We thus have
\begin{equation*}
s^a\rightarrow s'^a=s^a\cos 2\theta+t^a\sin 2\theta.
\end{equation*}
The interpretation that Penrose suggests is to imagine a univalent spin vector $\kappa^A$ as a flag whose flagpole is defined to be parallel to the direction of $k^a$, while the flag itself lies in the two plane spanned by $k^a$ and $s^a$. Thus the flagpole lies in the plane of the flag (see Figure \ref{fig:2.61}).
\begin{figure}[h]
\begin{center}
\includegraphics[scale=0.45]{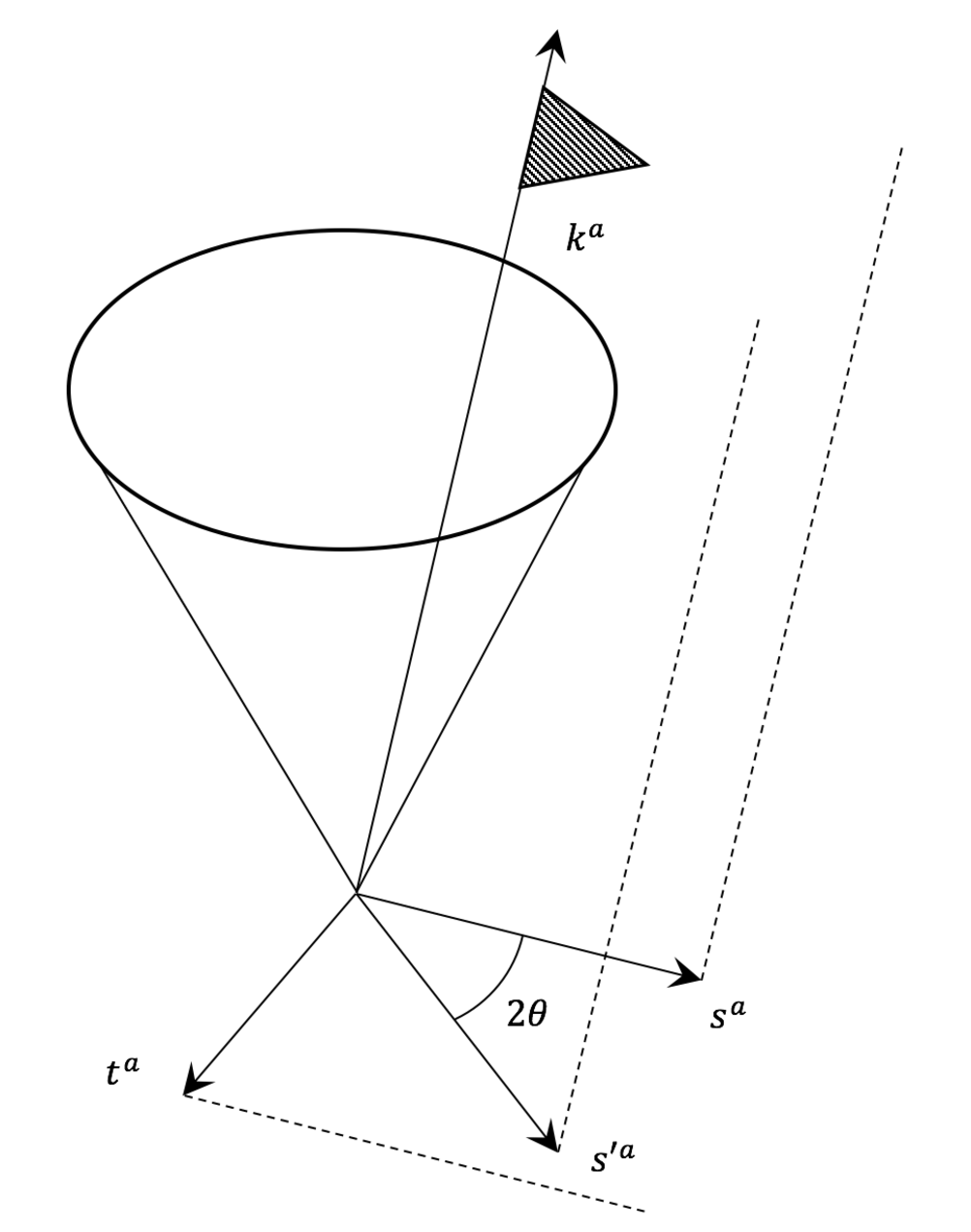}
\caption{The null flag representing $\kappa^A$, and its relation to $k^a$ and $s^a$}
\label{fig:2.61}
\end{center}
\end{figure}
Changing a phase by $\theta$ leaves the flagpole invariant but rotates the flag plane around the pole by an angle $2\theta$. If $\theta=\pi$ the flag plane is left invariant, but $\kappa^A$ takes a minus sign. This is deeply linked to the fact that the homomorphism between $\mathrm{SL}(2,\mathbb{C})$ and $\mathscr{L}$ is $(2,1)$.\\
 We investigate now what are the restrictions to put on $\mathscr{M}$ for it to allow objects like spin vectors to be defined globally. By theorem \eqref{eqn:34} we see that at a certain point $p$ of $\mathscr{M}$ we have an absolute distinction between the two null half-cones: we define the future-pointing and past-pointing null vectors to be those for which, respectively, the decomposition \eqref{eqn:34} holds with a plus or with a minus sign. Thus the request of existence of spin vectors allows us to divide continuously over $\mathscr{M}$ the null half-cones of $\mathscr{M}$ into two classes, \lq future\rq\hspace{0.1mm} and \lq past\rq\hspace{1mm}. Thus according to definition \ref{defn:to} we see that $\mathscr{M}$ must be time-orientable. Furthermore the effect of a phase change on a flag can be used to define a space orientation: the flag plane rotation corresponding to $\theta>0$ is defined to be right-handed. Hence we have a second restriction on $\mathscr{M}$, that has to be space-orientable. It follows that the tetrad of vectors in \eqref{eqn:35} has the standard orientation, i.e. $e_{\hat{0}}$ is future-pointing, and $(e_{\hat{1}},e_{\hat{2}},e_{\hat{3}})$ form a right-handed triad of vectors.\\
 So far we have shown that a necessary condition for the existence of spinor fields in some region of the space-time is that the manifold $\mathscr{M}$ must be both time- and space-orientable (i.e. \textit{orientable}). But those requirements are not sufficient. In fact $\mathscr{M}$ must also permit a \textit{spin structure} to be defined on it, which means, roughly speaking, a prescription for keeping track of the sign of a spin  vector not only if we move it around at a fixed point of $\mathscr{M}$, but also if we move it around from point to point within $\mathscr{M}$ (for a more accurate description of the problem see \cite{Pen67} or \cite{Penrin1}). It should be emphasized that the question of existence of spin structure on a manifold $\mathscr{M}$ is not the same question as that of the existence of certain spinor fields on $\mathscr{M}$. In fact without the spin structure, the concept of global spinor field does not exist. If $\mathscr{M}$ is orientable and admits a spin structure then we say that $\mathscr{M}$ has a \textit{spinor structure}.\\ The problem was deeply studied by \cite{Ger68} and \cite{Ger70}. It turns out that, assuming orientability, the condition on $\mathscr{M}$ for existence and uniqueness of spinor structure depend \textit{only} on the topology of $\mathscr{M}$ and not on the nature of its (Lorentzian) metric. In fact a topologically trivial $\mathscr{M}$ admits a unique spinor structure, while a topologically non-trivial one may or may not permit a consistent spinor structure, and if it does, the possible spin structure may or may not be unique. The result obtained by Geroch is the following
\begin{thm}{\citep{Ger68}} $\\ $
\label{thm:Geroch2}
If $\mathscr{M}$ is a non-compact 4-dimensional manifold, then a necessary and sufficient condition that it should have spinor structure is the existence of four continuous vector fields on $\mathscr{M}$ which constitute a Minkowski tetrad in the tangent space at each point of $\mathscr{M}$, i.e. if there exists on $\mathscr{M}$ a global system of orthonormal tetrads.
\end{thm}
Hence a manifold $\mathscr{M}$ can be equipped with a spinor structure if and  only if it is \textit{parallelizable}, i.e. it admits a set of $n$ (dimension of the manifold) vector fields defined through all $\mathscr{M}$ which at every $p\in\mathscr{M}$ constitute a basis for $T_p$.  A classical example of a $2$-dimensional manifold  which is not parallelizable is the sphere $S^2$. \\
We remind here that, from our definition \ref{defn:st}, a space-time is paracompact and hence non-compact. Thus theorem \ref{thm:Geroch2} can be applied to space-times. Furthermore we remark that the property of non-compactness, (i.e. losely speaking, \lq open\rq) is a very reasonable one for $\mathscr{M}$, preventing it to contain closed time-like curves, as explained in section \ref{sect:1.6}.\\
While theorem \ref{thm:Geroch2} represents a strong condition to be satisfied, it is not always the most convenient way to decide whether or not a given space-time has spinor structure. In \cite{Ger70}, it is developed some
criteria for the existence of spinor structure, based on the neighbourhoods of certain 2-spheres in $\mathscr{M}$. With each such 2-sphere, $S$, it can be associated an index, defined as the number of times that $S$ intersects a surface obtained by slightly deforming $S'$. That this index be even for each $S$ in $\mathscr{M}$ is a necessary and sufficient condition for $\mathscr{M}$ to have a spinor structure.
 \section{The Petrov Classification}
 \label{sect:2.5}
In tensor algebra we usually define the totally skew-symmetric tensor $\epsilon_{abcd}$ as \begin{equation*}
\epsilon_{abcd}=\epsilon_{[abcd]},\hspace{0.5cm}\epsilon_{abcd}{}^{abcd}=-24,\hspace{0.5cm}\epsilon_{\hat{0}\hat{1}\hat{2}\hat{3}}=1.
\end{equation*}
It can be shown that spinorially those properties are satisfied by  
\begin{equation*}
\epsilon_{abcd}=i(\epsilon_{AB}\epsilon_{CD}\epsilon_{A'C'}\epsilon_{B'D'}-\epsilon_{AC}\epsilon_{BD}\epsilon_{A'B'}\epsilon_{C'D'}).
\end{equation*}
\begin{thm} $\\ $
\label{thm:1}
Suppose the spinor $\tau_{AB...C}$ to be totally symmetric. Then there exist univalent spinors $\alpha_A$, $\beta_B$,...,$\gamma_C$, such that 
\begin{equation}
\label{eqn:20}
\tau_{AB...C}=\alpha_{(A}\beta_B{}_{...}\gamma_{C)}.
\end{equation}
$\alpha$, $\beta$,...,$\gamma$ are called the principal spinors of $\tau$. The corresponding real null vectors, obtained using \eqref{eqn:35}, are called principal null directions (PND) of $\tau$. 
\end{thm}
\begin{proof}[Proof:] Let $\tau$ have valence $n$ and let $\xi^A=(x,y)$. Define then 
\begin{equation*}
\tau(\xi)=\tau_{AB...C}\xi^A\xi^B...\xi^{C}.
\end{equation*}
This is a homogeneous polynomial of degree $n$ in the $(x,y)$ and so we may factorize it as
\begin{equation*}
\tau(\xi)=(\alpha_0x-\alpha_1y)(\beta_0x-\beta_1y)...(\gamma_0x-\gamma_1y),
\end{equation*}
which proves the result since $x$ and $y$ are arbitrary. 
\end{proof}
We observe that from \eqref{eqn:20} that if $\xi^A\neq 0$, then
\begin{equation*}
\tau_{AB...C}\xi^A\xi^B...\xi^{C}=0
\end{equation*}
if and only if $\xi^A$ is a principal spinor. We can say more in the case of a multiple PND. Suppose $\alpha_A$ is a $k$-fold principal spinor,
\begin{equation}
\label{eqn:92}
\tau_{AB...CD...L}=\alpha_{(A}\alpha_{B}...\alpha_{C}\eta_{D}...\lambda_{L)},
\end{equation}
so that $\alpha_A$ occurs $k$ times on the right, none of the spinors $\eta_A$,...,$\lambda_{A}$ being proportional to $\alpha_A$.  Then we have, multiplying \eqref{eqn:92} with the product $\alpha^D...\alpha^L$ of $n-k$ $\alpha$'s,
\begin{equation*}
\tau_{AB...CD...L}\alpha^D...\alpha^L=\kappa\alpha_A\alpha_B...\alpha_C,
\end{equation*}
where
\begin{equation*}
\kappa=\frac{k!(n-k)!}{n!}(\eta_D\alpha^D)...(\lambda_L\alpha^L)\neq 0.
\end{equation*}
If, on the other hand, we multiply \eqref{eqn:92} with $n-k+1$ $\alpha$s it is clear that the expression vanishes. Thus:
\begin{prop}$\\ $
\label{prop1}
A necessary and sufficient condition that $\xi_A\neq 0$ be a $k$-fold principal spinor of the non-vanishing symmetric spinor $\tau_{AB...L}$ is that
\begin{equation*}
\tau_{AB...CD...L}\xi^D...\xi^L
\end{equation*}
should vanish if $n-k+1$ $\xi$'s are transvected with $\tau_{AB...L}$ but not if only $n-k$ $\xi$'s are transvected with $\tau_{AB...L}$.
\end{prop}
As corollary we have 
\begin{prop} $\\ $
\label{thm:2}
If $\xi^A\neq 0$, $\tau_{A...G}=\tau_{(A...G)}$ and
\begin{equation*}
\xi^A...\xi^C\xi^D\tau_{A...CDE...G}=0,
\end{equation*}
then there exists a $\psi_{A...C}$ such that 
\begin{equation*}
\tau_{A...CDE...G}=\psi_{(A...C}\xi_D\xi_E...\xi_{G)}.
\end{equation*}
\end{prop}
We can see two important applications of theorem \ref{thm:1}. \\
As first consider the  totally skew-symmetric Maxwell tensor, $F_{ab}=-F_{ba}$ for the electromagnetic field. Introduce a spinor equivalent $F_{ABA'B'}=-F_{BAB'A'}$ and define
\begin{equation*}
\varphi_{AB}=\frac{1}{2}F_{ABC'}{}^{C'}.
\end{equation*}
Note that $\varphi_{AB}=\varphi_{BA}$ since, using \eqref{eqn:19},
\begin{equation}
\varphi_{[AB]}=\frac{1}{2}\epsilon_{AB}\varphi_C^{\enskip C}=\frac{1}{2}\epsilon_{AB}F_C^{\enskip C}{}_{C'}{}^{C'}=\frac{1}{2}\epsilon_{AB}\eta_{ab}F^{ab}=0.
\end{equation}
Using \eqref{eqn:15} we have
\begin{equation*}
F_{ABA'B'}=F_{AB(A'B')}+\varphi_{AB}\epsilon_{A'B'}.
\end{equation*}
A second application gives
\begin{equation*}
F_{ABA'B'}=F_{(AB)(A'B')}+\varphi_{AB}\epsilon_{A'B'}+\epsilon_{AB}\bar{\varphi}_{A'B'}=\varphi_{AB}\epsilon_{A'B'}+\epsilon_{AB}\bar{\varphi}_{A'B'},
\end{equation*}
the second equality resulting from the skew-symmetric nature of $F$. Using now the result \eqref{eqn:20} we get 
\begin{equation*}
\varphi_{AB}=\alpha_{(A}\beta_{B)}.
\end{equation*}
It is simple to show that, if we define $F^*{}_{ab}=\frac{1}{2}\epsilon_{ab}{}^{cd}F_{cd}$ we have
\begin{equation}
\label{eqn:70}
F_{ab}+i\
F^*_{ab}=2\varphi_{AB}\epsilon_{A'B'}=2\alpha_{(A}\beta_{B)}.
\end{equation}
There are now two possibilities. If $\alpha$ and $\beta$ are proportional then $\alpha$ is called a repeated spinor of $\varphi$ and $\varphi$ is said to be null, or of \textbf{\textit{type N}}. The vector $\alpha_a=\alpha_{A}\bar{\alpha}_{A'}$ is called a repeated null direction. If $\alpha$ and $\beta$ are not proportional $\varphi$ is said to be algebraically general or of \emph{\textbf{type I}}.\\
As a second example we consider the Weyl tensor $C_{abcd}$, the conformally invariant part of the Riemann tensor. As we will derive later $C_{abcd}$ has the following properties
\begin{equation*}
C_{abcd}=C_{[ab][cd]}=C_{cdab}.
\end{equation*}
It can be written as
\begin{equation*}
C_{abcd}=C_{ABCDA'B'C'D'}=\Psi_{ABCD}\epsilon_{A'B'}\epsilon_{C'D'}+\bar{\Psi}_{A'B'C'D'}\epsilon_{AB}\epsilon_{CD},
\end{equation*}
where $\Psi_{ABCD}$ is totally symmetric. Defining $C^*_{\enskip abcd}=\frac{1}{2}\epsilon_{cd}{}^{ef}C_{abef}$ we have similarly
\begin{equation*}
C_{abcd}+iC^*_{\enskip abcd}=2\Psi_{ABCD}\epsilon_{A'B'}\epsilon_{C'D'}=2\alpha_{(A}\beta_B\gamma_C\delta_{D)}.
\end{equation*}
The corresponding space-times can be classified as follows:
\begin{itemize}
\item \textit{\textbf{Type I or}} \{1,1,1,1\}. None of the four principal null directions coincide. This is the algebraically general case;\\
\item \textit{\textbf{Type II or}} \{2,1,1\}. Two directions coincide. This and the subsequent cases are algebraically special;\\
\item \textit{\textbf{Type D or}} \{2,2\}. Two different pairs of repeated principal null directions exist;\\
\item \textit{\textbf{Type III or}} \{3,1\}. Three principal null directions coincide;\\
\item \textit{\textbf{Type N or}} \{4\}. All four principal null directions coincide.
\end{itemize}
\begin{figure}[h]
\begin{center}
\includegraphics[scale=0.45]{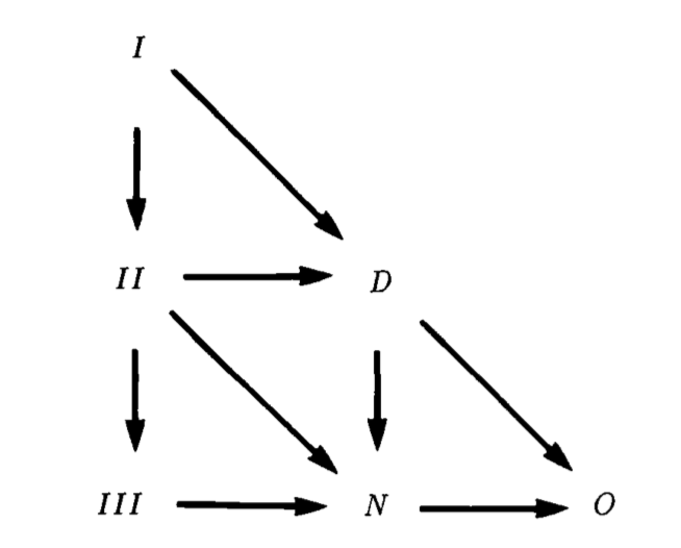}
\caption{A Penrose diagram showing the Petrov classification where the arrows indicate increasing specialization of the solution.}
\label{fig:2.7}
\end{center}
\end{figure}
The previous is called Petrov classification.  We note here that $\Psi$ has two scalar invariants
\begin{equation*}
I=\Psi_{ABCD}\Psi^{ABCD}\hspace{1cm}J=\Psi_{AB}{}^{CD}\Psi_{CD}{}^{EF}\Psi_{EF}{}^{AB}.
\end{equation*}
For type II and type III space-times it can be show that \citep{Penrin2} $I^3=6J^2$ and $I=J=0$, respectively. For type D we have
\begin{equation*}
\Psi_{PQR(A}\Psi_{BC}{}^{PQ}\Psi^R{}_{DEF)}=0,
\end{equation*}
and for type N
\begin{equation*}
\Psi_{(AB}{}^{EF}\Psi_{CD)EF}=0.
\end{equation*}
\section{Spinor Analysis and Curvature}
\label{sect:2.6}
To introduce the space-time curvature we need to define a spinor covariant derivative, whose definition is made axiomatically.
\begin{defn} $\\ $
 Let $\theta$, $\phi$ and $\psi$ be spinor fields of the same valence. Then the \textit{spinor covariant derivative} is a map 
\begin{equation*}
\nabla_{AA'}:\theta_{...}\longrightarrow \theta_{...;AA'}
\end{equation*}
which satisfies
\begin{enumerate}
\item $\nabla_{AA'}(\theta+\phi)=\nabla_{AA'}\theta+\nabla_{AA'}\phi;$
\item $\nabla_{AA'}(\theta\phi)=(\nabla_{AA'}\theta)\phi+\theta(\nabla_{AA'}\psi);$
\item $\psi=\nabla_{AA'}\theta$ implies $\bar{\psi}=\nabla_{AA'}\bar{\theta};$
\item $\nabla_{AA'}\epsilon_{BC}=\nabla_{AA'}\epsilon^{BC}=0;$
\item $\nabla_{AA'}$ commutes with any index substitution not involving $AA';$
\item $(\nabla_{a}\nabla_{b}-\nabla_{b}\nabla_{a})f=0$ for any scalar $f;$
\item for any derivation $D$ acting on spinor fields there exists a spinor $\xi^{AA'}$ such that $D\psi=\xi^{AA'}\nabla_{AA'}\psi$, for all $\psi.$
\end{enumerate}
\end{defn}
\begin{thm} $\\ $
The spinor covariant derivative defined above exists and is unique.
\end{thm}
The proof of this important theorem is given in \cite{Penrin1}.\\
The first two properties are respectively linearity and Leibniz rule, which characterize $\nabla_{AA'}$ to be a derivative. The third is the reality condition. The fourth property can be interpreted to be a sort of compatibility condition, that allows the symplectic form to raise or lower indices within spinor expressions acted upon by $\nabla_{AA'}$. Note that this is stronger than $\nabla g=0$. The sixth is the torsion-free condition. Sometimes this condition can be weakened and the right member of the equation can be replaced by a non-zero $2S_{ab}{}^{c}\nabla_{c}f$ where $S_{ab}{}^{c}$ is the torsion tensor. The last property is needed to ensure tha $\xi^{AA'}$ is a tangent vector in space-time.\\
As next we are going to derive the spinorial form of the Riemann curvature tensor, using its (skew-) symmetries:
\begin{equation}
\label{eqn:21}
R_{abcd}=-R_{bacd}=-R_{abdc},\hspace{1cm}R_{abcd}=R_{cdab},\hspace{1cm}R_{a[bcd]}=0.
\end{equation}
We have, applying \eqref{eqn:15} repeatedly and the first two of \eqref{eqn:21}
\begin{equation*}
R_{abcd}=R_{ABA'B'CC'DD'}=\frac{1}{2}\epsilon_{AB}R_{FA'}{}^{F}{}_{B'CC'DD'}+R_{(AB)A'B'CC'DD'}
\end{equation*}
\begin{equation*}
=\frac{1}{2}\epsilon_{AB}R_{FA'}{}^{F}{}_{B'CC'DD'}+\frac{1}{2}\epsilon_{A'B'}R_{ABF'}{}^{F'}{}_{CC'DD'}
\end{equation*}
\medskip
\begin{equation*}
=\frac{1}{4}\epsilon_{A'B'}\epsilon_{C'D'}R_{AF'B}{}^{F'}{}_{CL'D}{}^{L'}+\frac{1}{4}\epsilon_{A'B'}\epsilon_{CD}R_{AF'B}{}^{F'}{}_{LC'}{}^{L}{}_{D'}
\end{equation*}
\medskip
\begin{equation*}
+\frac{1}{4}\epsilon_{AB}\epsilon_{CD}R_{A'FB'}{}^{F}{}_{C'LD'}{}^{L}+\frac{1}{4}\epsilon_{AB}\epsilon_{C'D'}R_{A'FB'}{}^{F}{}_{L'C}{}^{L'}{}_{D}
\end{equation*}
\medskip
\begin{equation*}
\equiv \epsilon_{A'B'}\epsilon_{C'D'}X_{ABCD}+\epsilon_{A'B'}\epsilon_{CD}\Phi_{ABC'D'}
\end{equation*}
\medskip
\begin{equation}
\label{eqn:25}
+\epsilon_{AB}\epsilon_{CD}\bar{X}_{A'B'C'D'}+\epsilon_{AB}\epsilon_{C'D'}\bar{\Phi}_{A'B'CD}.
\end{equation}
\medskip
$X_{ABCD}$ and $\Phi_{ABC'D'}$ are uniquely defined \textit{curvature spinors}. Using again the first of \eqref{eqn:21} it is easy to prove that the following properties hold for the curvature spinors:
\begin{subequations}
\label{eqn:23}
\begin{align}
&X_{ABCD}=X_{BACD}=X_{ABDC};\\
&X_{ABCD}=X_{CDAB};\\
&\Phi_{ABC'D'}=\Phi_{BAC'D'}=\Phi_{ABD'C'};\\
&\Phi_{ABC'D'}=\bar{\Phi}_{ABC'D'}.
\end{align}
\end{subequations}
The first two properties imply that
\begin{equation}
\label{eqn:66}
X_{A(BC)}{}^{A}=0,
\end{equation}
while the last two properties force the tensor $\Phi_{ab}$ corresponding to the spinor $\Phi_{AA'BB'}$ to be traceless and real:
\begin{equation*}
\Phi_{ab}=\bar{\Phi}_{ab},\hspace{1cm}\Phi^a_{\enskip a}=0.
\end{equation*}
It is very useful to introduce a dual of $R_{abcd}$, defined as
\begin{equation*}
R^*{}_{abcd}\equiv \frac{1}{2}\epsilon_{cd}{}^{pq}R_{abpq}
\end{equation*}
\begin{equation*}
=\frac{i}{2}(\epsilon_{CD}\epsilon^{PQ}\epsilon_{C'}{}^{P'}\epsilon_{D'}{}^{Q'}-\epsilon_{C}{}^{P}\epsilon_{D}{}^{Q}\epsilon_{C'D'}\epsilon^{P'Q'})(\epsilon_{A'B'}\epsilon_{P'Q'}X_{ABPQ}+\epsilon_{A'B'}\epsilon_{PQ}\Phi_{ABP'Q'}
\end{equation*}
\medskip
\begin{equation*}
+\epsilon_{AB}\epsilon_{PQ}\bar{X}_{A'B'P'Q'}+\epsilon_{AB}\epsilon_{P'Q'}\bar{\Phi}_{A'B'PQ})
\end{equation*}
\medskip
\begin{equation*}
=-iX_{ABCD}\epsilon_{A'B'}\epsilon_{C'D'}+i\Phi_{ABC'D'}\epsilon_{A'B'}\epsilon_{CD}
\end{equation*}
\medskip
\begin{equation*}
-i\bar{\Phi}_{A'B'CD}\epsilon_{AB}\epsilon_{C'D'}+i\bar{X}_{A'B'C'D'}\epsilon_{AB}\epsilon_{CD}\
\end{equation*}
\medskip
\begin{equation}
\label{eqn:22}
=iR_{AA'BB'CD'DC'}.
\end{equation}
\medskip
The second of \eqref{eqn:21} is then equivalent to
\begin{equation}
R^*_{abc}{}^{b}=0,
\end{equation}
thus multiplying \eqref{eqn:22} by $\epsilon^{BD}\epsilon^{B'D'}$ 	we should get zero. Hence we have
\medskip
\begin{equation*}
-X_{ABC}{}^{B}\epsilon_{A'C'}-\Phi_{ACC'A'}+\bar{\Phi}_{A'C'CA}+\epsilon_{AC}\bar{X}_{A'B'C'}{}^{B'}=0,
\end{equation*}
that, using the first of \eqref{eqn:23}, becomes
\begin{equation*}
X_{AB}{}^{B}{}_{C}\epsilon_{A'C'}=\epsilon_{AC}\bar{X}_{A'B'}{}^{B'}{}_{C'}.
\end{equation*}
Thus, on defining 
\begin{equation}
\label{eqn:26}
\Lambda\equiv\frac{1}{6}X_{AB}{}^{AB},
\end{equation}
we get the reality condition
\begin{equation}
\label{eqn:24}
\Lambda=\bar{\Lambda}.
\end{equation}
Relations \eqref{eqn:23} and \eqref{eqn:24} are the only algebraic relations necessarily satisfied by $X_{ABCD}$ and $\Phi_{ABC'D'}$. However, $X_{ABCD}$ and $\Phi_{ABCD}$ also satisfy a differential relation obtained from the Bianchi identities\begin{equation}
\label{eqn36}
\nabla_{[a}R_{bc]de}=0.
\end{equation}
Introducing another dual of $R_{abcd}$ as
\begin{equation*}
{}^*R_{abcd}=\frac{1}{2}\epsilon_{ab}{}^{pq}R_{pqcd}=R^*{}_{cdab}=iR_{CC'DD'AB'BA'}
\end{equation*} 
\begin{equation*}
=-iX_{CDAB}\epsilon_{C'D'}\epsilon_{A'B'}+i\Phi_{CDA'B'}\epsilon_{C'D'}\epsilon_{AB}
\end{equation*}
\begin{equation*}
-i\bar{\Phi}_{C'D'AB}\epsilon_{CD}\epsilon_{A'B'}+i\bar{X}_{C'D'A'B'}\epsilon_{CD}\epsilon_{AB},
\end{equation*}
equation \eqref{eqn36} reads as 
\begin{equation*}
\nabla^a{}^{*}R_{abcd}=0,
\end{equation*}
from which we get (using properties \ref{eqn:23})
\begin{equation*}
-\nabla^A{}_{B'}X_{ABCD}\epsilon_{C'D'}+\nabla^{A'}{}_{B}\Phi_{A'B'CD}\epsilon_{C'D'}-\nabla^{A}{}_{B'}\Phi_{ABC'D'}\epsilon_{CD}+\nabla^{A'}{}_{B}\bar{X}_{A'B'C'D'}\epsilon_{CD}=0.
\end{equation*}
Separating this last equation into parts which are skew-symmetric and symmetric in $C'D'$, respectively, we find it to be equivalent to 
\begin{equation}
\label{eqn:47}
\nabla^A{}_{B'}X_{ABCD}=\nabla^{A'}{}_{B}\Phi_{CDA'B'},
\end{equation}
and its complex conjugate. Equation \eqref{eqn:47} is the spinor form of Bianchi identity.\\
Using \eqref{eqn:25} with \eqref{eqn:23} and \eqref{eqn:24} we get for the Ricci tensor the spinor form 
\begin{equation}
\label{eqn:32}
R_{ab}=R^c_{acb}=6\Lambda\epsilon_{AB}\epsilon_{A'B'}-2\Phi_{ABA'B'}=6\Lambda g_{ab}-2\Phi_{ab},
\end{equation} 
and a further contraction gives the scalar curvature as
\begin{equation}
\label{eqn:33}
R=24\Lambda.
\end{equation}
Eventually we get the Einstein tensor
\begin{equation}
\label{eqn:48}
G_{ab}=R_{ab}-\frac{1}{2}Rg_{ab}=-6\Lambda\epsilon_{AB}\epsilon_{A'B'}-2\Phi_{ABA'B'}=-6\Lambda g_{ab}-2\Phi_{ab}.
\end{equation}
We obtain now a more suitable form of the Riemann curvature.\\
First we note that $X_{ABCD}$ can be decomposed as 
\begin{equation*}
X_{ABCD}=\frac{1}{3}(X_{ABCD}+X_{ACDB}+X_{ADBC})+\frac{1}{3}(X_{ABCD}-X_{ACBD})
\end{equation*}
\begin{equation*}
+\frac{1}{3}(X_{ABCD}-X_{ADCB})
\end{equation*}
\begin{equation*}
=X_{(ABCD)}+\frac{1}{3}\epsilon_{BC}X_{AF}{}^{F}{}_{D}+\frac{1}{3}\epsilon_{BD}X_{AFC}{}^{F}.
\end{equation*}
Since, from \eqref{eqn:26} it follows that $X_{AFC}{}^{F}=3\Lambda\epsilon_{AF}$, we get
\begin{equation}
\label{eqn:43}
X_{ABCD}=\Psi_{ABCD}+\Lambda(\epsilon_{AC}\epsilon_{BD}+\epsilon_{AD}\epsilon_{BC}).
\end{equation}
Inserting this in \eqref{eqn:21} and taking into account \eqref{eqn:24} we get 
\begin{equation*}
R_{abcd}=\Psi_{ABCD}\epsilon_{A'B'}\epsilon_{C'D'}+\Lambda(\epsilon_{BC}\epsilon_{AD}+\epsilon_{BD}\epsilon_{AC})\epsilon_{A'B'}\epsilon_{C'D'}
\end{equation*}
\begin{equation*}
+\bar{\Psi}_{A'B'C'D'}\epsilon_{AB}\epsilon_{CD}+\Lambda(\epsilon_{B'C'}\epsilon_{A'D'}+\epsilon_{B'D'}\epsilon_{A'C'})\epsilon_{AB}\epsilon_{CD}
\end{equation*}
\begin{equation*}
+\Phi_{ABC'D'}\epsilon_{A'B'}\epsilon_{CD}+\bar{\Phi}_{A'B'CD}\epsilon_{AB}\epsilon_{C'D'}.
\end{equation*}
Using the Jacobi identity in the form 
\begin{equation*}
\epsilon_{AB}\epsilon_{CD}=\epsilon_{AC}\epsilon_{BD}-\epsilon_{AD}\epsilon_{BC},
\end{equation*}
we get the desired decomposition of the Riemann curvature
\begin{equation*}
R_{abcd}=\Psi_{ABCD}\epsilon_{A'B'}\epsilon_{C'D'}+\bar{\Psi}_{A'B'C'D'}\epsilon_{AB}\epsilon_{CD}
\end{equation*}
\begin{equation*}
+\Phi_{ABC'D'}\epsilon_{A'B'}\epsilon_{CD}+\bar{\Phi}_{A'B'CD}\epsilon_{AB}\epsilon_{C'D'}
\end{equation*}
\begin{equation}
\label{eqn:27}
+2\Lambda(\epsilon_{AC}\epsilon_{BD}\epsilon_{A'C'}\epsilon_{B'D'}-\epsilon_{AD}\epsilon_{BC}\epsilon_{A'D'}\epsilon_{B'C'}).
\end{equation}
The first two terms in \eqref{eqn:27} are respectively the self-dual and the anti-self-dual Weyl tensors,
\begin{equation*}
{}^{(-)}C_{abcd}=\Psi_{ABCD}\epsilon_{A'B'}\epsilon_{C'D'}, \hspace{1cm}{}^{(+)}C_{abcd}=\bar{\Psi}_{A'B'C'D'}\epsilon_{AB}\epsilon_{CD},
\end{equation*}
which form the Weyl tensor
\begin{equation}
\label{eqn:28}
C_{abcd}={}^{(-)}C_{abcd}+{}^{(+)}C_{abcd}.
\end{equation}
Purely tensorially, this is given by 
\begin{equation*}
C_{abcd}=R_{abcd}+4\Phi_{[a[c}g_{d]b]}-4\Lambda g_{[a[c}g_{d]b]}=
R_{abcd}-2R_{[a[c}g_{d]b]}+\frac{1}{3}Rg_{[a[c}g_{d]b]}.
\end{equation*}
This tensor has the same symmetries as $R_{abcd}$, i.e.
\begin{equation}
\label{eqn:29}
C_{abcd}=-C_{bacd}=-C_{abdc},\hspace{1cm}C_{abcd}=C_{cdab},\hspace{1cm}C_{a[bcd]}=0,
\end{equation}
and is in addition trace-free
\begin{equation}
\label{eqn:30}
C^c{}_{acb}=0.
\end{equation}
As we will see in the remainder, the Weyl tensor is the conformally invariant part of the curvature. Furthermore, introducing the following tensors,
\begin{align*}
&E_{abcd}=\Phi_{ABC'D'}\epsilon_{A'B'}\epsilon_{CD}+\bar{\Phi}_{A'B'CD}\epsilon_{AB}\epsilon_{C'D'},\\
&g_{abcd}=\epsilon_{AC}\epsilon_{BD}\epsilon_{A'C'}\epsilon_{B'D'}-\epsilon_{AD}\epsilon_{BC}\epsilon_{A'D'}\epsilon_{B'C'}=2g_{a[c}g_{d]b},
\end{align*}
equation \eqref{eqn:27} reads
\begin{equation}
\label{eqn:27.1}
R_{abcd}={}^{(-)}C_{abcd}+{}^{(+)}C_{abcd}+E_{abcd}+2\Lambda g_{abcd}.
\end{equation}
In the language of representation theory ${}^{(-)}C$, ${}^{(+)}C$, $E$ and $g$ of equation \eqref{eqn:27.1} belong to representation spaces for the $D^{(2,0)}$, $D^{(0,2)}$, $D^{(1,1)}$, $D^{(0,0)}$ irreducible representations of the Lorentz group. \\
So far, as we said before, we derived the spinorial form of curvature tensors only by using their (skew-)symmetry properties. However the curvature tensor $R_{abcd}$ appears when a commutator of derivatives $\nabla_a$ is applied to vectors and tensors. Thus, we may expect that the spinors which represent $R_{abcd}$ appear when such commutators are applied to spinors. This is indeed the case. We start by building the connection coefficients.\\
Consider a tetrad of vectors $e_a{}^{\hat{c}}$ and its dual basis of co-vectors as done in section \ref{sect:2.3}. 
Define 
\begin{equation}
\label{eqn:72}
\Gamma^{\hat{a}}{}_{\hat{b}\hat{c}}=e_c{}^{\hat{a}}e^{d}{}_{\hat{c}}\nabla_{d}e^{c}{}_{\hat{b}}\equiv e_c{}^{\hat{a}}\nabla_{\hat{c}}e^{c}{}_{\hat{b}}.
\end{equation}
These quantities are called \textit{Ricci rotation coefficients} in the case when the frame is chosen so that the metric has a specific form, for example \eqref{eqn:71}, or \textit{Christoffel symbols} when the basis is naturally derived from a coordinate system. 
To obtain the spinor equivalent of  \eqref{eqn:72} we need to introduce a spinor dyad $\epsilon^{A}{}_{\hat{A}}$ and its symplectic dual $\epsilon_{A}{}^{\hat{A}}$ by 
\begin{equation}
\label{eqn:74}
\epsilon^{A}{}_{\hat{A}}\epsilon_{B}{}^{\hat{A}}=\epsilon^A{}_{B}=-\delta^A{}_{B}.
\end{equation}
A suitable choice is $\epsilon^A{}_{\hat{0}}=o^A$, $\epsilon^A{}_{\hat{1}}=\iota^A$, $\epsilon_A{}^{\hat{0}}=-\iota^A$, $\epsilon_A{}^{\hat{1}}=o^A$.
We could now introduce the \textit{spinor Ricci rotation coefficients} as 
\begin{equation}
\label{eqn:73}
\Gamma_{\hat{A}\hat{B}\hat{C}\hat{C}'}=\epsilon_{A\hat{A}}\epsilon^{C}{}_{\hat{C}}\epsilon^{C'}{}_{\hat{C}'}\nabla_{CC'}\epsilon^{A}{}_{\hat{B}}\equiv\epsilon_{A\hat{A}}\nabla_{\hat{C}\hat{C}'}\epsilon^{A}{}_{\hat{B}}.
\end{equation}
It is immediate to see that the spinor Ricci rotation coefficients constitute the spinor analogous of the the Ricci rotation coefficients.\\
Define 
\begin{equation*}
\square_{CD}=\epsilon^{C'D'}\nabla_{[CC'}\nabla_{DD']}=\frac{1}{2}\epsilon^{C'D'}(\nabla_{CC'}\nabla_{DD'}-\nabla_{DD'}\nabla_{CC'})
\end{equation*}
\begin{equation}
\label{eqn:37}
=\frac{1}{2}(\nabla_{CC'}\nabla_{D}{}^{C'}+\nabla_{DD'}\nabla_{C}{}^{D'})=\nabla_{C'(C}\nabla_{D)}{}^{C'}.
\end{equation}
Consider now the commutator
\begin{equation*}
\Delta_{cd}=2\nabla_{[c}\nabla_{d]}=\nabla_{CC'}\nabla_{DD'}-\nabla_{DD'}\nabla_{CC'}.
\end{equation*}
By applying \eqref{eqn:15} to the previous equation we get, using \eqref{eqn:37} 
\begin{equation*}
\Delta_{cd}=\nabla_{(C'(C}\nabla_{D)D')}+\frac{1}{2}\epsilon_{C'D'}\nabla_{E'(C}\nabla_{D)}{}^{E'}+\frac{1}{2}\epsilon_{CD}\nabla_{E(C'}\nabla_{D')}{}^{E}
\end{equation*}
\begin{equation*}
+\frac{1}{4}\epsilon_{CD}\epsilon_{C'D'}\nabla_{EE'}\nabla^{EE'}-\nabla_{(D(D'}\nabla_{C')C)}-\frac{1}{2}\epsilon_{D'C'}\nabla_{E'(D}\nabla_{C)}{}^{E'}
\end{equation*}
\begin{equation*}
-\frac{1}{2}\nabla_{E(D'}\nabla_{C'}{}^{E}-\frac{1}{4}\epsilon_{DC}\epsilon_{D'C'}\nabla_{EE'}\nabla^{EE'}=
\end{equation*}
\begin{equation}
\label{eqn:38}
=\epsilon_{C'D'}\square_{CD}+\epsilon_{CD}\square_{C'D'}.
\end{equation}
To obtain the action on a spinor, $\kappa^A$,  we begin by forming the self-dual null bi-vector
\begin{equation*}
k^{ab}=\kappa^A\kappa^B\epsilon^{A'B'}.
\end{equation*}
We know that 
\begin{equation*}
\Delta_{ab}k^{cd}=R_{abe}{}^{c}k^{ed}+R_{abe}{}^{d}k^{ce}.
\end{equation*}
 Thus
 \begin{equation*}
 \kappa^C\epsilon^{C'D'}\Delta_{ab}\kappa^{D}+\kappa^D\epsilon^{C'D'}\Delta_{ab}\kappa^{C}=R_{abEE'}{}^{CC'}\kappa^E\kappa^D\epsilon^{E'D'}+R_{abEE'}{}^{DD'}\kappa^{C}\kappa^{E}\epsilon^{C'E'},
 \end{equation*}
 i.e.
 \begin{equation*}
 2\epsilon^{C'D'}\kappa^{(C}\Delta_{ab}\kappa^{D)}=-R_{abE}{}^{D'CC'}\kappa^E\kappa^D+R_{abE}{}^{C'DD'}\kappa^C\kappa^E.
 \end{equation*}
 Using decomposition \eqref{eqn:25} for the curvature we get 
 \begin{equation*}
 2\epsilon^{C'D'}\kappa^{(C}\Delta_{ab}\kappa^{D)}=2\epsilon^{C'D'}(\epsilon_{A'B'}X_{ABE}{}^{(C}\kappa^{D)}+\epsilon_{AB}\Phi_{A'B'E}{}^{(C}\kappa^{D)})\kappa^E
 \end{equation*}
 \begin{equation*}
 +2\epsilon_{E}{}^{[D}\kappa^{C]}\kappa^E(\epsilon_{AB}\bar{X}_{A'B'}{}^{D'C'}+\epsilon_{A'B'}\Phi_{AB}{}^{D'C'}).
 \end{equation*}
 Multiplying by $\epsilon_{C'D'}$ we get
 \begin{equation*}
\kappa^{(C}\Delta_{ab}\kappa^{D)}=(\epsilon_{A'B'}X_{ABE}{}^{(C}+\epsilon_{AB}\Phi_{A'B'E}{}^{(C})\kappa^{D)}\kappa^{E}
 \end{equation*}
 from which
 \begin{equation}
 \label{eqn:41}
 \Delta_{ab}\kappa^C=(\epsilon_{A'B'}X_{ABE}{}^{C}+\epsilon_{AB}\Phi_{A'B'E}{}^{C})\kappa^E.
 \end{equation}
On taking into account the decomposition \eqref{eqn:38} we get easily
\begin{equation}
\label{eqn:39}
\square_{AB}\kappa^C=X_{ABE}{}^{C}\kappa^E,\hspace{1cm}\square_{A'B'}\kappa^C=\Phi_{A'B'E}{}^{C}\kappa^E
\end{equation}
and by taking complex conjugates of \eqref{eqn:41} and \eqref{eqn:39} we get 
\begin{equation*}
\Delta_{ab}\kappa^{C'}=(\epsilon_{AB}\bar{X}_{A'B'E'}{}^{C'}+\epsilon_{A'B'}\Phi_{ABE'}{}^{C'})\kappa^{E'},
\end{equation*}
and
\begin{equation*}
\square_{AB}\kappa^{C'}=\Phi_{ABE'}{}^{C'}\kappa^{E'},\hspace{1cm}\square_{A'B'}\kappa^C=\bar{X}_{A'B'E'}{}^{C'}\kappa^{E'}.
\end{equation*}
These equations may be easily generalized to many-index spinors.\\ From the first of  \eqref{eqn:39} we get, lowering the index $C$ and substituting \eqref{eqn:43}
\begin{equation*}
\square_{AB}\kappa_{C}=-X_{ABC}{}^{D}\kappa_{D}=-\Psi_{ABC}{}^{D}\kappa_D-\Lambda(\epsilon_{AC}\kappa_B+\epsilon_{BC}\kappa_A),
\end{equation*}
and, by symmetrizing on $(ABC)$ and multiplying by $\epsilon^{BC}$, the terms in $\Psi_{ABCD}$ and $\Lambda$ can be respectively singled out
\begin{equation}
\label{eqn:44}
\square_{(AB}\kappa_{C)}=-\Psi_{ABC}{}^{D}\kappa_D,\hspace{1cm}\square_{AB}\kappa^B=-3\Lambda\kappa_A.
\end{equation}
Introducing a spinor dyad $\epsilon^{A}{}_{\hat{A}}$ and its symplectic dual $\epsilon_{A}{}^{\hat{A}}$ as done in \eqref{eqn:74} we get from \eqref{eqn:39}
\begin{equation}
\label{eqn:46}
X_{ABCD}=\epsilon_{D\hat{C}}\square_{AB}\epsilon_{C}{}^{\hat{C}},\hspace{1cm}\Phi_{A'B'CD}=\epsilon_{D\hat{C}}\square_{A'B'}\epsilon_{C}{}^{\hat{C}}
\end{equation}
and from \eqref{eqn:44}
\begin{equation}
\label{eqn:45}
\Psi_{ABCD}=\epsilon_{D\hat{C}}\square_{(AB}\epsilon_{C)}{}^{\hat{C}},\hspace{1cm}\Lambda=\frac{1}{6}\epsilon_{A\hat{C}}\square^{AB}\epsilon_{B}{}^{\hat{C}}.
\end{equation}
Furthermore it can be shown \citep{Penrin1} that the Bianchi identity \eqref{eqn:47} can be deduced by the action of commutators on spinors.\\
We can re-express now \eqref{eqn:47}
in terms of $\Psi_{ABCD}$ and $\Lambda$ by use of \eqref{eqn:43}:
\begin{equation*}
\nabla^{A}_{B'}\Psi_{ABCD}=\nabla^{A'}{}_{B}\Phi_{CDA'B'}-2\epsilon_{B(C}\nabla_{D)B'}\Lambda.
\end{equation*}
By splitting this equation in its symmetric and skew-symmetric part in $BC$ we get
\begin{equation}
\label{eqn:51}
\nabla^A_{B'}\Psi_{ABCD}=\nabla^{A'}{}_{(B}\Phi_{CD)A'B'},\hspace{1cm}\nabla^{CA'}\Phi_{CDA'B'}+3\nabla_{DB'}\Lambda=0.
\end{equation}
As we see from \eqref{eqn:48} this equation is the spinor form of the important result that the Einstein tensor is divergence free, $\nabla^aG_{ab}=0$.
An important point is that the Bianchi identity could be regarded as a field equation for the Weyl tensor. It might be useful here to point out that it is a misconception to consider the Bianchi identity as simply a tautology and to ignore it as contributing no further information, as it is done even today. It is an important piece of the structure on a Riemannian or Lorentzian manifold which relates the (derivatives of the) Ricci and Weyl tensors. If the Ricci tensor is restricted by the Einstein equations to equal the energy-momentum tensor, then the Bianchi identity provides a differential equation for the Weyl tensor. Its structure is very similar to the familiar zero rest-mass equation for a particle with spin 2. In fact, in a sense, one can consider
this equation as the essence of the gravitational theory, as we will see in section \ref{sect:2.9}.
\section{The Newman-Penrose Formalism}
\label{sect:2.7}
As shown in section \ref{sect:2.3}, $(o,\iota)$ induce four null vectors, a Newman-Penrose null tetrad \eqref{eqn:18}. We denote the directional derivatives along these directions by the conventional symbols
\begin{equation}
\label{eqn:101}
D=l^a\nabla_a,\hspace{0.8cm}\Delta=n^a\nabla_a,\hspace{0.8cm}\delta=m^a\nabla_a,\hspace{0.8cm}\bar{\delta}=\bar{m}^a\nabla_a.
\end{equation}
Clearly $\nabla_a$ is a combination of these operators. In fact, using \eqref{eqn:69} we have 
\begin{equation}
\label{eqn:68}
\nabla_a=g_a{}^b\nabla_b=(n_al^b+l_an^b-\bar{m}_am^b-m_a\bar{m}^b)\nabla_b=n_aD+l_a\Delta-\bar{m}_a\delta-m_a\bar{\delta}.
\end{equation}
The idea is now to replace $\nabla_a$ by \eqref{eqn:68} and then convert all the remaining tensor equations to sets of scalar ones by contraction with the N-P null tetrad. This can lead to a big number of equations, but they usually possess discrete symmetries, and involve only scalars so they are easier to handle in specific calculations.\\
We start by considering the connection. In spinor formalism this is described by the spinor Ricci rotation coefficients \eqref{eqn:73}. Each term is of the form
\begin{equation*}
\alpha^A\nabla\beta_A,
\end{equation*}
where $\alpha$ and $\beta$ are $(o,\iota)$ and $\nabla$ is one of $D$, $\Delta$, $\delta$, $\bar{\delta}$. We could also derive them vectorially using the N-P null tetrad of vectors \eqref{eqn:18}, e.g. 
$$\kappa=o^ADo_A=o^ADo_A=o^A\bar{o}_{A'}\bar{\iota}^{A'}o^ADo_A+o^Ao_A\bar{\iota}^{A'}D\bar{o}_{A'}=o^AD\bar{\iota}^{A'}D(o_A\bar{o}_{A'})=m^aDl_a.$$
Note that $\alpha$, $\beta$, $\gamma$, $\epsilon$ have particularly complicated vector descriptions. All of these are given in the box below.
\begin{figure}[h]
\begin{center}
\includegraphics[scale=0.5]{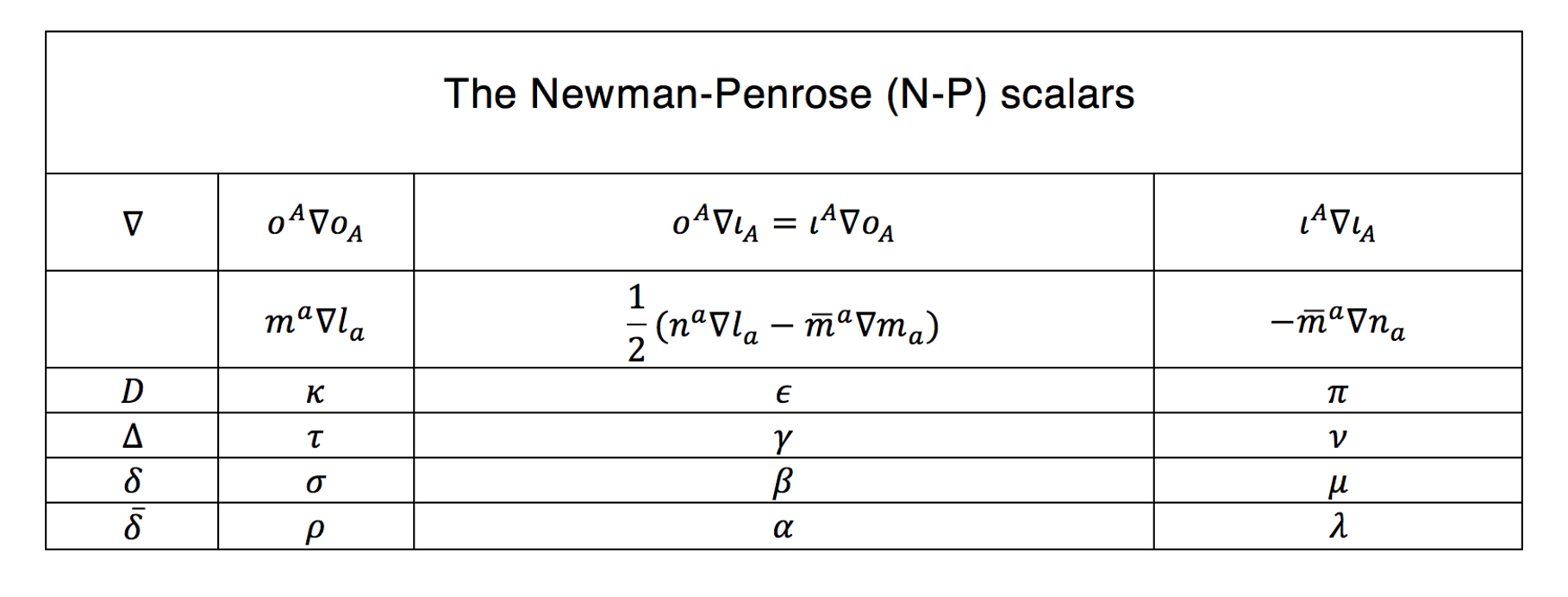}
\caption{A table with the connection coefficient.}
\label{fig:box}
\end{center}
\end{figure}
\\Note that for each scalar in the box there exists another one obtained from $(\bar{o}^{A'},\bar{\iota}^{A'})$, e.g. $\bar{\kappa}=\bar{o}^{A'}D\bar{o}_{A'}$, etc.
\\As a first application of this technique we derive the N-P description of electromagnetism. As we have seen in section \ref{sect:2.5} the Maxwell field tensor is described by a symmetric 2-spinor $\varphi_{AB}$. We can thus form 3 complex scalars
\begin{equation*}
\varphi_0=\varphi_{AB}o^Ao^B,\hspace{1cm}\varphi_1=\varphi_{AB}o^A\iota^B,\hspace{1cm}\varphi_2=\varphi_{AB}\iota^A\iota^B.
\end{equation*}
The generic term $\varphi_n$ has $n$ $\iota$'s and $2-n$ $o$'s. We have
\begin{equation*}
\varphi_{AB}=\epsilon_A{}^{C}\epsilon_B{}^{D}\varphi_{CD}=(o_A\iota^C-\iota_Ao^C)(o_B\iota^D-\iota_Bo^D)\varphi_{CD}
\end{equation*}
\begin{equation*}
=\varphi_2o_Ao_B-\varphi_1\iota_Ao_B-\varphi_1o_A\iota_B+\varphi_0\iota_A\iota_B
\end{equation*}
\begin{equation}
\label{eqn:128}
=\varphi_2o_Ao_B-2\varphi_1o_{(A}o_{B)}+\varphi_0\iota^A\iota^B.
\end{equation}
The Maxwell equations are equivalent to \begin{equation*}
\nabla^a(F_{ab}+iF^*{}_{ab})=0,
\end{equation*}
that, taking into account equation \eqref{eqn:70}, becomes
\begin{equation}
\label{eqn:75}
0=\nabla^{AA'}\varphi_{AB}=\epsilon^{AC}\nabla_{AA'}\varphi_{CB}=(o^A\iota^C-\iota^Ao^C)\nabla_{AA'}\varphi_{CB}.
\end{equation}
As $A',B$ take the values $0,1$, the previous equation splits into four complex equations corresponding to the eight real Maxwell equations. For example we may multiply it by $\bar{o}^{A'}o^B$ to get, after some calculation,
\begin{equation}
D\varphi_1-\bar{\delta}\varphi_0=(\pi-2\alpha)\varphi_0+2\rho\varphi_1-\kappa\varphi_2.
\end{equation}
Similarly, one can obtain the other three complex Maxwell equations, given in Appendix \ref{B}.\\
The same procedure can be followed for the trace-free Ricci tensor, $\Phi_{ABA'B'}$, that can be decomposed into 9 independent real quantities $\Phi_{ij}=\bar{\Phi}_{ji}$ and for the Weyl tensor $\Psi_{ABCD}$, that can be decomposed into 5 complex scalars $\Psi_n$. Note that the first index $i$ of $\Phi_{ij}$ is the number of contraction with $\iota$'s and the second, $j$, is the number of contractions with $\bar{\iota}$'s while the index $n$ of $\Psi_n$ is the number of contractions with $\iota$'s. For example, we have
\begin{equation*}
\Phi_{12}=\Phi_{ABA'B'}o^A\iota^B\bar{\iota}^{A'}\bar{\iota}^{B'},
\end{equation*} 
\begin{equation*}
\Psi_2=\Psi_{ABCD}o^Ao^B\iota^C\iota^D.
\end{equation*}
We know that the curvature tensor can be be expressed in terms of the Ricci rotation coefficients. It implies that there must be some relations in which both $\Phi$ and $\Psi$ can be linked to the N-P scalars. The equations defining the curvature tensor components in terms of derivatives and products of N-P scalars are called \textit{N-P field equations} and are 18 independent relations. They can be obtained with tedious calculations and are reported in Appendix \ref{B}. The Bianchi identities \eqref{eqn:51} too can be expressed in terms of the N-P scalars, giving 11 independent relations, that can be found, for example, in the Appendix B of \cite{Stew}.
\section{Null Congruences}
\label{sect:2.8}
A null congruence $\mathscr{C}$ is a congruence of null curves in space-time, i.e. a family of null curves with the property that precisely one member of the family passes through each point of a given domain under consideration. As we will see null congruences of rays (geodetic null curves) are very important in the gravitational radiation theory.\\
A congruence of curves is specified by giving a vector field $l^a$ on $\mathscr{M}$ and is defined to be the set of integral curves of $l^a$, i.e. the set of curves whose tangent vector is $l^a$. Choosing a parameter $u$ along each curve, the scaling of the vector $l^a$ is therefore defined by the relation
\begin{equation*}
l^a\nabla_au=1.
\end{equation*}
We are taking into account only null congruences, for which the tangent vector $l^a$ is null,
\begin{equation*}
l^al_a=0.
\end{equation*}
Using the map between spinors and vectors we may associate to $l^a$ a spin vector $o^A$ as
\begin{equation}
\label{eqn:98}
l^a=o^A\bar{o}^{A'}.
\end{equation}
If we are considering a geodetic congruence the vector $l^a$ has to be parallelly propagated along the curve,
\begin{equation}
\label{eqn:99}
l^a\nabla_al^b\propto l^b.
\end{equation}
The parameter $u$ is called affine if 
\begin{equation}
\label{eqn:100}
l^a\nabla_al^b=0.
\end{equation}
Using equations \eqref{eqn:98} and the first of \eqref{eqn:101} we have that \eqref{eqn:99} is
written in terms of the spinor $o^A$ as
\begin{equation}
\label{eqn:105}
Do^A\propto o^A,
\end{equation}
while \eqref{eqn:100} \textit{can} be written as
\begin{equation}
\label{eqn:104}
Do^A=0.
\end{equation}
Geometrically this equation tells us that the flag planes has to be parallel along $\mathscr{C}$.\\
Note that \eqref{eqn:105} is equivalent to
\begin{equation}
\label{eqn:102}
o^ADo_A=0
\end{equation}
and from the box of N-P scalars we see that this implies
\begin{equation*}
\kappa=0.
\end{equation*}
We can write down equation \eqref{eqn:102} entirely as  
\begin{equation*}
o^Ao^B\bar{o}^{A'}\nabla_{AA'}o_B=0,
\end{equation*}
and hence we obtain the following relations
\begin{equation}
\label{eqn:103}
o^A\bar{o}^{A'}\nabla_{AA'}o_B=\epsilon o_B,\hspace{0.65cm}o^B\bar{o}^{A'}\nabla_{AA'}o_B=\rho o_A,\hspace{0.65cm}o^Ao^B\nabla_{AA'}o_B=\sigma\bar{o}_{A'},
\end{equation}
$\epsilon$, $\rho$ and $\sigma$ being the N-P scalars given in the box \ref{fig:box}. This follows by transvecting the three previous equations by $\iota^B$, $\iota^A$ and $\bar{\iota}^{A'}$ respectively.\\
Notice that $\epsilon$, $\rho$ and $\sigma$ are defined without any reference to $\iota^A$, being referred to the geometry of the $o^A$ field alone. \\
From equation \eqref{eqn:100} we see that the condition for a null congruence of geodesics to be affinely parametrized is that 
\begin{equation*}
0=o^A\bar{o}^{A'}\nabla_{AA'}\left(o^B\bar{o}^{B'}\right),
\end{equation*}
hence
\begin{equation*}
0=\iota_B\bar{\iota}_{B'}o^A\bar{o}^{A'}\nabla_{AA'}\left(o^B\bar{o}^{B'}\right)=\epsilon+\bar{\epsilon}.
\end{equation*}
The condition for the geodetic congruence $\mathscr{C}$ to have both parallelly propagated flag planes alone and affine parametrization is, from equation \eqref{eqn:104}, 
\begin{equation*}
\epsilon=0.
\end{equation*} 
We now collect these results and one other in the following table:
\begin{itemize}
\item $\mathscr{C}$ geodetic $\Leftrightarrow$ $\kappa=0$;
\item $\mathscr{C}$ geodetic, u affine $\Leftrightarrow$ $\kappa=0$, $\epsilon+\bar{\epsilon}=0$;
\item $Do^A=0\Leftrightarrow\kappa=0$, $\epsilon=0$;
\item $D\iota^A=0\Leftrightarrow\pi=0$, $\epsilon=0$.
\end{itemize}
Consider now a null curve $\mu$ of the congruence $\mathscr{C}$ whose tangent null vector is $l^a=o^A\bar{o}^{A'}$. Complete $o^A$ to a spin basis $(o^A,\iota^A)$ at a point $p$ of $\mu$. We can propagate $o^A$ and $\iota^A$ along $\mu$ via
\begin{equation*}
Do^A=0,\hspace{1cm}D\iota^A=0,
\end{equation*}
where $D=l^a\nabla_a$. This means that $o^A$ and $\iota^A$ are parallelly propagated along $\mu$ and remain a spin basis at each of its points. Consider a connecting vector $\zeta^a$ of any two elements of $\mathscr{C}$, $\mu$ and $\mu'$. By definition of connecting vector, $\zeta^a$ satisfies
\begin{equation}
\label{eqn:s1}
[l,\zeta]_a=0\Rightarrow \nabla_l\zeta_a=\nabla_{\zeta}l_a.
\end{equation}
Suppose that at $p\in\gamma$ the vector $\zeta$ is orthogonal to $\gamma$. Then $l^a\zeta_a=0$ at $p$ and hence, from \eqref{eqn:s1}
\begin{equation*}
D(l^a\zeta_a)=l^aD\zeta_a=l^a\nabla_l\zeta_a=l^a\nabla_{\zeta}l_a=\frac{1}{2}\nabla_{\zeta}(l^al_a)=0.
\end{equation*}
Thus $\zeta^a$ remains always orthogonal to $\gamma$ in each of its points. Neighbouring pair of rays satisfying this property are called \textit{abreast}. Their physical meaning is the following. If we realize the congruence physically by a cloud of photons, then two abreast rays correspond to the world-lines of two neighbouring photons which in some observer's local 3-space lie in a 2-plane element perpendicular to their paths. Moreover, any two local observers will judge the photons to be at the same distance from each other if and only if the rays are abreast.\\ Construct now the N-P tetrad $(l,n,m,\bar{m})$ induced by the spin basis. Since $\zeta^a$ is real and orthogonal to $l^a$ there must exist a real $u$ and a complex $\zeta$ such that 
\begin{equation*}
\zeta^a=ul^a+\bar{\zeta}m^a+\zeta\bar{m}^a
\end{equation*}
\begin{equation*}
=uo^A\bar{o}^{A'}+\bar{\zeta}o^A\bar{\iota}^{A'}+\zeta\bar{o}^{A'}\iota^A.
\end{equation*}
It follows that 
\begin{equation*}
D\zeta^a=\nabla_l\zeta^a=\nabla_{\zeta}l^a=\zeta^b\nabla_{b}l^a=(ul^b+\bar{\zeta}m^b+\zeta\bar{m}^b)\nabla_bl^a=\bar{\zeta}\delta l^a+\zeta\bar{\delta}l^a
\end{equation*}
\begin{equation}
\label{eqn:s2}
=\bar{\zeta}o^A\delta\bar{o}^{A'}+\bar{\zeta}\bar{o}^{A'}\delta o^{A}+\zeta o^A\bar{\delta}\bar{o}^{A'}+\zeta\bar{o}^{A'}\bar{\delta}o^A.
\end{equation}
But we also have
\begin{equation}
\label{eqn:s3}
D\zeta^a=o^A\bar{o}^{A'}Du+o^A\bar{\iota}^{A'}D\bar{\zeta}+\iota^A\bar{o}^{A'}D\zeta.
\end{equation}
Thus, comparing LH sides of \eqref{eqn:s2} and \eqref{eqn:s3} we get
\begin{equation*}
o^A\bar{o}^{A'}Du+o^A\bar{\iota}^{A'}D\bar{\zeta}+\iota^A\bar{o}^{A'}D\zeta=\bar{\zeta}o^A\delta\bar{o}^{A'}+\bar{\zeta}\bar{o}^{A'}\delta o^{A}+\zeta o^A\bar{\delta}\bar{o}^{A'}+\zeta\bar{o}^{A'}\bar{\delta}o^A.
\end{equation*}
Multiplying by $o_A\bar{\iota}_{A'}$ we have
\begin{equation*}
-D\zeta=-\bar{\zeta}o_A\delta o^A-\zeta o_A\bar{\delta}o^A,
\end{equation*}
i.e.,
\begin{equation}
\label{eqn:s4}
D\zeta=-\rho\zeta-\sigma\bar{\zeta}.
\end{equation}
The interpretation of $\zeta$ is as follows. The projection of $\zeta^a$ onto the spacelike 2-plane spanned by $m^a$ and $\bar{m}^a$,  which we call $\Pi$, is
\begin{equation*}
\zeta^bm_bm^a+\zeta^b\bar{m}_b\bar{m}^a=
\zeta m^a+\bar{\zeta}\bar{m}^a.
\end{equation*}
Thus $\zeta$ describes the projection in an Argand 2-plane $\Pi$ spanned by $m^a$ and $\bar{m}^a$, see Figure \ref{fig:piani}.
\begin{figure}[h]
\begin{center}
\includegraphics[scale=0.5]{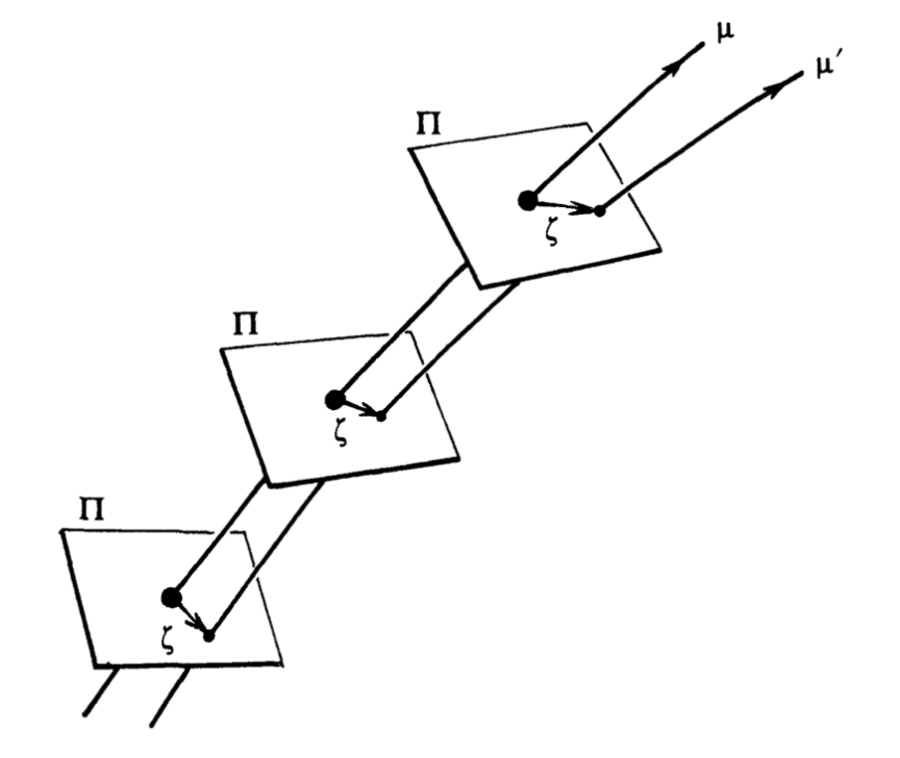}
\caption{Two abreast neighbouring rays. Their separation, as time progresses, is measured by $\zeta$.}
\label{fig:piani}
\end{center}
\end{figure}
\\Define
\begin{equation*}
\rho=k+it,\hspace{1cm}\sigma=se^{2i\theta}.
\end{equation*}
with $k$, $t$, $s$ and $\theta$ real. We consider three cases:
\begin{itemize}
\item For $t=s=0$ we have that \eqref{eqn:s4} reduces to 
\begin{equation*}
D\zeta=-k\zeta,
\end{equation*}
showing that $k=\mathrm{Re}(\rho)$ measures the rate of contraction of the simultaneous bundle of rays, i.e. the \textit{congruence} of $\mathscr{C}$.
\item For $k=s=0$ we have 
\begin{equation*}
D\zeta=-it\zeta,
\end{equation*}
showing that $t=\mathrm{Im}(\rho)$ measures the \textit{twist} (or rotation).
\item For $\rho=\theta=0$, setting $\zeta=x+iy$ we get
\begin{equation*}
Dx=-sx,\hspace{1cm}Dy=sy,
\end{equation*}
which represents a volume-preserving shear at a rate $s$ with principal axes along the $x$ and $y$ axes. Thus $s=\abs{\sigma}$ is a measure of \textit{degree of shearing}, i.e. the tendency of the initial sphere to become distorted into an ellipsoidal shape. If $s=0$, then a small spherical region will remain spherical, but if $s\neq 0$, it will be stretched in some directions. It can be shown that multiplying   $s$ by $e^{2i\theta}$ rotates the principal shear axes by $\theta$.
\end{itemize}
\begin{figure}[h]
\begin{center}
\includegraphics[scale=0.5]{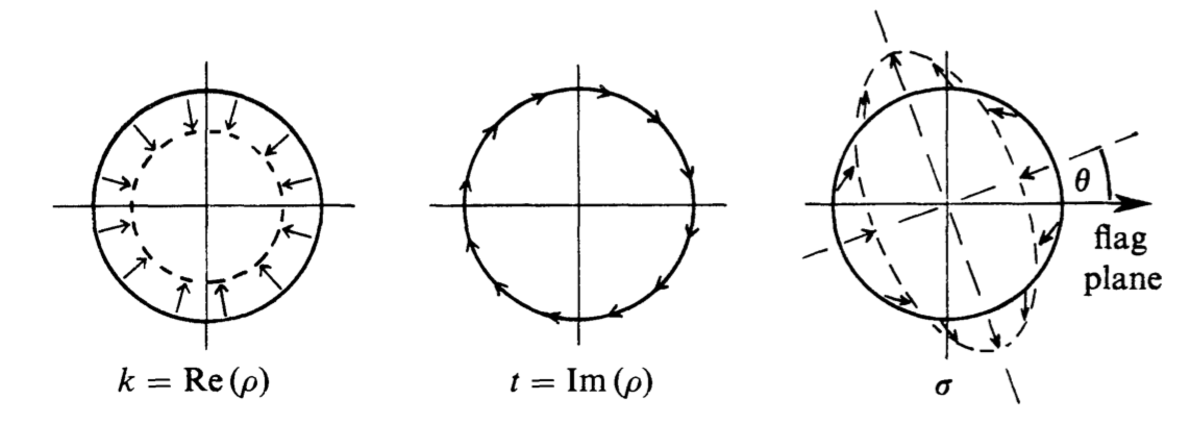}
\caption{The geometrical interpretation of $\rho$ and $\sigma$ in terms of behaviour in $\zeta$-plane.}
\end{center}
\end{figure}
\begin{thm}$\\ $
\label{thm:shear}
If $\tilde{\sigma}$ is the shear of $\tilde{l}^a$ with respect to $\tilde{g}_{ab}$, then the shear $\sigma$ of $l^a=\Omega^{-2}\tilde{l}^a$ with respect to $g_{ab}=\Omega^2\tilde{g}_{ab}$ is given by $\sigma = \Omega^{-2}\tilde{\sigma}$. In particular, shear-freeness is conformally invariant.
\end{thm}
Although the proof to this theorem is very simple \citep[see][pg.~110]{GRG}, it will play an important role in the description of the properties of $\mathscr{I}$.\\
Since \eqref{eqn:s4} is a linear equation, the general case is a superposition of these effects: the congruence, or more precisely the projection of the connecting vector onto an orthogonal spacelike 2-surface, is expanded, rotated and sheared.\\
\begin{defn}$\\ $
A null congruence $\mathscr{C}$ is said to be \textit{hypersurface-orthogonal} if its tangent vector field, $l^a$, is proportional to a gradient field, i.e. there exist $v$ and $f$ such that 
\begin{equation*}
l_a=v\nabla_af.
\end{equation*} 
\end{defn}
If $\mathscr{C}$ is hypersurface-orthogonal we have 
\begin{equation*}
l_{[a}\nabla_bl_{c]}=v\nabla_{[a}f\nabla_b(v\nabla_{c]}f)=v\nabla_{[a}f\nabla_bv\nabla_{c]}f+v^2\nabla_{[a}f\nabla_b\nabla_{c]}f
\end{equation*}
\begin{equation*}
=\frac{1}{3}v(\nabla_bv\nabla_{[a}f\nabla_{b]}f+\nabla_cv\nabla_{[b}f\nabla_{a]}f+\nabla_av\nabla_{[c}f\nabla_{b]}f)
\end{equation*}
\begin{equation*}
+\frac{1}{3}v^2(\nabla_af\nabla_{[b}\nabla_{c]}f+\nabla_bf\nabla_{[c}\nabla_{a]}f+\nabla_cf\nabla_{[a}\nabla_{b]}f)=0.
\end{equation*}
Note that the second term vanishes in virtue of the torsion free condition. The converse is also true, i.e. $\mathscr{C}$ is hypersurface-orthogonal if its tangent vector $l^a$ satisfies $ l_{[a}\nabla_bl_{c]}=0$. This is a particular consequence of the Frobenius theorem \citep[see][pg.~434-436]{Wald}. Hence we can state the following
\begin{prop}$\\ $
\begin{center}
$\mathscr{C}$ is hypersurface-orthogonal $\Leftrightarrow l_{[a}\nabla_bl_{c]}=0$.
\end{center}
\end{prop}
Furthermore we have
\begin{itemize}
\item $0=m^a\bar{m}^bn^cl_{[a}\nabla_bl_{c]}=m^a\bar{m}^bn^c(l_a\nabla_bl_c-l_a\nabla_cl_b+l_c\nabla_al_b-l_c\nabla_bl_a+l_b\nabla_cl_a-l_b\nabla_al_c)$\\
i.e.,\\
$0=m^a\bar{m}^b(\nabla_al_b-\nabla_bl_a)=\bar{m}^b\delta l_b-m^b\bar{\delta}l_b\Rightarrow\rho=\bar{\rho};$ 

\item $0=m^al^bn^cl_{[a}\nabla_bl_{c]}=m^a\bar{m}^bn^c(l_a\nabla_bl_c-l_a\nabla_cl_b+l_c\nabla_al_b-l_c\nabla_bl_a+l_b\nabla_cl_a-l_b\nabla_al_c)$\\
i.e.,\\
$0=l^b\delta l_b-m^bDl_b\Rightarrow 0=-m^bDl_b=-\kappa\Rightarrow \kappa=0.$ 
\end{itemize}
Thus we have
\begin{prop}$\\ $
\label{prop:equivalence}
\begin{center}
$\mathscr{C}$ is hypersurface-orthogonal $\Leftrightarrow l_{[a}\nabla_bl_{c]}=0\Leftrightarrow\left\{\begin{matrix}\kappa=0  \\ \rho=\bar{\rho}\end{matrix}\right.$
\end{center}
i.e. a null congruence $\mathscr{C}$ is hypersurface-orthogonal if and only if it is geodetic and twist-free.
\end{prop}
\begin{figure}[h]
\begin{center}
\includegraphics[scale=0.5]{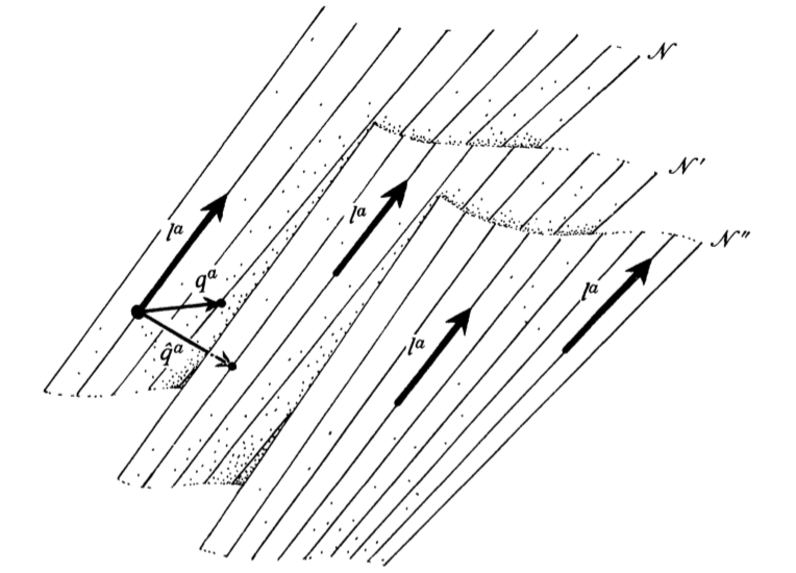}
\caption{A one-parameter family of null hypersurfaces, whose normals $l^a$ are null vectors and which are therefore also tangent vectors.}
\end{center}
\end{figure}
If a null congruence $\mathscr{C}$ with tangent vector $l^a$ is hypersurface-orthogonal, then, by definition, the hypersurfaces $\mathscr{N}$ to which it is orthogonal must be null, and $l^a$ is also tangent to them. Moreover, since the normal direction to a particular $\mathscr{N}$ is unique at each point, and since $l^a$ is the only null direction orthogonal to $l^a$, $l^a$ is the unique future-pointing null tangent direction at each point of $\mathscr{N}$. These directions in $\mathscr{N}$ have a two-parameter family of integral curves called \textit{generators}: they \lq form\rq\hspace{0.1mm} $\mathscr{N}$. Conversely, the generators of a one-parameter family of null hypersurfaces $\mathscr{N}$ constitute a three-parameter family of null lines which are hypersurface-orthogonal. Note that the generators of the hypersurfaces $\mathscr{N}$, being null and hypersurface-orthogonal, must be \textit{geodetic} by proposition \ref{prop:equivalence}. For the equivalence stated above the quantities $\rho$ and $\sigma$ refer as well to the geometry of $\mathscr{N}$ as to the entire congruence $\mathscr{C}$. We can speak of \textit{convergence and shear of a single null hypersurface}. Therefore we have the following characterization.
\begin{prop}$\\ $
A null congruence $\mathscr{C}$ with tangent vector $l^a$ is hypersurface-orthogonal if and only if it is null-hypersurface forming, i.e. there exists a one parameter family of null hypersurfaces to which $l^a$ is tangent at each point.
\end{prop}
\section{Einstein's Equations\hspace{1cm}}
\label{sect:2.9}
The decomposition \eqref{eqn:27} of the Riemann tensor into its irreducible spinorial parts allows us to discuss the structure of space-time  curvature, as is implied by Einstein's Field Equations,
\begin{equation}
\label{eqn:31}
R_{ab}-\frac{1}{2}Rg_{ab}+\lambda g_{ab}=-8\pi GT_{ab}.
\end{equation}
Here $\lambda$ is the cosmological constant and $T_{ab}$ the local stress-energy-momentum tensor. \\
Using \eqref{eqn:31} and \eqref{eqn:32}-\eqref{eqn:33} we get 
\begin{equation*}
8\pi GT_{ab}=(6\Lambda-\lambda)g_{ab}+2\Phi_{ab},
\end{equation*}
from which, taking into account the trace-free property of $\Phi_{ab}$,
\begin{equation}
\label{eqn:49}
\Lambda=\frac{\pi G}{3}T^{a}{}_{a}+\frac{\lambda}{6},
\end{equation}
and thus
\begin{equation}
\label{eqn:50}
\Phi_{ab}=4\pi G\left(T_{ab}-\frac{1}{4}g_{ab}T^{c}{}_{c}\right).
\end{equation}
Thus, in terms of the spinors and $\Phi_{ABC'D'}$ and $\Psi_{ABCD}$, introducing the spinor $T_{AA'BB'}$ equivalent to the tensor $T_{ab}$ and taking into account the first of \eqref{eqn:51}, the field equations become
\begin{subequations}
\label{eqn:53}
\begin{align}
&\nabla^A{}_{B'}\Psi_{ABCD}=4\pi G\nabla^{A'}{}_{(B}T_{CD)A'B'},\\
&\Phi_{ABC'D'}=4\pi G\left(T_{ABC'D'}-\frac{1}{4}\epsilon_{AB}\epsilon_{C'D'}T^{EE'}{}_{EE'}\right),\\
&\Lambda=\frac{\pi G}{3}T^{EE'}{}_{EE'}+\frac{\lambda}{6}.
\end{align}
\end{subequations}
If we assume now that $T_{ab}=0$, i.e. in absence of matter, the previous equations reduce to
\begin{subequations}
\label{eqn:52}
\begin{align}
\label{eqn:52a}
&\nabla^{AA'}\Psi_{ABCD}=0,\\
&\Phi_{ABC'D'}=0,\\
&\Lambda=\frac{1}{6}\lambda.
\end{align}
\end{subequations}
In particular equations \eqref{eqn:52a} is in a certain sense analogous to an actual field equation. It has a significance as being formally identical with the wave equation for a massless (zero rest-mass) spin 2-particle.
\chapter{Conformal Infinity}
\label{chap:3}
\begin{center}
\begin{large}
\textbf{Abstract}
\end{large}
\end{center}
In this part of the work the notion of conformal infinity, originally introduced by Penrose, is developed. The idea, which can be found in \citep{Pen62}, is that if the space-time is considered from the point of view of its conformal structure only, \lq points at infinity\rq\hspace{0.1mm} can be treated on the same basis as finite points. This can be done completing the space-time manifold to a highly symmetrical conformal manifold by the addition of a null cone at infinity, called $\mathscr{I}$. In this chapter we will first build up this kind of construction for Minkowski and Schwarzschild space-times and study their properties. Basing on these results we will give the definition of \lq asymptotic simplicity\rq\hspace{0.1mm}. Owing to their conformal invariance, zero rest-mass fields can be studied on the whole of this conformal manifold, as will be done in the next chapter. It is worth noting that this method, which is deeply geometrical and coordinate-free, and hence more elegant, allows to obtain many results about the gravitational radiation in a very simple and natural way.
\section{Introduction}
\label{sect:3.1}
The notion of \lq conformal infinity\rq\hspace{0.1mm} introduced by Penrose almost fifty years ago is one of the most fruitful concepts within Einstein's theory of gravitation. Most of the modern developments in the theory are based on or at least influenced in one way or another by the conformal properties of Einstein's equations in general or, in particular, by the structure of null infinity: the study of radiating solutions of the field equations and the question of fall-off conditions for them; the global structure of space times; the structure of singularities; conserved quantities; the null hypersurface formulation of General Relativity; the conformal field equations and their importance for the numerical evolution of space-times.\\
To introduce the notion of conformal infinity in a suitable way consider the situation of an isolated gravitating source. Then we might expect space-time to become \lq flat asymptotically\rq\hspace{0.1mm} as we move further away from it. How to formulate the concept of asymptotic flatness is a priori rather vague. Somehow we want to express the fact that the space-time \lq looks like\rq\hspace{0.1mm} Minkowski space-time at \lq large distances\rq\hspace{0.1mm} from the source. Obviously, the investigation of asymptotic properties depends critically on \textit{how} we approach infinity. In order to fix ideas, consider an isolated gravitational system that emits gravitational waves carrying positive energy, as consequence of its varying asymmetry \citep{Stew,Pen67}. To find out what energy we should assign to the waves, it is necessary to measure the masses $m_1$ and $m_2$  respectively before and after the emission and then evaluate the difference $m_1 - m_2$. One way to measure the mass $m_1$ would be to integrate some expression of mass density over a spacelike hypersurface $\mathscr{S}_1$; however in this  case we should take into account the non-local mass density of the gravitational field itself. To avoid those complicated effects a good idea would be to take the hypersurface $\mathscr{S}_1$, which can be chosen to be a 2-surface, to infinity, where the curvature should become  in some way\lq small \rq\hspace{0.1mm} for an asymptotically flat space-time. The mass measure obtained in this way is called \textit{ADM mass}. But if we were taking in the same way a spacelike 2-surface $\mathscr{S}_2$ that intercepts the source after the emission of gravitational radiation, we wouldn't get $m_2$ as our mass measurement, but again $m_1$. This is because such a $\mathscr{S}_2$ would intercept, in addiction to the source worldline, the radiation ones too, giving a total mass measure $m_2+(m_1-m_2) =m_1$. 
To obtain a correct measurement of $m_2$ we should bend $\mathscr{S}_2$ so that the already emitted radiation remains in its past, never intersecting it. As a consequence we pick $\mathscr{S}_2$ to be a null surface at far distances from the source and the corresponding mass we obtain is called \textit{Bondi mass}. The measurement of $m_1$ can be done by using a null 2-surface too. We denote the two null surfaces used by $\mathscr{N}_1$ e $\mathscr{N}_2$.\\
\begin{figure}[h]
\begin{center}
\includegraphics[scale=0.38]{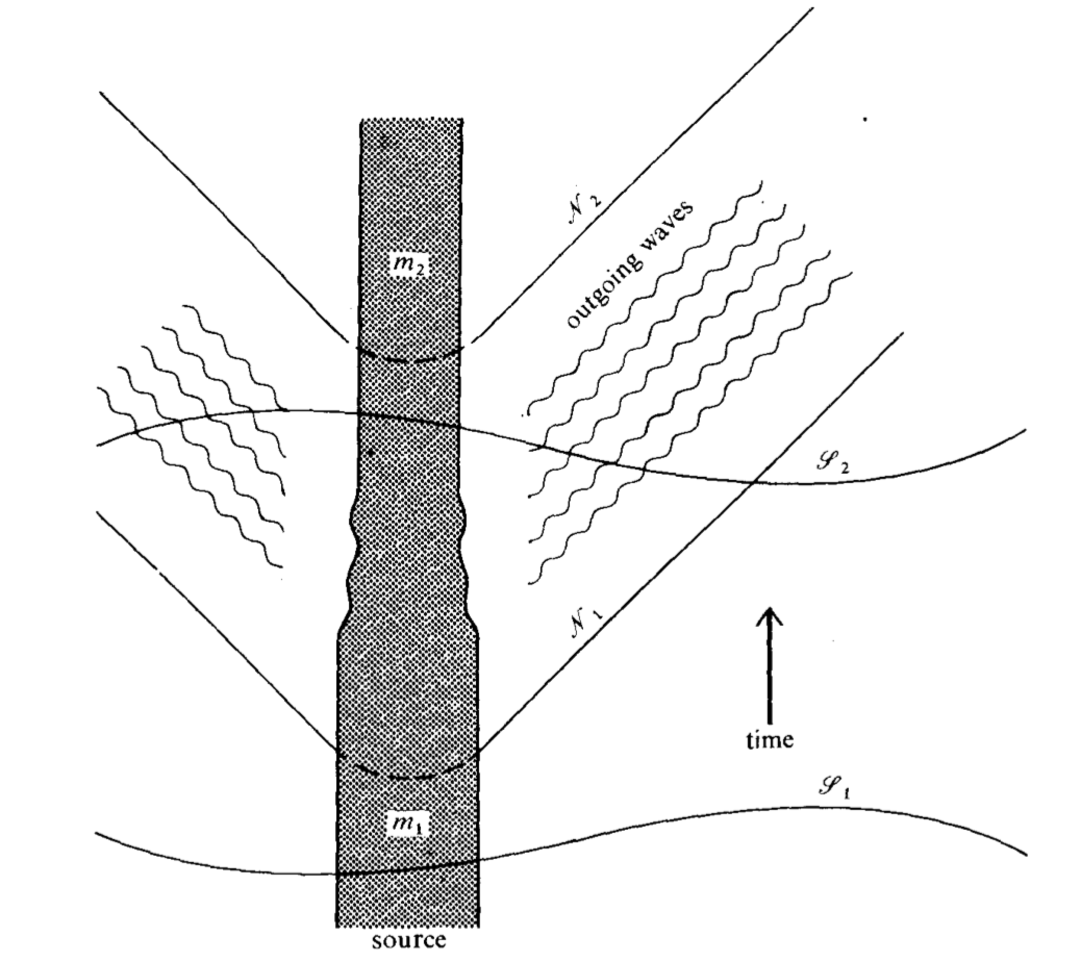}
\hspace{15mm}
\includegraphics[scale=0.5]{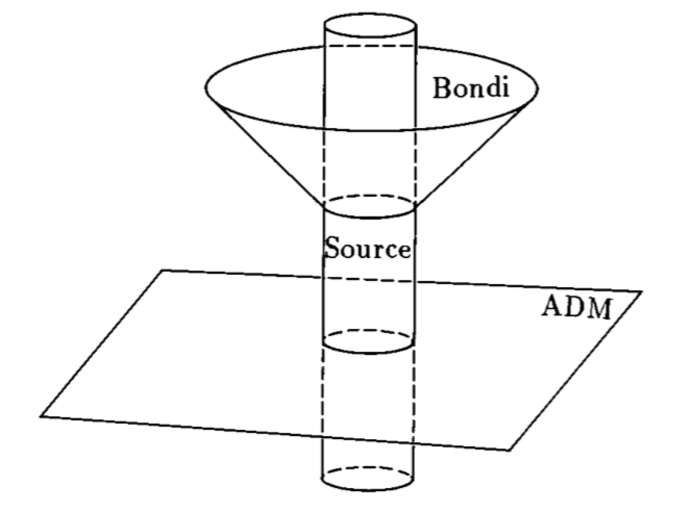}
\caption{On the left, the choice of the hypersurfaces $\mathscr{S}_1$,$\mathscr{S}_2$, $\mathscr{N}_1$ ed $\mathscr{N}_2$. On the right, the difference between ADM and Bondi mass.}
\label{fig:2.1}
\end{center}
\end{figure}
\\The main question now is to understand how to take the above limits to infinity in an appropriate way and to express properly the concept of asymptotic behaviour. While it is certainly possible to discuss asymptotic properties by taking carefully limits as \lq$r\rightarrow\infty$\rq , there is an equivalent way to approach the problem, introduced by Penrose, which allows us to avoid such limit operations and that possesses a more geometrical nature, based on the concept of space-time conformal rescalings. As we will see the main feature of those transformations will be to \lq make infinity finite\rq\hspace{0.1mm} in a certain way.
\section{Conformal Structure of Minkowski and\\Schwarzschild Space-Times\hspace{1cm}}
\label{sect:3.2}
The idea \citep{Pen62,Pen63,Pen64,Pen67} is to construct, starting from the \lq physical space-time\rq\hspace{0.1mm} $(\mathscr{\tilde M},\tilde{g})$, another \lq unphysical space-time\rq\hspace{0.1mm} $(\mathscr{M},g)$ with boundary $\mathscr{I}=\mathscr{\partial M}$ (see Appendix \ref{A}), such that $\mathscr{\tilde M}$ is conformally equivalent to the interior of $\mathscr{M}$ with $g_{ab}=\Omega^2\tilde g_{ab}$, given an appropriate function $\Omega$. The two metric $\tilde g_{ab}$ and $g_{ab}$ define on $\mathscr{\tilde{M}}$ the same null-cone structure. The function $\Omega$ has to vanish on $\mathscr{I}$, so that the physical metric would have to be infinite on it and cannot be extended. The boundary $\mathscr{I}$ can be thought as being at infinity, in the sense that any affine parameter in the metric $\tilde{g}$ on a null geodesic in $\mathscr{M}$ attains unboudedly large values near $\mathscr{I}$. This is because if we consider an affinely parametrized null geodesic $\gamma$ in the unphysical space-time $(\mathscr{M},g)$ with affine parameter $\lambda$, whose equation is\begin{equation*}
\frac{d^2x^a}{d\lambda^2}+\Gamma^a{}_{bc}\frac{dx^b}{d\lambda}\frac{dx^c}{d\lambda}=0,
\end{equation*}
it is easy to see that the corresponding geodesic $\tilde{\gamma}$ in the physical space-time $(\tilde{\mathscr{M}},\tilde{g})$ with affine parameter $\tilde{\lambda}(\lambda)$ is solution of the equation 
\begin{equation*}
\frac{d^2x^a}{d\tilde{\lambda}^2}+\tilde{\Gamma}^a{}_{bc}\frac{dx^b}{d\tilde{\lambda}}\frac{dx^c}{d\tilde{\lambda}}=-\frac{1}{\tilde{\lambda}'}\left(\frac{\tilde{\lambda}''}{\tilde{\lambda}'}+2\frac{\Omega'}{\Omega}\right)\frac{dx^a}{d\tilde{\lambda}},
\end{equation*} 
where a $'$ denotes a $\lambda$ derivative. If we want the parameter $\tilde{\lambda}$ to be affine the right hand side of the above equation must vanish, and hence we must have
\begin{equation*}
\frac{d\tilde{\lambda}}{d\lambda}=\frac{c}{\Omega^2},
\end{equation*}
where $c$ is an arbitrary constant. Since $\Omega=0$ on $\mathscr{I}$, $\tilde{\lambda}$ diverges and hence $\tilde{\gamma}$ never reaches $\mathscr{I}$, which apparently really is at infinity. Thus, from the point of view of the physical metric, the new points (i.e. those on $\mathscr{I}$) are infinitely distant from their neighbours and hence, physically, they represent \lq points at infinity\rq .\\ The advantage in studying the space-time $(\mathscr{M},g)$ instead of $(\tilde{\mathscr{M}},\tilde{g})$ is that the infinity of the latter gets represented by a finite hypersurface $\mathscr{I}$ and the asymptotic properties of the fields defined on it can be investigated by studying $\mathscr{I}$ and the  behaviour of such fields on $\mathscr{I}$. \\
However, there is a large freedom for the choice of the function $\Omega$. Anyway, it will turn out from general considerations that an appropriate behaviour for $\Omega$ is that it should approach zero (both in the past and in the future) like the reciprocal of an affine parameter $\lambda$ on a null geodesic of the space-time considered ($\lambda\Omega\rightarrow \mathrm{constant}$ as $\lambda\rightarrow\pm\infty$).
\begin{figure}[h]
\begin{center}
\includegraphics[scale=0.4]{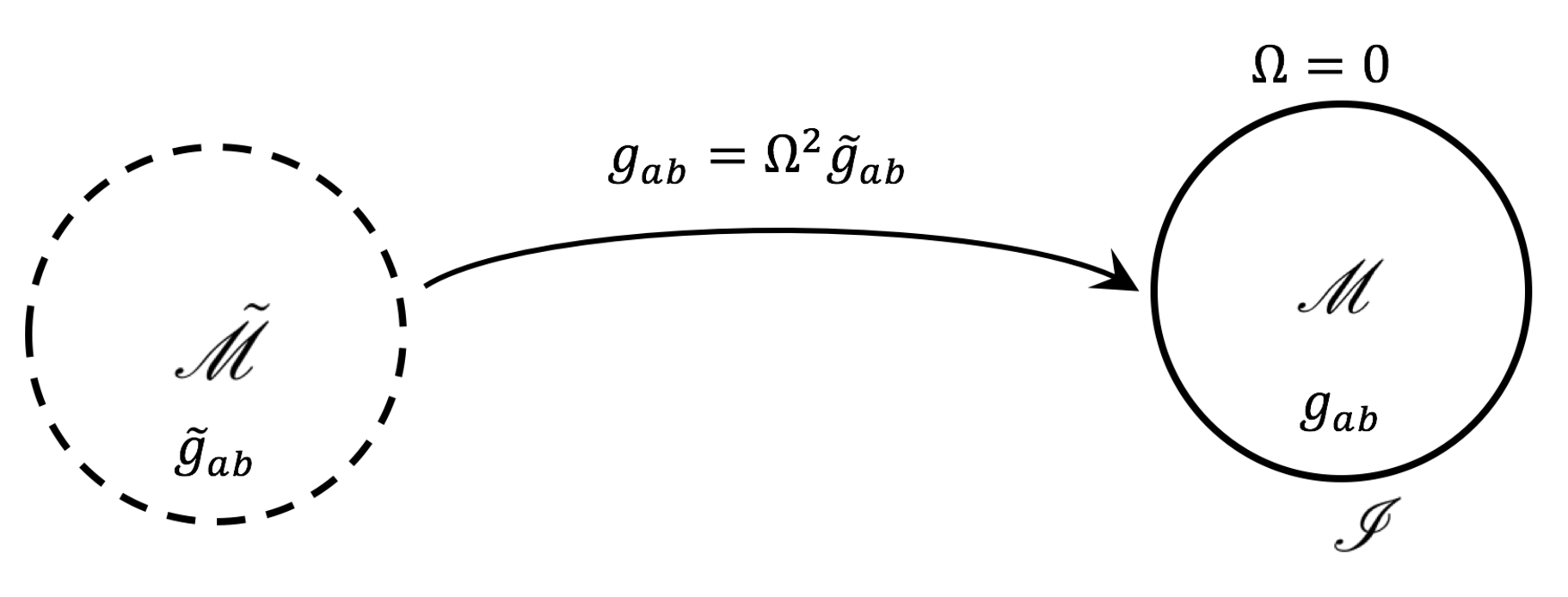}
\caption{Conformal transformation from $\mathscr{\tilde M}$ to $\mathscr{M}$}
\label{fig:2.2}
\end{center}
\end{figure}
\\Consider physical Minkowski space-time in spherical polar coordinates
\begin{equation}
\tilde{g}=dt\otimes dt -dr\otimes dr-r^2\Sigma_2,
\end{equation}
where
\begin{equation}
\label{eqn:sfer}
\Sigma_2=d\theta\otimes d\theta+\sin^2\theta d\phi\otimes d\phi.
\end{equation}
Introduce now the standard retarded and advanced null coordinates $(t,r)\rightarrow(u,v)$ defined by
\begin{equation*}
u=t-r,\hspace{1cm}v=t+r,\hspace{1cm} v\geq u.
\end{equation*}
The coordinates $u$ and $v$ serve as affine parameters into the past and into the future of null geodesics of Minkowski space-time.\\
The metric tensor becomes
\begin{equation*}
\tilde{g}=\frac{1}{2}(du\otimes dv+dv\otimes du)-\frac{1}{4}(v-u)^2\Sigma_2.
\end{equation*}
Consider now the unphysical metric
\begin{equation*}
g=\Omega^2\tilde{g},
\end{equation*}
with the choice
\begin{equation*}
\Omega^2=\frac{4}{(1+u^2)(1+v^2)}.
\end{equation*}
Note that for $u,v\rightarrow\pm\infty$  we have $\Omega u$, $\Omega v\rightarrow \mathrm{constant}$, as pointed out before.\\ Now to interpret this metric it is convenient to introduce new coordinates
\begin{equation*}
u=\tan p,\hspace{1cm}v=\tan q,\hspace{1cm}-\frac{\pi}{2}<p\leq q<\frac{\pi}{2},
\end{equation*}
such that we have
\begin{equation}
\label{eqn:2}
g=2(dp\otimes dq+dq\otimes dp)-\sin^2(p-q)\Sigma_2.
\end{equation}
It is possible to bring the metric \eqref{eqn:2} in a more familiar form by setting
\begin{equation*}
t'=q+p,\hspace{0.8cm}r'=q-p,\hspace{0.8cm}-\pi<t'<\pi,\hspace{0.8cm}-\pi<t'-r'<\pi,\hspace{0.8cm}0<r'<\pi,
\end{equation*}
from which follows 
\begin{equation}
\label{eqn:3}
g=dt'\otimes dt'-dr'\otimes dr'-\sin^2 (r')\Sigma_2.
\end{equation}
It's worth noting that the metric \eqref{eqn:3} is that of \textit{Einstein static universe}, $\mathscr{E}$, the cylinder obtained as product between the real line and the 3-sphere, $S^3\times\mathbb{R}$. However the manifold $\mathscr{M}$ represents just a finite portion of such cylinder.
\begin{figure}[h]
\begin{center}
\includegraphics[scale=0.35]{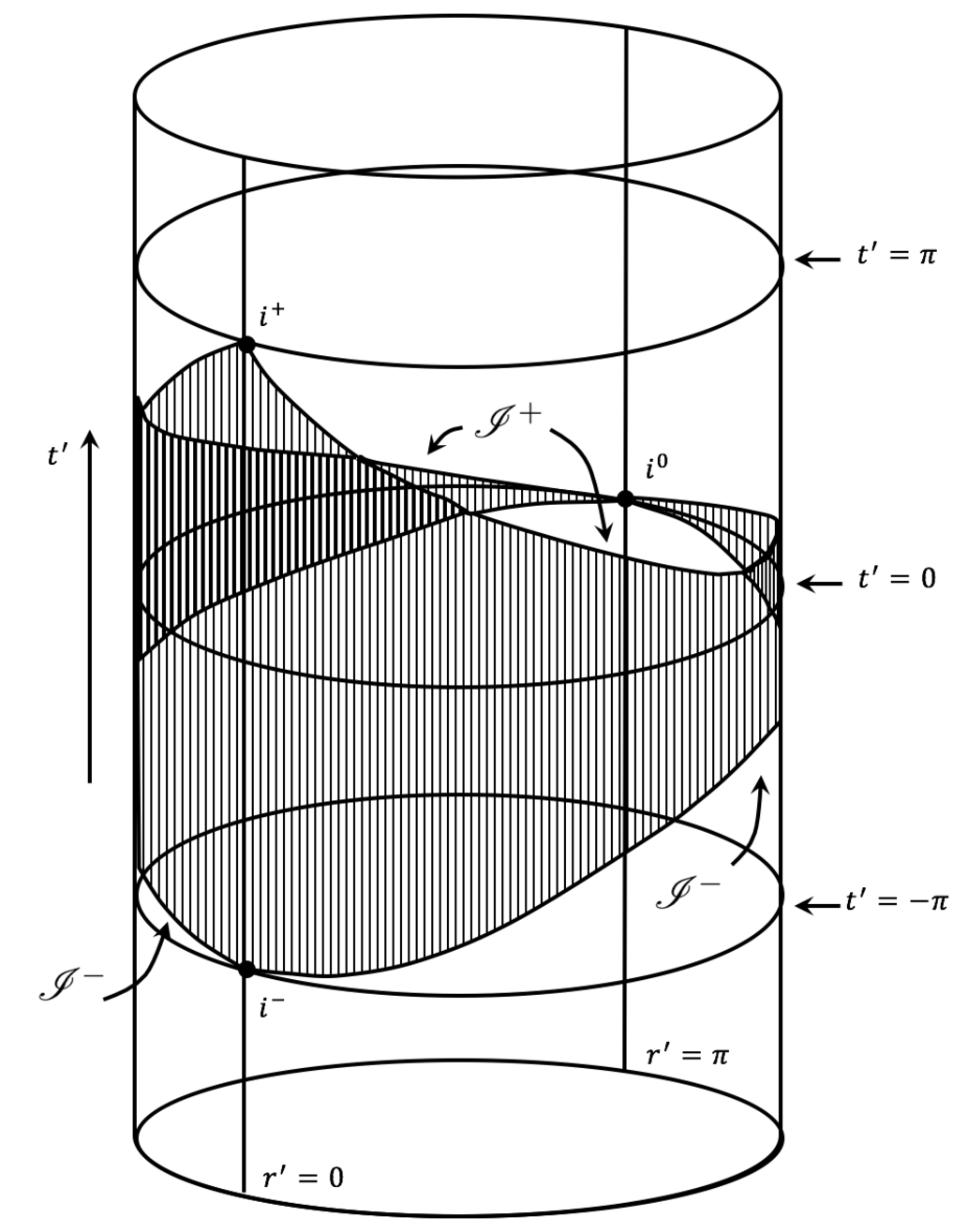}
\caption{The cylinder $\mathscr{E}=S^3\times\mathbb{R}$, of which $\mathscr{M}$ is just a finite portion, delimited by $\mathscr{I^+}$, $\mathscr{I^-}$, $i^+$, $i^-$ and $i^0$.
We note that the $(\theta,\phi)$ coordinates are suppressed, so that each point represents a 2-sphere of radius $\sin r'$. }
\label{fig:2.3}
\end{center}
\end{figure}
\\
The metric \eqref{eqn:2} is defined at $q=\pi/2$ and $p=-\pi/2$: those values correspond to the  infinity of $\mathscr{\tilde M}$ and therefore they represent the hypersurface $\mathscr{I}$. Hence we have defined a conformal structure on  $\mathscr{M}$, whose coordinates are free to move in the range  $-\pi/2\leq p\leq q\leq\pi/2$. The boundary is given by $p=-\pi/2$ or $q=\pi/2$ and the interior of $\mathscr{M}$ is conformally equivalent to Minkowski space-time.\\
We introduce the following points in $\mathscr{M}$:
\begin{itemize}
\item $i^+$, called \textit{future timelike infinity} given by the limits $t\pm r\rightarrow\infty$, $u,v\rightarrow\infty$, $p,q\rightarrow\frac{\pi}{2}$, $t'\rightarrow\pi$, $r'\rightarrow 0$. 
All the images in $\mathscr{M}$  of timelike geodesics terminate at this point;
\item $i^-$, called \textit{past timelike infinity} given by the limits $t\pm r\rightarrow -\infty$, $u,v\rightarrow-\infty$, $p,q\rightarrow-\frac{\pi}{2}$, $t'\rightarrow -\pi$, $r'\rightarrow 0$.  All the images in $\mathscr{M}$ of timelike geodesics originate at this point;
\item $i^0$, called \textit{spacelike infinity} given by the limits $t\pm r\rightarrow\pm \infty$, $u\rightarrow-\infty$, $v\rightarrow\infty$, $p\rightarrow-\frac{\pi}{2}$, $q\rightarrow\frac{\pi}{2}$, $t'\rightarrow0$, $r'\rightarrow \pi$. All spacelike geodesics originate and terminate at this point.
\end{itemize}
We also introduce the following hypersurfaces in $\mathscr{M}$:
\begin{itemize}
\item $\mathscr{I^+}$, called \textit{future null infinity}, is the null hypersurface where all the outgoing null geodesics terminate and is obtained in the following way. Null outgoing geodesics are described by $t=r+c$, with $c$ finite constant, from which $u=t-r=c$ and $v=t+r=2t-c$. Taking the limit $t\rightarrow\infty$ we get $u=c$ and $v=\infty$, hence $q=\pi/2$ and $p=\tan^{-1}c=p_0$ with $-\pi/2< p_0<\pi/2$. In $(t',r')$ coordinates $t'=\pi/2+p_0$ and $r'=\pi/2-p_0$. As $p_0$ runs in its range of values this is a point moving on the segment connecting $i^+$ e $i^0$. All outgoing null geodesics terminate on this segment, described by the equation $t'=\pi-r'$.
\item $\mathscr{I^-}$, called \textit{past null infinity}, is the hypersurface form which all null ingoing geodesics originate. It can be shown that this is given by the region $p=-\pi/2$ and $-\pi/2<q_0<\pi/2$ and is described, in terms of $(t',r')$ coordinates, by the segment of equation $t'=\pi+r'$ connecting $i^{-}$ and $i^{0}$.
\end{itemize}
  Putting 
\begin{equation*}
 f^{\pm}(t',r')=t'\pm r'-\pi,
\end{equation*} 
the two equations defining the hypersurfaces $\mathscr{I}^+$ and $\mathscr{I}^-$ are 
\begin{equation*}
f^{\pm}(t',r')=0,
\end{equation*}
respectively. 
The normal covectors to $\mathscr{I}^+$ and $\mathscr{I}^-$ are 
\begin{equation*}
n^{\pm}_a=\frac{\partial f^{\pm}}{\partial x^a}=(1,\pm1,0,0).
\end{equation*}
Since $g^{ab}n^{\pm}_an^{\pm}_b=0$ it follows that $\mathscr{I}^+$ and $\mathscr{I}^-$ are null hypersurfaces.\\
At this stage, we can build some useful representation of the space-time $\mathscr{M}$. One is a portion of the plane in $(t',r')$ coordinates, that is an example of Penrose diagram. Each point of the Penrose diagram represents a sphere $S^2$, and radial null geodesics are represented by straight lines at $\pm 45^{\circ}$, see Figure \ref{fig:2.3.1}
\begin{figure}[h]
\begin{center}
\includegraphics[scale=0.44]{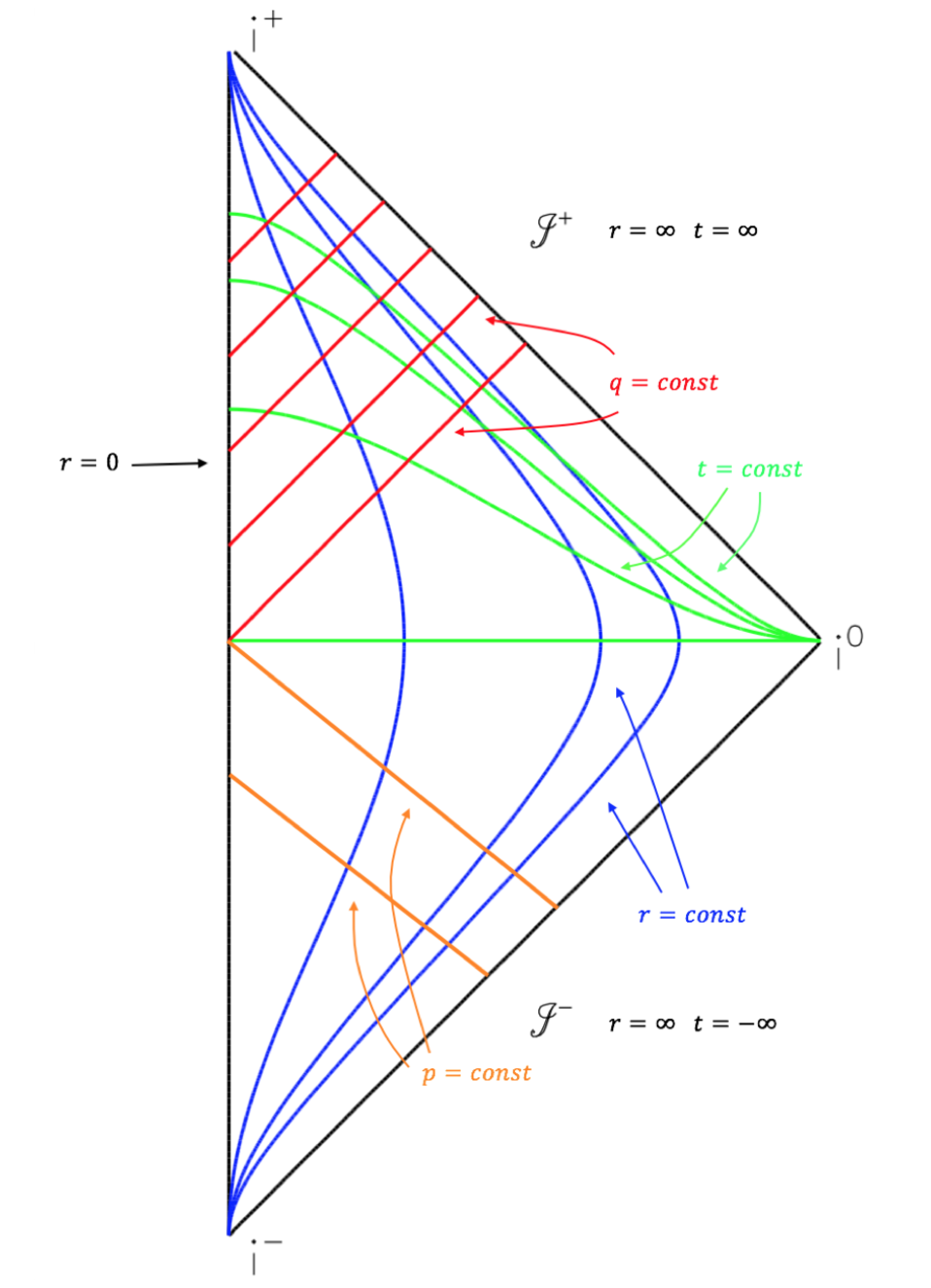}
\caption{A Penrose diagram for $\mathscr{M}$, using $(t',r')$ coordinates.} 
\label{fig:2.3.1}
\end{center}
\end{figure}
\\Another one is depicting $\mathscr{M}$ as a portion of the cylinder $\mathscr{E}=S^3\times E^1$, see Figure \ref{fig:2.3}.\\
One more representation for the Minkowski space-time is furnished by Figure \ref{fig:2.4}.
\begin{figure}[h]
\begin{center}
\includegraphics[scale=0.35]{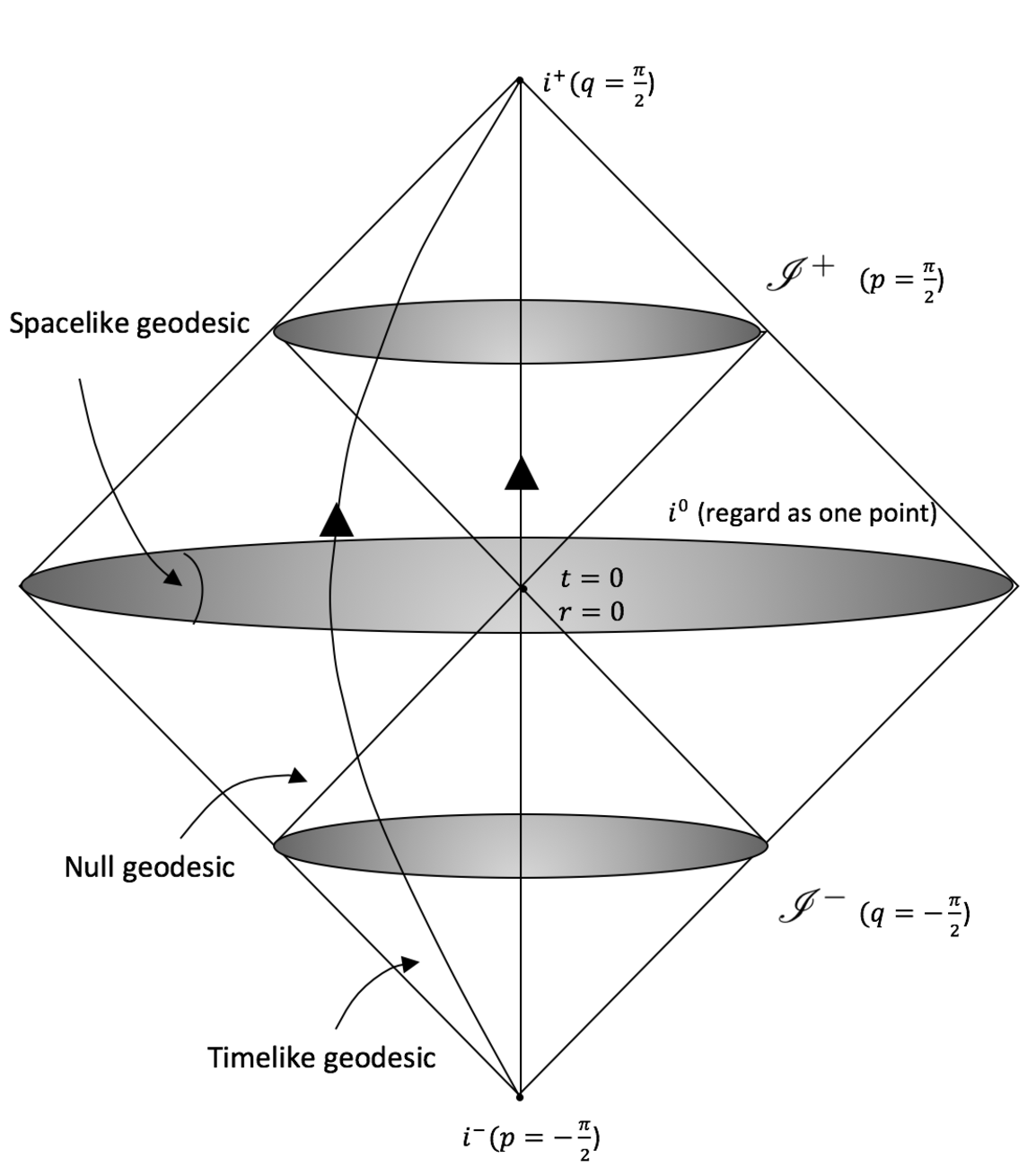}
\caption{This is another useful way of depicting $\mathscr{M}$ as the interior of two cones joined base to base. This picture however is not conformally accurate: in fact $i^0$ appears as an equatorial region whereas it should be a point.}
\label{fig:2.4}
\end{center}
\end{figure}
We note here that  for $\mathscr{M}$ the points $i^+$, $i^-$ and $i^0$ are regular and that $\mathscr{I^-}$ and $\mathscr{I^+}$ both have $S^2\times\mathbb{R}$ topology. Furthermore the boundary of $\mathscr{M}$ is given by $\mathscr{I}=\mathscr{I^+}\cup\mathscr{I^-}\cup i^+\cup i^-\cup i^0$.\\
Consider now Schwarzschild space-time, defined by the metric
\begin{equation}
\label{eqn:4}
\tilde{g}=dt\otimes dt\left(1-\frac{2m}{r}\right)-dr\otimes dr\left(1-\frac{2m}{r}\right)^{-1}-r^2\Sigma_2.
\end{equation}
 Introducing $(u,w)$ coordinates as \begin{equation}
 \label{eqn:5}
u=t-\left[r+2m\ln\left(\frac{r}{2m}-1\right)\right],\hspace{1cm}w=1/r,
\end{equation}
 we have
\begin{equation}
\label{eqn:5.1}
\tilde{g}=du\otimes du\left(1-2mw\right)-(du\otimes dw+dw\otimes du)\frac{1}{w^2}-\frac{1}{w^2}\Sigma_2.
\end{equation}
The first of \eqref{eqn:5} is just the null retarded coordinate, corresponding to a null outgoing geodesic. Note that the coordinate $r^*=r+2m\ln\left(r/2m-1\right)$ in \eqref{eqn:5} is the usual Wheeler-Regge \lq tortoise coordinate\rq\hspace{0.1mm} introduced in \cite{Wheel}.
Consider now the unphysical metric
\begin{equation*}
g=\Omega^2\tilde{g},\hspace{1cm}\Omega=w,
\end{equation*}
\begin{equation}
\label{eqn:6}
g=w^2(1-2mw)du\otimes du-(du\otimes dw+dw\otimes du)-\Sigma_2.
\end{equation}
Schwarzschild space-time, $\mathscr{\tilde M}$, is given by $0<w<1/2m$ because $2m<r<\infty$. We remark that the Schwarzschild solution can easily be extended beyond the event horizon, i.e. $0<r<\infty$ and $0<w<\infty$ because the apparent singular point $r=2m$ of the metric \eqref{eqn:4} is just a coordinate singularity and not a physical one, how can be noticed from \eqref{eqn:5.1}. The metric \eqref{eqn:6} is defined for $w=0$ (i.e. $r=\infty$) and hence for $\mathscr{M}$  we may take the range $0\leq w<1/2m$, such that the hypersurface $\mathscr{I^+}$ is given by $\Omega=w=0$. \\ Re-expressing \eqref{eqn:6} in terms of a null advanced coordinate
\begin{equation*}
v=u+2r+4m\ln\left(\frac{r}{2m}-1\right),
\end{equation*}
corresponding to a null ingoing geodesic we get
\begin{equation}
\label{eqn:7}
g=w^2(1-2mw)dv\otimes dv+(dv\otimes dw+dw\otimes dv)-\Sigma_2.
\end{equation}
By doing this it is now possible to introduce $\mathscr{I^-}$ as the hypersurface of $\mathscr{M}$ described by \eqref{eqn:7} for $w=0$.
\begin{figure}[h]
\begin{center}
\includegraphics[scale=0.51]{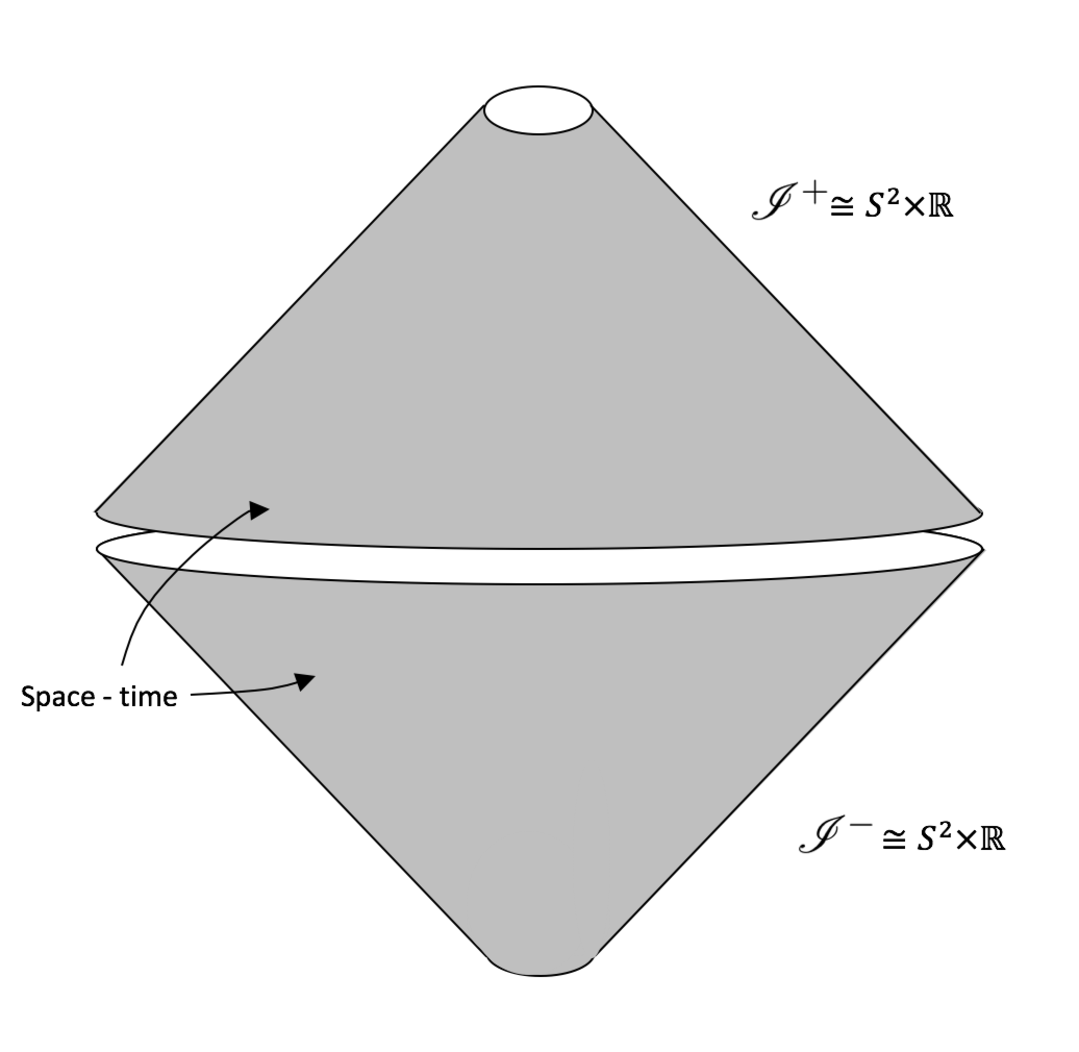}
\caption{Null infinity for the Schwarzschild space-time. Note that $w=0$ corresponds both to $\mathscr{I^+}$ and $\mathscr{I^-}$. The points $i^{\pm}$ and $i^0$ are singular and have been deleted.}
\label{fig:2.5}
\end{center}
\end{figure}
It is easy to check that the hypersurfaces $\mathscr{I}^+$ and $\mathscr{I}^-$, given by the equations $f^{\pm}(w)=w=0$ are again null hypersurfaces. \\The main difference between the Minkowski space-time case emerges from the fact that the points $i^+$, $i^-$ and $i^0$ in the Schwarzschild case are not regular, as could be deduced by the study of the eigenvalues of the Weyl tensor. However it should not be surprising that $i^+$ and $i^-$ turn out to be singular, since the source generating the gravitational field becomes concentrated at these points, at the two ends of its history. Thus we will omit  $i^+$, $i^-$ and $i^0$ from the definition of $\mathscr{I}$, that will just be $\mathscr{I}=\mathscr{I^-}\cup\mathscr{I^+}$. We have two disjoint boundary null hypersurfaces $\mathscr{I^-}$ and $\mathscr{I^+}$ each of which is a cylinder with topology $S^2\times\mathbb{R}$. These null hypersurfaces are generated by rays (given by $\theta$,$\phi=$constant, $w=0$) whose tangents are normals to the hypersurfaces. These rays may be taken to be the $\mathbb{R}'\mathrm{s}$ of the topological product $S^2\times\mathbb{R}$.\\
Take now into account a space-time $(\mathscr{\tilde M},\tilde g)$ with metric tensor \citep{Penrin2,Pen67}
\begin{equation}
\label{eqn:7.1}
\tilde g=r^{-2}Adr\otimes dr+B_i(dx^i\otimes dr+dr\otimes dx^i)+r^2C_{ij}dx^i\otimes dx^j,
\end{equation}
with $A$, $B_i$ and $C_{ij}$ sufficiently differentiable functions (say $C^3$) of $x^{\mu}$, with $x^0=r^{-1}$, on the hypersurface $\mathscr{I}$ defined by $x^0=0$ and in its neighbourhood. If the determinant 
\begin{equation*}
\mathrm{det}\left(\begin{matrix}A & B_i \\ B_j & C_{ij}\end{matrix}\right)
\end{equation*}
 does not vanish, the space-time $(\mathscr{M},g)$ with metric $g=\Omega^2\tilde{g}$, being  $\Omega=r^{-1}$,
\begin{equation*}
g=Adx^0\otimes dx^0-B_i(dx^i\otimes dx^0+dx^0\otimes dx^i)+C_{ij}dx^i\otimes dx^j
\end{equation*}
is regular on $\mathscr{I}$. It is clear that Schwarzschild space-time is just a particular case of this more general situation described by \eqref{eqn:7.1}. Furthermore this metric includes all the Bondi-Sachs type (with which we will deal in chapter \ref{chap:6}). These metrics describe a situation where there is an isolated source (with asymptotic flatness) and outgoing gravitational radiation. Hence a regularity assumption for $\mathscr{I}$ seems a not unreasonable one to impose if we wish to study asymptotically flat space-times and allow the possibility of gravitational radiation. In such situations, therefore, we expect a future-null conformal infinity $\mathscr{I}$ to exist. The choice made for $\Omega$ possesses the important property that its gradient at $\mathscr{I}$, $\partial\Omega/\partial x^{\mu}=(1,0,0,0)$, is not vanishing and hence defines a normal direction to $\mathscr{I}$ ($\mathscr{I}$ being described by the equation $\Omega=0$).\\
Roughly speaking, to say that a space-time is asymptotically flat means that its infinity is \lq similar\rq\hspace{0.1mm} in some way to the Minkowski's one. As a consequence we may expect the conformal structure at infinity of an asymptotically flat space-time to be similar to the one found for the Minkowski case.\\
With those ideas in mind we may now proceed to a rigorous definition of asymptotically simple space-time. However we must also bear in mind that asymptotically flatness is, in itself, a mathematical idealization, and so mathematical convenience and elegance constitute, in themselves, an important criteria for selecting the appropriate idealization.
\section{Aymptotic Simplicity and Weak Asymptotic Simplicity}
\label{sect:3.3}
\begin{defn} $\\ $
\label{defn:AS}
A space-time $(\mathscr{\tilde M},\tilde g)$ is \textit{$k$-asymptotically simple} if some $C^{k+1}$ smooth manifold-with-boundary $\mathscr{M}$ exists, with metric $g$ and smooth boundary $\mathscr{I}=\mathscr{\partial M}$ such that:
\begin{enumerate}
\item $\mathscr{\tilde M}$ is an open submanifold of $\mathscr{M}$; 
\item there exists a real valued and positive function $\Omega>0$, that is $C^k$ throughout $\mathscr{M}$, such that $g_{ab}=\Omega^2\tilde g_{ab}$ on $\mathscr{\tilde{M}}$;
\item $\Omega=0$ and $\nabla_a\Omega\neq 0$ on $\mathscr{I}$;
\item every null geodesic on $\mathscr{M}$ has two endpoints on $\mathscr{I}$.
\end{enumerate}
The space-time $(\mathscr{\tilde M},\tilde g)$ is called \textit{physical space-time}, while $(\mathscr{M},g)$ \textit{unphysical space-time}.
\end{defn}
\begin{defn}{\cite{HawEll}}$\\ $
\label{defn:EAS}
A space-time $(\mathscr{\tilde M},\tilde g)$ is \textit{$k$-asymptotically empty and simple} if it is $k$-asymptotically simple and if satisfies the additional condition
\begin{enumerate}[start=5]
\item $\tilde{R}_{ab}=0$ on an open neighbourhood of $\mathscr{I}$ in $\mathscr{M}$ (this condition can be modified to allow the existence of electromagnetic radiation near $\mathscr{I}$).
\end{enumerate}
\end{defn}
\begin{oss}$\\ $
Note that there are many different definitions of asymptotic simplicity. We used here the one which is due to \cite{Penrin2}, but others which slightly differ from this are possible \citep{Pen67,HawEll,Stew}.
\end{oss}
\begin{oss}$\\ $
Note that although the extended manifold $\mathscr{M}$ and its metric are called \lq unphysical\rq\hspace{0.1mm}, there is nothing unphysical in this construction. The boundary of $\mathscr{\tilde{M}}$ in $\mathscr{M}$ is uniquely determined by the conformal structure of $\mathscr{\tilde{M}}$ and, therefore, it is just as physical as $\mathscr{\tilde{M}}$. 
\end{oss}
Now we try to justify the previous assumptions.\\
Clearly with $1.$, $2.$ and $3.$ we mean to build $\mathscr{I}$ as the null infinity of $(\mathscr{\tilde M},\tilde{g})$, using the results obtained in the Minkowski case, with which must share some properties.  Condition $4.$ ensures that the whole of null infinity is included in $\mathscr{I}$. Furthermore null geodesics in $\mathscr{\tilde M}$ correspond to null geodesics in $\mathscr{M}$ because conformal transformations map null vectors to null vectors: the concept of null geodesic is conformally invariant. Thus we deduce that past and future infinity of any null geodesic in  
$\mathscr{\tilde M}$ are points of  $\mathscr{I}$. Condition $5.$ ensures that the physical Ricci curvature $\tilde{R}_{ab}$ vanishes in the asymptotic region far away from the source of the gravitational field. Finally note how the points $i^+$, $i^-$ and $i^0$ are excluded from the definition of $\mathscr{I}$, since $\mathscr{I}$ is not a smooth manifold at these points. The condition $5.$, together with $3.$, implies that for an asymptotically empty and simple space-time the conformal infinity $\mathscr{I}$ is always null. This is because it is easy to see that the Ricci scalar $R$ of the metric $g_{ab}$ is related to the Ricci scalar $\tilde{R}$ of the metric $\tilde{g}_{ab}$ by 
\begin{equation*}
\tilde{R}=\Omega^{-2}R -6\Omega^{-1}g^{cd}\nabla_c\nabla_d\Omega+3\Omega^{-2}g^{cd}\nabla_c\Omega\nabla_d\Omega,
\end{equation*}
and hence, by multiplying both members by $\Omega^2$, and by evaluating this equation on $\mathscr{I}$ where $\Omega=0$, it follows that $g^{cd}\nabla_c\Omega\nabla_d\Omega=0$. By condition $3.$, since $\nabla_c\Omega\neq 0$, it follows that 
$g^{cd}\nabla_c\Omega\nabla_d\Omega =0$ and thus $\nabla_c\Omega$, the normal vector to $\mathscr{I}$ is null and, by definition, $\mathscr{I}$ is a null hypersurface. Furthermore $R_{ab}$ is related to $\tilde{R}_{ab}$ by 
\begin{equation*}
\tilde{R}_{ab}=R_{ab}-2\Omega^{-1}\nabla_a\nabla_b\Omega-g_{ab}(\Omega^{-1}\nabla_c\nabla^c\Omega-3\Omega^{-2}\nabla_c\Omega\nabla^c\Omega).
\end{equation*}
Since $\nabla_c\Omega$ is null and $R_{ab}$ is defined on $	\mathscr{I}^+$, if condition $5.$ holds, the previous equation on $\mathscr{I}^+$ leads to
\begin{equation*}
2\nabla_a\nabla_b\Omega+g_{ab}\nabla_c\nabla^c\Omega=0. 
\end{equation*}
Contracting with $g^{ab}$ it gives 
\begin{equation*}
\nabla_c\nabla^c\Omega=0\Rightarrow\nabla_a\nabla_b\Omega=0.
\end{equation*}
Hence the normal vector to $\mathscr{I}$ is divergence- and shear-free. Furthermore it will be pointed out in chapter \ref{chap:4} that for any asymptotically simple space-time $\mathscr{I}$ consists of two disconnected components, $\mathscr{I}^+$ on which null geodesics in $\mathscr{\tilde{M}}$ have their future endpoints and $\mathscr{I}^-$ on which they have their past endpoints and both  $\mathscr{I}^+$ and $\mathscr{I}^-$ have the topology $S^2\times\mathbb{R}$, so that the structure of the conformal infinity found for Minkowski space-time is that of any asymptotically simple space-time.
The following theorem is very important because it states the relation that occurs between an asymptotically simple space-time and its causal structure. 
\begin{thm}$\\ $
\label{thm:ASGH}
An asymptotically simple and empty space-time $(\mathscr{\tilde{M}},\tilde g)$ is globally hyperbolic. 
\end{thm}
The proof can be found in \cite{HawEll}. The interpretation of this result is the following. If a space-time has a certain structure at infinity that resembles the one of Minkowski, then it is \lq predictive\rq\hspace{0.1mm} in the sense specified in section \ref{sect:1.7} and is causally stable and thus there cannot occur any causality violations.
The asymptotic simplicity is a very strong assumption for a space-time.
\\However condition $4.$ is difficult to verify in practice and is not even satisfied by some space-times that we would like to classify as asymptotically flat. As an example, for Schwarzschild space-time, it is known that there exist null circular orbits with radius $3m$, and hence don't terminate on $\mathscr{I^+}$. For this reasons condition $4.$ is often too strong and gets replaced by a weaker one that brings to the notion of weakly asymptotically simple space-time. 
\begin{defn} $\\ $
\label{defn:WAS}
A space-time ($\mathscr{\tilde M},\tilde g)$ is \textit{weakly asymptotically simple} if there exists an asymptotically simple space-time $(\mathscr{\tilde M'},\tilde g')$ with associate unphysical space-time $(\mathscr{M'},g')$, such that for a neighbourhood  $\mathscr{H'}$  of $\mathscr{I'}$ in $\mathscr{M'}$ , the region $\mathscr{\tilde M'}\cap\mathscr{H'}$ is isometric to a similar neighbourhood $\mathscr{\tilde H}$ of $\mathscr{\tilde M}$.
\end{defn}
In this way a weakly asymptotically simple space-time possesses the same properties of the conformal infinity of an asymptotically simple one, but the null geodesics do not necessary reach it because it may have other infinities as well. Such space-times are essentially required to be isometric to an asymptotically simple space-time in a neighbourhood of $\mathscr{I}$. A different condition has been proposed by \cite{GerHor78}
In the remainder for asymptotically flat space-time we will mean a weakly asymptotically simple one.
\chapter{Conformal Rescalings and $\mathscr{I}$}
\label{chap:4}
\begin{center}
\begin{large}
\textbf{Abstract}
\end{large}
\end{center}
In this chapter we will make a direct application of the spinor formalism to the conformal technique, both concepts being introduced in the previous parts of the work. We will study the behaviour of the spinor fields under conformal rescalings and, in particular, that of zero rest-mass fields. This is motivated by the fact that the equation of motion of the Weyl spinor, as shown in section \ref{sect:2.9}, is that of a spin 2 zero-rest mass field. Generally, if we find a way to solve an equation on the unphysical space-time and the fields involved behave in a suitable way under conformal rescalings, then we are able to reconstruct the solution on the physical space-time too. In the last section we show the basic properties of $\mathscr{I}$ for an asymptotically flat space-time. These are (under the assumption that the vacuum Einstein's equations hold and hence the cosmological constant equals zero):
\begin{itemize}
\item $\mathscr{I}$ is a null hypersurface;
\item $\mathscr{I}$ has two connected components, $\mathscr{I}^+$ and $\mathscr{I}^-$, each of which has topology $S^2\times\mathbb{R}$;
\item the conformal Weyl tensor $C_{abcd}$ vanishes on $\mathscr{I}$;
\item $\mathscr{I}$ is shear-free.
\end{itemize}
\section{Conformal Rescalings Formulae\hspace{1cm}}
\label{sect:4.1}
Before exploring asymptotic properties we need to establish the relationship between the connection and curvature tensors in the physical and unphysical space-times. 
As we have seen in chapters \ref{chap:2} and \ref{chap:3} a conformal rescaling is a transformation that maps the physical metric $\tilde{g}_{ab}$ to the unphysical one $g_{ab}$ via
\begin{equation}
\label{eqn:56}
g_{ab}=\Omega^2\tilde{g}_{ab},\hspace{1cm}g^{ab}=\Omega^{-2}\tilde{g}_{ab},
\end{equation}
$\Omega$ being a smooth and positive real valued function that vanishes on $\mathscr{I}$. Note that no transformation of points is involved. Thus we can set, consistently with the transformation of the metric and with \eqref{eqn:54}, 
\begin{equation}
\label{eqn:55}
\epsilon_{AB}=\Omega\tilde{\epsilon}_{AB},\hspace{1cm}\epsilon^{AB}=\Omega^{-1}\tilde{\epsilon}^{AB}.
\end{equation}
The only alternative to this choice, $\epsilon_{AB}=-\Omega\tilde{\epsilon}_{AB}$, is not continuous with the identity scaling and is therefore rejected.
We note here that we could have chosen a complex $\Omega$ in \eqref{eqn:55} and hence replaced $\Omega^2$ with $\Omega\bar{\Omega}$ in \eqref{eqn:56}, but this would naturally give rise to a torsion, as discussed in \cite{Pen83}. By choosing a real $\Omega$ we get for the conjugates
\begin{equation*}
\epsilon_{A'B'}=\Omega\tilde{\epsilon}_{A'B'},\hspace{1cm}\epsilon^{A'B'}=\Omega^{-1}\tilde{\epsilon}^{A'B'}.
\end{equation*}
We establish now the following convention: if the kernel letter of a spinor carries a tilde then $\tilde{\epsilon}^{AB}$, $\tilde{\epsilon}_{AB}$ are to be used; however if there is no such tilde then the standard $\epsilon^{AB}$, $\epsilon_{AB}$ are used.\\
As we mentioned above, a spin vector $\tilde{\kappa}^A$ has a definite geometric interpretation (flag and flagpole) which is quite independent of any rescaling. Hence we can suppose a conformal rescaling \eqref{eqn:55} to leave $\tilde{\kappa}^{A}$ unaffected
\begin{equation*}
\kappa^A=\tilde{\kappa}^A.
\end{equation*}
Then for the associate spin co-vector we have 
\begin{equation*}
\kappa_A=\epsilon_{BA}\kappa^B=\Omega\tilde{\epsilon}_{AB}\kappa^B=\Omega\tilde{\kappa}_A.
\end{equation*}
Thus $\tilde{\kappa}_A$ is a \textit{conformal density} of weight 1, i.e. a quantity that gets multiplied by $\Omega^1$ under a rescaling.\\More generally, it is convenient to work with conformal densities of arbitrary weight.
\begin{defn}$\\ $
Define $\tilde{\theta}^{A}$ to be a \textit{conformal density of weight k} if it is to change under a rescaling \eqref{eqn:55} to 
\begin{equation*}
\theta^A=\Omega^k\tilde{\theta}^A.
\end{equation*}
\end{defn}
Normally $k$ is integer or half-integer. Observe that $\tilde{g_{ab}}$, $\tilde{\epsilon}_{AB}$, $\tilde{\epsilon}^{AB}$, $\tilde{g}^{ab}$ have respective conformal weights 2,1,-1,-2. Consequently, whenever a spinor (tensor) index is raised its conformal weight is reduced by unity (2), and whenever a spinor (tensor) index is lowered its weight is increased by unity (2).

We also require now a  spinor covariant derivative for the unphysical space-time, $\nabla_{AA'}$. To find the right expression for such a covariant derivative we need to note that under the conformal rescaling \eqref{eqn:56} the Christoffel symbols
\begin{equation*}
\Gamma^{a}{}_{bc}=\frac{1}{2}g^{ad}(\partial_{b}g_{cd}+\partial_{c}g_{bd}-\partial_{d}g_{bc})
\end{equation*}
 transform as  
\begin{equation*}
\Gamma^{a}{}_{bc}=\tilde{\Gamma}^{a}{}_{bc}+2\Omega^{-1}\delta^a{}_{(b}\nabla_{c)}\Omega-\Omega^{-1}(\nabla_d\Omega)\tilde{g}^{ad}\tilde{g}_{bc}.
\end{equation*}
Some experimentation shows that the only plausible candidate is defined via the rules
\begin{enumerate}
\item $\tilde{\nabla}_{AA'}\phi=\nabla_{AA'}\phi$ for all scalars $\phi$;
\item $\tilde{\nabla}_{AA'}\xi_B=\nabla_{AA'}\xi_B+\Upsilon_{BA'}\xi_A,$ \medskip  \\  
$\tilde{\nabla}_{AA'}\eta_{B'}=\nabla_{AA'}\eta_{B'}+\Upsilon_{BA'}\eta_{B'};$
\item $\tilde{\nabla}_{AA'}\xi^B=\nabla_{AA'}\xi^B-\epsilon_A{}^{B}\Upsilon_{CA'}\xi^C,$ \medskip \\
$\tilde{\nabla}_{AA'}\eta^{B'}=\nabla_{AA'}\eta^{B'}-\epsilon_{A'}{}^{B'}\Upsilon_{AC'}\eta^{C'};$
\item $\tilde{\nabla}_{AA'}\xi^{B_1...B_r}_{C_1...C_s}=\nabla_{AA'}\xi^{B_1...B_r}_{C_1...C_s}+\Upsilon_{C_1A'}\xi^{B_1...B_r}_{A...C_s}+...+\Upsilon_{C_sA'}\xi^{B_1...B_r}_{C_1...
A}- \medskip \epsilon_{A}{}^{B_1}\Upsilon_{DA'}\xi^{D...B_r}_{C_1...C_s}+...-\epsilon_{A}{}^{B_r}\Upsilon_{DA'}\xi^{B_1...D}_{C_1...C_s}$;\medskip\\ 
$\tilde{\nabla}_{AA'}\xi^{B'_1...B'_r}_{C'_1...C'_s}=\nabla_{AA'}\xi^{B'_1...B'_r}_{C'_1...C'_s}+\Upsilon_{C'_1A}\xi^{B'_1...B'_r}_{A'...C'_s}+...+\Upsilon_{C'_sA}\xi^{B'_1...B'_r}_{C'_1...
A'}- \medskip \epsilon_{A'}{}^{B'_1}\Upsilon_{D'A}\xi^{D'...B'_r}_{C'_1...C'_s}+...-\epsilon_{A'}{}^{B'_r}\Upsilon_{D'A}\xi^{B'_1...D'}_{C'_1...C'_s},$
\end{enumerate}
where $\Upsilon_{AA'}=\nabla_{AA'}\ln\Omega$. The formal existence and uniqueness proof is given in \cite{Penrin1}. We know from the property $4.$ of $\nabla$ that
 \begin{equation*}
0=\tilde{\nabla}_{AA'}\tilde{\epsilon}_{BC}=\nabla_{AA'}\tilde{\epsilon}_{BC}+\Upsilon_{A'B}\tilde{\epsilon}_{AC}+\Upsilon_{A'C}\tilde{\epsilon}_{BA}
\end{equation*}
\begin{equation*}
=\nabla_{AA'}\left(\Omega^{-1}\epsilon_{BC}\right)+\Omega^{-1}\Upsilon_{A'B}\epsilon_{AC}+\Omega^{-1}\Upsilon_{A'C}\epsilon_{BA}
\end{equation*}
\begin{equation*}
=-\Omega^{-2}\epsilon_{BC}\nabla_{AA'}\Omega+\Omega^{-1}\nabla_{AA'}\epsilon_{BC}+\Omega^{-1}\Upsilon_{A'B}\epsilon_{AC}+\Omega^{-1}\Upsilon_{A'C}\epsilon_{BA}.
\end{equation*}
Hence, taking into account the definition of $\Upsilon_{AA'}$ and \eqref{eqn:19} we get
\begin{equation*}
\nabla_{AA'}\epsilon_{BC}=\Upsilon_{AA'}\epsilon_{BC}-2\Upsilon_{A'[B}\epsilon_{C]A}=\epsilon_{BC}\left(\Upsilon_{AA'}-\Upsilon_{A'}{}_{D}\epsilon^{D}{}_{A}\right)=0.
\end{equation*}
Thus, having defined the above transformation properties of the covariant derivative, it naturally follows that if the compatibility condition holds for $\tilde{\epsilon}_{AB}$ then it holds for $\epsilon_{AB}$ too.\\ The transformation formulae for the various parts of the Riemann tensor are quite complicated and deriving them from \eqref{eqn:39} and \eqref{eqn:44} is a tedious, but simple exercise. They are
\begin{align}
\label{eqn:57}
&\tilde{\Psi}_{ABCD}=\Psi_{ABCD},\\
\label{eqn:58}
&\tilde{\Lambda}=\Omega^2\Lambda-\frac{1}{4}\Omega\nabla_{CC'}\nabla^{CC'}\Omega+\frac{1}{2}(\nabla_{CC'}\Omega)(\nabla^{CC'}\Omega),\\
\label{eqn:59}
&\tilde{\Phi}_{ABA'B'}=\Phi_{ABA'B'}+\Omega^{-1}\nabla_{A(A'}\nabla_{B')B}\Omega.
\end{align}
In vector form the last two equations are
\begin{align}
\label{eqn:60}
&\tilde{\Lambda}=\Omega^2\Lambda-\frac{1}{4}\Omega\square\Omega+\frac{1}{2}(\nabla_c\Omega)(\nabla^c\Omega),\\
\label{eqn:61}
&\tilde{\Phi}_{ab}=\Phi_{ab}+\Omega^{-1}\nabla_{a}\nabla_b\Omega-\frac{1}{4}\Omega^{-1}(\square\Omega)g_{ab},
\end{align}
where $\square\equiv\nabla_a\nabla^a$.\\
We see from \eqref{eqn:57} that $\Psi_{ABCD}$ is \textit{conformally invariant} as claimed before. It is a measure of the part of the curvature that remains invariant under conformal rescalings. Since $\tilde{\Psi}_{ABCD}=0$ in flat space-time, it vanishes also in conformally flat space-time.\\
Note that due to equation \eqref{eqn:28} we have under conformal rescalings the following behaviours for the Weyl tensor
\begin{equation}
\label{eqn:82}
\tilde{C}_{abcd}=\Omega^{-2}C_{abcd},\hspace{0.8cm}\tilde{C}^{a}{}_{bcd}=C^a{}_{bcd}, \hspace{0.8cm}\mathrm{etc}.
\end{equation}
\section{Zero Rest-Mass Fields}
\label{sect:4.2}
As mentioned above, the first of equations \eqref{eqn:52} describes the dynamic of a massless free field of spin 2. It is useful to discuss zero rest-mass fields of arbitrary (half or half-integer) spin and their radiation properties in general. In fact they exhibit a certain characteristic asymptotic behaviour, that has been called by Sachs the \textit{peeling-off} property \citep{Sachs61,Sachs62}.\\
A field of spin $s=n/2$, where $n$ is a positive integer, is generally described by a spinor $\tilde{\phi}_{AB...L}$ with $n$ indices. Let $\tilde{\phi}_{AB...L}$ be totally symmetric in its $n$ indices:
\begin{equation}
\label{eqn:64}
\tilde{\phi}_{AB...L}=\tilde{\phi}_{(AB...L)}.
\end{equation}
The massless free-field equation for spin $n/2$ is then taken to be
\begin{equation}
\label{eqn:62}
\tilde{\nabla}^{AA'}\tilde{\phi}_{AB...L}=0.
\end{equation}
The Bianchi identity has this form in empty space, with $\tilde{\Psi}_{ABCD}$ taking the place of $\tilde{\phi}$. It is thus a curved-space spin-2 field equation. Similarly, the source free Maxwell equations have this form with $\tilde{\varphi}_{AB}$ (spin 1) taking the place of $\tilde{\phi}$ (equation \eqref{eqn:75}). The Dirac-Weyl equation for the neutrino also falls into this category, with $\tilde{\phi}=\tilde{\nu}_A$ (spin 1/2). \\
To establish the conformal invariance of \eqref{eqn:62}, it is convenient first to re-express that equation in a different form. 
We have 
\begin{equation*}
\tilde{\nabla}_{M'M}\tilde{\phi}_{AB...L}=\tilde{\nabla}_{M'(M}\tilde{\phi}_{A)B...L}+\tilde{\nabla}_{M'[M}\tilde{\phi}_{A]B...L}.
\end{equation*}
Using \eqref{eqn:19} and \eqref{eqn:62} 
\begin{equation*}
\tilde{\nabla}_{M'[M}\tilde{\phi}_{A]B...L}=-\frac{1}{2}\tilde{\epsilon}_{MA}\tilde{\nabla}_{M'}{}^{C}\tilde{\phi}_{CB...L}=0,
\end{equation*}
thus \eqref{eqn:62} is, taking into account the symmetry property \eqref{eqn:64}, equivalent to
\begin{equation}
\label{eqn:63}
\tilde{\nabla}_{M'M}\tilde{\phi}_{AB...L}=\tilde{\nabla}_{M'(M}\tilde{\phi}_{A)B...L}=\tilde{\nabla}_{M'(M}\tilde{\phi}_{AB...L)}.
\end{equation}
Now choose $\tilde{\phi}_{AB...L}$ to be a \textit{conformal density of weight} -1.
\begin{equation}
\label{eqn:83}
\phi_{AB...L}=\Omega^{-1}\tilde{\phi}_{AB...L}.
\end{equation}
Then we have, using the property $2.$ of the covariant derivative of the unphysical space-time,
\begin{equation*}
\Omega\nabla_{MM'}\phi_{AB...L}=\Omega\nabla_{MM'}\left(\Omega^{-1}\tilde{\phi}_{AB...L}\right)
\end{equation*}
\begin{equation}
\label{eqn:65}
=\tilde{\nabla}_{M'M}\tilde{\phi}_{AB...L}-\Upsilon_{MM'}\tilde{\phi}_{AB...L}-\Upsilon_{AM'}\tilde{\phi}_{MB...L}-...-\Upsilon_{LM'}\tilde{\phi}_{AB...M}.
\end{equation}
Now the RHS of \eqref{eqn:65} beyond the first term is automatically symmetric in $MAB...L$. Consequently the LHS is symmetric in $MAB...L$ if and only if \eqref{eqn:63} holds. But that means that \eqref{eqn:63} is conformally invariant, i.e.
\begin{equation*}
\nabla_{M'M}\phi_{AB...L}=\nabla_{M'(M}\phi_{AB...L)}.
\end{equation*}
We can thus give the following:
\begin{prop}$\\ $
If a totally symmetric spinor $\tilde{\phi}_{AB...L}$ is a conformal density of weight -1 then equation \eqref{eqn:62} is conformally invariant.
\end{prop}
Hence for a spin-2 massless field we should have
\begin{equation}
\label{eqn:85}
\phi_{ABCD}=\Omega^{-1}\tilde{\phi}_{ABCD}.
\end{equation}
However from equation \eqref{eqn:57} we see that $\tilde{\Psi}_{ABCD}$ is a conformal density of weight 0, i.e. conformally invariant. Thus equation \eqref{eqn:52}
\begin{equation}
\label{eqn:84}
\tilde{\nabla}^{AA'}\tilde{\Psi}_{ABCD}=0
\end{equation}
is not invariant under rescalings. This is related to the fact that the equations of General Relativity are not conformally invariant. Suppose we have a solution $\tilde{\Psi}_{ABCD}$ of \eqref{eqn:84}. Putting $\tilde{\phi}_{ABCD}=\tilde{\Psi}_{ABCD}$ and transforming according to \eqref{eqn:85} we have
\begin{equation}
\label{eqn:86}
\nabla^{AA'}\phi_{ABCD}=\nabla^{AA'}\left(\Omega^{-1}\tilde{\Psi}_{ABCD}\right)=0,
\end{equation}
from which
\begin{equation*}
\nabla^{AA'}\tilde{\Psi}_{ABCD}=\Upsilon^{AA'}\tilde{\Psi}_{ABCD},
\end{equation*}
i.e. 
\begin{equation*}
\nabla^{AA'}{\Psi}_{ABCD}=\Upsilon^{AA'}{\Psi}_{ABCD}.
\end{equation*}
A conformal rescaling applied to a solution of Einstein's equations will, therefore, generally destroy the satisfaction of the vacuum equations.\\
Furthermore in curved space-times there is an algebraic consistency condition for equation \eqref{eqn:62} that holds for $n>2$, called the Buchdahl constraint \citep{Buch}. To obtain this relations apply $\tilde{\nabla}^B{}_{A'}$ to \eqref{eqn:62}
\begin{equation*}
0=\tilde{\nabla}^B{}_{A'}\tilde{\nabla}^{AA'}\tilde{\phi}_{ABC...L}=\tilde{\nabla}^{(B}{}_{A'}\tilde{\nabla}^{A)A'}\tilde{\phi}_{ABC...L}=\tilde{\square}^{AB}\tilde{\phi}_{ABC...L},
\end{equation*}
where the symmetry of $\tilde{\phi}$ in $AB$ and definition \eqref{eqn:37} have been used. Now using the generalization to many index spinors of first of \eqref{eqn:39} we get 
\begin{equation*}
\tilde{\square}^{AB}\tilde{\phi}_{ABC...L}=-\tilde{X}^{ABE}{}_{A}\tilde{\phi}_{EBC...L}-\tilde{X}^{ABE}{}_{B}\tilde{\phi}_{AEC...L}
\end{equation*}
\begin{equation*}
-\tilde{X}^{ABE}{}_{C}\tilde{\phi}_{ABE...L}-...-\tilde{X}^{ABE}{}_{L}\tilde{\phi}_{ABC...E}.
\end{equation*}
Each of the first two terms of the LHS involves, because of the symmetry of $\tilde{\phi}$, the term $\tilde{X}^{A(BE)}{}_{A}$ that vanishes because of \eqref{eqn:66}. The other terms of LHS involve $\tilde{X}^{(ABE)}{}_{C}=\tilde{\Psi}^{ABE}{}_{C}$, this equation coming from the definition \eqref{eqn:43}. Hence we have, for $n\geq 2$, 
\begin{equation}
\label{eqn:67}
(n-2)\tilde{\phi}_{ABE(C...K}\tilde{\Psi}_{L)}{}^{ABE}=0.
\end{equation}
Thus, if $n>2$, equation \eqref{eqn:67} implies an interconnection between $\tilde{\phi}$ and the conformal curvature spinor $\tilde{\Psi}_{ABCD}$. 
\section{Properties of $\mathscr{I}$}
\label{sect:4.3}
Assume now that Einstein Field Equations
\begin{equation}
\tilde R_{ab}-\frac{1}{2}\tilde R\tilde g_{ab}+\tilde\lambda \tilde g_{ab}=-8\pi G\tilde T_{ab},
\end{equation}
hold in an weakly asymptotically simple space-time $(\mathscr{\tilde M},\tilde g)$ with conformal infinity $\mathscr{I}$, given by the equation 
\begin{equation*}
\Omega=0.
\end{equation*}
Hence, putting
\begin{equation}
\label{eqn:gen}
N_a\equiv-\nabla_a\Omega,
\end{equation}
by the property $3.$ of the definition of asymptotic simplicity it follows that $N_a\neq 0$ on $\mathscr{I}$ and, being orthogonal to each locus $\Omega=\mathrm{constant}$, it constitutes a normal to $\mathscr{I}$ at each of its points.\\
For the time being, we shall allow some massless matter fields in the neighbourhood of $\mathscr{I}$, so that $T^{a}{}_{a}=0$. From equations \eqref{eqn:49} and \eqref{eqn:33} we get 
\begin{equation}
\label{eqn:76}
\tilde{R}=4\tilde{\lambda},\hspace{1cm}\tilde{\lambda}=6\tilde{\Lambda}
\end{equation}
near $\mathscr{I}$ (where \lq near $\mathscr{I}$\rq\hspace{0.1mm} means in $\mathscr{\tilde{H}}\cap\mathscr{\tilde{M}}$, for some neighbourhood $\mathscr{\tilde{H}}$ of $\mathscr{I}$ in the unphysical space-time $\mathscr{M}$).
Now define the quantity $P_{AA'BB'}$
\begin{equation*}
\tilde{P}_{AA'BB'}=\tilde{\Phi}_{AA'BB'}-\tilde{\Lambda}\tilde{\epsilon}_{AB}\tilde{\epsilon}_{A'B'}.
\end{equation*}
Using equation \eqref{eqn:32} and \eqref{eqn:33} we get 
\begin{equation*}
\tilde{P}_{ab}=\frac{1}{12}\tilde{R}\tilde{g}_{ab}-\frac{1}{2}\tilde{R}_{ab}.
\end{equation*}
From \eqref{eqn:76} if we evaluate the trace we get
\begin{equation}
\label{eqn:77}
\tilde{g}^{ab}\tilde{P}_{ab}=\tilde{P}^{a}{}_{a}=-\frac{2}{3}\tilde{\lambda}.
\end{equation}
But from the transformation properties \eqref{eqn:60} and \eqref{eqn:61} 
\begin{equation*}
\tilde{P}_{ab}=P_{ab}+\Omega^{-1}\nabla_a\nabla_b\Omega-\frac{1}{2}\Omega^{-2}\left(\nabla_c\Omega\right)\left(\nabla^c\Omega\right)g_{ab},
\end{equation*}
then evaluating the trace again gives
\begin{equation}
\label{eqn:78}
\tilde{P}^{a}{}_{a}=\Omega^2 P^{a}{}_{a}+\Omega\nabla_a\nabla^a\Omega-2\nabla^a\nabla_a\Omega.
\end{equation}
Comparing \eqref{eqn:77} and \eqref{eqn:78} we obtain on $\mathscr{I}$ (set $\Omega=0$) 
\begin{equation*}
-\frac{2}{3}\tilde{\lambda}=-2\nabla^a\Omega\nabla_a\Omega,
\end{equation*}
and thus
\begin{equation}
\label{eqn:79}
N_aN^a=\frac{1}{3}\tilde{\lambda}.
\end{equation}
Equation \eqref{eqn:79} yields the following: 
\begin{prop} $\\ $
\label{prop:1}
If the trace of the energy tensor vanishes near $\mathscr{I}$, then $\mathscr{I}$ is spacelike, timelike, or null according as $\tilde{\lambda}$ is positive, negative, or zero. 
\end{prop}
When $\mathscr{I}$ is null, it consists naturally of two disconnected pieces, $\mathscr{I^+}$ and $\mathscr{I^-}$, since $\mathscr{\tilde{M}}$ lies locally to the past or future of it. A point of $\mathscr{I}$ lies on $\mathscr{I^+}$ ($\mathscr{I^-}$) if the interior of its past (future) light cone lies in $\mathscr{\tilde{M}}$.\\
There is an important theorem which regards the structure of $\mathscr{I}$ when it is null. 
\begin{thm} $\\ $
In any asymptotically simple space-time for which $\mathscr{I}$ is everywhere null, the topology of each of $\mathscr{I^{\pm}
}$ is given by 
\begin{equation}
\label{eqn:89}
\mathscr{I^+}\cong\mathscr{I^-}\cong S^2\times\mathbb{R}
\end{equation}
and the rays generating $\mathscr{I^{\pm}
}$ can be taken to be the $\mathbb{R}$ factors. 
\end{thm}
We give here a sketch of the proof.
\begin{proof}
If the space-time is asymptotically simple, it is globally hyperbolic (\ref{thm:ASGH}) and thus it has a Cauchy surface $S$. Let $N$ denote the collection of all null geodesic in $\tilde{\mathscr{M}}$ and $p$ be a point of $S$. The set of null geodesics emanating from $p$ is just the set of null directions at $p$ and  hence is, topologically, a 2-sphere $S^2$. Therefore the topology of $N$ is $S\times S^2$ since, by theorem \ref{thm:ngcs}, all null geodesics intersect $S$. Let $q$ be a point of $\mathscr{I}^+$. The set of null geodesics emanating from $q$ is a 2-sphere. However, one of these null geodesics lies in $\mathscr{I}^+$, and so does not enter $\tilde{\mathscr{M}}$. Hence, the null geodesics from $q$ which do not enter $\tilde{\mathscr{M}}$ are, topologically, a 2-sphere minus a single point, i.e. a plane $\mathbb{R}^2$. Hence $N$ is topologically $\mathscr{I}^+\times \mathbb{R}^2$ since all null geodesics intersect $\mathscr{I}^+$. But $\mathscr{I}^+$ is the union of its null generators, and so we have $\mathscr{I}^+=\mathbb{R}\times K$, where $K$ is some 2-manifold. It follows that $N=K\times\mathbb{R}^3$ and thus $S\times S^2=K\times\mathbb{R}^3$. It can be shown that the last equation can hold only if $K$ is topologically a 2-sphere and $S$ is topologically $\mathbb{R}^3$ \citep[see][pg.~99]{Ger71}. We have obtained $\mathscr{I}^+=S^2\times\mathbb{R}$. 
\end{proof}
\begin{oss}$\\ $
The first proof of this theorem, involving sophisticated arguments, is due to \cite{Pen65}. However, as remarked in \cite{NewRP}, the arguments carried out by Penrose were incorrect, and a more rigorous proof can be found in \cite{Ger71} or in \cite{HawEll}.
\end{oss}
\begin{oss}$\\ $
Since $S$, as shown, is topologically $\mathbb{R}^3$, $\mathscr{M}$ is topologically $\mathbb{R}^4$. That is, every asymptotically simple space-time is topologically the same as Minkowski space. 
\end{oss}
Thus $\mathscr{I}$, when it is null, it fairly resembles the infinity of Minkowski space-time, each of $\mathscr{I^{\pm}}$ containing $S^2$ null generators. The occurrence of this situation distinguishes an asymptotically simple space-time in which $\mathscr{I}$ is null from one in which $\mathscr{I}$ is spacelike.\\
In order to proceed further and to obtain more results concerning the structure of $\mathscr{I}$ we shall impose that Einstein vacuum equations hold near $\mathscr{I}$. It can be shown that in other certain cases when $T_{ab}$ does not vanish near $\mathscr{I}$, i.e. the vacuum equations are not appropriate to describe the space-time, the consequences are almost the same, but the derivation is more involved \citep{Pen65}. An example of this is the case in which Einstein-Maxwell equations hold near $\mathscr{I}$. Hence we restrict ourselves to the Einstein vacuum equations, assuming that when matter is present the results still hold for $\mathscr{I}$.\\
From equation \eqref{eqn:61} and the second of \eqref{eqn:52} we have \begin{equation*}
\Phi_{ab}=-\Omega^{-1}\nabla_a\nabla_b\Omega+\frac{1}{4}\Omega^{-1}\left(\square\Omega\right)g_{ab},\hspace{0.8cm}\Phi_{ABA'B'}=-\Omega^{-1}\nabla_{A(A'}\nabla_{B')B}\Omega.
\end{equation*}
Multiplying by $\Omega$ both sides of previous equation, assuming $k\geq 2$ for $\mathscr{\tilde{M}}$ and hence that $\Phi_{ab}$ must be continuous at $\mathscr{I}$, we obtain on $\mathscr{I}$ the \textit{asymptotic Einstein condition} 
\begin{equation}
\label{eqn:81}
\nabla_a\nabla_b\Omega\approx\frac{1}{4}g_{ab}\nabla^c\nabla_c\Omega,\hspace{0.8cm}\nabla_{A'(A}\nabla_{B)B'}\Omega\approx 0,
\end{equation}
where we introduced the \lq weak equality\rq\hspace{0.1mm} symbol $\approx$. Considering two spinor fields $\psi^{...}_{...}$ and $\phi^{...}_{...}$, saying that 
\begin{equation}
\label{eqn:80}
\psi^{...}_{...}\approx\phi^{...}_{...}
\end{equation}
means that $\psi^{...}_{...}-\phi^{...}_{...}=0$ on $\mathscr{I}$. We have to be careful when taking the derivatives of a weak equation, bearing in mind that only tangential derivatives can be relied upon to obtain a new weak equation. Hence from \eqref{eqn:80} we could get correctly
\begin{equation*}
N_{[a}\nabla_{b]}\psi^{...}_{...}\approx N_{[a}\nabla_{b]}\phi^{...}_{...}
\end{equation*}
but not $\nabla_a\psi^{...}_{...}\approx \nabla_{b}\phi^{...}_{...}$. In the remainder we could drop the phrase \lq near $\mathscr{I}$\rq, replacing the equality symbol with the weak equality one. Obviously we have 
\begin{equation*}
\Omega\approx0.
\end{equation*}
Equation \eqref{eqn:79} can be written in the form
\begin{equation}
\label{eqn:88}
N_aN^a\approx\frac{1}{3}\tilde{\lambda},
\end{equation}
and \eqref{eqn:81} as 
\begin{equation}
\label{eqn:94}
\nabla_a N_b\approx\frac{1}{4}g_{ab}\nabla_cN^c,\hspace{0.8cm}\nabla_{A'(A}N_{B)B'}\approx 0.
\end{equation}
In the case of null $\mathscr{I}$ we can put
\begin{equation}
\label{eqn:90}
N^a\approx A o^A\bar{o}^{A'}=Al^a,
\end{equation}
introducing a N-P tetrad of vectors as done in \eqref{eqn:18} and with $A$ a non-zero scalar. Multiplying the first of \eqref{eqn:94} by $m^am^b$ we get 
\begin{equation}
\label{eqn:94a}
m^am^b\nabla_aN_b=Am^b\delta l_b=A\sigma\approx 0\Rightarrow\sigma\approx 0.
\end{equation}
i.e. the null congruence with tangent vector $N^a$ is shear-free. Also, from \eqref{eqn:gen}, this congruence is rotation free, i.e. $\rho=\bar{\rho}$ (this last condition is trivial and follows from the fact that $\mathscr{I}$ is null). \\
We can finally say that when $\mathscr{I}$ is a null hypersurface it is generated by the two-parameter family of integral curves of $N^a$, whose null congruence is shear- and rotation-free. The fact that $\mathscr{I}^+$ and $\mathscr{I}^-$ have this kind of structure is essential for the definition of the BMS group, which will be done later, in chapter \ref{chap:6}.
\begin{thm}$\\ $
\label{thm:EAC}
If the vacuum equations $\tilde{R}_{ab}=\tilde{\lambda}\tilde{g}_{ab}$ hold near $\mathscr{I}$, then the Weyl tensor $C_{abcd}\approx 0.$
\end{thm}
\begin{proof}[Proof:]From the Bianchi identity in spinor form (i.e. equation \eqref{eqn:52}) we have for the physical space-time
$\tilde{\nabla}^{AA'}\tilde{\Psi}_{ABCD}=0$. By \eqref{eqn:86} and \eqref{eqn:57} we obtain\\
\begin{equation*}
\nabla^{AA'}\left(\Omega^{-1}\Psi_{ABCD}\right)=0
\end{equation*}
in $\mathscr{M}$, that implies
\begin{equation}
\label{eqn:93}
\Omega\nabla^{AA'}\Psi_{ABCD}=\Psi_{ABCD}\nabla^{AA'}\Omega,\medskip
\end{equation}
which, by continuity, holds on $\mathscr{I}$. Hence\medskip
\begin{equation}
\label{eqn:87}
\Psi_{ABCD}N^{AA'}\approx 0.\medskip
\end{equation}
From \eqref{eqn:88} we have $N^{AA'}6\tilde{\lambda}^{-1}N_{A'E}=\epsilon_{E}{}^{A}$. Thus, if $\tilde{\lambda}\neq 0$, the matrix $N^{AA'}$ is non-singular and can be inverted so that we have
\begin{equation*}
\Psi_{ABCD}\approx 0
\end{equation*} 
from which the result follows.\\
The case in which $\tilde{\lambda}=0$ is more difficult. The sketch of the proof we are giving here depends upon a global result requiring the topology \eqref{eqn:89}. From equations \eqref{eqn:87} and \eqref{eqn:90} we have
\begin{equation*}
\Psi_{ABCD}\iota^A\approx 0,
\end{equation*}
i.e., by \ref{thm:2}
\begin{equation}
\label{eqn:91}
\Psi_{ABCD}\approx\Psi\iota_A\iota_B\iota_C\iota_D
\end{equation}
for some $\Psi$. Applying $\nabla_{EE'}$ to \eqref{eqn:93} we get
\begin{equation*}
N_{EE'}\nabla^{AA'}\Psi_{ABCD}\approx\nabla_{EE'}\Psi_{ABCD}N^{AA'}+\Psi_{ABCD}\nabla_{EE'}N^{AA'}.
\end{equation*}
Lowering $A'$ and symmetrizing over $A'E'$ it becomes
\begin{equation*}
N{}_{E(E'}\nabla_{A')}{}^{A}\Psi_{ABCD}\approx N^A{}_{(A'}\nabla_{E')E}\Psi_{ABCD}+\Psi_{ABCD}\nabla_{E(E'}N_{A')}{}^{A}.
\end{equation*}
Using the complex conjugate of equation \eqref{eqn:94} the last term vanishes and we obtain
\begin{equation*}
N{}_{E(E'}\nabla_{A')}{}^{A}\Psi_{ABCD}\approx N^A{}_{(A'}\nabla_{E')E}\Psi_{ABCD},
\end{equation*}
hence, raising $E$
\begin{equation*}
N^{[E}{}_{(E'}\nabla_{A')}{}^{A]}\Psi_{ABCD}\approx 0.
\end{equation*}
By theorem \ref{thm:asy} we obtain
\begin{equation*}
\epsilon^{EA}N^{F}_{(E'}\nabla_{A')F}\Psi_{ABCD}\approx 0,
\end{equation*}
thus
\begin{equation*}
N^A{}_{(A'}\nabla_{E')A}\Psi_{EBCD}\approx 0.
\end{equation*}
From symmetry properties of spinors \citep{Penrin1} it can be shown that the previous equation implies
\begin{equation*}
\iota^A\nabla_{AE'}\Psi\iota_E\iota_B\iota_C\iota_D\approx 0.
\end{equation*}
Again it can be shown \citep{Penrin1} that this equation on any spherical cross-section of $\mathscr{I}$ admits the solution $\Psi=0$ on the sphere, whence the proof is complete.
\end{proof}
We refer to the conditions
\begin{equation}
\label{eqn:95}
\Psi_{ABCD}\approx 0,\hspace{1cm}\nabla_{A'(A}N_{B)B'}\approx 0,
\end{equation}
as the \textit{strong asymptotic Einstein condition}.\\
It is now important to mention the following general result.
\begin{lem}$\\ $
\label{thm:EAC1}
Let $(\mathscr{\tilde{M}},\tilde{g})$ be a weakly asymptotically simple space-time and let $\mathscr{H}$ be a neighbourhood of $\mathscr{I}$ in $\mathscr{M}$. Suppose $T^{\mathscr{A}}$ to be a $C^r[\mathscr{H}]$ (with $r\leq k$) spinor that satisfies $T^{\mathscr{A}}\approx 0$. Then there exists a $C^{r-1}[\mathscr{H}]$ spinor $U^{\mathscr{A}}$ such that $\Omega U^{\mathscr{A}}=T^{\mathscr{A}}$.
\end{lem}
An outline of the proof can be found in \cite{Penrin2}, pg. 357.\\
In section \ref{sect:4.2} we noted the difference in the conformal behaviour of the Weyl spinor $\Psi_{ABCD}$ with that of a massless spin-2 field $\phi_{ABCD}$. We \textit{define} now a specific massless spin-2 field
\begin{equation*}
\tilde{\psi}_{ABCD}=\tilde{\Psi}_{ABCD}
\end{equation*}
that, under a rescaling \eqref{eqn:55},  transforms as
\begin{equation}
\label{eqn:96}
\psi_{ABCD}=\Omega^{-1}\tilde{\psi}_{ABCD}=\Omega^{-1}\tilde{\Psi}_{ABCD}=\Omega^{-1}\Psi_{ABCD}.
\end{equation}
In that way the zero rest-mass field equation for $\psi_{ABCD}$ is conformally invariant. We thus have
\begin{equation}
\label{eqn:96.1}
\nabla^{AA'}\psi_{ABCD}=0\hspace{1cm}\mathrm{and}\hspace{1cm}\tilde{\nabla^{AA'}}\tilde{\psi}_{ABCD}=0.
\end{equation}
As it stands $\psi_{ABCD}$ is not defined on $\mathscr{I}$. But theorem \ref{thm:EAC}, lemma \ref{thm:EAC1} and the assumption of smoothness allow the extension of $\psi_{ABCD}$ to the boundary as a smooth field on $\mathscr{M}$. Note that applying $\nabla_{AA'}$ to \eqref{eqn:96} we obtain
\begin{equation}
\label{eqn:97}
-\nabla_{AA'}\Psi_{BCDE}\approx N_{AA'}\psi_{BCDE},
\end{equation}
and thus, on $\mathscr{I}$ we obtain $\psi_{ABCD}$ from the derivative of the Weyl spinor. \\
$\psi_{ABCD}$ will be called the \textit{gravitational field}. From the above discussion $\psi_{ABCD}$ is a genuine spin-2 field with the natural conformal behaviour and can be thought to describe the gravitational effects. In particular its values on the boundary are closely related to the gravitational radiation which escapes from the system under consideration and hence the behaviour of its components at $\mathscr{I}$ is fundamental in the understanding of gravitational radiation theory, which can be analysed with the Sachs peeling property that we are going to investigate for a general spin field with a generic conformal weight. The basic idea, which will be developed in the next chapter is the following: consider a space-time $(\tilde{\mathscr{M}},\tilde{g})$ which allows us to attach a conformal boundary, thus defining an unphysical space-time $(\mathscr{M},g)$ conformally related to the given space-time. Suppose we are also given a solution of a conformally invariant equation on this unphysical manifold, as the first of \eqref{eqn:96.1}. Then there exists a rescaling of that unphysical field with power of the conformal factor, which produces a solution of the equation on the physical manifold. Now suppose that the unphysical field is smooth on the boundary. Then the physical solution will have a characteristic asymptotic behaviour which is entirely governed by the conformal weight of the field, i.e. by the power of the conformal factor used for the rescaling. Thus, the regularity requirement of the unphysical field translates into a characteristic asymptotic fall-off or growth behaviour of the physical field, depending on its conformal weight. Penrose used this idea to show that solutions of the zero rest-mass equations for arbitrary spin on a space-time, which can be compactified by a conformal rescaling, exhibit the peeling property in close analogy to the gravitational case as discovered by Sachs. Take as an example the spin-2 zero rest-mass conformally invariant equation \eqref{eqn:96.1} for the field $\tilde{\psi}_{ABCD}$ which is a conformal density of weight 1. Then, assuming a regularity condition of the unphysical field $\psi_{ABCD}$, we would have that $\tilde{\psi}_{ABCD}$ falls asymptotically as $\Omega$ does.
\chapter{Peeling Properties}
\label{chap:5}
\begin{center}
\begin{large}
\textbf{Abstract}
\end{large}
\end{center}
In this part of the work we carry out a detailed proof of the so-called \textit{peeling property}. This will be done using the spinor formalism only. Hence, many of the notions introduced in chapter \ref{chap:2} will be useful here. The peeling property is a characteristic fall-off behaviour of the fields at infinity, in asymptotically flat space-times. Of particular importance for our purposes is the application of the peeling property to zero rest-mass fields. In fact it allows us to find the behaviour of the Weyl spinor that, as already remarked, can be used to build a spin 2 massless gravitational field. We will show that this behaviour is described by a sum of negative powers of $\tilde{r}$, where $\tilde{r}$ is the affine parameter along null geodesics. 
\section{Introduction}
The peeling, or peeling-off of principal null directions is a generic asymptotic behaviour which was firstly developed for spin 1 and spin 2 fields in the flat case by \cite{Sachs61} and in the asymptotically flat case by \cite{Sachs62}. In his works Sachs proposed an invariant condition for outgoing gravitational waves. The intuitive idea was that at large distances from the source, the gravitational field, i.e. the Riemann tensor of outgoing radiation, should have approximately the same algebraic structure as does the Riemann tensor for a plane wave. As one approaches the source, deviations from the plane wave should appear. Sachs analysed these deviations in detail, using the geometry of congruences of null curves, and obtained rather pleasing qualitative insights into the behaviour of the curvature in the asymptotic regime. The notion was explored further in \cite{New62}, in which the authors formulated what has become known as the Newman-Penrose formalism (that we developed in section \ref{sect:2.7}) that combined the spinor methods, which had been developed earlier in \cite{Pen60}, with the (null-)tetrad calculus already used. In particular they applied their formalism to the problem of gravitational radiation using a coordinate system which was very similar to the one used by \cite{Sachs61} and showed that the single assumption $\tilde{\Psi}_0=o(\tilde{r}^{-4})$ already implied the peeling property as stated by Sachs. In the next years Penrose, guided by the idea of \lq following the field along null directions\rq\hspace{0.1mm}, formulated the concept of the conformal structure of the space-time, and in \cite{Pen65} the conformal method is used, together with the spinor formalism, to deduce the peeling properties. In this work the basic qualitative picture we have today is developed. In the remainder of this chapter we will use Penrose approach (that was further improved in \cite{Penrin2}) to derive the peeling property. The proof of the peeling properties depends strongly on the comparison between two parallelly propagated spin frames along a null geodesic $\gamma$ with respect to the two metrics $\tilde{g}_{ab}$ and $g_{ab}$.\\
\section{Parallelly Propagated Spin-Frames}
\label{sect:5.1}
Let $\gamma$ be a null geodesic in $\mathscr{\tilde{M}}$ reaching $\mathscr{I}$ at the point $p$. We choose a spin frame $(\tilde{o}^A,\tilde{\iota}^A)$ at one point of $\gamma$, where the tangent vector is $\tilde{l}^a=\tilde{o}^A\tilde{\bar{o}}^{A'}$, and propagate the basis parallelly along $\gamma$, as in section \ref{sect:2.8}, via
\begin{equation}
\label{eqn:107}
\tilde{D}\tilde{o}_A=0,
\end{equation}
where $\tilde{D}=\tilde{l}^a\tilde{\nabla}_a$. Ler $\tilde{r}$ be an affine parameter on $\gamma$:
\begin{equation*}
\tilde{D}\tilde{r}=\tilde{l}^a\tilde{\nabla}_a\tilde{r}=1.
\end{equation*}
We can choose whatever conformal transformation for behaviour for $\tilde{o}^A$ and $\tilde{\iota}^A$  provided that $\tilde{\epsilon}^{AB}=\tilde{o}^A\tilde{\iota}^B-\tilde{\iota}^A\tilde{o}^B$ transforms according to \eqref{eqn:55}. If we take
\begin{equation}
\label{eqn:106}
o_A=\tilde{o}_A,\hspace{1cm}o^A=\Omega^{-1}\tilde{o}^A,\hspace{1cm}\mathrm{i.e.}\hspace{1cm}
l_a=\tilde{l}_a,\hspace{1cm}
l^a=\Omega^{-2}\tilde{l}^a,
\end{equation}
then the propagation equation \eqref{eqn:107} is preserved. In fact we have, using rule $2.$ of section \ref{sect:4.1} and equations \eqref{eqn:107} and \eqref{eqn:106}
\begin{equation*}
Do_A=l^b\nabla_bo_A=\Omega^{-2}\tilde{l}^b(\tilde{\nabla}_bo_A-\Upsilon_{B'A}o_B)=\Omega^{-2}\tilde{D}\tilde{o}_A=0.
\end{equation*}
We can complete $o^A$ to a spin frame $(o^A,\iota^A)$ and arrange 
\begin{equation}
\label{eqn:108}
Do^A=0,\hspace{1cm}D\iota^A=0,
\end{equation}
where we set, to preserve \eqref{eqn:55},
\begin{equation}
\label{eqn:113}
\iota^A=\tilde{\iota}^A-\nu\tilde{o}^A,
\end{equation}
with $\nu$ scalar function to be determined. Using the second of \eqref{eqn:108} we get 
\begin{equation*}
0=\tilde{D}\tilde{\iota}^A=\tilde{D}(\iota^A+\nu\tilde{o}^A)=\tilde{l}^b(\nabla_b\iota^A-\epsilon_B{}^{A}\Upsilon_{B'C}\iota^C)+\tilde{o}^A\tilde{D}\nu+\nu\tilde{D}\tilde{o}^A
\end{equation*}
\begin{equation*}
=-\tilde{o}^A\tilde{\bar{o}}^{B'}\iota^C\Omega^{-1}\nabla_{B'C}\ln\Omega+\tilde{o}^AD\nu,
\end{equation*}
hence
\begin{equation}
\label{eqn:114}
D\nu=\Omega^{-2}\eta,
\end{equation}
where
\begin{equation}
\label{eqn:115}
\eta=\iota^C\bar{o}^{B'}\nabla_{B'C}\Omega.
\end{equation}
We compare now affine parameters on $\gamma$. We take $r$ to be an affine parameter on $\gamma$, with origin at $p$ and with tangent vector $l^a=o^A\bar{o}^{A'}$. We have
\begin{equation*}
r\approx 0,\hspace{1cm}Dr=1.
\end{equation*}
The symbol $\approx$ now means equality at $p$. We have of course $\Omega\approx0$ and by condition $3.$ of definition \ref{defn:AS} $D\Omega=d\Omega/dr\neq 0$ on $\mathscr{I}$. Thus we can set
\begin{equation}
\label{eqn:118}
\frac{d\Omega}{dr}\approx-A.
\end{equation}
Note that
\begin{equation}
\label{eqn:116}
-l^aN_a=l^a\nabla_a\Omega=D\Omega\approx-A,
\end{equation}
so this $A$ coincides with the one of \eqref{eqn:90}. Using the $C^k$ smoothness of $\Omega$ it follows that 
\begin{equation}
\label{eqn:110}
\Omega=-Ar-A_2r^2-A_3r^3-...-A_kr^k+o(r^k),
\end{equation} 
$A$, $A_1$, $A_2$...being constant. \\
From the asymptotic Einstein condition \eqref{eqn:81} we have, using equation \eqref{eqn:100} ($r$ is an affine parameter for $\gamma$):
\begin{equation*}
0\approx l^al^b\nabla_a\nabla_b\Omega=l^a\nabla_a(l^b\nabla_b\Omega)=D^2\Omega.
\end{equation*}
Thus in expansione \eqref{eqn:110} we have $A_2=0$.
Consider
\begin{equation*}
\frac{dr}{d\tilde{r}}=\tilde{D}r=\Omega^2Dr=\Omega^2,
\end{equation*}
hence
\begin{equation*}
\tilde{r}=\int\Omega^{-2}dr=\int r^{-2}\left[A+A_3r^2+A_4r^3+...+A_kr^{k-1}+o(r^{k-1})\right]^{-2}dr
\end{equation*}
\begin{equation*}
=\int A^{-2}r^{-2}\left[1+B_2r^2+B_3r^3+...B_{k-1}r^{k-1}+o(r^{k-1})\right]dr
\end{equation*}
\begin{equation}
\label{eqn:111}
=- A^{-2}r^{-1}+C_0+C_1r+...+C_{k-2}r^{k-2}+o(r^{k-2}),
\end{equation}
all $B$ and $C$ coefficients being constant on $\gamma$ and $C_0$ being the constant of integration. For large values of $\tilde{r}$ we can invert \eqref{eqn:111} to obtain
\begin{equation}
\label{eqn:121}
r=-A^{-2}\tilde{r}^{-1}+D_2\tilde{r}^{-2}+D_3\tilde{r}^{-3}+...+D_k\tilde{r}^{-k}+o(\tilde{r}^{-k}),
\end{equation}
which, substituted back into \eqref{eqn:110}, yelds the following expansion for $\Omega$:
\begin{equation}
\label{eqn:112}
\Omega=A^{-1}\tilde{r}^{-1}+E_2\tilde{r}^{-2}+E_3\tilde{r}^{-3}+...+E_k\tilde{r}^{-k}+o(\tilde{r}^{-k}),
\end{equation}
the $D$ and $E$ coefficients being again constant on $\gamma
$.\\
The form we just found for $\Omega$ justifies the assumption made in section \ref{sect:2.2}, according to which the conformal factor $\Omega$ should behave like the reciprocal of an affine parameter along any null geodesic. In fact here we have, by \eqref{eqn:112}, $\Omega\tilde{r}\rightarrow A^{-1}$ as $\tilde{r}\rightarrow\infty$.\\
\section{Comparison Between the Spin-Frames}
\label{sect:5.2}
We proceed now to the comparison between the spin frames. By \eqref{eqn:113} we see that 
\begin{equation}
\label{eqn:122}
\tilde{\iota}^A=\iota^A+\nu\Omega o^A.
\end{equation}
We would like to choose the function $\nu$ such that $\nu\Omega\rightarrow 0$ at the point $p$, so that the two spin frames coincide at that point, i.e. $\tilde{o}^A\approx o^A$ and $\tilde{\iota}^A\approx\iota^A$. Integrating equation \eqref{eqn:114} would give
\begin{equation*}
\nu=\int \eta\Omega^{-2}dr,
\end{equation*}
that, using expansion \eqref{eqn:110}, yields the behaviour
\begin{equation*}
\nu\approx-\eta A^{-2}r^{-1}\approx\eta A^{-1}\Omega^{-1}.
\end{equation*}
Hence we are forced to require $\eta\approx 0$ to ensure that the two spin frames coincide at $p$. By equation \eqref{eqn:115} this condition means that $N^{BB'}$ must be, at $p$, a linear combination of $o^B\bar{o}^{B'}=l^b$ and of $\iota^B\bar{\iota}^{B'}=n^b$, so that its transvection with $\iota^B\bar{o}^{B'}$ vanishes.
We thus write
\begin{equation}
\label{eqn:117}
N^b\approx An^b+\frac{1}{6}\tilde{\lambda} A^{-1}l^b.
\end{equation}
We have $l_bN^b\approx A$ so that $A$ of \eqref{eqn:117} is the same of \eqref{eqn:116}. The second term has been chosen to have that form because  it satisfies  $N^bN_b\approx\tilde{\lambda}/3$, just as in equation \eqref{eqn:88}. \\
With this choice for $N^b$, equation $\eta\approx 0$ is satisfied. Hence it is reasonable to assume $\eta\propto\Omega$, using lemma \ref{thm:EAC1}. With this assumption integration of \eqref{eqn:114} yields, using \eqref{eqn:118}
\begin{equation*}
\nu\propto\int \Omega^{-1}dr=\int A^{-1}\Omega^{-1}d\Omega=-A^{-1}\ln\Omega.
\end{equation*}
This would give, as required, $\nu\Omega\rightarrow 0$, but actually it turns out that $\eta$ vanishes to second order at $p$, i.e. $\eta\approx\Omega^2$, so the logarithm is eliminated. This elimination is very important becauase the presence of a logarithm would destroy the power series we are looking for.\\
Consider in fact the difference between the two members of \eqref{eqn:117}. It has to vanish at $p$. Hence, by \ref{thm:EAC1}, there exists a $C^{k-2}$ ($N^b$ is $C^{k-1}$ being the derivative of $\Omega$, that is $C^{k}$) covector $Q_b$, defined along $\gamma$, such that
\begin{equation}
\label{eqn:120}
\Omega Q_b=N_b-An_b-\frac{1}{6}\tilde{\lambda}A^{-1}l_b.
\end{equation}
Acting on this with $\nabla_c$, we get
\begin{equation*}
-N_cQ_b+\Omega\nabla_cQ_b=\nabla_cN_b-A\nabla_cn_b-\frac{1}{6}A^{-1}l_b\nabla_c\tilde{\lambda}-\frac{1}{6}A^{-1}\tilde{\lambda}\nabla_cl_b,
\end{equation*}
and, after transvecting with $l^c\bar{o}^{B'}\iota^{B}$,
\begin{equation*}
-l^cN_c\bar{o}^{B'}\iota^{B}Q_b+\bar{o}^{B'}\iota^{B}\Omega DQ_b=\bar{o}^{B'}\iota^{B}DN_b-A\bar{o}^{B'}\iota^{B}Dn_b-\bar{o}^{B'}\iota^{B}\frac{1}{6}A^{-1}\tilde{\lambda}Dl_b.
\end{equation*}
The last two terms vanish by \eqref{eqn:108}. Using equation \eqref{eqn:116} we get $l^cN_c\approx A$ and using the Einstein asymptotic condition \eqref{eqn:94} $\iota^B\bar{o}^{B'}DN_b\approx 0$, hence on $p$
\begin{equation*}
-A\iota^B\bar{o}^{B'}Q_b\approx 0.
\end{equation*}
Then by lemma \ref{thm:EAC1} there exists some $C^{k-3}$ function $\mu$ on $\gamma$ such that
\begin{equation}
\label{eqn:119}
-Q_{BB'}\iota^Bo^{B'}=\Omega\mu.
\end{equation}
Using equation \eqref{eqn:119} and \eqref{eqn:120} then we get that \eqref{eqn:115} becomes
\begin{equation*}
\eta=\Omega^2\mu.
\end{equation*}
Hence
\begin{equation*}
\nu=\int\mu dr.
\end{equation*}
that is regular $(C^{k-2})$ in $p$. Thus we can write the expansion 
\begin{equation*}
\nu=c_0+c_1r+c_2r^2+...+c_{k-2}r^{k-2}+o(r^{k-2}), 
\end{equation*}
$c_0$ being the constant of integration, that becomes, using \eqref{eqn:121}
\begin{equation*}
\nu=c_0+d_1\tilde{r}^{-1}+d_2\tilde{r}^{-2}+...+d_{2-k}\tilde{r}^{2-k}+o(\tilde{r}^{2-k}).
\end{equation*}
Thus, using \eqref{eqn:112}
\begin{equation*}
\Omega\nu=\nu_1\tilde{r}^{-1}+\nu_2\tilde{r}^{-2}+...+\nu_{1-k}\tilde{r}^{1-k}+o(\tilde{r}^{1-k}).
\end{equation*}
Eventually equation \eqref{eqn:122} and the second of \eqref{eqn:106} become
\begin{equation}
\label{eqn:123}
\tilde{\iota}^A=\iota^A+\left[\nu_1\tilde{r}^{-1}+\nu_2\tilde{r}^{-2}+...+\nu_{k-1}\tilde{r}^{1-k}+o(\tilde{r}^{1-k})\right]o^A,\hspace{1cm}\tilde{o}^A=\Omega o^{A},
\end{equation}
that are the relations between the spin frames we were looking for.\\
\section{Proof of the Peeling Property}
\label{sect:5.3}
Consider now a conformal density $\tilde{\theta}_{A...HK'...Q'}$ of weight $-w$
\begin{equation}
\label{eqn:125}
 \theta_{A...HK'...Q'}=\Omega^{-w}\tilde{\theta}_{A...HK'...Q'},
\end{equation}
and suppose $\theta_{A...HK'...Q'}$ to be $C^h$ with $0\leq h\leq k-1\geq2$ at $p\in\mathscr{I}$. Under this hypothesis  the following expansion holds
\begin{equation}
\label{eqn:124}
\theta_{{}_{...}}=\theta_{{0}_{...}}+r\theta_{{1}_{...}}+r^2\theta_{{2}_{...}}+...+r^h\theta_{{h}_{...}}+o(r^h).
\end{equation}
Consider a typical component
\begin{equation*}
\tilde{\theta}=\tilde{\theta}_{\tilde{0}...\tilde{0}\tilde{1}...\tilde{1}\tilde{0'}...\tilde{0'}\tilde{1'}...\tilde{1'}}=\tilde{\theta}_{A...BC...HK'...D'E'...Q'}\tilde{o}^{A}...\tilde{o}^{B}\tilde{\iota}^{C}...\tilde{\iota}^{H}\tilde{\bar{o}}^{K'}...\tilde{\bar{o}}^{D'}\tilde{\bar{\iota}}^{E'}...\tilde{\bar{\iota}}^{Q'},
\end{equation*}
with a total number $q$ of zero indices ($0$ and $0'$) and the corresponding component with respect to the frame $(o^A,\iota^A)$, 
\begin{equation*}
\theta=\theta_{0...01...10...01...1}.
\end{equation*}
The expansion \eqref{eqn:124} applies to every component, hence:
\begin{equation*}
\theta=\theta_0+r\theta_1+r^2\theta_2+...+r^h\theta_h+o(r^h),
\end{equation*}
which, in terms of $\tilde{r}$ (using equation \ref{eqn:121}) becomes, renaming the constants of the expansion
\begin{equation}
\label{eqn:126}
\theta=\theta_0+\tilde{r}^{-1}\theta_1+\theta_2\tilde{r}^{-2}+...+\theta_h\tilde{r}^{-h}+o(\tilde{r}^{-h}).
\end{equation}
We have, transvecting \eqref{eqn:125} with the basis vector $\tilde{o}^A$ and $\tilde{\iota}^A$,
\begin{equation*}
\tilde{\theta}=\Omega^w\theta_{\tilde{0}...\tilde{0}\tilde{1}...\tilde{1}\tilde{0'}...\tilde{0'}\tilde{1'}...\tilde{1'}}.
\end{equation*}
As next we insert relations \eqref{eqn:123} in the previous equation:
\begin{equation*}
\tilde{\theta}=\Omega^{w+q}\theta_{A...BC...HK'...D'E'...Q'}o^A...o^B\left(\iota^C+\sum_{i=1}^{k-1}[\nu_i\tilde{r}^{-i}+o(\tilde{r}^{1-k})]o^C\right)
\end{equation*}
\begin{equation*}
...\left(\iota^H+\sum_{j=1}^{k-1}[\nu_j\tilde{r}^{-j}+o(\tilde{r}^{1-k})]o^H\right)
\end{equation*}
\begin{equation*}
\times\bar{o}^{K'}...\bar{o}^{D'}\left(\bar{\iota}^{E'}+\sum_{n=1}^{k-1}[\nu^*_n\tilde{r}^{-n}+o(\tilde{r}^{1-k})]\bar{o}^{E'}\right)...\left(\bar{\iota}^{Q'}+\sum_{m=1}^{k-1}[\nu^*_m\tilde{r}^{-m}+o(\tilde{r}^{1-k})]\bar{o}^{Q'}\right),
\end{equation*}
which, for large $\tilde{r}$'s becomes
\begin{equation*}
\tilde{\theta}\simeq\Omega^{w+q}\theta_{0...01...10'...0'1'...1'}=\Omega^{w+q}\theta.
\end{equation*}
Using the expansions \eqref{eqn:112} and \eqref{eqn:126} we get
\begin{equation*}
\tilde{\theta}\simeq\left(\sum_{j=1}^{k}A_j\tilde{r}^{-j}+o(\tilde{r}^{-k})\right)^{w+q}\left(\sum_{i=0}^{h}\theta_i\tilde{r}^{-i}+o(\tilde{r}^{-h})\right).
\end{equation*}
Again, for large values of $\tilde{r}$ we can neglect the terms beyond the first in the first sum and thus obtain, renaming the parameters of the expansion (which remain constant along $\gamma$)
\begin{equation*}
\tilde{\theta}\simeq\sum_{i=0}^{h}\theta_i\tilde{r}^{-i-w-q}=\sum_{i=w+q}^{w+q+h}\tilde{\theta}_i\tilde{r}^{-i}.
\end{equation*} 
We can eventually state the peeling property for the field $\theta_{A...HK'...Q'}$:
\begin{equation}
\label{eqn:127}
\tilde{\theta}=\sum_{i=w+q}^{w+q+h}\tilde{\theta}_i\tilde{r}^{-i}+o(\tilde{r}^{-w-q-h}),
\end{equation}
where $\tilde{\theta}_i$ is constant along $\gamma$. We note that the leading term in the expansion \eqref{eqn:127} is a multiple of $1/\tilde{r}^{w+q}$.
\section{Applications to Massless Fields}
\label{sect:5.4}
Consider a massless field of arbitrary spin $n/2$, described by a totally symmetric spinor field $\tilde{\phi}_{AB...L}$ with $n$ indices, of weight $w=-1$, so that the massless field equation is conformally invariant. Suppose that $\phi_{AB...L}$ is $
C^{0}$ at $p$. Then if we consider the various components $\tilde{\phi}_0:=\tilde{\phi}_{00...0}$, $\tilde{\phi}_1:=\tilde{\phi}_{10...0}$,...,$\tilde{\phi}_n:=\tilde{\phi}_{11...1}$ their behaviour is the following:
\begin{align*}
&\tilde{\phi}_0=\tilde{\phi}^{(0)}_0\tilde{r}^{-1-n}+o(\tilde{r}^{-1-n}),\\
&\tilde{\phi}_1=\tilde{\phi}^{(0)}_1\tilde{r}^{-n}+o(\tilde{r}^{-n}),
\\
&.\\
&.\\
&.\\
&\tilde{\phi}_n=\tilde{\phi}^{(0)}_n\tilde{r}^{-1}+o(\tilde{r}^{-1}),
\end{align*}
where $\tilde{\phi}^{(0)}_{i}$ are constant on $\gamma$.\\
Consider now electromagnetism. In section \ref{sect:2.7} we have seen that, starting from $\tilde{\varphi}_{AB}$, we could build three scalars, $\tilde{\varphi}_0=\tilde{\varphi}_{AB}\tilde{o}^A\tilde{o}^B$, $\tilde{\varphi}_1=\tilde{\varphi}_{AB}\tilde{o}^A\tilde{\iota}^B$ and $\tilde{\varphi}_2=\tilde{\varphi}_{AB}\tilde{\iota}^A\tilde{\iota}^B$. Using \eqref{eqn:127}  we get that each of this scalars has the following expansion in decreasing powers of $\tilde{r}$:
\begin{align*}
&\tilde{\varphi}_0\simeq \tilde{\varphi}_0^{(0)} \tilde{r}^{-3},\\
&\tilde{\varphi}_1\simeq\tilde{\varphi}_1^{(0)}\tilde{r}^{-2}+\tilde{\varphi}_1^{(1)}\tilde{r}^{-3},\\
&\tilde{\varphi}_2\simeq\tilde{\varphi}_2^{(0)}\tilde{r}^{-1}+\tilde{\varphi}_2^{(1)}\tilde{r}^{-2}+\tilde{\varphi}_2^{(2)}\tilde{r}^{-3}.
\end{align*}
Furthermore, using \eqref{eqn:128} we can write
\begin{equation}
\label{eqn:131}
\tilde{\varphi}_{AB}=\frac{\tilde{\varphi}_2^{(0)}\tilde{o}_A\tilde{o}_B}{\tilde{r}}+\frac{\tilde{\varphi}_2^{(1)}\tilde{o}_A\tilde{o}_B-2\tilde{\varphi}_1^{(0)}\tilde{o}_{(A}\tilde{\iota}_{B)}}{\tilde{r}^2}+o(\tilde{r}^{-2}),
\end{equation}
which schematically is 
\begin{equation*}
\tilde{\varphi}_{AB}=\frac{[N]_{AB}}{\tilde{r}}+\frac{[I]_{AB}}{\tilde{r}^2}+o(\tilde{r}^{-2}).
\end{equation*}
The leading term proportional to $\tilde{r}^{-1}$ is the radiation part of the electromagnetic field. It is of type N or null, according to the Petrov classification, and is the component of $\tilde{\varphi}_{AB}$ totally contracted with $\tilde{\iota}$.\\
In the case of gravity we have, following the same procedure of \eqref{eqn:128}, 
\begin{equation*}
\tilde{\Psi}_{ABCD}=\tilde{\Psi}_{0}\tilde{\iota}_A\tilde{\iota}_B\tilde{\iota}_C\tilde{\iota}_D-3!\tilde{\Psi}_1\tilde{o}_{(A}\tilde{\iota}_B\tilde{\iota}_C\tilde{\iota}_{D)}+3!\tilde{\Psi}_2\tilde{o}_{(A}\tilde{o}_B\tilde{\iota}_C\tilde{\iota}_{D)}
\end{equation*}
\begin{equation}
\label{eqn:129}
-3!\tilde{\Psi}_3\tilde{o}_{(A}\tilde{o}_B\tilde{o}_C\tilde{\iota}_{D)}+\tilde{\Psi}_4\tilde{o}_A\tilde{o}_B\tilde{o}_C\tilde{o}_D.
\end{equation}
Due to the peeling property \eqref{eqn:127} we have the following expansions:
\begin{align*}
&\tilde{\Psi}_0\simeq\frac{\tilde{\Psi}_0^{(0)}}{\tilde{r}^5},\\
&\tilde{\Psi}_1\simeq\frac{\tilde{\Psi}_1^{(0)}}{\tilde{r}^4}+\frac{\tilde{\Psi}_1^{(1)}}{\tilde{r}^5},\\
&\tilde{\Psi}_2\simeq\frac{\tilde{\Psi}_2^{(0)}}{\tilde{r}^3}+\frac{\tilde{\Psi}_2^{(1)}}{\tilde{r}^4}+\frac{\tilde{\Psi}_2^{(2)}}{\tilde{r}^5},\\
&\tilde{\Psi}_3\simeq\frac{\tilde{\Psi}_3^{(0)}}{\tilde{r}^2}+\frac{\tilde{\Psi}_3^{(1)}}{\tilde{r}^3}+\frac{\tilde{\Psi}_3^{(2)}}{\tilde{r}^4}+\frac{\tilde{\Psi}_3^{(3)}}{\tilde{r}^5},\\
&\tilde{\Psi}_4\simeq\frac{\tilde{\Psi}_4^{(0)}}{\tilde{r}}+\frac{\tilde{\Psi}_4^{(1)}}{\tilde{r}^2}+\frac{\tilde{\Psi}_4^{(2)}}{\tilde{r}^3}+\frac{\tilde{\Psi}_4^{(3)}}{\tilde{r}^4}+\frac{\tilde{\Psi}_4^{(4)}}{\tilde{r}^5}.
\end{align*}\\ \\ \\
Hence we can express \eqref{eqn:129} as 
\begin{equation*}
\tilde{\Psi}_{ABCD}\simeq\frac{\tilde{\Psi}_4^{(0)}\tilde{o}_A\tilde{o}_B\tilde{o}_C\tilde{o}_D}{\tilde{r}}+\frac{\tilde{\Psi}_4^{(1)}\tilde{o}_A\tilde{o}_B\tilde{o}_C\tilde{o}_D-3!\tilde{\Psi_3^{(0)}}\tilde{o}_{(A}\tilde{o}_B\tilde{o}_C\tilde{\iota}_{D)}}{\tilde{r}^2}
\end{equation*}
\begin{equation*}
+\frac{\tilde{\Psi}_4^{(2)}\tilde{o}_A\tilde{o}_B\tilde{o}_C\tilde{o}_D-3!\tilde{\Psi}_3^{(1)}\tilde{o}_{(A}\tilde{o}_B\tilde{o}_C\tilde{\iota}_{D)}+3!\tilde{\Psi}_2^{(0)}\tilde{o}_{(A}\tilde{o}_B\tilde{\iota}_C\tilde{\iota}_{D)}}{\tilde{r}^3}
\end{equation*}
\begin{equation}
\label{eqn:132}
+\frac{\tilde{\Psi}_4^{(3)}\tilde{o}_A\tilde{o}_B\tilde{o}_C\tilde{o}_D-3!\tilde{\Psi}_3^{(2)}\tilde{o}_{(A}\tilde{o}_B\tilde{o}_C\tilde{\iota}_{D)}+3!\tilde{\Psi}_2^{(1)}\tilde{o}_{(A}\tilde{o}_B\tilde{\iota}_C\tilde{\iota}_{D)}-3!\tilde{\Psi}_1^{(0)}\tilde{o}_{(A}\tilde{\iota}_B\tilde{\iota}_C\tilde{\iota}_{D)}}{\tilde{r}^4},
\end{equation}
which can be written schematically as 
\begin{equation}
\label{eqn:130}
\tilde{\Psi}_{ABCD}=\frac{[N]_{ABCD}}{\tilde{r}}+\frac{[III]_{ABCD}}{\tilde{r}^{2}}+\frac{[II]_{ABCD}}{\tilde{r}^3}+\frac{[I]_{ABCD}}{\tilde{r}^4}+o(\tilde{r}^{-4}).
\end{equation}
\begin{figure}[t]
\begin{center}
\includegraphics[scale=0.5]{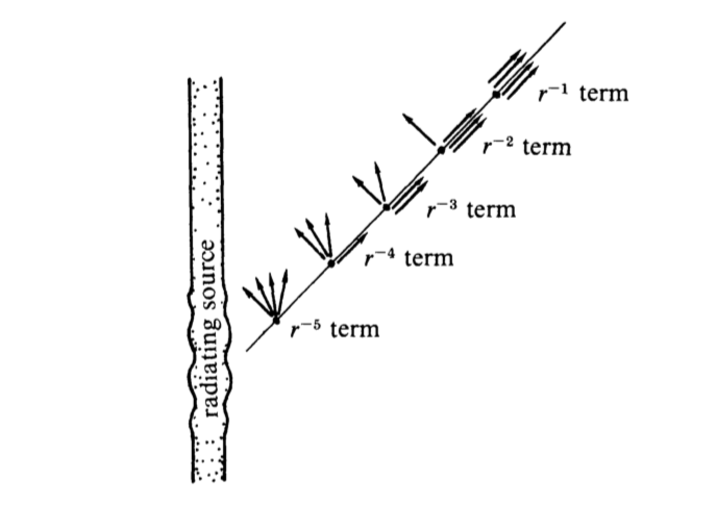}
\caption{The Sachs peeling property, expressed by equation \eqref{eqn:130}, illustrating the multiplicity of the radial PND of the Weyl curvature for the various terms in the expansion in negative powers of $\tilde{r}$.}
\label{fig:2.8}
\end{center}
\end{figure}
In particular, $\tilde{\Psi}_4^{(0)}$, which may be thought of as describing the gravitational radiation field, can be identified with that component of $\tilde{\psi}_{ABCD}$ on $\mathscr{I}$ which is  totally contracted with $\tilde{\iota}$. Note that the part of the curvature which has no relation to the null direction of the outgoing geodesic goes as $1/\tilde{r}^5$, and hence to that order the curvature is not related to the geodesic.\\
This argument can be carried out more generally. In fact, we can always decompose a $n/2$ spin massless field $\tilde{\phi}_{AB...L}$ as
\begin{equation*}
\tilde{\phi}_{AB...L}=\sum_{i=0}^{n-1}\sum_{k=0}^ic_{ik}\tilde{\phi}_{n-k}^{(i-k)}\underbrace{\tilde{o}_A...\tilde{o}_B}_{n-k}\underbrace{\tilde{\iota}_C...\tilde{\iota}_L}_{k}\tilde{r}^{-(i+1)}+o(\tilde{r}^{-n})=\sum_{i=0}^{n-1}\tilde{\phi}_i{}_{AB...L}\tilde{r}^{-(i+1)}+o(\tilde{r}^{-n}),
\end{equation*}
where
\begin{equation*}
\tilde{\phi}_i{}_{AB...L}\equiv\sum_{k=0}^ic_{ik}\tilde{\phi}_{n-k}^{(i-k)}\underbrace{\tilde{o}_A...\tilde{o}_B}_{n-k}\underbrace{\tilde{\iota}_C...\tilde{\iota}_L}_{k}.
\end{equation*}
This is the generalization of equations \eqref{eqn:131} and \eqref{eqn:132} to an arbitrary spin field.\\
Consider now the transvection of $\tilde{\phi}_i{}_{AB...L}$ with $i+1$ $\tilde{o}$'s:
\begin{equation}
\label{eqn:133}
\tilde{\phi}_i{}_{AB...C}\underbrace{{}_{D...L}}_{i+1}\underbrace{\tilde{o}^D...\tilde{o}^L}_{i+1}=0.
\end{equation}
This transvection always vanishes because in $\tilde{\phi}_i{}_{AB...L}$ there are at least $n-i$ $\tilde{o}$'s and hence in \eqref{eqn:133} there must be at least one term of the form $\tilde{o}_A\tilde{o}^A=0$. From equation \eqref{eqn:133} and proposition \ref{prop1} we see that $\tilde{\phi}_i{}_{AB...L}$, and hence the term of $\tilde{\phi}_{AB...L}$ that behaves like $\tilde{r}^{-(i+1)}$, has always at least $n-i$ principal null directions pointing along the direction of $\gamma$, i.e. along $\tilde{l}^a$.  In particular the $\tilde{r}^{-1}$ part, what we call the radiation field, is always null.
\chapter{Bondi-Metzner-Sachs Group}
\label{chap:6}
\begin{center}
\begin{large}
\textbf{Abstract}
\end{large}
\end{center}
Minkowski space-time has an interesting and useful group of isometries. But for a general space-time, the isometry group is simply the identity and hence provides no significant informations. Yet symmetry groups have important role to play in physics; in particular, the Poincar\'e group, describing the isometries of Minkowski space-time plays a role in the standard definitions of energy-momentum and angular-momentum. For this reason alone it would seem to be important to look for a generalization of the concept of isometry group that can apply in a useful way to suitable curved space-times. The curved space-times that will be taken into account are the ones that suitably approach, at infinity, Minkowski space-time. In particular we will focus on asymptotically flat space-times. In this chapter the concept of \textit{asymptotic symmetry group} of those space-times will be studied. In the first two sections we derive the asymptotic group, which is referref to as \lq BMS\rq, by Bondi, Metzner and Sachs, following the classical approach which is basically due to \cite{Bondi62,Sachs62,Sachs1}. This is essentially the group of transformations between coordinate systems of a certain type in asymptotically flat space-times. In the third section the derivation is made following arguments developed by Penrose, which involve the conformal structure (which we carried out in chapter \ref{chap:3}), and is thus more geometrical and fundamental \citep{NewPen66,Pen72,Penrin2,Stew}. In the remaining sections we will discuss the properties of the BMS group, such as its group structure, its algebra and the possibility to obtain as its subgroup the Poincar\'e group, as we may expect.
\section{Introduction}
\label{sect:6.1}
The importance of the concept of energy within a physical theory, if introduced correctly, arises from 
the fact that it is a conserved quantity in time and hence a very useful tool. 
Thus, in general relativity one of the most interesting questions is related to the meaning of gravitational energy.\\
Starting from any vector $J^a$ that satisfies a local conservation equation, that can be put in the form
\begin{equation}
\label{eqn:0}
\nabla_aJ^a=0,
\end{equation}
one can deduce an integral conservation law which states that the integral over the boundary 
$\partial\mathscr{D}$ of some compact region $\mathscr{D}$ of the flux of the vector $J^a$ across this 
boundary necessarily vanishes. In fact, using Gauss' theorem we have
\begin{equation}
\label{eqn:0.1}
\int_{\partial\mathscr{D}} J^ad\sigma_a=\int_{\mathscr{D}}\nabla_aJ^adv=0.
\end{equation}
Now we know that in General Relativity the energy-momentum tensor $T_{ab}$  satisfies the local conservation law
\begin{equation}
\label{eqn:1}
\nabla_aT^{ab}=0,
\end{equation}
which follows directly from the Einstein field equations. However from \eqref{eqn:1} we cannot deduce any 
conservation law. This is because in this case the geometric object to integrate over a 4-volume (as on the 
right-hand side of \eqref{eqn:0.1}) would be a vector and we can not take the sum of two vectors at different 
points of a manifold. This picture is ameliorated if space-time possesses symmetries, 
i.e. Killing vectors. If $K^a$ is a Killing vector,
\begin{equation*}
\nabla_{(a}K_{b)}=0,
\end{equation*}
we may build the vector 
\begin{equation*}
P^a=T^{ab}K_b,
\end{equation*}
that satisfies \eqref{eqn:0}, since
\begin{equation*}
\nabla_a P^a=\nabla_aT^{ab}K_b+T^{ab}\nabla_aK_b=0.
\end{equation*}
The second term vanishes because $T^{ab}$ is symmetric and so $T^{ab}\nabla_aK_b=T^{ab}\nabla_{(a}K_{b)}=0$.
Therefore the presence of Killing vectors for the metric leads to an integral conservation law. In  
flat Minkowski space-time we know that there are 10 Killing vectors:
\begin{equation*}
\textbf{L}_{\alpha}=\frac{\partial}{\partial x^{\alpha}},\hspace{2.2cm}(\alpha=0,1,2,3)
\end{equation*}
\begin{equation*}
\textbf{M}_{\alpha\beta}=e_{\alpha}x^{\alpha}\frac{\partial}{\partial x^{\beta}}
-e_{\beta}x^{\beta}\frac{\partial}{\partial x^{\alpha}},\hspace{2cm}
(\mathrm{no}\hspace{1.3mm}\mathrm{summation};\alpha,\beta =0,1,2,3)
\end{equation*}
where $e_{\alpha}$ is +1 if $\alpha=0$ and -1 if $\alpha=1,2,3$. The first four generate space-time translations 
and the second six \lq rotations\rq\hspace{0.1mm} in space-time (these are just the usual ten generators of 
the inhomogeneous Lorentz group). One may use them to define ten vectors $P^a_{\alpha}$ and 
$P^a_{\alpha\beta}$ which will obey \eqref{eqn:0}. We can think of $\textbf{P}_0$ as representing the flow of 
energy, and $\textbf{P}_1$, $\textbf{P}_2$ and $\textbf{P}_3$ as the flow of the three components of linear 
momentum. The $\textbf{P}_{\alpha\beta}$ can be interpreted as the flow of angular momentum. If the metric is 
not flat there will not, in general, be any Killing vectors. It is worth noting that the diffeomorphism group 
has, for historical reasons, frequently been invoked as a possible substitute for the Poincar\'e group for a 
general space-time. However, it is not really useful in this context, being much too large and preserving only 
the differentiable structure of the space-time manifold rather than any of its physically more important property.
\\ However, one could introduce in a suitable neighbourhood of a point $q$ normal coordinates $\{x^a\}$ so that 
the components $g_{ab}$ of the metric are $e_a\delta_{ab}$ (no summation) and that the components of 
$\Gamma^a{}_{bc}$ are zero at $q$. One may take a neighbourhood $\mathscr{D}$ of $q$ in which $g_{ab}$ and 
$\Gamma^a{}_{bc}$ differ from their values at $q$ by an arbitrary small amount. Then 
$\nabla_{(a}L_{\alpha\hspace{1mm}b)}$ and $\nabla_{(a}M_{\alpha\beta\hspace{1mm}b)}$ will not exactly vanish in 
$\mathscr{D}$, but will in this neighbourhood differ from zero by an arbitrary small amount. Thus\begin{equation*}
\int_{\partial\mathscr{D}}P^b_{\alpha}d\sigma_b\hspace{1cm}\mathrm{and}\hspace{1cm}
\int_{\partial\mathscr{D}}P^b_{\alpha\beta}d\sigma_b
\end{equation*}
will still be zero in the first approximation.
Hence the best we can get from \eqref{eqn:1} is an approximate integral conservation law, if we integrate 
over a region whose typical dimensions are very small compared with the radii of curvature involved in 
$R_{abcd}$. We can interpret this by thinking the space-time curvature as giving a non-local contribution 
to the energy-momentum, that has to be considered in order to obtain a correct integral conservation law.\\
From the above discussion we deduce that no exact symmetries can be found for a generic space-time. 
However, if we turn to the concept of asymptotic symmetries and we apply it to asymptotically flat space-times, 
we will see that the picture is not so bad and that we can still talk about the Poincar\'e group. The basic 
idea, developed in the remainder of the chapter, is that, since we are taking into account asymptotically 
flat space-times, we may expect that by going to \lq infinity\rq\hspace{0.1mm} one might acquire the Killing 
vectors necessary for stating integral conservation laws.  
\section{Bondi-Sachs Coordinates and Boundary Conditions} 
\label{sect:6.2}
Consider the Minkowski metric
\begin{equation*}
g=\eta_{ab}dx^a\otimes dx^b=dt\otimes dt-dx\otimes dx-dy\otimes dy-dz\otimes dz.
\end{equation*}
We introduce new coordinates 
\begin{equation}
\label{eqn:135}
u=t-r,\hspace{1cm}r\cos\theta=z,\hspace{1cm}r\sin\theta e^{i\phi}=x+iy,
\end{equation}
in terms of which the Minkowski metric takes the form
\begin{equation}
\label{eqn:137}
g=du\otimes du+du\otimes dr+dr\otimes du-r^2(d\theta\otimes d\theta+\sin^2\theta d\phi\otimes d\phi).
\end{equation}
which can also be written as
\begin{equation}
\label{eqn:139}
g=du\otimes du+du\otimes dr+dr\otimes du-r^2q_{AB}dx^A\otimes dx^B, 
\end{equation}
where 
\begin{equation*}
q_{AB}=\left(\begin{matrix} 1& 0\\ 0& \sin^2\theta\end{matrix}\right),\hspace{0.5cm}A,B,...=2,3.
\end{equation*}
Note that $q_{AB}$ represents the metric on the unit sphere. The coordinate $u$ is called \textit{retarded time}.\\
We proceed to the interpretation of the coordinates \eqref{eqn:135}. The hypersurfaces given by the equation $u=\mathrm{const}$ are null hypersurfaces, since their normal co-vector $k_a=\nabla_a u$ is null. They are everywhere tangent to the light-cone. Note that it is a peculiar property of null hypersurfaces that their normal direction is also tangent to the hypersurface. The coordinate $r$ is such that the area of the surface element $u=\mathrm{const}$, $r=\mathrm{const}$ is $r^2\sin\theta d\theta d\phi$. Define a ray as the line with tangent $k^a=g^{ab}\nabla_bu$. Then the scalars $\theta$ and $\phi$ are constant along each ray.\\Now we would like to introduce for a generic metric tensor a set of coordinates $(u,r,x^A)$ which has the same properties as the ones of \eqref{eqn:135}. These coordinates are known as \textit{Bondi-Sachs coordinates} \citep{Bondi62,Sachs1,Sachs62}. The hypersurfaces $u=\mathrm{const}$ are null, i.e. the normal co-vector $k_a=\nabla_au$ satisfies $g^{ab}(\nabla_au)(\nabla_bu)=0$, so that $g^{uu}=0$, and the corresponding future-pointing vector $k^a=g^{ab}\nabla_bu$ is tangent to the null rays. Two angular coordinates $x^A$, with $A,B,...=2,3$, are constant along the null rays, i.e. $k^a\nabla_ax^A=g^{ab}(\nabla_au)\nabla_bx^A=0$, so that $g^{uA}=0$. The coordinate $r$, which varies along the null rays, is chosen to be an areal coordinate such that $\mathrm{det}[g_{AB}]=r^4\mathrm{det}[q_{AB}]$, where $q_{AB}$ is the unit sphere metric associated with the angular coordinates $x^A$, e.g. $q_{AB}=\mathrm{diag}(1,\sin^2\theta)$ for standard spherical coordinates $x^A=(\theta,\phi)$. The contravariant components $g^{ab}$ and covariant components $g_{ab}$ are related by $g^{ac}g_{cb}=\delta^a_b$, which in particular implies $g_{rr}=0$ (from $\delta_{ur}=0$) and $g_{rA}=0$ (from $\delta_{uA}=0$). See Figure \ref{fig:3.1}.
\begin{figure}[h]
\begin{center}
\includegraphics[scale=0.49]{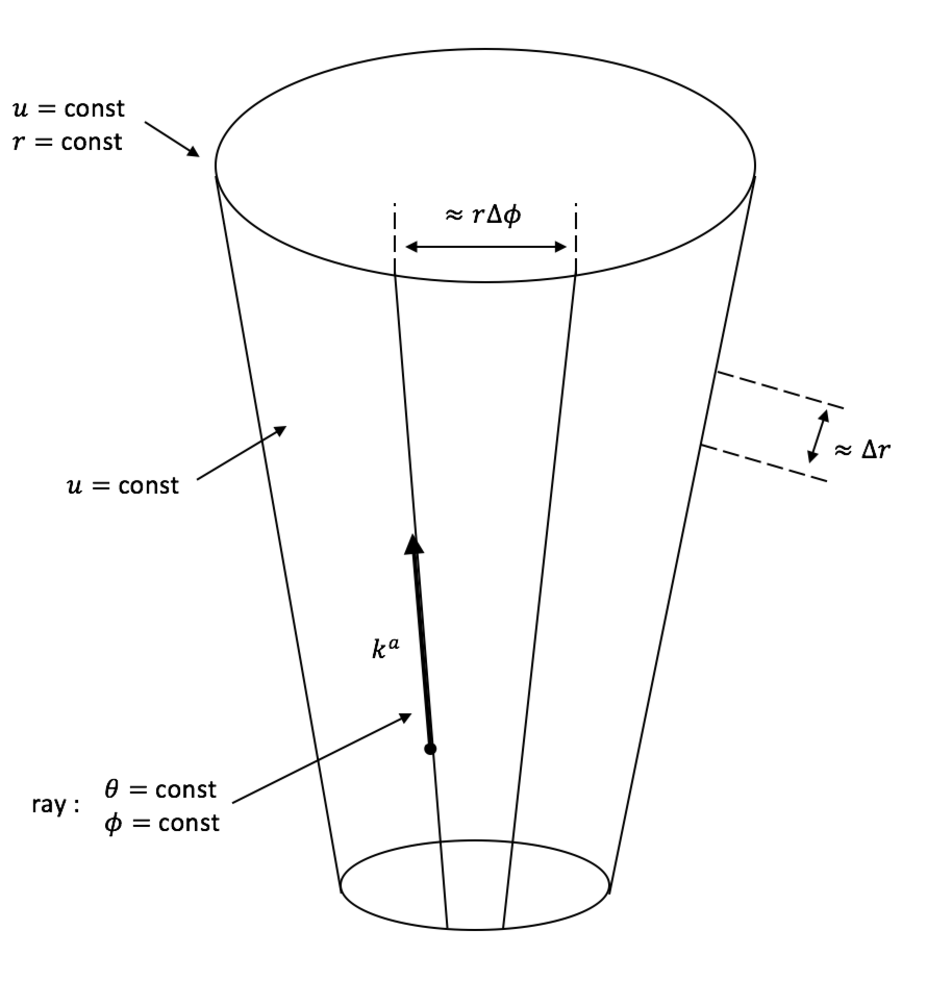}
\caption{The Bondi-Sachs coordinate system. The coordinates $u$, $r$, and $\phi$ and the vector $k^a$ are shown in the hypersurface $\theta=\mathrm{const}$.}
\label{fig:3.1}
\end{center}
\end{figure}
\\It can be shown \citep{Bondi62} that the metric takes the form 
\begin{equation}
\label{eqn:136}
g=g_{ab}dx^a\otimes dx^b=e^{2\beta}\frac{V}{r}du\otimes du+e^{2\beta}(du\otimes dr+dr\otimes du)
\end{equation}
\begin{equation*}
+g_{AB}(dx^A-U^Adu)\otimes(dx^B-U^Bdu),
\end{equation*}
where
\begin{equation}
\label{eqn:168}
g_{AB}=r^2h_{AB},\hspace{1cm}\mathrm{det}[h_{AB}]=h(x^A).
\end{equation}
Using Jacobi's formula for the derivative of a determinant for a generic matrix $g_{\mu\nu}$,
\begin{equation*}
 \partial_{\rho}\mathrm{det}[g_{\mu\nu}]=\partial _{\rho}g=gg^{\mu\nu}\partial_{\rho}g_{\mu\nu},
\end{equation*}
we have from the second of \eqref{eqn:168}
\begin{equation}
\label{eqn:169}
\partial_u h=0\Rightarrow h^{AB}\partial_uh_{AB}=0,\hspace{1cm}\partial_r h=0\Rightarrow h^{AB}\partial_rh_{AB}=0.
\end{equation}
We also have
\begin{equation*}
g^{ur}=e^{-2\beta},\hspace{0.5cm}g^{rr}=-\frac{V}{r}e^{-2\beta},\hspace{0.5cm}g^{rA}
=U^Ae^{-2\beta},\hspace{0.5cm}g^{AB}=-\frac{1}{r^2}h^{AB}.
\end{equation*}
A suitable representation for $h_{AB}$ is the following:\\
\begin{equation}
\label{eqn:141}
h_{AB}=\left(\begin{matrix}\cosh2\delta e^{2\gamma} & \sin\theta\sinh2\delta \\ \\ \\\sin\theta\sinh2\delta 
& \sin^2\theta\cosh2\delta e^{-2\gamma}\end{matrix}\right)\Rightarrow \mathrm{det}[h_{AB}]=\sin^2\theta.
\end{equation}
Here $V$, $\beta$, $U^A$, $\gamma$ and $\delta$ are any six functions of the coordinates. The form 
\eqref{eqn:136} holds if and only if $(u,r,\theta,\phi)$ have the properties stated above. Note that this 
form differs from the original form of Sachs \cite{Sachs62} by the transformation $\gamma\rightarrow(\gamma+\delta)/2$ 
and $\delta\rightarrow(\gamma-\delta)/2$. The original axisymmetric Bondi metric \cite{Bondi62} with rotational 
symmetry in the $\phi$-direction was characterized by $\delta=U^{\phi}=0$ and $\gamma=\gamma(u,r,\theta)$, 
resulting in a metric with reflection symmetry $\phi\rightarrow-\phi$ so that it is not suitable for describing 
an axisymmetric rotating star. \\
The next step is to write down the Einstein vacuum field equations in the above coordinate system in order to find 
the equations that rule the evolution of the six arbitrary functions on which the metric depends. 
As shown in \cite{Sachs62} or \cite{BSF} the Einstein vacuum field equations
\begin{equation*}
G_{ab}=R_{ab}-\frac{1}{2}Rg_{ab}=0,
\end{equation*}
separate into the \textit{Hypersurface equations},
\begin{equation*}
G^u_a=0,
\end{equation*}
and the \textit{Evolution equations},
\begin{equation*}
G_{AB}-\frac{1}{2}g_{AB}g^{CD}G_{CD}=0.
\end{equation*}
The former determines $\beta$ along the null rays ($G^u_r=0$), $U^A$ ($G^u_A=0$) and $V$ ($G^u_u=0$), 
while the latter gives informations about the retarded time derivatives of the two degrees of freedom 
contained in $h_{AB}$. Usually one requires the following conditions:
\begin{enumerate}
\item For any choice of $u$ one can take the limit $r\rightarrow\infty$ along each ray;
\item For some choice of $\theta$ and $\phi$ and the above choice of $u$ the metric \eqref{eqn:136} 
should approach the Minkowski metric \eqref{eqn:137}, i.e. 
\begin{equation}
\label{eqn:138}
\lim_{r\to\infty}\beta=\lim_{r\to\infty}U^A=0,\hspace{0.5cm}\lim_{r\to\infty}\frac{V}{r}
=1,\hspace{0.5cm}\lim_{r\to\infty}h_{AB}=q_{AB}.
\end{equation}
Note that these conditions, as pointed out in \cite{Sachs62}, are rather unsatisfactory from a geometrical 
point of view. They will be completely justified later, using the method of the conformal structure, introduced by Penrose;
\item Over the coordinate ranges $u_0\leq u\leq u_1$, $r_0\leq r\leq\infty$, $0\leq\theta\leq\pi$ 
and $0\leq\phi\leq 2\pi$ all the metric functions can be expanded in series of $r^{-1}$. 
\end{enumerate}
Using the Einstein equations with these assumptions it can be shown \cite{Sachs62,BSF} that the 
following asymptotic behaviours hold:
\begin{subequations}
\label{eqn:145}
\begin{align}
&V=r-2M+O(r^{-1}),\\
&h_{AB}=q_{AB}+\frac{c_{AB}}{r}+O(r^{-2}),\\
&\beta=-\frac{c^{AB}c_{AB}}{32r^2}+O(r^{-3}),\\
&U^A=-\frac{D_{B}c^{AB}}{2r^2}+O(r^{-3}),
\end{align}
\end{subequations}
i.e. the metric \eqref{eqn:136} admits the asymptotic expansion
\begin{eqnarray}
\label{eqn:140}
g&=& du\otimes du+du\otimes dr+dr\otimes du-r^2q_{AB}dx^A\otimes dx^B
\nonumber \\
&-& \frac{2M}{r}du\otimes du-\frac{c^{AB}c_{AB}}{4r^2}(du\otimes dr+dr\otimes du)
\nonumber \\
&-& rc_{AB}dx^A\otimes dx^B-\frac{D_Fc^{F}_{A}}{2}(du\otimes dx^A+dx^A\otimes du)+...
\end{eqnarray}
Here the function $M=M(u,\theta,\phi)$ is called the \textit{mass aspect}, $c_{AB}=c_{AB}(u,\theta,\phi)$ 
represents the $O(r^{-1})$ correction to $h_{AB}$ and $D_A$ is the covariant derivative with respect to the metric 
on the unit $2$-sphere, $q_{AB}$ \citep{Haw17}. Capital letters A, B,... can be raised and lowered with respect 
to $q_{AB}$. In carrying out the $1/r$ expansion of the field equations the covariant derivative $D_A$ corresponding 
to the metric $h_{AB}$ is related to the covariant derivative ${\cal D}_A$ 
corresponding to the unit sphere metric $q_{AB}$ by 
\begin{subequations}
\label{eqn:160}
\begin{equation}
D_AV^B={\cal D}_AV^B+C^B{}_{AE}V^E,
\end{equation}
where
\begin{equation}
C^B{}_{AE}=\frac{1}{2r}q^{BF}\Bigr({\cal D}_Ac_{FE}+{\cal D}_Ec_{FA}-{\cal D}_Fc_{AE}\Bigr)
+O(r^{-2}).
\end{equation}
\end{subequations}
This property will be useful later.
\begin{defn}$\\ $
\label{defn:AF1}
A space-time $(\mathscr{M},g)$ is \textit{asymptotically flat} if the metric tensor $g$ and its components 
satisfy the conditions \eqref{eqn:145} and \eqref{eqn:140}. These conditions are often 
referred to \textit{boundary conditions}.
\end{defn} 
\begin{oss}$\\ $
Note that this definition seems to be completely different from \ref{defn:AS} and \ref{defn:EAS}, given in section \ref{sect:3.3}, that are based on the works of Penrose \citep{Pen63,Pen64,Pen65,Pen67}, in which the conformal technique was first developed. Definition \ref{defn:AF1} is based mainly on the works of \cite{Sachs61,Sachs62,Sachs1,Bondi62}. However the two approaches are completely equivalent, as shown in \cite{New62,NewUn}, since they lead to the same asymptotic properties, using two different ways. It must be pointed out that the conformal method introduced by Penrose represents a \lq natural evolution\rq\hspace{0.1mm} of the previous one, being it more geometrical. It is worth remarking that the peeling property, which was developed in chapter \ref{chap:5} using the Penrose formalism, can be deduced from this approach \citep{Sachs62,NewUn}.
\end{oss}
\section{Bondi-Metzner-Sachs Group}
\label{sect:6.3}
In this section our purpose is to find the coordinate transformations which preserve the asymptotic 
flatness condition. In other words we want to find the asymptotic isometry group of the metric \eqref{eqn:136} and we must demand some conditions to hold in order for the coordinate conventions and boundary conditions 
to remain invariant. It is clear that, from \eqref{eqn:140}, the corresponding changes suffered from the metric 
must therefore obey certain fall-off conditions, i.e.
\begin{equation}
\label{eqn:146}
\delta g_{rr}=0,\hspace{1cm}\delta g_{rA}=0,\hspace{1cm} g^{AB}\delta g_{AB}=0.
\end{equation} 
and 
\begin{subequations}
\label{eqn:181}
\begin{equation}
\label{eqn:147}
\delta g_{uu}=O(r^{-1}),\hspace{1cm}\delta g_{uA}=O(1),
\end{equation} 
\begin{equation}
\label{eqn:180}
\delta g_{ur}=O(r^{-2}),\hspace{1cm}\delta g_{AB}=O(r).
\end{equation} 
\end{subequations}
The third of \eqref{eqn:146} expresses the fact that we \textit{don't} want the angular metric $g_{AB}$ to undergo 
any conformal rescaling under the transformation. However a generalization which includes conformal rescalings 
of $g_{AB}$ can be found in \cite{Barn2010}.\\
We know that the infinitesimal change $\delta g_{ab}$ in the metric tensor is given by the Lie derivative of 
the metric along the $\xi^a$ direction, $\xi^a$ being the generator of the transformation of coordinates:
\begin{equation}
\label{eqn:148}
\delta g_{ab}=-\nabla_a\xi_b-\nabla_b\xi_a.
\end{equation}
Clearly the vector $\xi^a$ obeys \textit{Killing's equation},
\begin{equation*}
\nabla_a\xi_b+\nabla_b\xi_a=0,
\end{equation*}
if and only if the corresponding transformations are isometries. What we want to solve now is an 
\textit{asymptotic Killing's equation}, obtained putting together \eqref{eqn:146} and \eqref{eqn:181} 
with \eqref{eqn:148}. We get from the first of \eqref{eqn:146} 
\begin{equation*}
\nabla_r\xi_r=\partial_r\xi_r-\Gamma^{u}{}_{rr}\xi_u-\Gamma^{r}{}_{rr}\xi_r-\Gamma^{A}{}_{rr}\xi_A=0,
\end{equation*}
and using the Christoffel symbols given in Appendix \ref{C}, we get 
\begin{equation*}
\partial\xi_r=2\partial_r\beta,
\end{equation*}
and hence
\begin{equation}
\label{eqn:143}
\xi_r=f(u,x^A)e^{2\beta},
\end{equation}
where $f$ is a suitably differentiable function of its arguments.\\
From the second of \eqref{eqn:146} we obtain
\begin{equation*}
\nabla_r\xi_A+\nabla_A\xi_r=\partial_r\xi_A+\partial_A\xi_r-2\Gamma^{u}{}_{rA}
\xi_u-2\Gamma^{r}{}_{rA}\xi_r-2\Gamma^{B}{}_{rA}\xi_B=0,
\end{equation*}
and thus, using \eqref{eqn:143} we get 
\begin{equation*}
\partial_r\xi_A-r^2h_{AB}f\left(\partial_rU^B\right)-\frac{2\xi_A}{r}
-\left(\partial_rh_{AC}\right)h^{BC}\xi_B=-\left(\partial_Af\right)e^{2\beta},
\end{equation*}
and after some manipulation
\begin{equation*}
\partial_r\left(\xi_Bg^{BD}+fU^D\right)=-e^{2\beta}g^{AD}\left(\partial_Af\right),
\end{equation*}
which leads to 
\begin{equation*}
\xi_A=-h_{DA}f^Dr^2+fU^Dh_{DA}r^2+r^2h_{DA}\left(\partial_Bf\right)\int_r^{\infty}\frac{e^{2\beta}h^{BD}}{r'^2}dr'
\end{equation*}
\begin{equation}
\label{eqn:144}
=-f_Ar^2+fU_Ar^2+I_Ar^2+O(r),
\end{equation}
where 
\begin{equation*}
I^D(u,r,x^A)=(\partial_Bf)\int_r^{\infty}\frac{e^{2\beta}h^{BD}}{r'^2}dr'=\frac{\partial^Df}{r}+O(r^{-2}),
\end{equation*}
where $f^D$ are suitably differentiable functions of their arguments and the indices A, B etc. 
are raised and lowered with respect to the metric $q_{AB}$.\\ We can solve algebraically 
the third equation in \eqref{eqn:146} to obtain $\xi_u$:
\begin{equation*}
\xi_u=-\frac{e^{2\beta}}{2r}\left(-\partial_a\xi_B+\Gamma^{r}{}_{AB}\xi_r+\Gamma^{C}{}_{AB}\xi_C\right)h^{AB}.
\end{equation*}
Working with Christoffel symbols we get the following expression for $\xi_u$:
\begin{eqnarray}
\label{eqn:149}
\xi_u &=& -\frac{e^{2\beta}r}{4}\partial_D\left(h_{AB}f^D\right)h^{AB}+\frac{e^{2\beta}r}{2}
\left(\partial_Af\right)U^{A}-\frac{e^{2\beta}r}{4}\partial_D\left(h_{AB}I^D\right)h^{AB}
\nonumber \\
&+& e^{2\beta}\frac{V}{r}+r^2h_{AB}(U^Af^B-r^2U^AU^Bf-r^2U^AI^B).
\end{eqnarray}
Now equations \eqref{eqn:181} can be used to give constraints on the arbitrary functions 
$f$ and $f^A$. From the second of \eqref{eqn:180} we get
\begin{equation*}
\nabla_A\xi_B+\nabla_B\xi_A=\partial_A\xi_B+\partial_B\xi_A-2\Gamma^{u}{}_{AB}
\xi_u-2\Gamma^{r}{}_{AB}\xi_r-2\Gamma^C{}_{AB}\xi_C=O(r).
\end{equation*}
Using asymptotic expansions \eqref{eqn:145}, taking the order $r^2$ of the previous equation and putting it equal to zero we get
\begin{equation*}
-\partial_Af_B-\partial_Bf_A+\frac{1}{2}q_{AB}\partial_D\left(q_{CE}f^D\right)q^{CE}+q^{CD}
\left(\partial_Aq_{DB}+\partial_Bq_{DA}-\partial_Dq_{AB}\right)f_C=0,
\end{equation*}
thus
\begin{equation*}
-\partial_Af_B+\gamma^{C}{}_{AB}f_C-\partial_Af_B+\gamma^{C}_{AB}f_C=-\frac{1}{2}q_{AB}\partial_D\left(q_{CE}\right)q^{CE},
\end{equation*}
where $\gamma^A{}_{BC}$ are the Christoffel symbols with respect to the metric on the unit sphere $q_{AB}$. We eventually get
\begin{equation}
\label{eqn:150}
D_Af_B+D_Bf_A=\frac{1}{2}q_{AB}\partial_D\left(q_{CE}f^D\right)q^{CE}.
\end{equation}
and hence
\begin{equation*}
D_Af_B+D_Bf_A=f^D\frac{1}{2}q_{AB}\left(\partial_Dq_{CE}\right)q^{CE}+\left(\partial_Df^D\right)q_{AB}=q_{AB}D_Cf^C.
\end{equation*}
Thus $f^B$ are the conformal Killing vectors of the unit 2-sphere metric $q_{AB}$.\\
From the second of \eqref{eqn:147} we get 
\begin{equation*}
\nabla_u\xi_A+\nabla_A\xi_u=\partial_u\xi_A+\partial_A\xi_u-2\Gamma^{u}{}_{Au}
\xi_u-2\Gamma^{r}{}_{Au}\xi_r-2\Gamma^{B}{}_{Au}\xi_B=O(1).
\end{equation*}
Putting the order $r^2$ of this equation equal to zero we obtain
\begin{equation}
\label{eqn:151}
\partial_uf_A=0. 
\end{equation}
From the first of \eqref{eqn:180} we get
\begin{equation*}
\nabla_u\xi_r+\nabla_r\xi_u=\partial_u\xi_r+\partial_r\xi_u-2\Gamma^{u}_{ur}
\xi_u-2\Gamma^{r}_{ur}\xi_r-2\Gamma^{A}_{ur}	\xi_A=O(r^{-2}).
\end{equation*}
Putting the term of order $r^0$ of the previous equation equal to zero we get 
\begin{equation}
\label{eqn:152}
\partial_uf=\frac{1}{4}\partial_D\left(q_{AB}f^D\right)q^{AB}.
\end{equation}
Putting all the results together we have
\begin{subequations}
\label{eqn:153}
\begin{equation}
\partial_uf_A=0\Rightarrow f_A=f_A(x^B),\end{equation}
\begin{align}
D_Af_B+D_Bf_A=2q_{AB}\partial_uf\Rightarrow \left\{\begin{matrix} \partial^2_uf=0,&\\
\partial_uf=\frac{1}{2}D_Af^A.&\end{matrix}\right.
\end{align}
\end{subequations}
We get for $f$ the following expansion
\begin{equation}
f=\alpha+\frac{u}{2}D_Af^A,
\end{equation}
where $\alpha$ is a suitably differentiable function of $x^A$.\\
Consider now
\begin{equation*}
\xi^a=g^{ab}\xi_b,
\end{equation*}
from which we get 
\begin{eqnarray}
\label{eqn:154}
\xi^u&=&f=\alpha+\frac{u}{2}D_Af^A,\\
\label{eqn:155}
\xi^A&=&f^A-I^A=f^A-\frac{D^A\alpha}{r}-u\frac{D^AD_Cf^C}{2r}+O(r^{-2}),\\
\label{eqn:161}
\xi^r &=& -\frac{r}{2}\left[D_A\xi^A-U^A\partial_Af\right]=-\frac{r}{2}D_{C}\xi^C+O(r^{-1})
\nonumber \\
&=& -\frac{r}{2}D_Cf^C+\frac{D_CD^C\alpha}{2}+u\frac{D_CD^CD_Af^A}{4}+O(r^{-1}).
\end{eqnarray}
The second equality in \eqref{eqn:161} follows from \eqref{eqn:160} and from 
\begin{equation*}
q^{AB}c_{AB}=0,
\end{equation*}
which follows from satisfying at order $r^{-2}$ the second of \eqref{eqn:169} in the form 
\begin{equation*}
0=h^{AB}\partial_rh_{AB}=[q^{AB}-\frac{c^{AB}}{r^2}+O(r^{-3})][-\frac{c_{AB}}{r^2}+O(r^{-3})].
\end{equation*}
As $r\rightarrow\infty$  \eqref{eqn:154} and \eqref{eqn:155} become, respectively
\begin{align}
\label{eqn:156}
&\xi^u=\alpha+\frac{u}{2}D_Af^A,\\
\label{eqn:157}
&\xi^A=f^A.
\end{align}
Finally we can state that the asymptotic Killing vector is of the form \begin{equation}
\label{eqn:158}
\xi=\xi^a\partial_a=\left[\alpha(x^C)+\frac{u}{2}D_Af^A(x^C)\right]\partial_u+f^A(x^C)\partial_A,
\end{equation}
where $\alpha$ is arbitrary and $f^A$ are the conformal Killing vectors of the metric of the unit sphere. In order to fix ideas, set $x^A=(\theta,\phi)$.
It is clear then that $\theta$ and $\phi$ undergo a finite conformal transformation, i.e. 
\begin{subequations}
\label{eqn:159}
\begin{equation}
\label{eqn:182}
\theta\rightarrow\theta'=F(\theta,\phi),
\end{equation}
\begin{equation}
\label{eqn:183}
\phi\rightarrow\phi'=G(\theta,\phi),
\end{equation}
for which
\begin{equation*}
d\theta'^2+\sin^2\theta' d\phi'^2=K^2(\theta,\phi)(d\theta^2+\sin^2\theta d\phi^2),
\end{equation*}
and hence
\begin{equation}
\label{eqn:171}
K^4=J^{2}(\theta,\phi;\theta',\phi')\sin^2\theta\left(\sin\theta'\right)^{-2},\hspace{0.5cm}J
=\mathrm{det}\left(\begin{matrix}\frac{\partial F}{\partial\theta} & \frac{\partial F}
{\partial\phi}\\\frac{\partial G}{\partial\theta} & \frac{\partial G}{\partial\phi}
\end{matrix}\right).
\end{equation}
By definition of conformal Killing vector we also have
\begin{equation}
\label{eqn:214}
K^2=e^{D_Af^A}.
\end{equation}
The finite form of the transformation of the coordinate $u$ is given, as can be easily checked, by 
\begin{equation}
\label{eqn:184}
u\rightarrow u'=K[u+\alpha(\theta,\phi)].
\end{equation}
\end{subequations}
\begin{defn}$\\ $
\label{defn:BMS}
The transformations \eqref{eqn:159} are called \textit{BMS (Bondi-Metzner-Sachs) transformations}, and 
are the set of diffeomorphisms which leave the asymptotic form of the metric of 
an asymptotically flat space-time unchanged.
\end{defn}
The BMS transformations form a group. In fact, as is known, the conformal transformations form a group, so that 
$F$, $G$, and $K$ have all the necessary properties. Thus, one must only check the fact that if one 
carries out two transformation \eqref{eqn:184} the corresponding $\alpha$ for the product is again 
a suitably differentiable function of $\theta$ and $\phi$. If
\begin{equation*}
u_1\rightarrow u_2=K_{12}[u_1+\alpha_{12}]
\end{equation*} 
and 
\begin{equation*}
u_2\rightarrow u_3=K_{23}[u_2+\alpha_{23}]
\end{equation*}
then we have 
\begin{equation*}
u_1\rightarrow u_3=K_{13}[u_1+\alpha_{13}],\hspace{0.5cm}K_{13}=K_{12}K_{23},\hspace{0.5cm}
\alpha_{13}=\alpha_{12}+\frac{\alpha_{23}}{K_{12}}.
\end{equation*}
Since $\alpha_{13}$ is a suitably differentiable function it follows that 
\begin{prop}$\\ $
The BMS transformations form a group, denoted with $\mathscr{B}$.
\end{prop}
\begin{defn}$\\ $
The BMS transformations for which the determinant $J$, defined in \eqref{eqn:171}, 
is positive form the \textit{proper} subgroup of the BMS group.
\end{defn}
In the remainder we will omit the word \lq proper\rq, even if all of our considerations will regard this component of $\mathscr{B}$.
\begin{oss}$\\ $
Note that the $r$ coordinate too may be involved in the BMS group of transformations, but such a 
transformation is somewhat arbitrary since it depends on the precise type of radial coordinate used and 
it is not relevant to the structure of the group. Clearly the BMS group is infinite-dimensional since 
the transformations depend upon a suitably differentiable function $\alpha(\theta,\phi)$.
\end{oss}
\section{Symmetries on $\mathscr{I}$}
\label{sect:6.4}
The geometrical approach to asymptotic flatness, discussed in chapter \ref{chap:3}, affords us a much more vivid picture of the significance of the BMS group.\\
The idea is that by adjoining to the physical space-time $(\tilde{\mathscr{M}},\tilde{g})$ an appropriate conformal boundary $\mathscr{I}$, as done in chapter \ref{chap:3}, we may obtain the asymptotic symmetries as conformal transformations of the boundary, the boundary having a much better chance of having a meaningful symmetry group than $\mathscr{\tilde{M}}$. \\ We start by making an example to better understand the nature of the problem, which is due to \cite{Pen72}. Consider Minkowski space-time with standard coordinates 
$(t,x,y,z)$, the metric being given by
\begin{equation*}
g=\eta_{ab}dx^a\otimes dx^b=dt\otimes dt-dx\otimes dx-dy\otimes dy-dz\otimes dz,
\end{equation*} 
and consider the null cone $\mathscr{N}$ through the origin, given by the equation 
\begin{equation}
\label{eqn:162}
t^2-x^2-y^2-z^2=0.
\end{equation}
The generators of $\mathscr{N}$ are the null rays through the origin, given by
\begin{equation*}
t:x:y:z=\mathrm{const},
\end{equation*}
with $t,x,y,z$ satisfying \eqref{eqn:162}. Let us consider $S^2$ to be the section of $\mathscr{N}$ by 
the spacelike 3-plane $t=1$. Then there exists a (1-1)-correspondence between the generators of $\mathscr{N}$ 
and the points of $S^2$ (i.e. that given by the intersections of the generators with $t=1$). We may regard 
$S^2$ as a realization of the space of generators of $\mathscr{N}$. However, we could have used any 
other cross-section $\hat{S}^2$ of $\mathscr{N}$ to represent this space. The important point is to 
realize that the map which carries any one such cross-section into another, with points on the same 
generator of $\mathscr{N}$ corresponding to one another, is a conformal map. The situation is 
given in Figure \ref{fig:6.1}.:
\begin{figure}[h]
\begin{center}
\includegraphics[scale=0.4]{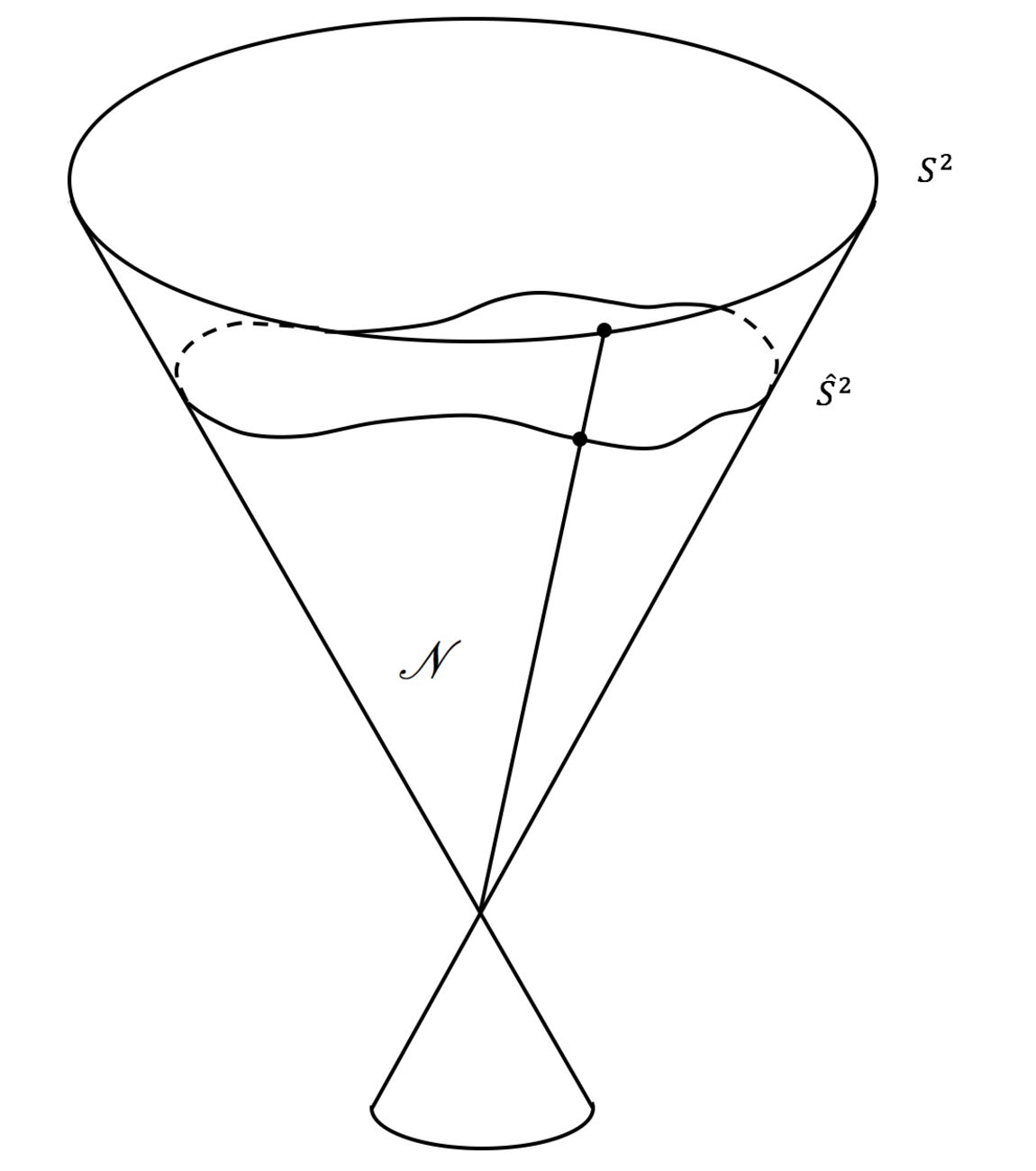}
\caption{The generators of the null cone $\mathscr{N}$ establish a 1-1 map between any two cross-sections of $\mathscr{N}$.}
\label{fig:6.1}
\end{center}
\end{figure}
\\The above mentioned map being conformal, the space of generators of $\mathscr{N}$ may itself be assigned a conformal structure, i.e. that of any of these sections. To see that the map is conformal we may re-express the metric induced on $\mathscr{N}$ in a form
\begin{equation}
\label{eqn:163}
g_{_{\mathscr{N}}}=-r^2\gamma_{\alpha\beta}(x^{\gamma})dx^{\alpha}\otimes dx^{\beta}+0\cdot dr\otimes dr,
\end{equation}
where $x^{\alpha}$ and $r$ are coordinates on $\mathscr{N}$, the generators being given by the coordinate lines $x^{\alpha}=\mathrm{const}$ (the term \lq$0$\rq\hspace{0.1mm} takes into account that, the surface $\mathscr{N}$ being null, its induced metric is degenerate, i.e. with vanishing determinant). There exist obviously many ways of attaining the form \eqref{eqn:163}. One is to use ordinary spherical coordinates for Minkowski space-time, giving $g_{_{\mathscr{N}}}=-r^2(d\theta\otimes d\theta+\sin^2\theta d\phi\otimes d\phi)+0\cdot dr\otimes dr$. Since a cross-section of $\mathscr{N}$ is given by specifying $r$ as function of $x^{\alpha}$ it is clear that any two cross-sections give conformally related metrics, being mapped to one another by the generators of $\mathscr{N}$. It is now obvious that many other cone-like null surfaces will share this property of $\mathscr{N}$, provided their metrics can be put in the form \eqref{eqn:163}. Now if we suppose here to deal with an empty asymptotically simple space-time $(\mathscr{\tilde{M}},\tilde{g})$ (according to definition \ref{defn:EAS}, with associated unphysical space-time $(\mathscr{M},g)$) we know that, if $\mathscr{I}$ is null,  it has the important property to be shear-free, as shown in section \ref{sect:4.3}. Physically, the shear-free nature of the generators of $\mathscr{I}$ tells us that small shapes are preserved as we follow these generators along $\mathscr{I}$. Hence any diffeomorphism which maps each null generator of $\mathscr{I}^+$ into itself is a conformal transformation for any metric on $\mathscr{I}^+$. That is to say, if we take any two cross-sections $S_1$ and $S_2$ of $\mathscr{I}^+$ or $\mathscr{I}^-$, then the correspondence between $S_1$ and $S_2$ established by  the generators is a conformal one. This is exactly the same situation we encountered in the example with $\mathscr{N}$. We have the following
\begin{prop}$\\ $
\label{prop:gen}
If $\mathscr{I}$ is null, then any two cross-sections of $\mathscr{I}^{\pm}$ are mapped to one another conformally by the generators of $\mathscr{I}^{\pm}$. 
\end{prop}
In section \ref{sect:4.3} we have shown that the topology of $\mathscr{I}^{\pm}$ is $S^2\times\mathbb{R}$, where the $\mathbb{R}$ factor may be taken as the null-geodesic generator  $\mathscr{I}^{\pm}$. Hence these generators, by proposition \ref{prop:gen}, establish a conformal mapping between any two $S^2$ cross-sections of $\mathscr{I}^{\pm}$, these sections being of course conformal spheres. It is a theorem that any conformal 2-surface with the topology of a sphere $S^2$ is conformal to the unit 2-sphere in Euclidean 3-space. Thus we can assume without loss of generality, that the conformal factor $\Omega$ has been chosen so that some cross-section $S$ has unphysical square line element $-ds^2$ of a unit 2-sphere. Given one choice of $\Omega$, we can always make a new choice $\Omega'=\Theta\Omega$ which again has the property of vanishing at $\mathscr{I}$ with non-zero gradient there. The factor $\Theta$ has to be an arbitrary smooth positive function on $\mathscr{I}$ and can be chosen to rescale the metric on $\mathscr{I}$ as we please. It is worth noting that the shear-free condition can be saved by the change $\Omega'=\Theta\Omega$, as stated by theorem \ref{thm:shear}. \citep[see also][pg.~132-133]{Stew}. This property can be interpreted as a \lq gauge freedom\rq\hspace{0.1mm} in the choice of the conformal factor $\Omega$. We can use this freedom to set the metric of a continuous sequence of cross-sections along the generators equal to that of $S$. Hence, in spherical polar coordinates the induced metric on $\mathscr{I}^{+}$ is 
\begin{equation}
\label{eqn:164}
g_{_{\mathscr{I^+}}}=d\theta\otimes d\theta+\sin^2\theta d\phi\otimes d\phi+0\cdot du\otimes du,
\end{equation}
where $u$ is a retarded time coordinate, i.e. a parameter defined along each generator increasing monotonically with time from $-\infty$ to $+\infty$, the corresponding form with an advanced time coordinate $v$ in place of $u$ holding for $\mathscr{I}^{-}$. The surfaces $u=\mathrm{const}$ are cross-sections of $\mathscr{I}^+$, each of which has the metric of a unit 2-sphere, as is clear from \eqref{eqn:164}. \\
From the above discussion it follows that the metric on $\mathscr{I}^+$ belongs to an equivalence class of metrics, two elements being equivalent if they are conformally related one to the other. Hence, the form of the metric \eqref{eqn:164} is just one element of this equivalence class that we have chosen as representative. Let us consider the group of conformal transformations of $\mathscr{I}^{+}$, i.e. the group of transformations which conformally preserve the metric \eqref{eqn:164}. It is clear that any smooth transformation which maps each generator into itself will be allowable:
\begin{equation}
\label{eqn:165}
u\rightarrow u'=F(u,\theta,\phi),
\end{equation}
with $F$ smooth on the whole $\mathscr{I}^+$ and $\partial F/\partial u >0$, since it has to map the whole range for $u$ to itself, for any $\theta$ and $\phi$. In addition, we can allow conformal transformations of the $(\theta,\phi)$-sphere into itself. These transformations can be regarded as those of the compactified complex plane $\mathbb{C}\cup\{\zeta=\infty\}$ into itself. Introducing the complex stereographic coordinate
\begin{equation*}
\zeta=e^{i\phi}\cot\frac{\theta}{2},
\end{equation*}
we have that \eqref{eqn:164} may be written as 
\begin{equation}
\label{eqn:167}
g_{_{\mathscr{I^+}}}=\frac{2(d\zeta \otimes d\bar{\zeta}+d\bar{\zeta}\otimes d\zeta)}{(1+\zeta\bar{\zeta})^2}+0\cdot du\otimes du.
\end{equation}
Then the most general conformal transformation of the compactified plane is given by
\begin{equation}
\label{eqn:166}
\zeta\rightarrow\zeta'=\frac{a\zeta+b}{c\zeta+d},
\end{equation}
where $a$,$b$,$c$,$d\in\mathbb{C}$, that can be normalized to satisfy $ad-bc=1$.
\begin{oss}$\\ $
Since conformal transformations can be equivalently expressed in terms of $x^A$ or $\zeta$ coordinates, in the remainder we will use both of them, depending on the convenience.
\end{oss}
The particular functional form of the transformations in \eqref{eqn:166} results from the request that 
they must be diffeomorphisms of the compactified plane $\mathbb{C}\cup\{\zeta=\infty\}$  into itself. 
Hence the transformations must have at least one pole, at $\zeta^*$ say, corresponding to the point 
that is mapped to the north pole $F(\zeta^*)=\infty$ and at least one zero, at $\zeta^{**}$ say, 
corresponding to the point that is mapped to the south pole $F(\zeta^{**})=0$. Thus, the transformations 
must be some rational complex function where the roots of the numerator and the denominator correspond 
to the points that are mapped to the south and the north pole, respectively. Since the transformation 
must be injective there must be one, and only one, point that is mapped to the south pole, and also 
exactly one other point that is mapped to the north pole. This requires that both numerator and denominator 
be linear functions of $\zeta$. Requiring this map to be surjective finally imposes that the complex 
numbers $a,b,c,d$ in \eqref{eqn:166} must satisfy $ad-bc\neq0$ (all of these parameters can be appropriately 
rescaled to get $ad-bc=1$ leaving the transformation unchanged). It is worth remarking that in pure 
mathematics these transformations were studied by Poincar\'e and other authors when they developed the 
theory of what are nowadays called automorphic functions, i.e. meromorphic functions such that 
$f(z)=f((az+b)/(cz+d))$. It is easy to see that these transformations contain:
\begin{itemize}
\item Translations $\zeta\rightarrow \zeta'=\zeta+b,\hspace{1cm}b\in\mathbb{C}$;
\item Rotations $\zeta\rightarrow \zeta'=e^{i\theta}\zeta,\hspace{1cm}\theta\in\mathbb{R}$;
\item Dilations $\zeta\rightarrow \zeta'=e^{-\chi}\zeta,\hspace{1cm}\chi\in\mathbb{R}$;
\item Special transformations $\zeta\rightarrow \zeta'=
-\displaystyle{\frac{b^2}{\zeta^2}},\hspace{1cm}b\in\mathbb{C}$;
\end{itemize}
Any transformation of the form \eqref{eqn:166} can be obtained as the composition of a special transformation, 
a translation, a rotation and a dilation.\\
Usually, transformations \eqref{eqn:166} are referred to as the \textit{conformal group} (in two dimensions), 
the \textit{projective linear group}, the \textit{M\"{o}bius transformations} or the \textit{fractional linear 
transformations}, and is denoted by PSL$(2,\mathbb{C})\cong\mathrm{SL}(2,\mathbb{C})/\mathbb{Z}_2$ 
(as will be discussed in the next section). Under these transformations we have
\begin{equation*}
\frac{2(d\zeta'\otimes d\bar{\zeta'}+d\bar{\zeta'}\otimes d\zeta')}{(1+\zeta'\bar{\zeta'})^2}
=K^2(\zeta,\bar{\zeta})\frac{2(d\zeta\otimes d\bar{\zeta}+d\bar{\zeta}\otimes d\zeta)}
{(1+\zeta\bar{\zeta})^2}\Rightarrow g'_{_{\mathscr{I}^+}}=K^2g_{_{\mathscr{I}^+}},
\end{equation*}
with 
\begin{equation}
\label{eqn:170}
K(\zeta,\bar{\zeta})=\frac{1+\zeta\bar{\zeta}}{(a\zeta+b)(\bar{a}\bar{\zeta}+\bar{b})
+(c\zeta+d)(\bar{c}\bar{\zeta}+\bar{d})}.
\end{equation}
It can be shown that transformations \eqref{eqn:166} are equivalent to \eqref{eqn:182} and \eqref{eqn:183}, 
and that the conformal factor $K$ in \eqref{eqn:170} is the same as one in \eqref{eqn:171}, 
expressed in terms of the $(\theta,\phi)$ variables.\\
\begin{defn}$\\ $
The group of transformations
\begin{subequations}
\begin{align}
\label{eqn:172}
&\zeta\rightarrow\zeta'=\frac{a\zeta+b}{c\zeta+d},\\
\label{eqn:172a}
&u\rightarrow u'=F(u,\zeta,\bar{\zeta}),
\end{align}
\end{subequations}
with $ad-bc=1$ and with $F$ smooth and $\partial F/\partial u>0$ is the \textit{Newman-Unti (NU) group}.
\end{defn}
\begin{oss}$\\ $
Note that \eqref{eqn:172} are the non-reflective conformal transformations of the $S^2$-space of 
generators of $\mathscr{I}^+$ (the conformal structure being defined equivalently by one of its cross-sections), 
while \eqref{eqn:172a}, when \eqref{eqn:172} is the identity ($a=d=1$, $b=c=0$), give the general non-reflective 
smooth transformations of the generators to  themselves.
\end{oss}
The conformal metric \eqref{eqn:167} is considered to be part of the universal intrinsic structure of 
$\mathscr{I}^+$ (universal, in the sense that any space-time which is asymptotically simple and vacuum near 
$\mathscr{I}$ has a $\mathscr{I}^+$ metric and similarly a $\mathscr{I}^-$ metric which is conformal to 
\eqref{eqn:167}). Hence the NU group can be regarded as the group of non-reflective transformations of 
$\mathscr{I}^+$ preserving its intrinsic (degenerate) conformal metric \cite{Pen72,Pen82}.\\
However, the NU group is different from the BMS group, the former being larger than the latter. In fact the NU 
group allows a greater freedom in the function $F$, while in the BMS group $F$ is constrained to be of the 
form \eqref{eqn:184}. Thus, we want to be somehow able to reduce this freedom, assigning a further geometric 
structure to $\mathscr{I}^+$, the preservation of which will furnish the BMS group, restricting exactly the form of $F$ to be the one of 
\eqref{eqn:184}. This additional structure is referred to as the \textit{strong conformal geometry} 
\cite{Penrin2,Pen72}. The most direct way to specify this structure is the following. Consider a 
replacement of the conformal factor,
\begin{equation}
\label{eqn:173}
\Omega\rightarrow\Omega'=\Theta\Omega.
\end{equation}
We choose the function $\Theta$ to be smooth and positive on $\mathscr{M}$ and nowhere vanishing on 
$\mathscr{I}^+$. Under \eqref{eqn:173} the metric transforms as 
\begin{equation*}
g_{ab}\rightarrow g'_{ab}=\Theta^2g_{ab},\hspace{1cm}g^{ab}\rightarrow g'^{ab}=\Theta^{-2}g^{ab},
\end{equation*}
and the normal co-vector to $\mathscr{I}^{+}$ as
\begin{equation*}
N_a=-\nabla_a\Omega\rightarrow N'_a=-\nabla'_a\Omega'=-\nabla_a\Omega'=-\Omega\nabla_a\Theta
-\Theta\nabla_a\Omega\approx\Theta N_a,
\end{equation*}
while the vector
\begin{equation*}
N^a=g^{ab}\partial_b\Omega\rightarrow N'^a=g'^{ab}N'_b\approx\Theta^{-1}N^a,
\end{equation*}
where we introduced the \lq weak equality\rq\hspace{0.1mm} symbol $\approx$. Considering two fields 
$\psi^{...}_{...}$ and $\phi^{...}_{...}$, saying that 
\begin{equation}
\label{eqn:80}
\psi^{...}_{...}\approx\phi^{...}_{...}
\end{equation}
means that $\psi^{...}_{...}-\phi^{...}_{...}=0$ on $\mathscr{I}$. 
The line element $dl$ of $\mathscr{I}^+$ rescales according to 
\begin{equation}
\label{eqn:176}
dl\rightarrow dl'=\Theta dl.
\end{equation}
Having done any allowable choice of the conformal factor $\Omega$, through the function $\Theta$, and hence 
some specific choice of the metric $dl$ for cross-sections of $\mathscr{I}^+$, then it is defined, from 
$N_a=-\nabla_a\Omega$, a precise scaling for parameters $u$ on the generators of $\mathscr{I}^+$, fixed by
\begin{equation*}
\frac{\partial}{\partial u}=N^a\nabla_a,\hspace{0.7cm}\mathrm{i.e.}\hspace{0.7cm}N^a\nabla_a u=1.
\end{equation*}
Under \eqref{eqn:173} we see that to keep the scaling of the parameters $u$ along the generators fixed we must choose
\begin{equation}
\label{eqn:175}
du\rightarrow du'=\Theta du, 
\end{equation}
so that 
\begin{equation*}
N^a\nabla_a u\rightarrow N'^a\nabla'_au'=N'^a\nabla_au'=N'^a\frac{\partial u'}{\partial x^a}
=\Theta^{-1}\Theta N^a\nabla_a u=1.
\end{equation*}
All the parameters $u$, linked by \eqref{eqn:175}, scale in the same way along the generators of $\mathscr{I}^+$. 
From \eqref{eqn:176} and \eqref{eqn:175} we see that the ratio
\begin{equation}
\label{eqn:177}
dl:du
\end{equation}
remains invariant and it is independent of the choice of the conformal factor $\Omega$. It is the invariant 
structure provided by \eqref{eqn:177} that can be taken to define the strong conformal geometry. To better 
reformulate this invariance we introduce the concept of \textit{null angle} \cite{Pen63,Penrin2,Pen72,BMS75}. 
Consider two non-null tangent directions at a point $P$ of $\mathscr{I}^+$. Let $[X]$ and $[Y]$ be such 
directions. If no linear combination of $X\in [X]$ and $Y\in [Y]$ is the null tangent direction at $P$, 
then the angle between $[X]$ and $[Y]$ is defined by the metric \eqref{eqn:164}. However, if the null tangent 
direction at $P$ is contained in the plane spanned by $[X]$ and $[Y]$, then the angle between $[X]$ and $[Y]$ 
always vanishes. To see this choose $X\in [X]$, $Y\in [Y]$ and $N\in [N]$ ($[N]$ being the null direction 
tangent at $P$), such that $Y=X+N$. Then since $N$ is null we have, using the metric $g$ on $\mathscr{I}^+$ 
given in \eqref{eqn:164} (and hence any other one of its equivalence class)
\begin{equation*}
0=g(N,X)=g(Y-X,X)=g(Y,X)-g(X,X),
\end{equation*}
and 
\begin{equation*}
0=g(N,Y)=g(Y-X,Y)=g(Y,Y)-g(X,Y),
\end{equation*}
from which the angle $\theta$ between $[X]$ and $[Y]$, given by 
\begin{equation*}
\cos\theta=\frac{g(X,Y)}{\sqrt{g(X,X)g(Y,Y)}}=1,
\end{equation*}
vanishes. However, if we require the strong conformal geometry structure to hold and hence the invariance of 
the ratio \eqref{eqn:177} we can numerically define the null angle $\nu$ between two tangent directions at 
a point $P$ of $\mathscr{I}^{+}$ by 
\begin{equation}
\label{eqn:178}
\nu=\frac{\delta u}{\delta l},
\end{equation}
where the infinitesimal increments $\delta u$ and $\delta l$ are as indicated in Figure \ref{fig:6.3} \cite{Penrin2}.
\begin{figure}[h]
\begin{center}
\includegraphics[scale=0.4]{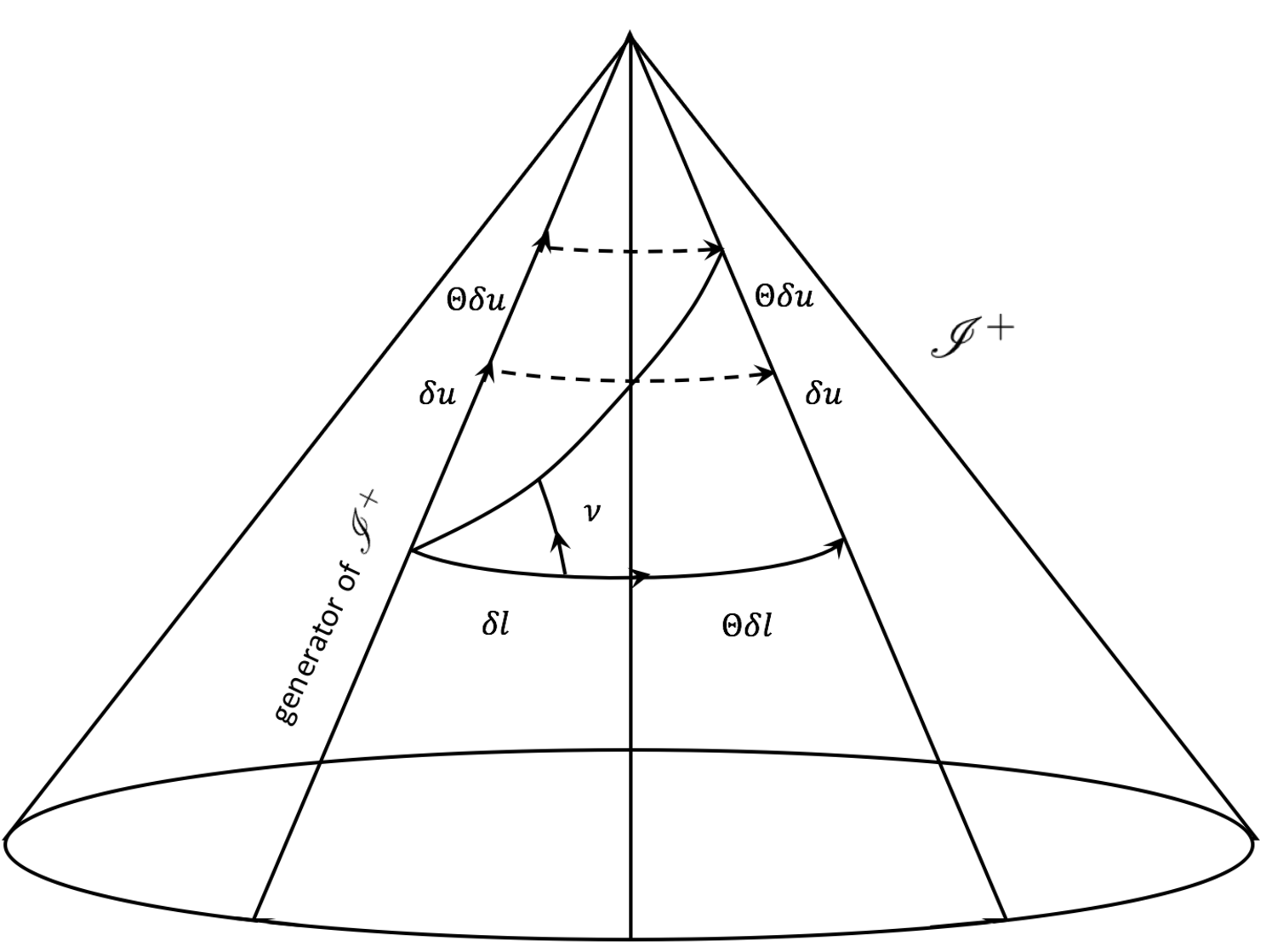}
\caption{A null angle $\nu$ on $\mathscr{I}^+$, given by $\nu=\delta u/\delta l$, is defined between a pair of directions on $\mathscr{I}^+$ whose span contains the null normal direction to $\mathscr{I}^+$.}
\label{fig:6.3}
\end{center}
\end{figure}
\\In virtue of the strong conformal geometry, under change of the conformal factor for the metric of $\mathscr{I}^+$, null angles 
remain invariant. For further insights about the strong conformal geometry and the interpretation of null angles we suggest to read \cite{Pen72} or \cite{BMS75}. \\A transformation of $\mathscr{I}^+$ to itself which 
preserves angles and null angles, i.e. that respects the strong conformal geometry structure, must have 
the effect that any expansion (or contraction) of the spatial distances $dl$ is accompanied by an equal 
expansion (or contraction) of the scaling of the special $u$ parameters. The allowed transformations
have the form \eqref{eqn:172}, where function $F$ must now have the  precise form that allows the ratio 
$du:dl$ to remain invariant. Under the transformation \eqref{eqn:173} we have, as seen, that the sphere 
of the cross-section of $\mathscr{I}^+$ undergoes a conformal mapping, i.e.
\begin{equation*}
dl\rightarrow dl'=\Theta dl.
\end{equation*}
Since $\Theta$ is the conformal factor of the transformation, it depends only on $\theta$ and $\phi$ or, 
equivalently, on $\zeta$ and $\bar{\zeta}$ and must have the form given in \eqref{eqn:170}. 
We must therefore also have 
\begin{equation*}
du\rightarrow du'=\Theta du, 
\end{equation*}
Integrating we get
\begin{equation*}
u\rightarrow u'=\Theta[u+\alpha(\zeta,\bar{\zeta})],
\end{equation*}
where $\Theta$ assumes the form
\begin{equation*}
\Theta(\zeta,\bar{\zeta})=\frac{1+\zeta\bar{\zeta}}{(a\zeta+b)(\bar{a}\bar{\zeta}+\bar{b})
+(c\zeta+d)(\bar{c}\bar{\zeta}+\bar{d})},
\end{equation*}
with $a,b,c,d\in\mathbb{C}$ and $ad-bc=1$. In virtue of definition \ref{defn:BMS} we have obtained the following
\begin{prop}$\\ $
The group of conformal transformations of $\mathscr{I}^+$ which preserve the strong conformal geometry, i.e. 
both angles and null angles, is the BMS group.
\end{prop}
\begin{oss}$\\ $
Conformal transformations, and hence the NU group, always preserve finite angles, but null angles, i.e. 
angles between tangent vectors of which $N^a$ is a linear combination, are preserved by the BMS group only. 
\end{oss}
The general form of a BMS transformation is thus
\begin{subequations}
\begin{equation}
\label{eqn:200}
\zeta\rightarrow\zeta'=\frac{a\zeta+b}{c\zeta+d},
\end{equation}
\begin{equation}
\label{eqn:201}
u\rightarrow u'=\frac{(1+\zeta\bar{\zeta})[u+\alpha(\zeta,\bar{\zeta})]}{(a\zeta+b)(\bar{a}\bar{\zeta}
+\bar{b})+(c\zeta+d)(\bar{c}\bar{\zeta}+\bar{d})},
\end{equation}
\end{subequations}
\\with $a,b,c,d\in\mathbb{C}$ and $ad-bc=1$. Clearly the BMS group is a subgroup of the NU group, the 
function $F$ having its form fixed. However it is still an infinite-dimensional function-space group.
\section{Structure of the BMS Group}
\label{sect:6.5}
We discuss first the BMS transformations obtained by setting $\alpha=0$,
\begin{equation}
\label{eqn:220}
u\rightarrow u'=Ku,\hspace{1cm}\zeta\rightarrow\zeta'=\frac{a\zeta+b}{c\zeta+d},
\end{equation}
with 
\begin{equation*}
K(\zeta,\bar{\zeta})=\frac{1+\zeta\bar{\zeta}}{(a\zeta+b)(\bar{a}\bar{\zeta}+\bar{b})
+(c\zeta+d)(\bar{c}\bar{\zeta}+\bar{d})},
\end{equation*}
\\i.e. a rescaling for $u$ and a transformation of the group $\mathrm{PSL}(2,\mathbb{C})$ for $\zeta$.  Any of these transformations is specified by the 4 constants $a,b,c,d\in\mathbb{C}$, satisfying $ad-bc=1$. 
Hence there are only 3 independent complex parameters, i.e. 6 independent real parameters. 
Any element $f\in \mathrm{PSL}(2,\mathbb{C})$ reads as
\begin{equation*}
f=\frac{a\zeta+b}{c\zeta+d}\equiv\{a,b,c,d\}.
\end{equation*}
Note that under simultaneous change $a\rightarrow-a$, $b\rightarrow-b$, $c\rightarrow-c$, 
$d\rightarrow-d$ any element $f\in\mathrm{PSL}(2,\mathbb{C})$ remains unaffected, i.e.
\begin{equation}
\label{eqn:203}
f=\{a,b,c,d\}=\{-a,-b,-c,-d\}.
\end{equation}
Now take into account the group $\mathrm{SL}(2,\mathbb{C})$ of $(2\times 2)$ complex matrices 
$Q$ with $\mathrm{det}[Q]=1$:
\begin{equation*}
Q=\left(\begin{matrix}A& B\\ C& D
\end{matrix}
\right),\hspace{1cm}\mathrm{det}[Q]=AD-BC=1,\hspace{1cm}A,B,C,D\in\mathbb{C}.
\end{equation*} 
Clearly the dimension of $\mathrm{SL}(2,\mathbb{C})$ is 6. Hence we can consider a map $\tilde{\varphi}$, 
between $\mathrm{PSL}(2,\mathbb{C})$ and $\mathrm{SL}(2,\mathbb{C})$ defined by
\begin{equation}
\label{eqn:202}
\tilde{\varphi}:f=\{a,b,c,d\}\in \mathrm{PSL}(2,\mathbb{C})\longrightarrow\tilde{\varphi}(f)
=\left(\begin{matrix}a & b\\
c & d\end{matrix}\right)\in\mathrm{SL}(2,\mathbb{C}).
\end{equation}
It is easy to show that, since the group operation of $\mathrm{PSL}(2,\mathbb{C})$ is the function 
composition $\circ$, given $f=\{a,b,c,d\}$ and $g=\{a',b',c',d'\}\in\mathrm{PSL}(2,\mathbb{C})$ we have
\begin{equation*}
f\circ g=\{aa'+bc',ab'+bd',ca'+dc',cb'+dd'\}\in\mathrm{PSL}(2,\mathbb{C}).
\end{equation*}
Then taking the images of $f$ and $g$ through $\tilde{\varphi}$,
\begin{equation*}
\tilde{\varphi}(f)=\left(\begin{matrix}a & b\\
c & d\end{matrix}\right),\hspace{1cm}\tilde{\varphi}(g)=\left(\begin{matrix}a' & b'\\
c' & d'\end{matrix}\right),
\end{equation*}
we have, since the operation in $\mathrm{SL}(2,\mathbb{C})$ is the ordinary matrix product,
\begin{equation}
\label{eqn:204}
\tilde{\varphi}(f)\cdot\tilde{\varphi}(g)=\left(\begin{matrix}aa'+bc' &ab'+bd'\\
ca'+dc' & cb'+dd'\end{matrix}\right)=\tilde{\varphi}(g\circ f).
\end{equation}
Note that, in virtue of \eqref{eqn:203}, to the same element $f$ there correspond, through $\tilde{\varphi}$, 
two different elements, $\tilde{\varphi}(f)$ and $-\tilde{\varphi}(f)$. If we consider now the map $\varphi$:
\begin{equation*}
\varphi:f=\{a,b,c,d\}\in \mathrm{PSL}(2,\mathbb{C})\longrightarrow\varphi(f)=\left(\begin{matrix}a & b\\
c & d\end{matrix}\right)\in\mathrm{SL}(2,\mathbb{C})/\mathbb{Z}_2.
\end{equation*}
it is clear that property \eqref{eqn:204} holds for $\varphi$ as well. This map is a group isomorphism, 
$\varphi(f)$ and $-\varphi(f)$ being now identified in $\mathrm{SL}(2,\mathbb{C})/\mathbb{Z}_2$. We can state
\begin{equation}
\label{eqn:205}
\mathrm{PSL}(2,\mathbb{C})\cong\mathrm{SL}(2,\mathbb{C})/\mathbb{Z}_2.
\end{equation}
The group $\mathrm{SL}(2,\mathbb{C})$ is the double covering of $\mathrm{PSL}(2,\mathbb{C})$. As already pointed out in the introduction of chapter \ref{chap:2}, we also have 
\begin{equation}
\label{eqn:210}
\mathscr{L}\cong\mathrm{SL}(2,\mathbb{C})/\mathbb{Z}_2,
\end{equation}
where $\mathscr{L}$ is the connected component of the Lorentz group. Thus
\begin{equation}
\label{eqn:211}
\mathrm{PSL}(2,\mathbb{C})\cong \mathscr{L}.
\end{equation}
We have the following
\begin{prop}$\\ $
The connected component of the Lorentz group is isomorphic with the subgroup $\mathrm{PSL}(2,\mathbb{C})$ of the BMS group.
\end{prop}
To make this isomorphism explicit take an element $f=\{a,b,c,d\}\in\mathrm{PSL}(2,\mathbb{C})$ and through $\varphi$ assign it an element of $\mathrm{SL}(2,\mathbb{C})$,
\begin{equation*}
\varphi(f)=\left(\begin{matrix}a & b\\ c & d
\end{matrix}\right)\in\mathrm{SL}(2,\mathbb{C}),\hspace{1cm}ad-bc=1.
\end{equation*}
Then, using the isomorphism of \eqref{eqn:210} it can be shown with tedious calculations \citep{Oblak16} that to $\varphi(f)$ there corresponds an element of $\mathscr{L}$ through the isomorphism $\psi$ of \eqref{eqn:10}:
\begin{equation}
\label{eqn:iso}
\psi(\varphi(f))=
\end{equation}
\begin{equation*}
\left(\begin{matrix}\frac{1}{2}\left(|a|^2+|b|^2+|c|^2+|d|^2\right) & -\mathrm{Re}\left\{a\bar{b}+c\bar{d}\right\} & \mathrm{Im}\left\{a\bar{b}+c\bar{d}\right\} & \frac{1}{2}\left(|a|^2-|b|^2+|c|^2-|d|^2\right) \\ -\mathrm{Re}\left\{\bar{a}c+\bar{b}d\right\} & \mathrm{Re}\left\{\bar{a}d+\bar{b}c\right\} & -\mathrm{Im}\left\{a\bar{d}-b\bar{c}\right\} & - \mathrm{Re}\left\{\bar{a}c+\bar{b}d\right\}\\
 \mathrm{Im}\left\{\bar{a}c+\bar{b}d\right\} & -\mathrm{Im}\left\{\bar{a}d+\bar{b}c\right\} & \mathrm{Re}\left\{a\bar{d}-b\bar{c}\right\} &  \mathrm{Im}\left\{\bar{a}c+\bar{b}d\right\}\\
 \frac{1}{2}\left(|a|^2+|b|^2-|c|^2-|d|^2\right) & -\mathrm{Re}\left\{a\bar{b}-c\bar{d}\right\} & \mathrm{Im}\left\{a\bar{b}-c\bar{d}\right\} & \frac{1}{2}\left(|a|^2-|b|^2-|c|^2+|d|^2\right)
\end{matrix}\right).
\end{equation*}\\
At this stage, the relation between the Lorentz group and the sphere appears as a mere coincidence. 
In particular, since the original Lorentz group is defined by its linear action on a four-dimensional space, 
there is no reason for it to have anything to do with certain non-linear transformations of a two-dimensional 
manifold such as the sphere. However, it can be shown that it is not accidental. Following \cite{Oblak16}, 
we can suppose to perform a Lorentz transformation in Minkowski space-time 
equipped with standard coordinates $(t,x,y,z)$, i.e.
\begin{equation*}
x'^{\mu}=\Lambda^{\mu}_{\nu}x^{\nu},\hspace{1cm}\Lambda\in \mathscr{L}.
\end{equation*} 
We may introduce Bondi coordinates $(u,r,x^A)$ for Minkowski space-time as done in \eqref{eqn:135}. 
Then if we evaluate the limit for large values of the radial coordinate $r=\sqrt{x^2+y^2+z^2}$ 
keeping the value of $u=t-r$ fixed (i.e. on $\mathscr{I}$) and use the isomorphism \eqref{eqn:210} 
and hence \eqref{eqn:iso} we obtain the following behaviour:
\begin{equation*}
\zeta'=\frac{a\zeta+b}{c\zeta+d}+O(r^{-1}).
\end{equation*}
where
$\zeta=e^{i\phi}\cot\frac{\theta}{2}$. Furthermore it can be checked that both $u$ and $r$, under 
the effect of a Lorentz transformation on $\mathscr{I}$, undergo an angle-dependent rescaling. Hence 
we have obtained a fundamental result: Lorentz transformations acting on $\mathscr{I}$, expressed in 
terms of the parameters $a,b,c,d$ coincide with conformal transformations of $\mathrm{PSL}(2,\mathbb{C})$. 
Since asymptotically flat space-times have the same structure of a Minkowski space-time at infinity, 
this argument can be extended to all of them too. In the remainder we will use $\mathscr{L}$ to describe 
the group structure of $\mathscr{B}$, the isomorphism with $\mathrm{PSL}(2,\mathbb{C})$ being implicit.\\
Now we turn to analyse the transformations which involve a non-vanishing $\alpha(\theta,\phi)$. 
\begin{defn}$\\ $
\label{defn:suptra}
The Abelian subgroup of BMS transformations for which
\begin{equation}
\label{eqn:225}
\theta'=\theta,\hspace{0.5cm}\phi'=\phi,\hspace{0.5cm}u'=u+\alpha(\theta,\phi),
\end{equation}
is called \textit{supertranslations subgroup} and is denoted by $\mathscr{S}$.
\end{defn}
Under such a transformation the system of null hypersurfaces $u=\mathrm{const}$ is transformed into another 
system of null hypersurfaces $u'=\mathrm{const}$.  \\
To proceed further in the analysis of the structure of the BMS group we need to recall the concepts of 
\textit{right} and \textit{left} \textit{cosets} and, hence, that of \textit{normal subgroup} \cite{Algebra}. 
Consider a group $G$ and a subgroup $H$ of $G$. Introduce in $G$ the equivalence relation 
$\sim$ defined, for $g$, $a\in G$, as
\begin{equation*}
g\sim a \Longleftrightarrow ag^{-1}\in H\Longleftrightarrow a\in Hg.
\end{equation*}
It is easy to verify that the previous relation is reflexive, symmetric and transitive. 
\begin{defn}$\\ $
The equivalence class with respect to $\sim$ is called \textit{right coset of $H$ in $G$ with respect to $g$} 
and is denoted by $Hg$:
\begin{equation*}
[g]=\{hg : h\in H\}= Hg.
\end{equation*}
\end{defn}
Similarly, the \textit{left coset of $H$ in $G$ with respect to $g$} can be introduced as
\begin{equation*}
[g]^*=\{gh : h\in H\}=gH.
\end{equation*}
Generally, the right and the left cosets are different sets. 
\begin{defn}$\\ $
A subgroup  $N$ of $G$ which defines a unique partition,
\begin{equation*}
[g]=[g]^*\Longleftrightarrow gN=Ng\hspace{0.5cm}\forall g \in G
\end{equation*}
is called \textit{normal subgroup of} $G$. 
\end{defn}
Clearly it follows that for every $n\in N$ and $g\in G$ the product $gNg^{-1}\subseteq N$. Note that every group 
$G$ possesses normal subgroups, since $G$ and the identity are normal subgroups. \\
Now consider for a general subgroup $H$ of $G$ the \textit{quotient group} (or \textit{factor group}) $G/H$, defined as
\begin{equation*}
G/H=\{[g] : g\in G\}.
\end{equation*}
If $H$ is normal the elements of $G/H$ are, indistinctly, the right and left cosets. Furthermore, 
under this hypothesis, the set $G/H$ can be equipped with a group structure in a natural way by 
defining the product $\ast$:
\begin{align*}
\ast : &G/H\times G/H\longrightarrow G/H \\
&gH\ast g'H \equiv gg'H,
\end{align*}
i.e., 
\begin{equation*}
[g]\ast [g']\equiv [gg'].
\end{equation*}
It can be shown that $G/H$ equipped with the product $\ast$ satisfies the group axioms. We are now ready to 
discuss further the BMS properties.\\
Any element $b$ of $\mathscr{B}$ can be written as
\begin{equation*}
b=(\Lambda,\alpha).
\end{equation*}
Note that with this nomenclature any element $\Lambda$ of $\mathscr{L}$ (or, equivalently, of 
$\mathrm{PSL}(2,\mathbb{C})$) can be written as $\Lambda=(\Lambda,0)$ and any element $s$ of $\mathscr{S}$ 
as $s=(\mathbb{I},\alpha)$, where $\mathbb{I}$ denotes the identity in $\mathscr{L}$.\\ 
The action of $b$ on the variables $(\zeta,u)$ is 
\begin{equation*}
b(\zeta,u)=(f(\zeta),K[u+\alpha(\zeta,\bar{\zeta})]),
\end{equation*}
where $f$ is the element of $\mathrm{PSL}(2,\mathbb{C})$ which corresponds to $\Lambda$ through the above 
discussed isomorphism and $K$ its conformal factor.\\
It is easy to show that, with this notation, we have for the inverse of $b$:
\begin{equation}
\label{eqn:208}
b^{-1}=(\Lambda^{-1},-K\alpha).
\end{equation}
Considering an element $s=(\mathbb{I},\beta)$ of $\mathscr{S}$ we have
\begin{equation*}
bsb^{-1}(\zeta,u)=(\zeta,u+K\beta(\zeta,\bar{\zeta}))=s'(\zeta,u),
\end{equation*}
with 
\begin{equation*}
s'=(\mathbb{I},K\beta)\in\mathscr{S}.
\end{equation*}
From the above discussion we have the following
\begin{prop}$\\ $
\label{prop:BMS1}
The supertranslations $\mathscr{S}$ form an Abelian normal, infinite-parameter, subgroup of the BMS group:
\begin{equation*}
b\mathscr{S}b^{-1}=\mathscr{S}\hspace{0.5cm}\mathrm{for}\hspace{1.2mm} \mathrm{all}\hspace{1.2mm} b\in\mathscr{B}.
\end{equation*}
\end{prop}
Under the assumption that the function $\alpha$ is twice differentiable, we can expand it into spherical harmonics as
\begin{equation}
\label{eqn:206}
\alpha(\theta,\phi)=\sum_{l=0}^{\infty}\sum_{m=-l}^{l}\alpha_{l,m}Y_{lm}(\theta,\phi),\hspace{1cm}\alpha_{l,-m}
=(-1)^m\bar{\alpha}_{l,m}.
\end{equation}
\begin{defn}$\\ $
If in decomposition \eqref{eqn:206} $\alpha_{l,m}=0$ for $l>2$, i.e.
\begin{equation}
\label{eqn:207}
\alpha\equiv\alpha_t=\epsilon_0+\epsilon_1\sin\theta\cos\phi+\epsilon_2\sin\theta\sin\phi+\epsilon_3\cos\theta,
\end{equation}
then the supertranslations reduce to a special case, called \textit{translation} subgroup, denoted by 
$\mathscr{T}$, with just four independent parameters $\epsilon_0,...,\epsilon_3$. 
\end{defn}
It is easy to show that $\zeta=e^{i\phi}\cot\frac{\theta}{2}$ implies
\begin{equation*}
\cos\phi=\frac{\zeta+\bar{\zeta}}{2\sqrt{\zeta\bar{\zeta}}},\hspace{1cm}\sin\phi
=\frac{i(\bar{\zeta}-\zeta)}{2\sqrt{\zeta\bar{\zeta}}},
\end{equation*}
\begin{equation*}
\cos\theta=\frac{\zeta\bar{\zeta}-1}{1+\zeta\bar{\zeta}},\hspace{1cm}\sin\theta
=\frac{2\sqrt{\zeta\bar{\zeta}}}{1+\zeta\bar{\zeta}}.
\end{equation*}
Then equation \eqref{eqn:207} becomes
\begin{equation*}
\alpha_t=\epsilon_0+\epsilon_1\frac{\zeta+\bar{\zeta}}{1+\zeta\bar{\zeta}}+\epsilon_2
\frac{(i\zeta-i\bar{\zeta})}{1+\zeta\bar{\zeta}}+\epsilon_3\frac{\zeta\bar{\zeta}-1}{1+\zeta\bar{\zeta}}
\end{equation*}
\begin{equation*}
=\frac{A+B\zeta+\bar{B}\bar{\zeta}+C\zeta\bar{\zeta}}{1+\zeta\bar{\zeta}},
\end{equation*}
with $A$ and $C$ real. 
Hence in terms of $\zeta$ and $\bar{\zeta}$ a translation is
\begin{equation*}
u=t-r\rightarrow u'=u+\frac{A+B\zeta+\bar{B}\bar{\zeta}+C\zeta\bar{\zeta}}{1+\zeta\bar{\zeta}},
\end{equation*}
\begin{equation*}
\zeta\rightarrow\zeta'=\zeta.
\end{equation*}
If we let $t,x,y,z$ be Cartesian coordinates in Minkowski space-time, it is easy to see that
\begin{equation*}
Z^2\zeta=\frac{(x+iy)(1-z/r)}{4r},\hspace{1cm}x=r(\zeta+\bar{\zeta})Z,
\end{equation*}
\begin{equation*}
y=-ir(\zeta-\bar{\zeta})Z,\hspace{1cm}z=r(\zeta\bar{\zeta}-1)Z,
\end{equation*}
where $Z=1/(1+\zeta\bar{\zeta})$. Now if we perform a translation
\begin{equation*}
t\rightarrow t'=t+a,\hspace{0.5cm}x\rightarrow x'=x+b,\hspace{0.5cm}y\rightarrow y'=y+c,
\hspace{0.5cm}z\rightarrow z'=z+d,
\end{equation*}
it is easy to get 
\begin{align*}
&u=t-r\rightarrow u'=u+Z(A+B\zeta+\bar{B}\bar{\zeta}+C\zeta\bar{\zeta})+O(r^{-1}),\\
&\zeta\rightarrow\zeta'=\zeta+O(r^{-1}).
\end{align*}
with $A=a+d$, $B=b-ic$ and $C=a-d$. Thus, the nomenclature \lq translation\rq\hspace{0.1mm} is consistent 
with that for the space-time translations in Minkowski space-time. In fact we have just shown that any 
translation in the ordinary sense induces a translation (i.e. an element of $\mathscr{T}$) on $\mathscr{I}^+$.\\
It is easy to verify that for any $b=(\Lambda,\alpha)\in\mathscr{B}$ and for any 
$t=(\mathbb{I},\alpha_t)\in\mathscr{T}$ we have 
\begin{equation*}
btb^{-1}(\zeta,u)=(\zeta,u+K\alpha_t)=t'(\zeta,u),
\end{equation*}
with 
\begin{equation*}
t'=(\mathbb{I},K\alpha_t(\zeta,\bar{\zeta}))\in\mathscr{T}.
\end{equation*}
Note that it is not obvious that $K\alpha_t$ is still a function of $\theta$ and $\phi$ containing only 
zeroth- and first-order harmonics. A proof of this result will be given in section \ref{sect:6.7}. Taking for the moment this result for true it follows that:
\begin{prop}$\\ $
The translations $\mathscr{T}$ form a normal four-dimensional subgroup of $\mathscr{B}$:
\begin{equation*}
b\mathscr{T}b^{-1}=\mathscr{T}\hspace{0.5cm}\mathrm{for}\hspace{1.2mm} \mathrm{all}\hspace{1.2mm} b\in\mathscr{B},
\end{equation*}
and clearly
\begin{equation*}
s\mathscr{T}s^{-1}=\mathscr{T}\hspace{0.5cm}\mathrm{for}\hspace{1.2mm} \mathrm{all}\hspace{1.2mm} s\in\mathscr{S},
\end{equation*}
\end{prop}
We have the following inclusion relations:
\begin{equation*}
\mathscr{T}\subset\mathscr{S}\subset\mathscr{B}.
\end{equation*}
The next step will be to investigate the group structure of $\mathscr{B}$. It is easy to show 
that for any $b\in\mathscr{B}$ there exists a unique $\Lambda\in\mathscr{L}$ and $s\in\mathscr{S}$ 
such that $b=\Lambda s$. In fact given 
\begin{equation*}
\Lambda=(\Lambda,0)\in\mathscr{L},\hspace{1cm}s=(\mathbb{I},\alpha)\in\mathscr{S},
\end{equation*}
we have that
\begin{equation*}
\Lambda s(\zeta,u)=\Lambda(\zeta,u+\alpha)=(f(\zeta),K[u+\alpha(\zeta,\bar{\zeta})])=b(\zeta,u)
\end{equation*}
with $b=(\Lambda,\alpha)$. The uniqueness results from the observation that $\mathscr{L}\cap\mathscr{S}=\{e\}$ where 
$e=(\mathbb{I},0)$ is the identity in $\mathscr{B}$, since if $fs=f's'$, then $f'^{-1}f=s's^{-1}\in\mathscr{L}\cap\mathscr{S}$ 
implying $f'=f$ and $s'=s$. Hence we have
\begin{equation}
\label{eqn:233}
\mathscr{B}=\mathscr{L}\mathscr{S}.
\end{equation}
Furthermore, the supertranslations $\mathscr{S}$ form an (Abelian) normal subgroup of $\mathscr{B}$, according 
to \ref{prop:BMS1}. Thus we can already state that, by definition of semi-direct product, 
\begin{prop}$\\ $
The BMS group is a semi-direct product of the conformal group of the unit $2$-sphere with the 
supertranslations group, i.e.
\begin{equation*}
\mathscr{B}=\mathscr{L}\rtimes\mathscr{S}.
\end{equation*}
\end{prop}
We can say more by specifying an action of $\mathscr{L}$ on $\mathscr{S}$ and, hence, by specifying a product rule 
for two elements of $\mathscr{B}$. Let $\mathscr{S}$ be the vector space of real functions on the Riemann sphere. 
Let $\sigma$ be a smooth right action of $\mathscr{L}$ on $\mathscr{S}$ defined as 
\begin{subequations}
\label{eqn:212}
\begin{equation}
\sigma:(\Lambda,\alpha)\in\mathscr{L}\times\mathscr{S}\longrightarrow\sigma_{\Lambda}(\alpha)
\equiv \alpha\Lambda\in\mathscr{S}
\end{equation}
such that 
\begin{equation}
\alpha(\zeta,\bar{\zeta})\Lambda=K^{-1}\alpha(f(\zeta),\bar{f}(\bar{\zeta})),
\end{equation}
\end{subequations}
where $K$ is the conformal factor associated with $f$, the element of $\mathrm{PSL}(2,\mathbb{C})$ that 
corresponds to $\Lambda$. Then it is easy to verify that the composition law for the elements of $\mathscr{B}$ is 
\begin{equation*}
b_1\cdot b_2=(\Lambda_1,\alpha_1)\cdot(\Lambda_2,\alpha_2)=(\Lambda_1\cdot\Lambda_2,\alpha_2+\alpha_1\Lambda_2).
\end{equation*}
Note that the inverse of $b$ in \eqref{eqn:208} may be written as
\begin{equation*}
b^{-1}=(\Lambda^{-1},\left[\alpha \Lambda^{-1}\right]^{-1}).
\end{equation*}
Thus the BMS group is the \textit{right} semi-direct product \cite{Col14} of $\mathscr{L}$ with $\mathscr{S}$ 
under the action $\sigma$, i.e.
\begin{equation}
\label{eqn:209}
\mathscr{B}=\mathscr{L}\rtimes_{\sigma}\mathscr{S} .
\end{equation}
Historically, this semi-direct product structure was realized by \cite{Cant}. Succesively this idea was developed 
by \cite{GerBMS} who gave an incorrect formula for the action \eqref{eqn:212}. Eventually the mistake was amended 
by \cite{McC2}, who defined a good action to describe the semi-direct product structure of the BMS group. 
However, the idea used here of giving a right action and hence of describing the BMS group as a \textit{right} 
semi-direct product was not developed by any of these authors.\\
Furthermore, it can be shown that from the above discussion it follows that, by virtue 
of the first isomorphism theorem,
\begin{equation}
\label{eqn:234}
\mathscr{L}\cong\mathscr{B}/\mathscr{S},
\end{equation}
i.e. $\mathscr{L}$ is the factor group of $\mathscr{B}$ with respect to its normal subgroup $\mathscr{S}$. 
\begin{oss}$\\ $
Note that the structure of the BMS group is similar to that of Poincar\'e group, denoted by $\mathscr{P}$. 
In fact the Poincar\'e group can be expressed as the semi-direct product of the connected component 
of the Lorentz group $\mathscr{L}$ and the translations group $\mathscr{T}$, the former being the factor group of 
$\mathscr{P}$ with respect to the latter, i.e. $\mathscr{L}\cong\mathscr{P}/\mathscr{T}$. The action of 
$\mathscr{L}$ on $\mathscr{T}$ is the \lq natural\rq\hspace{0.1mm} one, i.e. the usual multiplication of 
an element $\Lambda\in \mathscr{L}$ with a vector $b\in\mathscr{T}$.
\end{oss}
\begin{thm}$\\ $
If $N'$ is a 4-dimensional normal subgroup of $\mathscr{B}$, then $N'$ is contained in $\mathscr{S}$.
\end{thm}
\begin{proof}
Consider the image $N'/\mathscr{S}$ of $N'$ under the homomorphism $\mathscr{B}\rightarrow\mathscr{B}/\mathscr{S}$. 
Since $N'$ by hypothesis is a normal subgroup of $\mathscr{B}$, $N'/\mathscr{S}$ is a normal subgroup of 
$\mathscr{B}/\mathscr{S}$ and hence, by proposition \ref{prop:BMS1}, a subgroup of the connected component 
of the Lorentz group $\mathscr{L}$. However, the only normal subgroups of $\mathscr{L}$ are $\mathscr{L}$ 
itself and the identity $e$ of $\mathscr{L}$. Then $N'$ must be $6$-dimensional, contrary to hypothesis. 
Therefore $N'/\mathscr{S}=e$; $N'$ is thus contained in $\mathscr{S}$. 
\end{proof}
\section{BMS Lie Algebra}
\label{sect:6.6}
In this section we are going to investigate the Lie Algebra of the BMS group. At first we consider the 
generators of $\mathrm{PSL}(2,\mathbb{C})$ of \eqref{eqn:220}. For an infinitesimal conformal transformation we know 
that the $x^A$ coordinates change as
\begin{equation*}
x^A\rightarrow x'^A=x^A+f^A,
\end{equation*}
where $f^A$ is a conformal Killing vector of \eqref{eqn:157} of the unit 2-sphere. Furthermore, from 
\eqref{eqn:220} and taking into account \eqref{eqn:214} an infinitesimal transformation for $u$ reads as
\begin{equation*}
u\rightarrow u'=Ku=e^{\frac{1}{2}D_Af^A}u\simeq u+\frac{u}{2}D_Af^A.
\end{equation*}
Thus, the generator of transformation \eqref{eqn:220} is
\begin{equation}
\label{eqn:221}
\xi_R=f^A\partial_A+\frac{u}{2}D_Af^A\partial_u.
\end{equation}
To see how their Lie algebra closes consider the Lie brackets of two of them, $\xi_{R_1}$ and $\xi_{R_2}$:
\begin{equation*}
[\xi_{R_1},\xi_{R_2}]=\left[f_1^A\partial_A+\frac{u}{2}D_Bf_1^B\partial_u,f_2^C\partial_C
+\frac{u}{2}D_Cf_2^C\partial_u\right]
\end{equation*}
\begin{equation*}
=(f_1^A\partial_Af_2^C-f_2^A\partial_Af_1^C)\partial_C+\frac{u}{2}(f_1^A\partial_AD_Cf_2^C
-f_2^C\partial_CD_Bf_1^B)\partial_u.
\end{equation*}
After some calculation the term proportional to $\partial_u$ becomes
\begin{equation*}
\frac{u}{2}D_C(f_1^A\partial_Af_2^C-f_2^A\partial_Af_1^C)+\frac{u}{2}(\partial_A\Gamma^{C}{}_{CB})
(f_1^Af_2^B-f_1^Bf_2^A).
\end{equation*}
The last term in the previous equation vanishes since it can be shown by direct calculation 
that for the metric $q_{AB}$
\begin{equation*}
\partial_A\Gamma^{C}{}_{CB}=-\frac{1}{\sin^2\theta}\delta^{\theta}_{A}\delta_B^{\theta},
\end{equation*}
and hence
\begin{equation*}
\left(\partial_A\Gamma^{C}{}_{CB}\right)(f_1^Af_2^B-f_1^Bf_2^A)=-\frac{1}{\sin^2\theta}(f_1^{\theta}f_2^{\theta}
-f_1^{\theta}f_2^{\theta})=0.
\end{equation*}
Finally we get 
\begin{subequations}
\label{eqn:222}
\begin{equation}
\label{eqn:223}
[\xi_{R_1},\xi_{R_2}]=\xi_{\hat{R}}=\hat{f}^A\partial_A+\frac{u}{2}D_A\hat{f}^A\partial_u
\end{equation}
where
\begin{equation}
\label{eqn:224}
\hat{f}^A=f_1^B\partial_Bf_2^A-f_2^B\partial_Bf_1^A.
\end{equation}
\end{subequations}
We take now into account the generators of supertranslations. It is clear from \eqref{eqn:225} that these are
\begin{equation*}
\xi_T=\alpha\partial_u.
\end{equation*}
where $\alpha$ is an arbitrary function of $\theta$ and $\phi$. It follows that the Lie brackets of two 
generators, $\xi_{T_1}$ and $\xi_{T_2}$ vanish, i.e.
\begin{equation}
\label{eqn:228}
[\xi_{T_1},\xi_{T_2}]=0,
\end{equation}
that is just a restatement that the supertranslations group is Abelian. The only thing left to do is to 
calculate the Lie brackets of $\xi_{R}$ and $\xi_{T}$. It is easy to see that 
\begin{subequations}
\label{eqn:226}
\begin{equation}
\label{eqn:227}
[\xi_R,\xi_T]=[f^A\partial_A+\frac{u}{2}D_Bf^B\partial_u,\alpha\partial_u]
=\xi_{\hat{T}}=\hat{\alpha}\partial_u,
\end{equation}
where
\begin{equation}
\label{eqn:228}
\hat{\alpha}=f^A\partial_A\alpha-\frac{\alpha}{2}D_Bf^B.
\end{equation}
\end{subequations}
If we consider now $\xi$ as defined in \eqref{eqn:158} it turns out $\xi=\xi_R+\xi_T$. 
From the above discussions one obtains that
\begin{eqnarray}
\label{eqn:229}
[\xi_1,\xi_2]&=& [\xi_{R_1},\xi_{R_2}]+[\xi_{R_1},\xi_{T_2}]+[\xi_{T_1},\xi_{R_2}]
\nonumber \\
&=& \hat{f}^A\partial_A+\frac{u}{2}D_A\hat{f}^A\partial_u+(\hat{\alpha}_{2}-\hat{\alpha}_1)\partial_u,
\end{eqnarray}
with 
\begin{equation*}
\hat{\alpha}_{2}=f^A_1\partial_A\alpha_2-\frac{\alpha_2}{2}D_Bf^B_1,\hspace{1cm}\hat{\alpha}_{1}
=f^A_2\partial_A\alpha_1-\frac{\alpha_1}{2}D_Bf^B_2.
\end{equation*}
To sum up the Lie algebra of the BMS group, $\mathfrak{bms_4}$, is
\begin{align*}
&[\xi_{R_1},\xi_{R_2}]=\xi_{\hat{R}},\hspace{0.5cm}\mathrm{with}\hspace{0.5cm}\hat{f}^A
=f_1^B\partial_Bf_2^A-f_2^B\partial_Bf_1^A; \\
&[\xi_{T_1},\xi_{T_2}]=0;\\
&[\xi_R,\xi_T]=\xi_{\hat{T}},\hspace{0.78cm}\mathrm{with}\hspace{0.5cm}\hat{\alpha}
=f^A\partial_A\alpha-\frac{\alpha}{2}D_Bf^B.
\end{align*}
Since, as shown in the previous section, the BMS group is a semi-direct product, it follows that the 
BMS Lie algebra, $\mathfrak{bms_4}$, should be taken to be the semi-direct sum of the Lie algebra of 
conformal Killing vectors $X=f^A\partial_A$ of the Riemann sphere, denoted by $\mathfrak{so(1,3)}$ (since it 
can be taken to be the algebra of $\mathscr{L}$) with that of the functions $\alpha(x^A)$ on the Riemann 
sphere, which we denote by $\mathscr{S}$, the supertranslations group being Abelian. Given an element 
$X=f^A\partial_A\in\mathfrak{so(1,3)}$ ($f^A$ being a generator of conformal transformations in \eqref{eqn:157}) 
we know that the exponential map associated to it, $e^{X}$, is an element of $\mathscr{L}$. 
Then consider the 1-parameter group of transformations in $\mathscr{S}$ defined as
\begin{equation}
\label{eqn:219}
\sigma_{e^{tX}}(\alpha)=\alpha e^{tX},
\end{equation}
where $\sigma$ is that of \eqref{eqn:212}. Consider the map
\begin{equation*}
\Sigma:f^A\partial_A\in\mathfrak{so(1,3)}\longrightarrow\Sigma_{f^A\partial_A}\in\mathrm{End}\mathscr{S},
\end{equation*}
such that 
\begin{equation*}
\Sigma_{f^A\partial_A}:\alpha\in\mathscr{S}\longrightarrow\Sigma_{f^A\partial_A}(\alpha)=\frac{d}{dt}	
\left.(\sigma_{e^{tf^A\partial_A}}(\alpha))\right|_{t=0}\in\mathscr{S}.
\end{equation*}
\\
Note that $\Sigma_{f^A\partial_A}(\alpha)$ is the infinitesimal generator of \eqref{eqn:219}. 
Hence \cite{OblakBMS} we have 
\begin{equation}
\label{eqn:213}
\mathfrak{bms_4}=\mathfrak{so(1,3)}\oplus_{\Sigma}\mathscr{S}.
\end{equation}
The Lie algebra $\mathfrak{bms_4}$ is determined by three arbitrary functions $f^A$ and $\alpha$ on the circle. 
Thus, defining $X=f^A\partial_A$ and labelling elements of \eqref{eqn:213} as pairs $(X,\alpha)$, 
we know that the Lie brackets in $\mathfrak{so(1,3)}\oplus_{\Sigma}\mathscr{S}$ are
\begin{equation}
\label{eqn:215}
[(X_1,\alpha_1),(X_2,\alpha_2)]=([X_1,X_2],\Sigma_{f_1^A\partial_A}(\alpha_2)-\Sigma_{f^B_2\partial_B}(\alpha_1)).
\end{equation}
Equation \eqref{eqn:215} follows from the fact the $\mathscr{S}$ is Abelian, otherwise there would be an extra 
term involving the commutator of the two elements $\alpha_1$ and $\alpha_2$.
Since we have, using \eqref{eqn:212} and \eqref{eqn:214}
\begin{equation*}
\Sigma_{f^A\partial_A}(\alpha)(x^B)=\frac{d}{dt}\left.(K^{-1}_{e^{tf^A\partial_A}}\alpha
(e^{tf^C\partial_C}x^B))\right|_{t=0}=\frac{d}{dt}\left.(e^{-\frac{1}{2}tD_Af^A}
\alpha(e^{tf^C\partial_C}x^B))\right|_{t=0}
\end{equation*}
\begin{equation*}
=-\frac{\alpha}{2}D_Af^A+ f^B\partial_B\alpha,
\end{equation*}
then \eqref{eqn:215} may be written as
\begin{subequations}
\label{eqn:230}
\begin{equation}
\label{eqn:231}
[(X_1,\alpha_1),(X_2,\alpha_2)]=(\hat{X},\hat{\alpha}),
\end{equation}
with 
\begin{equation}
\label{eqn:232}
\hat{f}^A=f_1^B\partial_Bf_2^A-f_2^B\partial_Bf_1^A,\hspace{1cm}\hat{\alpha}=f_1^B\partial_B\alpha_2
-\frac{\alpha_2}{2}D_Af^A_1-(1\leftrightarrow2),
\end{equation}
\end{subequations}
as remarked in \cite{Barn2010,Barn2010a}.
\begin{oss}$\\ $
Note that this result, obtained from the theory of semi-direct product of groups and their Lie algebra, 
is in complete agreement with that obtained just by looking at the generators, expressed in \eqref{eqn:229}. 
Note also that $\hat{f}^A$ of \eqref{eqn:232} coincides with of \eqref{eqn:224} and that 
$\hat{\alpha}=\hat{\alpha}_2-\hat{\alpha}_1$.
\end{oss}
Depending on the space of functions under consideration, there are many options which define what is actually 
meant by $\mathfrak{bms_4}$. The approach that will be followed in this work is originally due to 
\cite{Sachs1} and successively amended by \cite{Ant91}. Another approach, based on the Virasoro algebra, 
can be found in \cite{Barn2010}. \\
In general, we consider any $S$-dimensional Lie transformation group of a $R$-dimensional space. Let 
the coordinates of the space be $y^{\alpha}$ $(\alpha=1,..,R)$ and the parameters of the group be 
$z^{\mu}$ $(\mu=1,...,S)$, where $z^{\mu}=0$ is the identity of the group. Then the transformations have the form 
\begin{equation*}
y'^{\alpha}=f^{\alpha}(y^{\beta};z^{\mu}),\hspace{1cm}\mathrm{where}\hspace{1cm}f^{\alpha}(y^{\beta};0)=y^{\alpha}.
\end{equation*}
The functions $f^{\alpha}$ are assumed to be twice differentiable. The $S$ generators of the group are 
the vector fields given by
\begin{equation}
\label{eqn:217}
P_{\mu}=\left.\frac{\partial f^{\alpha}}{\partial z^{\mu}}\right|_{z^{\mu}=0}\frac{\partial}{\partial y^{\alpha}}.
\end{equation}
Applying these ideas to the BMS group, with the Sachs's notations \cite{Sachs1}, 
it gives for the supertranslations, using the expansion \eqref{eqn:206}:
\begin{equation*}
P_{lm}=Y_{lm}(\theta,\phi)\frac{\partial}{\partial u},\hspace{1cm}P_{lm}=(-1)^m\bar{P}_{l-m},
\end{equation*}
and hence
\begin{equation*}
[P_{lm},P_{nr}]=0,
\end{equation*}
i.e. two supertranslations commute. \\
To find the generators of conformal transformations we have to be careful. We know that any conformal 
transformation has the form 
\begin{equation*}
\zeta'=\frac{a\zeta+b}{c\zeta+d},\hspace{1cm}\zeta=e^{i\phi}\cot\frac{\theta}{2}.
\end{equation*}
By direct calculations one obtains that the following equations hold for $\theta'$, $\phi'$ and $u'$:
\begin{subequations}
\label{eqn:216}
\begin{align}
&\theta'=2\arctan\left[\left(\frac{|c\zeta+d|^2}{|a\zeta+b|^2}\right)^{1/2}\right],\\ \nonumber\\
&\phi'=\arctan\left[\frac{\mathrm{Im}\{(a\zeta+b)(\bar{c}\bar{\zeta}+\bar{d})\}}
{\mathrm{Re}\{(a\zeta+b)(\bar{c}\bar{\zeta}+\bar{d})\}}\right],\\ \nonumber\\
&u'=\frac{1+|\zeta|^2}{|a\zeta+b|^2+|c\zeta+d|^2}u.
\end{align}
\end{subequations}
On denoting by $x$ the parameter of the transformation it is clear that $a,b,c,d$ are functions of $x$ such 
that $a(0)=d(0)=1$ and $c(0)=b(0)=0$. We have to apply \eqref{eqn:217} to \eqref{eqn:216}. 
It is easy to verify that 
\begin{align*}
\left.\frac{d\theta'}{dx}\right|_{x=0}&=\frac{\cos^3\frac{\theta}{2}}{\sin\frac{\theta}{2}}
\frac{d}{dx}\left.\left[\frac{|c|^2|\zeta|^2+|d|^2+c\bar{d}\zeta+\bar{c}d\bar{\zeta}}
{|a|^2|\zeta|^2+|b|^2+a\bar{b}\zeta+\bar{a}b\bar{\zeta}}\right]\right|_{x=0},\\ \\
\left.\frac{d\phi'}{dx}\right|_{x=0}&=\cos^2\phi\frac{d}{dx}\left.\left[\frac
{\mathrm{Im}\{a\bar{c}|\zeta|^2+a\bar{d}\zeta+b\bar{c}\bar{\zeta}+b\bar{d}\}}
{\mathrm{Re}\{a\bar{c}|\zeta|^2+a\bar{d}\zeta+b\bar{c}\bar{\zeta}+b\bar{d}\}}\right]\right|_{x=0},\\ \\
\left.\frac{du'}{dx}\right|_{x=0}&=\sin^2\frac{\theta}{2}\frac{d}{dx}\left.\left[(|a|^2+|c|^2)|\zeta|^2
+(a\bar{b}+c\bar{d})\zeta+(b\bar{a}+d\bar{c})\bar{\zeta}+|b|^2+|d|^2\right]\right|_{x=0}u.
\end{align*}
These equations hold in general for any conformal transformation. We choose now to work with Lorentz 
transformations, and thus to use Lorentz generators $L_i$ and $R_i$ of rotations and boosts, respectively. 
To know the coefficients $a,b,c,d$ corresponding to a Lorentz transformation we need to use the isomorphism 
\eqref{eqn:210}. Any rotation of an angle $\varphi$ about an axis $\hat{n}$ and any boost of rapidity 
$\chi$ about an axis $\hat{e}$ can be done by using a $\mathrm{SL}(2,\mathbb{C})$ matrix given by
\begin{subequations}
\label{eqn:236}
\begin{align}
&U_{\hat{n}}(\varphi)=e^{\frac{i}{2}\varphi\hat{n}\cdot\vec{\sigma}}=\mathbb{I}\cos\frac{\varphi}{2}
+i\hat{n}\cdot\vec{\sigma}\sin\frac{\varphi}{2},\\
&H_{\hat{e}}(\chi)=e^{\frac{1}{2}\chi\hat{e}\cdot\vec{\sigma}}=\mathbb{I}\cosh\frac{\chi}{2}
+\hat{e}\cdot\vec{\sigma}\sinh\frac{\chi}{2},
\end{align}
\end{subequations}
respectively, where $\vec{\sigma}=(\sigma_x,\sigma_y,\sigma_z)$ are the Pauli matrices of \eqref{eqn:pauli}. The parameter 
$x$ of the two transformations is $\varphi$ and $\chi$, respectively. After some calculations we find that the vector fields that generate the transformations are 
\begin{align}
L^{23}=L_x&=-\sin\phi\frac{\partial}{\partial\theta}-\cot\theta\cos\phi\frac{\partial}{\partial \phi},\\
L^{13}=L_y&=-\cos\phi\frac{\partial}{\partial\theta}+\cot\theta\sin\phi\frac{\partial}{\partial \phi},\\
L^{12}=L_z&=\frac{\partial}{\partial \phi},\\
L^{10}=R_x&=\cos\theta\cos\phi\frac{\partial}{\partial\theta}-\frac{\sin\phi}{\sin\theta}\frac{\partial}{\partial\phi}
-u\sin\theta\cos\phi\frac{\partial}{\partial u},\\
L^{20}=R_y&=-\cos\theta\sin\phi\frac{\partial}{\partial\theta}-\frac{\cos\phi}{\sin\theta}\frac{\partial}{\partial\phi}
+u\sin\theta\sin\phi\frac{\partial}{\partial u},\\
L^{30}=R_z&=-\sin\theta\frac{\partial}{\partial \theta}-u\cos\theta\frac{\partial}{\partial u}.
\end{align}
Note that rotations are characterized by $K=1$. The $\{P_{lm}\}$
and $\{L^{ab}\}$ form a complete set of linearly independent vector fields for the Lie algebra 
$\mathfrak{bms}_4$. We can find now the commutators
\begin{equation*}
[L^{ab},L^{cd}]=\eta^{ac}L^{bd}+\eta^{bd}L^{ac}-\eta^{ad}L^{bc}-\eta^{bc}L^{ad},
\end{equation*}
\begin{equation*}
[L_i,L_j]=\epsilon_{ijk}L_k,\hspace{1cm}[R_i,R_j]=-\epsilon_{ijk}L_k,\hspace{1cm}[L_i,R_j]=-\epsilon_{ijk}R_k,
\end{equation*}
where $\eta^{ab}=\mathrm{diag}(1,-1,-1,-1)$ and $\epsilon_{ijk}$ is the Levi-Civita symbol. Note that we have 
just obtained the classical Lorentz algebra. Furthermore, it is easy to derive the following commutator:
\begin{equation}
\label{eqn:218}
\left[L^{ab},\alpha(\theta,\phi)\frac{\partial}{\partial u}\right]=\left[L^{ab}\alpha(\theta,\phi)
-\alpha(\theta,\phi)W(L^{ab})\right]\frac{\partial}{\partial u},
\end{equation}
where $W(L^{ab})$ is defined by the relation 
\begin{equation*}
\frac{\partial}{\partial u}(L^{ab}f)=L^{ab}\frac{\partial f}{\partial u}+W(L^{ab})\frac{\partial f}{\partial u},
\end{equation*}
for arbitrary $f(u)$. \\
For convenience we introduce the raising and lowering operators,
\begin{align*}
L^{\pm}&=L_y\pm iL_x=-e^{\pm i\phi}\left(\frac{\partial}{\partial \theta}
\pm i\cot\theta\frac{\partial}{\partial \phi}\right),\\
R^{\pm}&=R_x\mp i R_y=e^{\pm i\phi}\left(\cos\theta\frac{\partial}{\partial \theta}\pm\frac{i}{\sin\theta}
\frac{\partial}{\partial \phi}-u\sin\theta\frac{\partial}{\partial u}\right), 
\end{align*}
in terms of which, using equation \eqref{eqn:218}, we give the commutation relations 
with the generators of the supertranslations:
\begin{align*}
&[L_z,P_{lm}]=imP_{lm},\\ \nonumber\\
&[L^+,P_{lm}]=-\sqrt{(l-m)(l+m+1)}P_{l,m+1},\\ \nonumber\\
&[L^-,P_{lm}]=\sqrt{(l+m)(l-m+1)}P_{l,m-1},
\end{align*}
\begin{align*}
&[R_z,P_{lm}]=-(l-1)\sqrt{\frac{(l+m+1)(l-m+2)}{(2l+1)(2l+3)}}P_{l+1,m}\\
&+(l+2)\sqrt{\frac{(l+m)(l-m)}{4l^2-1}}P_{l-1,m},\\ \\
&[R^+,P_{lm}]=(l-1)\sqrt{\frac{(l+m+1)(l+m+2)}{(2l+1)(2l+3)}}P_{l+1,m+1}\\
&+(l+2)\sqrt{\frac{(l-m-1)(l-m)}{4l^2-1}}P_{l-1,m+1},\\ \\
&[R^-,P_{lm}]=-(l-1)\sqrt{\frac{(l-m+1)(l-m+2)}{(2l+1)(2l+3)}}P_{l+1,m-1}\\
&-(l+2)\sqrt{\frac{(l+m-1)(l+m)}{4l^2-1}}P_{l-1,m-1}.
\end{align*}
The form of the commutation relations shows that the BMS algebra is the semi-direct sum of the Lorentz algebra 
$\mathfrak{so(1,3)}$ with the infinite Lie algebra $\mathscr{T}$, as remarked before.
\section{Good and Bad Cuts}
\label{sect:6.7}
We begin this section by citing a remarkable result obtained by Sachs.
\begin{thm}{\cite{Sachs1}}$\\ $
\label{thm:transl}
The only $4$-dimensional normal subgroup of the BMS group is the translation group.
\end{thm}
Theorem \ref{thm:transl} characterizes translations uniquely: the translations normal subgroup of the BMS group 
is singled out by its group-theoretic properties. Since we have shown that the translations $\mathscr{T}$ are 
the BMS transformations induced on $\mathscr{I}^+$ by translations in Minkowski space-time, theorem 
\ref{thm:transl} makes it possible for  us to define the asymptotic translations 
of a general asymptotically flat space-time 
as the BMS elements belonging to this normal subgroup. However a similar procedure for $\mathscr{L}$, 
i.e. rotations and boosts, fails. Thus, as we will discuss in this section, there are several problems in 
identifying the Poincar\'e group as a subgroup of $\mathscr{B}$.\\
The Poincar\'e group is the symmetry group of flat space-time, hence it might have been thought that a suitably 
asymptotically flat space-time should, in some appropriate sense, have the Poincar\'e group as an asymptotic 
symmetry group. Instead, it turns out that in general we seem only to obtain the BMS group (which has  
the unpleasant feature of being an infinite-dimensional group) as the asymptotic symmetry group of an 
asymptotically flat space-time. \\
To better understand the nature of this problem we revert to Minkowski space-time and see how the Poincar\'e 
group arises in that case as a subgroup of the BMS group. The BMS group was defined as the group of 
transformations which conformally preserves the induced metric on $\mathscr{I}^+$ and the strong conformal 
geometry. However, the BMS group is much larger than the Poincar\'e group and thus the former must preserve 
less structure on $\mathscr{I}^+$ than does the latter. The preservation of this additional structure, 
in the case of Minkowski space-time, should allow us to restrict the BMS transformations to Poincar\'e 
transformations, since we know that $\mathscr{P}$ in that case \textit{is} a subgroup of $\mathscr{B}$.\\
In Minkowski space-time a null hypersurface is said to be a \textit{good cone} if it is the future light cone 
of some point, and a \textit{bad cone} if its generators do not meet at a point. Consequently we define a 
\textit{good cross-section}, often called a \textit{good cut}, a cross-section of $\mathscr{I^+}$ which is 
the intersection of a future light cone of some point and the null hypersurface $\mathscr{I}^+$. 
A \textit{bad cut} is, on the other hand, the intersection of $\mathscr{I}^+$ with some null hypersurface 
which does not come together cleanly at a single vertex. The situation is reported in Figure \ref{fig:6.4}.
\begin{figure}[h]
\begin{center}
\includegraphics[scale=0.5]{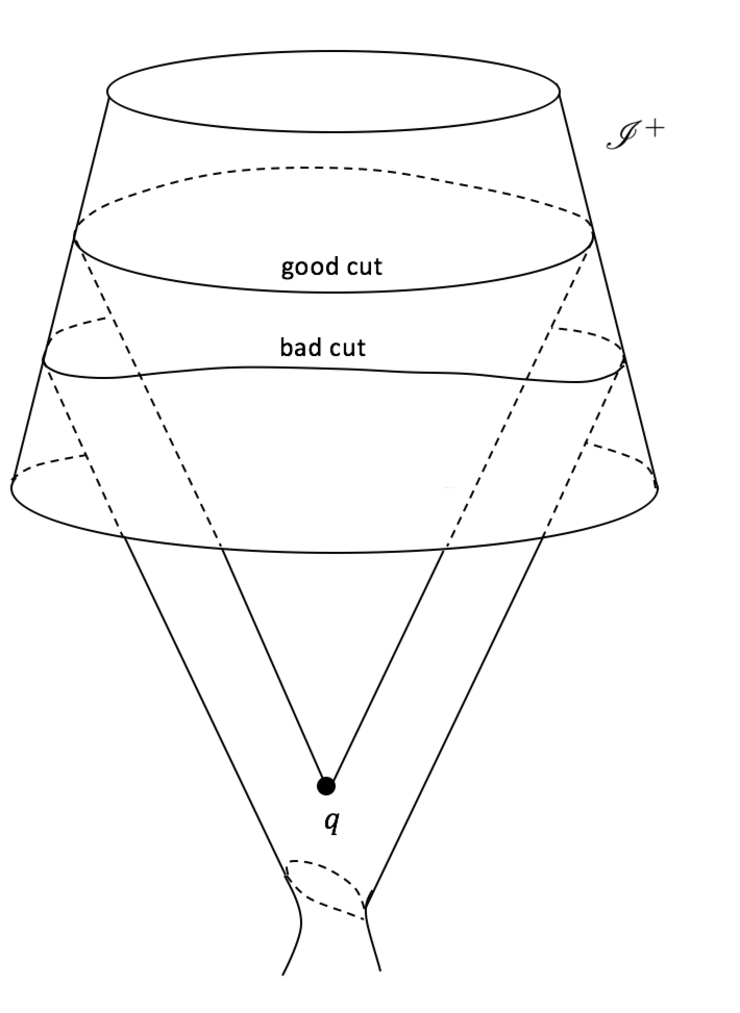}
\caption{A good cross-section of $\mathscr{I}^+$ for Minkowski space-time is one arising as the intersection 
of $\mathscr{I}^+$ with the future light cone of a point.}
\label{fig:6.4}
\end{center}
\end{figure}
\medskip \medskip \\ Using Bondi-Sachs coordinates the Minkowski metric tensor takes the form
\begin{equation*}
g=du\otimes du+du\otimes dr+dr\otimes du-r^2(d\theta\otimes d\theta+\sin^2\theta d\phi\otimes d\phi).
\end{equation*}
We see that each cut of $\mathscr{I}^+$ given by $u=t-r=\mathrm{const}$ is a good cut, since it arises from 
the future light cone of a point on the origin-axis $r=0$. In particular, all the good cuts can be obtained 
from the one given by $u=0$ by means of a space-time translation. Hence, as discussed in section \ref{sect:6.5} we 
obtain that every good cut can be expressed in the form
\begin{equation*}
u=\frac{A+B\zeta+\bar{B}\bar{\zeta}+C\zeta\bar{\zeta}}{1+\zeta\bar{\zeta}}
\end{equation*}
\begin{equation}
\label{eqn:235}
=\left(\frac{A+C}{2}\right)+\left(\frac{C-A}{2}\right)\cos\theta+\left(\frac{B+\bar{B}}{2}\right)
\sin\theta\cos\phi+i\left(\frac{B-\bar{B}}{2}\right)\sin\theta\sin\phi,
\end{equation}
$A,C$ being real and $B$ being complex. Hence the equations describing good cuts are given, generally, 
by setting $u$ equal to a function of $\theta$ and $\phi$ which consists only of zeroth- 
and first-order spherical harmonics.\\
The effect of a transformation of the connected component of the Lorentz group is to leave invariant the 
particular good cut $u=0$. Such a transformation is
\begin{subequations}
\label{eqn:237}
\begin{align}
&\zeta\rightarrow\zeta'=\frac{a\zeta+b}{c\zeta+d},\\
&u\rightarrow u'=\frac{1+\zeta\bar{\zeta}}{|a\zeta+b|^2+|c\zeta+d|^2}u,
\end{align}
\end{subequations}
with $a,b,c,d\in\mathbb{C}$ such that $ad-bc=1$. Note that transformations \eqref{eqn:237} preserve the 
functional form of good cuts given in \eqref{eqn:235}. In fact we have, applying \eqref{eqn:237}, that
\begin{equation*}
u'=\left(\frac{A+B\zeta+\bar{B}\bar{\zeta}+C\zeta\bar{\zeta}}{1+\zeta\bar{\zeta}}\right)
\frac{1+\zeta\bar{\zeta}}{|a\zeta+b|^2+|c\zeta+d|^2},
\end{equation*}
where $\zeta$ and $\bar{\zeta}$ have now to be expressed as functions of $\zeta'$ and $\bar{\zeta}'$. 
It is straightforward to show that
\begin{equation*}
u'=\frac{A'+B'\zeta'+\bar{B}'\bar{\zeta}'+C'\zeta'\bar{\zeta}'}{1+\zeta'\bar{\zeta}'},
\end{equation*}
where 
\begin{align*}
&A'=A|a|^2-Bb\bar{a}-\bar{B}\bar{b}a+C|b|^2,\\
&B'=Bd\bar{a}+\bar{B}\bar{b}c-A\bar{a}c-Cd\bar{b}	,\\
&C'=A|c|^2-B\bar{c}d-\bar{B}c\bar{d}+C|d|^2.
\end{align*}
For example, if we perform a boost in the $z$ direction we have from \eqref{eqn:236} 
$a=e^{\chi/2},d=e^{-\chi/2},c=b=0$. Hence we get $A'=e^{\chi}A,B'=B,C'=e^{-\chi}C$. \\
Now it is clear that the general BMS transformation which maps good cuts into good cuts must obtain the particular 
good cut $u=0$ from some other good cut. We can therefore express the BMS transformation as the composition 
of a translation which maps this other good cut into $u=0$, with a Lorentz transformation which leaves 
$u=0$ invariant. Thus, the BMS transformation is
\begin{subequations}
\label{eqn:238}
\begin{align}
&\zeta\rightarrow\zeta'=\frac{a\zeta+b}{c\zeta+d},\\
&u\rightarrow u'=\left(\frac{1+\zeta\bar{\zeta}}{|a\zeta+b|^2+|c\zeta+d|^2}\right)\left(u+\frac{A+B\zeta
+\bar{B}\bar{\zeta}+C}{1+\zeta\bar{\zeta}}\right).
\end{align}
\end{subequations}
These BMS transformations form a $10$-real-parameters group. This is exactly the Poincar\'e group of Minkowski 
space-time, being a composition of a Lorentz transformation and a translation. We have obtained the following 
\begin{prop}$\\ $
The Poincar\'e group $\mathscr{P}$ is the group of transformations which maps good cuts into good 
cuts in Minkowski space-time.
\end{prop}
However, there are many other subgroups of the BMS group which can be expressed in the form \eqref{eqn:238} and 
which are therefore isomorphic with the Poincar\'e group. In fact $\mathscr{L}$ is not a normal subgroup of the 
BMS group since for any $b=(\Lambda,\alpha)\in\mathscr{B}$ and for any $\Lambda'=(\Lambda',0)\in\mathscr{L}$ 
the product $b\Lambda'b^{-1}$ is not necessarily an element of $\mathscr{L}$, as can be easily verified, and 
hence $\mathscr{L}$ does not get canonically singled out, occurring only as a factor group of $\mathscr{B}$ 
by the infinite-parameter Abelian group of supertranslations $\mathscr{S}$. In particular, Lorentz transformations 
do not commute with supertranslations and if we take any supertranslation $s$ and consider the group 
$\mathscr{L}'=s\mathscr{L}s^{-1}$ then it is a subgroup of the BMS group which is distinct from $\mathscr{L}$ 
but still isomorphic, and thus equivalent, to it. Explicitly, having fixed a supertranslation $s=(\mathbb{I},\alpha)$, 
a transformation of $\mathscr{L}'$ reads as
\begin{subequations}
\label{eqn:239}
\begin{align}
&\zeta\rightarrow\zeta'=\frac{a\zeta+b}{c\zeta+d},\\
&u\rightarrow u'=\left(\frac{1+\zeta\bar{\zeta}}{|a\zeta+b|^2+|c\zeta+d|^2}\right)\left(u-\alpha\right)+\alpha.
\end{align}
\end{subequations}
If we start with the good cut described by the equation $u=0$, which is left invariant by $\mathscr{L}$, and 
perform the supertranslation $s$, we obtain a new (bad) cut given by $u=\alpha$. This is the cut which is left 
invariant by \eqref{eqn:239} and hence by $\mathscr{L}'$. Hence $\mathscr{L}'$ maps bad cuts into bad cuts.
It follows that if we conjugate the whole Poincar\'e group $\mathscr{P}$ of \eqref{eqn:238} with respect to 
any supertranslation $s$ which is not a translation obtaining $\mathscr{P}'=s\mathscr{P}s^{-1}$ we get a 
distinct subgroup of $\mathscr{B}$, but isomorphic and completely equivalent to $\mathscr{P}$, which maps 
bad cuts into bad cuts. Of course, for a general $s$, $\mathscr{P'}$ and $\mathscr{P}$ have only the 
translations $\mathscr{T}$ in common. There exist \textit{many} subgroups of $\mathscr{B}$ which are 
isomorphic with $\mathscr{P}$ and hence the Poincar\'e group \textit{is not} a subgroup of the BMS group 
in a canonical way. However we have just seen that in Minkowski space-time, if we require the group of 
transformations to preserve the conformal nature of $\mathscr{I}^+$ and the strong conformal geometry, together 
with the property of mapping good cuts into good cuts, just one of the several copies of the 
Poincar\'e group gets singled out.\\
\begin{oss}$\\ $
Note that, again, the situation is similar to what happens for $\mathscr{L}$ within $\mathscr{P}$. In fact 
$\mathscr{L}$ does not arise naturally as a subgroup of $\mathscr{P}$, since if we form the group 
$\mathscr{L}'=t\mathscr{L}t^{-1}$, where $t$ is a translation, it is a different subgroup of $\mathscr{P}$ 
but isomorphic to $\mathscr{L}$. We can say, since the commutator of a Lorentz transformation and a 
translation is a translation, that $\mathscr{L}$, as a subgroup of $\mathscr{P}$, depends on the choice 
of an arbitrary origin in Minkowski space-time. 
\end{oss}
We turn now to the case when the space-time is asymptotically flat. The difficulty here is that there seems 
to be no suitable family of cuts that can properly take over the role of Minkowskian good cuts. This means 
that, although the translation elements of $\mathscr{B}$ are canonically singled out, there is no canonical 
concept of a \lq supertranslation-free\rq\hspace{0.1mm} Lorentz transformation. Hence the notion of a 
\lq pure translation\rq\hspace{0.1mm} still makes sense, but that of \lq pure rotation\rq\hspace{0.1mm} 
or \lq pure boost\rq\hspace{0.1mm} does not. \\
The most obvious generalization, for an asymptotically flat space-time, of the Minkowskian definition 
of a good cut, i.e. the intersection of future light cone of a point with $\mathscr{I}^+$, is totally 
inappropriate. One first reason is that there are many perfectly reasonable asymptotically flat space-times 
in which no cuts of $\mathscr{I}^+$ at all would arise in this way, e.g. \cite{Pen72}. Even if we 
restrict attention only to asymptotically flat space-times which do contain a reasonable number of good cuts 
of this kind, we are not likely to obtain any of the BMS transformations (apart from the identity) which 
maps this system of cuts into itself. The difficulty lies in the fact that the detailed irregularities of 
the interior of the space-time would be reflected in the definition of \lq goodness\rq\hspace{0.1mm} 
of a cut. In other words, the light-cone cuts are far more complicated than those \eqref{eqn:235} of flat space. \\
However there is a more satisfactory way to characterize good cuts of $\mathscr{I}^+$, based on the shear (see section \ref{sect:2.8}) of null hypersurfaces intersecting $\mathscr{I}^+$. Consider now that the (physical) space-time under consideration contains a null curve $\mu$ of a null geodesics congruence $\mathscr{C}$ affinely parametrized by $\tilde{r}$ and whose tangent null vector is $\tilde{l}^a=\tilde{o}^A\bar{\tilde{o}}^{A'}$. Complete $\tilde{o}^A$ to a spin basis $(\tilde{o}^A,\tilde{\iota}^A)$ at a point $p$ of $\mu$. We can propagate $\tilde{o}^A$ and $\tilde{\iota}^A$ along $\mu$ via
\begin{equation*}
\tilde{D}\tilde{o}^A=0,\hspace{1cm}\tilde{D}\tilde{\iota}^A=0,
\end{equation*}
where $\tilde{D}=\tilde{l}^a\tilde{\nabla}_a$. From the results of section \ref{sect:2.8} we have that $\tilde{\kappa}=\tilde{\epsilon}=\tilde{\pi}=0$. From the first two equations Newman-Penrose field equations in Appendix \ref{B} we obtain for the shear $\tilde{\sigma}$ and for the divergence $\tilde{\rho}$ the so-called \textit{optical equations} of Sachs \citep{Sachs61,Sachs62}:
\begin{subequations}
\label{eqn:240}
\begin{align}
&\frac{\partial \tilde{\rho}}{\partial \tilde{r}}=\tilde{\rho}^2+\tilde{\sigma}\bar{\tilde{\sigma}}+\tilde{\Phi}_{00},\\ \nonumber\\
&\frac{\partial\tilde{\sigma}}{\partial \tilde{r}}=(\tilde{\rho}+\bar{\tilde{\rho}})\tilde{\sigma}+\tilde{\Psi}_0.
\end{align}
\end{subequations}
Introducing the matrices
\begin{equation*}
\textbf{Q}=\left(\begin{matrix}\tilde{\Phi}_{00} & \tilde{\Psi}_0\\\bar{\tilde{\Psi}}_0 & \tilde{\Phi}_{00}\end{matrix}\right),\hspace{1cm}\textbf{P}=\left(\begin{matrix}
\tilde{\rho} & \tilde{\sigma} \\ \bar{\tilde{\sigma}}& \bar{\tilde{\rho}}\end{matrix}\right),
\end{equation*}
then, equations \eqref{eqn:240} may be written as
\begin{equation}
\label{eqn:241}
D\textbf{P}=\textbf{P}^2+\textbf{Q}.
\end{equation}
Consider first the case of a flat space-time, i.e. $\textbf{Q}=0$.
Then equation \eqref{eqn:241} reduces to 
\begin{equation*}
D\textbf{P}=\textbf{P}^2\Longrightarrow\textbf{P}^{-1}(D\textbf{P})\textbf{P}^{-1}=\textbf{I},
\end{equation*}
where $\textbf{I}$ is the ($2\times 2$) unit matrix. But
\begin{equation*}
0=D(\textbf{P}^{-1}\textbf{P})=(D\textbf{P}^{-1})\textbf{P}+\textbf{P}^{-1}D\textbf{P},
\end{equation*}
hence we obtain
\begin{equation*}
D\textbf{P}^{-1}=-\textbf{I}\Longrightarrow\textbf{P}^{-1}=\textbf{A}-r\textbf{I},
\end{equation*}
where $\textbf{A}=\mathrm{const}$ and $r$ is affine along the congruence, i.e. $Dr=1$.
\begin{equation*}
\textbf{P}=\left(\textbf{A}-r\textbf{I}\right)^{-1}.
\end{equation*}
Writing 
\begin{equation*}
\textbf{A}=\left(\begin{matrix} \tilde{\rho}_0& \tilde{\sigma}_0\\ \bar{\tilde{\sigma}}_0 & \bar{\tilde{\rho}}_0\end{matrix}\right)^{-1},
\end{equation*}
where $\tilde{\rho}_0$ and $\tilde{\sigma}_0$ are the values of $\tilde{\rho}$ and $\tilde{\sigma}$ at $r=0$, we get, explicitly
\begin{subequations}
\label{eqn:242}
\begin{align}
\tilde{\rho}=\frac{\tilde{\rho}_0-\tilde{r}(\tilde{\rho}_0\bar{\tilde{\rho}}_0-\tilde{\sigma}_0\bar{\tilde{\sigma}}_0)}{1-\tilde{r}\left(\tilde{\rho}_0+\bar{\tilde{\rho}}_0\right)+\tilde{r}^2\left(\tilde{\rho}_0\bar{\tilde{\rho}}_0-\tilde{\sigma}_0\bar{\tilde{\sigma}}_0\right)},\\
\tilde{\sigma}=\frac{\tilde{\sigma}_0}{1-\tilde{r}\left(\tilde{\rho}_0+\bar{\tilde{\rho}}_0\right)+\tilde{r}^2\left(\tilde{\rho}_0\bar{\tilde{\rho}}_0-\tilde{\sigma}_0\bar{\tilde{\sigma}}_0\right)}.
\end{align}
\end{subequations}
Thus $\tilde{\sigma}$ has the following asymptotic behaviour:
\begin{equation*}
\tilde{\sigma}=\frac{\tilde{\sigma}_0}{\tilde{r}^2\left(\tilde{\rho}_0\bar{\tilde{\rho}}_0-\tilde{\sigma}_0\bar{\tilde{\sigma}}_0\right)}+O(\tilde{r}^{-4})=\frac{\tilde{\sigma}^0}{\tilde{r}^2}+O(\tilde{r}^{-4}),
\end{equation*}
where $\tilde{\sigma}^0=\tilde{\sigma}_0/(\tilde{\rho}_0\bar{\tilde{\rho}}_0-\tilde{\sigma}_0\bar{\tilde{\sigma}}_0)$ is called \textit{asymptotic shear}. Note that in a flat space-time the vanishing of the asymptotic shear implies the vanishing of the whole shear. If we now turn to the case of asymptotically flat space-times with non-vanishing $\tilde{\Phi}_{00}$ and $\tilde{\Psi}_{0}$, it can be shown, in virtue of their asymptotic behaviours, that the leading term of the asymptotic behaviour of $\tilde{\sigma}$ does not change. However the vanishing of $\tilde{\sigma}^0$ does not, in this non-flat case, imply that $\tilde{\sigma}$ vanishes \citep{New2012}.\\
Note that, in virtue of thorem \ref{thm:shear}, if we consider the unphysical space-time obtained with conformal factor $\Omega=\tilde{r}^{-1}$ the shear transforms as
\begin{equation*}
\sigma=\Omega^{-2}\tilde{\sigma}\Longrightarrow\left.\sigma\right|_{_{\mathscr{I}^+}}=\tilde{\sigma}^0.
\end{equation*}
In Minkowski space-time the good cones are characterized locally by the fact that the null rays generating them possess no shear and it can be shown that the cuts of $\mathscr{I}^+$ which we defined earlier are precisely the ones arising from the intersection of $\mathscr{I}^+$ with null hypersurfaces characterized by $\tilde{\sigma}^0=0$. Thus a definition of \lq goodness\rq\hspace{0.1mm} is provided, for Minkowski space-time, which refers only to quantities defined asymptotically. \\
We would like to extend this definition of good cut to asymptotically flat space-times too. To do that all we need to do is to set up $l^a$ and $m^a$ (see section \ref{sect:2.3}) at each point of the cut at $\mathscr{I}^+$, where $m^a$ is complex null, with real and imaginary parts both tangent to the cut and where $l^a$ is real null and orthogonal to $m^a$, i.e. $m_al^a=0$. The cut is a good cut if the complex shear $\sigma=m^am^b\nabla_al_b$ so obtained equals zero (since on $\mathscr{I}^+$ we have $\sigma=\tilde{\sigma}^0$).\\
We cannot, however, define good cones simply by requiring $\tilde{\sigma}^0=0$. In many cases it is not possible to arrange $\tilde{\sigma}^0=0$ for all values of $\theta$ and $\phi$. But even in cases where it is possible we have another problem, which is due to the presence of gravitational radiation. To make this point clear, we cite now some important results regarding the relation between asymptotic shear and gravitational radiation which are basically due to \cite{Bondi62,NewUn,Sachs62}. At first $\tilde{\sigma}^0$ forms part of the initial data on $u=0$ used to determine the space-time asymptotically. We introduce the \textit{Bondi news function} $N$:
\begin{equation*}
N=-\frac{1}{2}R_{ab}\bar{m}^a\bar{m}^b.
\end{equation*}
Clearly, $R_{ab}$ need not vanish near $\mathscr{I}^+$ even though $\tilde{R}_{ab}$ does. It turns out that
\begin{equation*}
\frac{\partial \tilde{\sigma}^0}{\partial u}=-\bar{N},\hspace{1cm}\frac{\partial^2 \tilde{\sigma}^0}{\partial u^2}=\frac{d\bar{N}}{du}=-\bar{\Psi}^{(0)}_4
\end{equation*}
where $\bar{\Psi}^{(0)}_4$ is the (complex conjugate) gravitational radiation field introduced in section \ref{sect:5.4}. The news function $N$ enters in the Bondi-Sachs definition of mass-momentum. In fact let us consider a hypersurface $\mathscr{S}$ which spans some two-dimensional cross-section $S$ of $\mathscr{I}^+$. That is to say, $S=\dot{\mathscr{S}}=\mathscr{S}\cap\mathscr{I}^+$. The total energy-momentum intercepted by $\mathscr{S}$ is equal to the following integral over $S$.
\begin{equation*}
P^a=\frac{1}{4\pi}\int W^a\left(\tilde{\sigma}^0N-\Psi_2^{(0)}\right)dS,
\end{equation*}
where 
\begin{equation*}
W^0=1,\hspace{0.5cm}W^1=\sin\theta\cos\phi,\hspace{0.5cm}W^2=\sin\theta\sin\phi,\hspace{0.5cm}W^3=\cos\theta.
\end{equation*}
Note that, $S$ is topologically a sphere $S^2$ and can always be transformed, by the introduction of a suitable conformal factor into a metric sphere of unit radius. Hence $dS$ can be taken to be	
\begin{equation*}
dS=\sin\theta d\theta d\phi.
\end{equation*}
Thus, the \textit{Bondi mass} (see section \ref{sect:3.1}) at a cut $S$ is
\begin{equation*}
m=\frac{1}{4\pi}\int\left(\tilde{\sigma}^0N-\Psi^{(0)}_2\right)dS.
\end{equation*}
The rate of energy-momentum loss due to gravitational radiation is
\begin{equation}
\label{eqn:244}
\frac{dP^a}{du}=-\frac{1}{4\pi}\int W^a|N|^2dS.
\end{equation}
\begin{figure}[h]
\begin{center}
\includegraphics[scale=0.38]{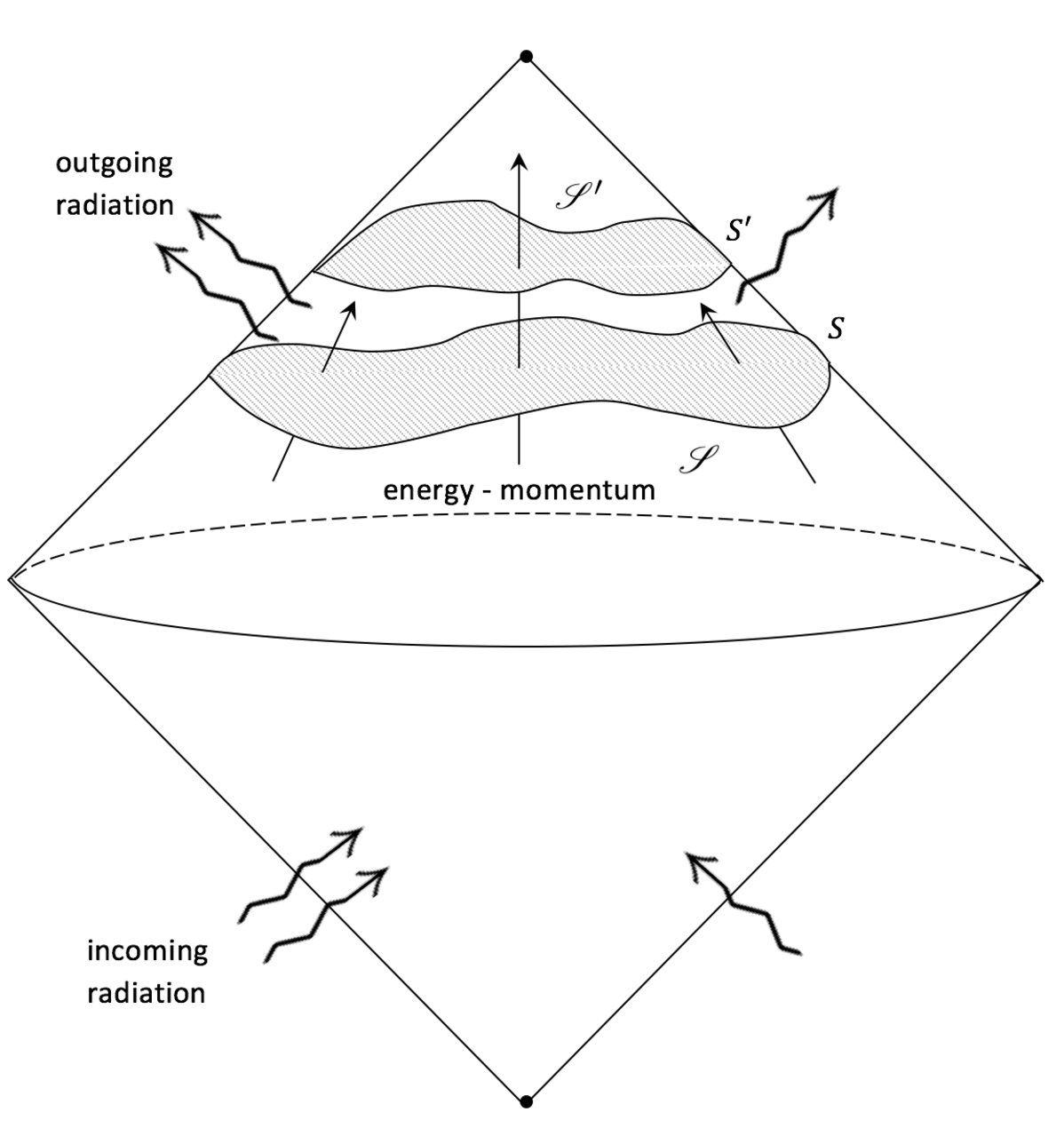}
\caption{The outgoing gravitational radiation through two hypersurfaces $\mathscr{S}$ and $\mathscr{S}'$ with associated cross-sections $S$ and $S'$ respectively.}
\label{fig:6.5}
\end{center}
\end{figure}
\\Thus, the squared modulus of $N$ represents the flux of energy-momentum of the outgoing gravitational radiation. The time component of \eqref{eqn:244} gives the famous \textit{Bondi-Sachs mass-loss formula}:
\begin{equation}
\label{eqn:243}
\frac{dm}{du}=-\frac{1}{4\pi}\int |N|^2dS\leq 0.
\end{equation}
The positivity of the integrand in \eqref{eqn:243} shows that if a system emits gravitational waves, i.e. if there is news, then its Bondi mass must decrease. If there is no news, i.e. $N=0$, the Bondi mass is constant. The reason we get a mass loss rather than a mass gain is simply that the above discussion has been applied to $\mathscr{I}^+$ instead of $\mathscr{I}^-$.\\
From the above discussion it follows that, if the cut $S$ is given by $u=0$ and is shear-free, the cross-sections $u=\mathrm{const}$, which are translations of $\mathscr{I}^+$, will not be shear-free in the presence of gravitational radiation. In other words, if $\tilde{\sigma}^0=0$ for one value of $u$, we will generally have $\tilde{\sigma}^0\neq 0$ for a later value of $u$, i.e. \lq goodness\rq\hspace{0.1mm} would not be invariant under translation. However, there is a more serious difficulty even than this. Since to specify a cut we just need to specify the value of $u$ on each generator of $\mathscr{I}^+$, the freedom in choosing a cut is one real number per point of the cut. On the other hand the quantity $\tilde{\sigma}^0$ is complex, its vanishing therefore, representing two real numbers per point of the section. To solve this problem we need to define the magnetic and the electric part of $\tilde{\sigma}^0$. Under a rotation of the spacelike vectors $\mathrm{Re}(m^{a})$, $\mathrm{Im}(m^{a})$ in their plane, given by 
\begin{equation}
\label{eqn:245}
m^{a}{}'=e^{i\psi}m^a,
\end{equation}
we have, from the definition of shear,
\begin{equation}
\label{eqn:246}
\tilde{\sigma}^{0}{}'=e^{2i\psi}\tilde{\sigma}^{0}.
\end{equation}
We say that $\tilde{\sigma}^0$ has spin weight 2. Furthermore it can be shown \citep{Sachs62} that, under a BMS transformation for which $\alpha=0$, i.e. a conformal transformation, we have
\begin{equation}
\label{eqn:247}
\tilde{\sigma}^{0}{}'=K^{-1}\tilde{\sigma}^0.
\end{equation}
Thus we say that $\tilde{\sigma}^0$ has conformal weight $-1$. Generally we give the following 
\begin{defn}$\\ $
A scalar $\eta$ will be said to have \textit{spin weight} $s$ if it transforms as
\begin{equation*}
\eta'=e^{is\psi}\eta.
\end{equation*}
under \eqref{eqn:245} and a \textit{conformal weight} $w$ if it transforms under BMS conformal transformations as
\begin{equation*}
\eta'=K^{w}\eta.
\end{equation*}
\end{defn}
We define the differential operator $\eth$ \citep{New67,Penrin1} acting on a scalar of spin weight $s$, in the particular coordinate system $(\theta,\phi)$:
\begin{equation*}
\eth\eta=-(\sin\theta)^s\left[\frac{\partial}{\partial \theta}+\frac{i}{\sin\theta}\frac{\partial}{\partial \phi}\right](\sin\theta)^{-s}\eta.
\end{equation*}
It is easily seen that, under \eqref{eqn:245}, we have 
\begin{equation*}
(\eth\eta)'=e^{i(s+1)\psi}(\eth\eta),
\end{equation*}
i.e. $\eth$ raises the spin weight by 1. Similarly $\bar{\eth}$ lowers the spin weight by 1. We also have
\begin{equation*}
(\bar{\eth}\eth-\eth\bar{\eth})\eta=2s\eta.
\end{equation*}
We can define \textit{spin-$s$ spherical harmonics} ${}_sY_{lm}(\theta,\phi)$ for integral $s,l$ and $m$ by
\begin{equation*}
{}_sY_{lm}(\theta,\phi)=\left\{\begin{matrix}\displaystyle{\left[\frac{(l-s)!}{(l+s)!}\right]}^{1/2}\eth^s Y_{lm}(\theta,\phi)\hspace{2cm}(0\leq s\leq l), \\
\\
(-1)^s\displaystyle{\left[\frac{(l+s)!}{(l-s)!}\right]}^{1/2}\bar{\eth}^{-s} Y_{lm}(\theta,\phi)\hspace{0.6cm}\hspace{0.6cm}(-l\leq s\leq 0),
\end{matrix}\right.
\end{equation*}
where $Y_{lm}(\theta,\phi)$ are ordinary spherical harmonics. Note that ${}_sY_{lm}(\theta,\phi)$ are not defined for $|s|>l$. It can be shown that the ${}_sY_{lm}$ form a complete orthonormal set for each value of $s$, i.e. any function of spin weight $s$ can be expanded in a series in ${}_sY_{lm}(\theta,\phi)$. If we pass now to complex stereographic coordinates $(\zeta,\bar{\zeta})$, we have that 
\begin{equation*}
\eth\eta=2P^{1-s}\frac{\partial(P^s\eta)}{\partial \zeta},
\end{equation*}
\begin{equation*}
\bar{\eth}\eta=2P^{1+s}\frac{\partial(P^{-s}\eta)}{\partial \bar{\zeta}},
\end{equation*}
with $P=\frac{1}{2}(1+\zeta\bar{\zeta})$. The spin-$s$ spherical harmonics take now the form
\begin{equation*}
{}_sY_{lm}(\theta,\phi)=(-)^{l-m}\frac{[(l+m)!(l-m)!(2l+1)]^{1/2}}{[(l-s)!(l+s)!4\pi]^{1/2}}(1+\zeta\bar{\zeta})
\end{equation*}
\begin{equation*}
\times\sum_{p}\left(\begin{matrix}l-s\\p
\end{matrix}\right)\left(\begin{matrix}l+s\\p+s-m
\end{matrix}\right)\zeta^p(-\bar{\zeta})^{p+s-m}.
\end{equation*}
We have
\begin{align*}
&\eth\left({}_sY_{lm}(\theta,\phi)\right)=[(l-s)(l+s+1)]^{1/2}{}_{s+1}Y_{lm}(\theta,\phi),\\
&\bar{\eth}\left({}_sY_{lm}(\theta,\phi)\right)=-[(l+s)(l-s+1)]^{1/2}{}_{s-1}Y_{lm}(\theta,\phi),\\
&\bar{\eth}\eth\left({}_sY_{lm}(\theta,\phi)\right)=-(l-s)(l+s+1){}_{s}Y_{lm}(\theta,\phi).
\end{align*}
In particular we see from the above relations that $\eth$ annihilates ${}_sY_{lm}$ whenever $l=s$ and $\bar{\eth}$ annihilates ${}_sY_{lm}$ when if $l=-s$. Furthermore the ${}_sY_{lm}$ are eigenfunctions of $\bar{\eth}\eth$ for each spin weight $s$. If $s=0$, $\bar{\eth}\eth$ is essentially the total angular momentum.\\ 
We give now three important results, whose proof can be found in \citep{NewPen66}. 
\begin{thm}$\\ $
\label{thm:conf0}
The operator $\eth$ is invariant, with spin weight unity, under change of coordinate system which preserves the sphere metric
\begin{equation*}
g_{_{S^2}}=d\theta\otimes d\theta+\sin^2\theta d\phi\otimes d\phi=P^{-2}d\zeta\otimes d\bar{\zeta},
\end{equation*}
i.e., under $\eta'=e^{is\psi}\eta$ we have $\eth'\eta'=e^{i(s+1)\psi}\eth\eta$.
\end{thm}
\begin{thm}$\\ $
\label{thm:conf1}
Let $\eta$ have conformal weight $w$ and spin weight $s$, where $w\geq s$. Then $\eth^{w-s+1}\eta$ is a quantity of conformal weight $s-1$ and of spin weight $w+1$. Furthermore we have for the commutator 
\begin{equation*}
(\bar{\eth}^s\eth^s-\eth^s\bar{\eth}^s)\eta=2s\eta.
\end{equation*}
\end{thm}
\begin{thm}$\\ $
\label{thm:conf2}
Given any suitably regular $\eta$ on the sphere of spin weight $s$, there exists $\xi$ of spin weight zero for which
\begin{equation*}
\eta=\eth^s\xi.
\end{equation*}
\end{thm}
Using theorem \ref{thm:conf2}, for a quantity $\eta$ of spin weight $s$ on the sphere, we can define its electric and magnetic part, denoted by $\eta_e$ and $\eta_m$ respectively, given by 
\begin{equation*}
\eta_e=\eth^s\mathrm{Re}(\xi),\hspace{1cm}\eta_m=i\eth^s\mathrm{Im}(\xi).
\end{equation*}
such that
\begin{equation*}
\eta=\eta_e+\eta_m,
\end{equation*}
with (see theorem \ref{thm:conf1})
\begin{equation*}
\bar{\eth}^s\eta_e=\eth^s\bar{\eta}_e,\hspace{1cm}\bar{\eth}^s\eta_m=-\eth^s\bar{\eta}_m.
\end{equation*}
The invariance properties of $\eth$ imply that the concepts of magnetic and electric parts of $\eta$ are invariant under rotations of the sphere or of the $m^{a}$ vectors. Furthermore, in virtue of theorem \ref{thm:conf1}, only if $\eta$ has conformal weight $-1$, does the split of $\eta$ onto its electric and magnetic parts turn out to be invariant also under BMS conformal transformations.\\
We apply now these concepts to $\tilde{\sigma}^0$ which has $s=2$ and $w=-1$. Thus, for each hypersurface $u=\mathrm{const}$ we have a splitting $\tilde{\sigma}^0=\tilde{\sigma}^0_e+\tilde{\sigma}^0_m$, which is invariant under Lorentz transformations, which leave the surface $u=0$ invariant. \\
Now, it can be shown \citep{NewPen66} that in the Minkowski case, the asymptotic shear $\tilde{\sigma}^0=\tilde{\sigma}^0_e+\tilde{\sigma}^0_m$ behaves, when $u\rightarrow-\infty$ as
\begin{subequations}
\label{eqn:250}
\begin{align}
&\tilde{\sigma}^{0}_e(u,\theta,\phi)\xrightarrow[u \rightarrow -\infty]{}S_e(\theta,\phi),\\
&\tilde{\sigma}^{0}_m(u,\theta,\phi)\xrightarrow[u \rightarrow -\infty]{}0,
\end{align}
\end{subequations}
with $S_e(\theta,\phi)$ purely electric and independent of $u$. The magnetic part $\tilde{\sigma}^0_m(u,\theta,\phi)$, as $u\rightarrow-\infty$ vanishes, i.e. $S_m(\theta,\phi)=0$.\\
On the basis of what happens in the Minkowski case we wish to impose a physical restriction on the behaviour of $\tilde{\sigma}^0$ as $u\rightarrow-\infty$ for a general space-time. Although no actual cuts of $\mathscr{I}^+$ may be shear-free, it is reasonable to expect that in the limit $u\rightarrow-\infty$ on $\mathscr{I}^+$, such cuts will exist. Requiring that this limiting shear-free cuts be sent into one another, we can actually restrict the BMS transformations to obtain a canonically defined subgroup of the BMS group, which is isomorphic to the Poincar\'e group. The Poincar\'e group which emerges in this way, in virtue of the considerations we have done on the gravitational radiation, may be thought of as that which has relevance to the remote past, \textit{before} all the gravitational radiation has been emitted. In analogy with \eqref{eqn:250} we require that
\begin{equation}
\label{eqn:248}
\tilde{\sigma}^0(u,\theta,\phi)\longrightarrow S(\theta,\phi).
\end{equation}
If the analogy with the Minkowski theory is to be trusted, we would expect $S(\theta,\phi)$ to be purely electric. However, it is not essential since it will be possible to extract the Poincar\'e group only on the basis of \eqref{eqn:248}. Hence we will treat the case in which $S(\theta,\phi)$ could have a magnetic part too, i.e. $S(\theta,\phi)=S_e(\theta,\phi)+S_m(\theta,\phi)$ with $S_m(\theta,\phi)\neq 0$. \\
It can be shown \citep{Sachs62} that, under a full BMS transformation, the asymptotic shear transforms as
\begin{equation}	
\label{eqn:249}
\tilde{\sigma}^0{}'(u,\theta,\phi)=K^{-1}e^{2i\psi}\left[\tilde{\sigma}^0(u,\theta,\phi)+\frac{1}{2}\eth^2\alpha(\theta,\phi)\right].
\end{equation}
It has to be remarked that $\tilde{\sigma}^0{}'(u,\theta,\phi)$ refers to the asymptotic shear of the hypersurface $u'=\mathrm{const}$ of the transformed coordinate system evaluated at $(u,\theta,\phi)$. The complete transformation $\tilde{\sigma}^0{}'(u',\theta',\phi')$ is more complicated. Applying \eqref{eqn:248} to \eqref{eqn:249} gives
\begin{align*}
&S'_e(\theta,\phi)=K^{-1}e^{2i\psi}\left[S_e(\theta,\phi)-\frac{1}{2}\eth^2\alpha(\theta,\phi)\right],\\
&S'_m(\theta,\phi)=K^{-1}e^{2i\psi}S_m(\theta,\phi),
\end{align*}
since $\alpha(\theta,\phi)$ is real. Using theorem \ref{thm:conf2} we can set
\begin{equation*}
S_e(\theta,\phi)=\eth^2G(\theta,\phi),
\end{equation*}
for some real $G(\theta,\phi)$. Since $\alpha(\theta,\phi)$ can be chosen arbitrarily on the sphere, it follows that a BMS transformations for which $\alpha(\theta,\phi)=2G(\theta,\phi)$ imposes $S'_e(\theta,\phi)=0$. Thus we have introduced coordinate conditions for which $S_e(\theta,\phi)=0$ at $u=-\infty$. Now the BMS transformations which preserve the condition $S_e(\theta,\phi)=0$ are those for which
\begin{equation*}
\eth^2\alpha(\theta,\phi)=0.
\end{equation*}
It can be easily shown that this condition restricts $\alpha(\theta,\phi)$ to be of the form of \eqref{eqn:207} and thus the allowed supertranslations are simply the translations. The Lorentz transformations, given by $\alpha(\theta,\phi)=0$ do not spoil the coordinate conventions. We have finally obtained the following 
\begin{thm}$\\ $
The group of asymptotic isometries of an asymptotically flat space-time which preserves the condition $S_e(\theta,\phi)=0$ at $u=-\infty$ is isomorphic with the Poincar\'e group $\mathscr{P}$.
\end{thm}
\begin{oss}$\\ $
We could have carried out the same arguments by taking the limit $u\rightarrow +\infty$, and it would have been an independent choice. Thus, in a similar way, we could have extracted another Poincar\'e group which has relevance to the remote future, i.e. \textit{after} all the gravitational radiation has been emitted. There seems to be no reason to believe that these two Poincar\'e groups will be the same, in general. 
\end{oss}
\section{Applications of the BMS Group\hspace{1cm}}
There are many very recent applications of the BMS formalism and of particular interest are the ones that have to do with black-hole physics in quantum gravity. Maybe the most significant of them is the one which is due to \cite{Stro14}. In this work, the author identifies the two different copies of the BMS group of $\mathscr{I}^+$ and of $\mathscr{I}^-$ as $\mathrm{BMS}^+$ and $\mathrm{BMS}^-$, respectively. Starting from the idea that $\mathrm{BMS}^+$ and $\mathrm{BMS}^-$ act non-trivially on outgoing and ingoing gravitational scattering data, it is argued that in Christodoulou-Klainerman space-times \citep{ChrKl} it is possible to make a canonical identification between the generators $\mathrm{B}^+$ and $\mathrm{B}^-$ of $\mathrm{BMS}^+$ and $\mathrm{BMS}^-$ (using the link between $\mathscr{I}^+$ and $\mathscr{I}^-$), and hence to find what are called the \lq diagonal\rq\hspace{0.1mm} generators $\mathrm{B}^0$ of $\mathrm{BMS}^+\times\mathrm{BMS}^-$. Then, since in asymptotically flat quantum gravity one seeks an $\cal{S}$-matrix relating the \lq in\rq\hspace{0.1mm} and \lq out\rq\hspace{0.1mm} Hilbert states, i.e.
\begin{equation*}
\left| \mathrm{in}\right>=\cal{S}\left|\mathrm{out}\right>,
\end{equation*}
it is asked whether or not the BMS group could provide a symmetry of 
the $\cal{S}$-matrix and it is guessed a relation of the form 
\begin{equation*}
\mathrm{B}^+\cal{S}-\cal{S} \mathrm{B}^{-}=\mathrm{0}.
\end{equation*}
The classical limit of such a relation would give symmetries of classical gravitational scattering. The generators $\mathrm{B}^0$ are made up of Lorentz generators, which provide the usual Lorentz invariance of the $\cal{S}$-matrix, and of the supertranslations generators which provide that the local energy is conserved at each angle. In fact the simple translations, that are supertranslations whose magnitude does not depend on the angle or position $(\zeta,\bar{\zeta})$ on $S^2$, lead, as known, to global energy conservation. Then it is plausible that the supertranslations which act at only one angle lead to the conservation of energy at that one angle. Furthermore in quantum theory, matrix elements of the conservation laws give an infinite number of exact relations between scattering amplitudes in quantum gravity. These relations turned out \citep{MitraStro} to have been previously discovered by \cite{Wein65} using Feynman diagrammatics and are known as the \textit{soft graviton theorem}. \\The results obtained by Strominger paved the way for a deeper investigation of the so called \lq soft hair\rq\hspace{0.1mm} of black holes. In \cite{Haw161} the authors showed that the Minkowski vacuum in quantum gravity is not invariant under supertranslations. In particular such transformations map the Minkowski vacuum into a physically inequivalent zero-energy one. Hence, the supertranslation symmetry is spontaneously broken and the \lq soft\rq\hspace{0.1mm} (i.e. zero-energy) gravitons are the associated Goldstone bosons. Since the information paradox \citep{Haw1975,Haw1976} relies on the fact that the vacuum is unique, this would be a motivation to doubt of its validity. Furthermore, the information loss argument assumes the fact that the black holes are only characterized by their mass $M$, charge $Q$ and angular momentum $J$. But supertranslations map a stationary black hole to a physically inequivalent one and in the process of Hawking evaporation, supertranslation charge will be radiated through null infinity. Since this charge is conserved, the sum of the black hole and radiated supertranslation charge is fixed at all times. This requires that black holes carry soft hair (i.e. additional data, besides mass, charge and angular momentum) arising from supertranslations. Moreover, when the black hole has fully evaporated, the net supertranslation charge in the outgoing radiation must be conserved. This leads in turn to correlations between the early- and late-time Hawking radiation. In \cite{Haw162} the same authors showed that black holes in General Relativity are characterized by an infinite head of supertranslation hair and that distinct black holes are characterized by different classical superrotation charges measured at infinity.

\appendix
\chapter{Topological spaces}
\label{A}
\begin{defn}$\\ $
\label{defn:top}
A \textit{topological space} $(X,\mathscr{T})$ consists of a set $X$ together with a collection $\mathscr{T}$ of subsets of $X$ satisfying the following three properties:
\begin{itemize}
\item The union of an arbitrary collection of subsets, each of which is in $\mathscr{T}$, is in $\mathscr{T}$, i.e. if $O_{\alpha}\in\mathscr{T}$ for all $\alpha$, then $\cup_{\alpha}O_{\alpha}\in\mathscr{T}$;
\item The intersection of a finite number of subsets in $\mathscr{T}$ is in $\mathscr{T}$, i.e. if $O_1,...,O_n\in\mathscr{T}$ then \begin{equation*}
\bigcap_{i=1}^nO_i\in\mathscr{T};
\end{equation*}
\item The entire set $X$ and the empty set $\emptyset$ are in $\mathscr{T}$.
\end{itemize}
\end{defn}
Usually $\mathscr{T}$ is referred to as topology on $X$, and subsets of $X$ which are listed in the collection $\mathscr{T}$ are called \textit{open sets}.\\
Let $(X,\mathscr{T})$ and $(Y,\mathscr{S})$ be topological spaces and a map $f:X\rightarrow Y$ between them.
\begin{defn}$ \\ $
\label{defn:map}
$f$ is said to be \textit{continuous} if the inverse image, $f^{-1}[O]\equiv\{\left.x\in X\right|f(x)\in O\}$, of every open set $O$ in $Y$ is an open set in $X$.
\end{defn}
\begin{defn}$\\ $
\label{defn:hom}
If $f$ is continuous, one-to-one, onto and its inverse is continuous, $f$ is called a \textit{homomorphism} and $(X,\mathscr{T})$ and $(Y,\mathscr{S})$  are said to be homeomorphic.
\end{defn}
\begin{defn} $\\ $
\label{defn:clo}
If $(X,\mathscr{T})$ is a topological space, a subset $C$ of $X$ is said to be \textit{closed} if its complement $X-C=\{\left.x\in X\right| x\notin C\}$ is open.
\end{defn}
Note that a subset may be neither open nor closed or may be both open and closed. This last possibility gives rise to the definition of connectedness. 
\begin{defn}
\label{defn:conn} $\\ $
A topological space $(X,\mathscr{T})$ is said to be \textit{connected} if the only subsets which are both open and closed are $X$ and $\emptyset$.
\end{defn}
Consider a topological space $(X,\mathscr{T})$ and an arbitrary subset of $X$, say $A$.
\begin{defn} $\\ $
\label{defn:closure}
The \textit{closure} of $A$, $\overline{A}$, is defined as the intersection of all closed sets containing $A$.
\end{defn}
Obviously $\overline{A}$ is closed, contains $A$ and equals $A$ if and only if $A$ is closed.
\begin{defn} $\\ $
\label{defn:int}
The \textit{interior} of $A$, $\mathrm{int}[A]$, is defined as the intersection of all open sets contained within $A$.
\end{defn}
Clearly $\mathrm{int}[A]$ is open, is contained in $A$ and equals $A$ if and only if $A$ is open.
\begin{defn} $\\ $
\label{defn:bou}
The \textit{boundary} of $A$, $\dot{A}$, consists of all points which lie in $\overline{A}$ but not in $\mathrm{int}[A]$:
\begin{equation*}
\dot{A}=\overline{A}-\mathrm{int}[A].
\end{equation*}
\end{defn}
\begin{defn} $\\ $
\label{defn:hau}
A topological space $(X,\mathscr{T})$ is said to be \textit{Hausdorff} if for each pair of distinct points $p$, $q\in X$, $p\neq q$, one can find open sets $O_p$, $O_q\in\mathscr{T}$ such that $p\in O_p$ and $q\in O_q$, and $O_p\cap O_q=\emptyset$.
\end{defn}
Let $(X,\mathscr{T})$ be a topological space, $\{O_{\alpha}\}$ be a collection of open sets and $A$ a subset of $X$.
\begin{defn} $\\ $
\label{defn:cov}
$\{O_{\alpha}\}$ is said to be a \textit{open cover} of $A$ if the union of these sets contains $A$.
\end{defn}
\begin{defn} $\\ $
\label{defn:sub}
A subcollection of the sets $\{O_{\alpha}\}$ which also covers $A$ is referred to a \textit{subcover}.
\end{defn}
The two previous definitions allow us to give the following
\begin{defn} $\\ $
\label{defn:com}
The set $A$ is said to be \textit{compact} if every open cover of $A$ has a finite subcover, i.e. a subcover consisting of only a finite number of sets.
\end{defn}
The general relation between compact and closed sets is described by the following two theorems.
\begin{thm}$\\ $
\label{thm:cc1}
Let $(X,\mathscr{T})$ be Hausdorff and let $A\subset X$ be compact. Then $A$ is closed.
\end{thm}
\begin{thm}$\\ $
\label{thm:cc2}
Let $(X,\mathscr{T})$ be compact and let $A\subset X$ be closed. Then $A$ is compact.
\end{thm}
We need now to give the notion of convergence of sequences.
\begin{defn} $\\ $
\label{defn:con}
A sequence $\{x_n\}$ of points in a topological space $(X,\mathscr{T})$ is said to \textit{converge} to point $x$ if given any open neighbourhood $O$ of $x$ (i.e. an open set containing $x$) there is an $N$ such that $x_n\in O$ for all $n>N$. The point $x$ is said to be the \textit{limit} of this sequence. 
\end{defn}
\begin{defn} $\\ $
\label{defn:acc}
A point $y\in X$ is said to be an \textit{accumulation point} of $\{x_n\}$ if every open neighbourhood of $y$ contains infinitely many points of the sequence. 
\end{defn}
It follows that if $\{x_n\}$ converges to $x$, then $x$ is an accumulation point of the sequence. However, in a general topological space, if $y$ is an accumulation point of $\{x_n\}$ it may not even be possible to find a subsequence $\{y_n\}$ of the sequence $\{x_n\}$ which converges to $y$. Hence we give the following
\begin{defn} $\\ $
\label{fco}
A topological space $(X,\mathscr{T})$ is \textit{first countable} if for each $p\in X$ there is a countable collection $\{O_n\}$ of open sets such that every open neighbourhood, $O$, of $p$ contains at least one member of this collection.
\end{defn}
A stronger requirement is the following
\begin{defn} $\\ $
\label{sco}
A topological space $(X,\mathscr{T})$ is \textit{second countable} if for each $p\in X$ there is a countable collection $\{O_n\}$ of open sets such that every open neighbourhood, $O$, of $p$ can be expressed as a union of sets in the collection.
\end{defn}
The Bolzano-Weierstrass theorem expresses the important relation between compactness and convergence of sequences.
\begin{thm}{Bolzano-Weierstrass}$\\ $
\label{thm:BW}
Let $(X,\mathscr{T})$ be a topological space and let $A\subset X$. If $A$ is compact then every sequence $\{x
_n\}$ of points in $A$ has an accumulation point lying in $A$. Conversely if $(X,\mathscr{T})$ is second countable and every sequence of points in $A$ has an accumulation point in $A$, then $A$ is compact. Thus in particular, if $(X,\mathscr{T})$ is second countable, $A$ is compact if and only if every subsequence in $A$ has a convergent subsequence whose limit lies in $A$.
\end{thm}
We give now the definition of paracompactness. Let $(X,\mathscr{T})$ be a topological space and let $\{O_{\alpha}\}$ be an open cover of $X$. 
\begin{defn} $\\ $
\label{defn:ref}
An open cover $\{V_{\beta}\}$ is said to be a \textit{refinement} of $\{O_{\alpha}\}$ if for each $V_{\beta}$ there exists an $O_{\alpha}$ such that $V_{\beta}\subset O_{\alpha}$.
\end{defn} 
\begin{defn} $\\ $
\label{defn:lof}
The cover $\{V_{\beta}\}$ is said to be \textit{locally finite} if each $x\in X$ has an open neighbourhood $W$  such that only finitely many $V_{\beta}$ satisfy $V_{\beta}\cap W\neq\emptyset$. 
\end{defn}
\begin{defn} $\\ $
\label{par}
A topological space $(X,\mathscr{T})$ is said to be \textit{paracompact} if every open cover $\{O_{\alpha}\}$ of $X$ has a locally finite refinement $\{V_{\beta}\}$.
\end{defn}
We can now give the definition of a manifold. 
\begin{defn}$\\ $
\label{defn:tma}
A $n$-dimensional \textit{topological manifold} $\mathscr{M}$ is a topological space $(X,\mathscr{T})$ with a collection of open subsets $\{O_{\alpha}\}$ of $X$, such that 
\begin{itemize}
\item $\{O_{\alpha}\}$ is an open cover of $\mathscr{M}$, i.e. $\mathscr{M}=\cup_{\alpha}O_{\alpha}$;
\item For each $\alpha$, there is a one-to-one, onto, homomorphism $\psi_{\alpha}:O_{\alpha}\rightarrow U_{\alpha}$, where $U_{\alpha}$ is an open subset of $\mathbb{R}^n$;
\end{itemize}
\end{defn}
\begin{oss}$\\ $
Usually the pair $(O_{\alpha},\psi_{\alpha})$ is called \textit{chart} or \textit{coordinate system} and the collection $\{O_{\alpha},\psi_{\alpha}\}$ is called \textit{atlas}. If $O_{\alpha}\cap O_{\beta}\neq\emptyset$ we can consider the map $\psi_{\beta}\circ\psi_{\alpha}^{-1}$ which takes points in $\psi_{\alpha}[O_{\alpha}\cap O_{\beta}]\subset U_{\alpha}\subset\mathbb{R}^n$ to points in  $\psi_{\beta}[O_{\alpha}\cap O_{\beta}]\subset U_{\beta}\subset\mathbb{R}^n$. This map is an homomorphism and is often called \textit{coordinate change}.
\end{oss}
Consider now two topological manifolds $\mathscr{M}$ and $\mathscr{M'}$, with respective dimensions $n$ and $n'$, and a map $f:\mathscr{M}\longrightarrow\mathscr{M'}$ between them. Let $\{O_{\alpha},\psi_{\alpha}\}$ and $\{O'_{\beta},\psi'_{\beta}\}$ be respectively charts of $\mathscr{M}$ and $\mathscr{M'}$. 
\begin{defn}$\\ $
\label{defn:cka}
$f$ is said to be \textit{$C^{k}$} if the map $\psi'_{\beta}\circ f\circ\psi_{\alpha}^{-1}:U_{\alpha}\subset\mathbb{R}^{n}\longrightarrow U'_{\beta}\subset\mathbb{R}^{n'}$ is $C^k$ in the sense of usual differential calculus.
\end{defn}
Let now $f$ be an homomorphism, between $\mathscr{M}$ and $\mathscr{M'}$.
\begin{defn}$\\ $
$f$ is said to be a \textit{diffeomorphism} if is $C^{\infty}$ with its inverse (i.e. the map $\psi_{\alpha}\circ f\circ\psi'_{\beta}{}^{-1}$).
\end{defn}
We can now give the definition of differentiable manifold.
\begin{defn}
\label{defn:dma}
A $n$-dimensional \textit{$C^k$ differentiable manifold} $\mathscr{M}$ is a topological manifold $\mathscr{M}$ such that all the coordinate changes are $C^k$ maps with their inverses of $\mathscr{M}$ into itself.
\end{defn}
\begin{oss}$\\ $
It is clear that for a $C^{\infty}$ differentiable manifold $\mathscr{M}$ all the coordinate changes must be diffeomorphisms of $\mathscr{M}$ into itself.
\end{oss}
\begin{figure}[h]
\begin{center}
\includegraphics[scale=0.5]{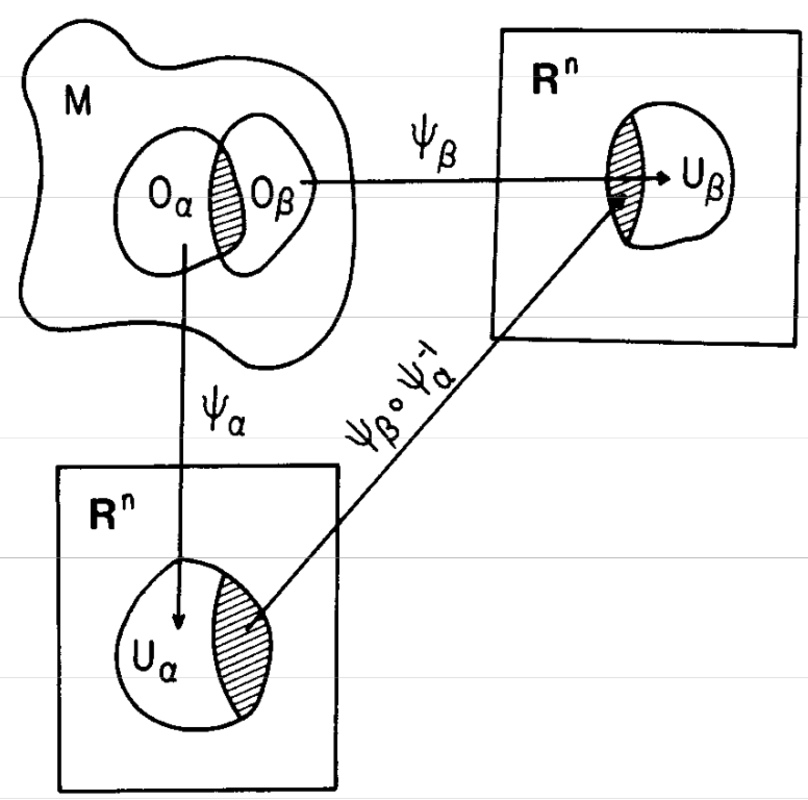}
\caption{An illustration of the map $\psi_{\beta}\circ\psi_{\alpha}^{-1}$.}
\label{fig:2.9}
\end{center}
\end{figure}
Before we state some important theorems about manifolds we want to briefly comment \textit{manifolds-with-boundary} without going into details. It is remarkable that the boundary of a manifold $\mathscr{M}$ is not defined to be the boundary of the topological space $X$. If $\mathscr{M}$ is covered by a family of open sets $\{O_{\alpha}\}$, suppose that some of them is homeomorphic not to $\mathbb{R}^n$, but to an open set of  $\mathbb{H}^n=\{\left.(x^1,...,x^n)\in\mathbb{R}^n\right| x^n\geq 0\}$ ($\mathbb{H}^n$ is called \textit{Euclidean half-space}). Then the set of points which are mapped to points with $x^n=0$ is called the \textit{boundary of} $\mathscr{M}$ and is denoted by $\partial\mathscr{M}$. The coordinates of $\partial\mathscr{M}$ may be given by $n-1$ numbers $(x^1,...,x^{n-1},0)$. A Further details about the smoothness of coordinate changes and examples can be found in \cite{Lang}, in \cite{Nak} or \cite{Spiv}.\\
Consider now a paracompact manifold $\mathscr{M}$, i.e. a manifold whose topological space is paracompact.
\begin{thm} $\\ $
\label{thm:sco}
A paracompact manifold $\mathscr{M}$ is second countable.
\end{thm}
\begin{cor}$\\ $
\label{cor:sco}
A paracompact manifold $\mathscr{M}$ can be covered by a locally finite, countable family of charts $(O_i,\psi_i)$ with each $\overline{O}_i$ compact.
\end{cor}
Consider a locally finite cover $\{O_{\alpha}\}$ of $\mathscr{M}$.
\begin{defn}$\\ $
\label{defn:par}
A \textit{partition of unity} subordinate to $\{O_{\alpha}\}$ is a collection of smooth functions $\{f_{\alpha}\}$ such that 
\begin{enumerate}
\item the support of $f_{\alpha}$ is contained within $O_{\alpha}$;
\item $0\leq f_{\alpha}\leq 1$;
\item $\sum_{\alpha}f_{\alpha}=1$.
\end{enumerate}
\end{defn}
\begin{thm}$ \\ $
\label{thm:par}
Every locally finite open cover $\{O_{\alpha}\}$ of a paracompact manifold $\mathscr{M}$, such that $\overline{O}_{\alpha}$ is compact, admits a subordinate partition of unity.
\end{thm}
\begin{cor} $\\ $
\label{cor:met}
Using this theorem we can define a local Riemannian metric $(g_{\alpha})_{ab}$ on each $O_{\alpha}$ and build a Riemannian metric $g_{ab}$ defined all through $\mathscr{M}$ by setting $g_{ab}=\sum_{\alpha}f_{\alpha}(g_{\alpha})_{ab}$.
\end{cor}
The proof of the last two theorems can be found in \cite{KobNom}.
\begin{thm}$\\ $
\label{thm:cba}
Let $\mathscr{M}$ be a manifold which is locally compact Haussdorff and  whose topology has a countable basis. Then $\mathscr{M}$ is paracompact. 
\end{thm}
The proof and insights of this theorem can be found in \cite{Lang} or in \cite{Rys}.
\chapter{The Newman-Penrose formalism}
\label{B}
The source-free Maxwell equations in the Newman-Penrose formalism are:
\begin{enumerate}
\item $D\phi_1-\bar{\delta}\phi_0=(\pi-2\alpha)\phi_0+2\rho\phi_1-\kappa\phi_2$;
\item $D\phi_2-\bar{\delta}\phi_1=-\lambda\phi_0+2\pi\phi_1+(\rho-2\epsilon)\phi_2$;
\item $\Delta\phi_0-\delta\phi_1=(2\gamma-\mu)\phi_0-2\tau\phi_1+\sigma\phi_2$;
\item $\Delta\phi_1-\delta\phi_2=\nu\phi_0-2\mu\phi_1+(2\beta-\tau)\phi_2$.
\end{enumerate}
The Newman-Penrose field equations are:
\begin{enumerate}
\item $D\rho-\bar{\delta}\kappa = (\rho^2+\sigma\bar{\sigma})+(\epsilon+\bar{\epsilon})\rho -\bar{\kappa}\tau-\kappa(3\alpha+\bar{\beta}-\pi)+\Phi_{00};$
\item $D\sigma-\delta\kappa=(\rho+\bar{\rho}+3\epsilon-\bar{\epsilon})\sigma-(\tau-\bar{\pi}+\bar{\alpha}+3\beta)\kappa+\Psi_0; $
\item $D\tau -\Delta\kappa=(\tau+\bar{\pi})\rho+(\bar{\tau}+\pi)\sigma+(\epsilon-\bar{\epsilon})\tau-(3\gamma+\bar{\gamma})\kappa+\Psi_1+\Phi_{01}; $
\item $D\alpha-\bar{\delta}\epsilon=(\rho+\bar{\epsilon}-2\epsilon)\alpha+\beta\bar{\sigma}-\bar{\beta}\epsilon-\kappa\lambda-\bar{\kappa}\gamma+(\epsilon+\rho)\pi+\Phi_{10}; $
\item $D\beta-\delta\epsilon=(\alpha+\pi)\sigma+(\bar{\rho}-\bar{\epsilon})\beta-(\mu+\gamma)\kappa-(\bar{\alpha}-\bar{\pi})\epsilon+\Psi_1$;
\item $D\gamma-\Delta\epsilon=(\tau+\bar{\pi})\alpha+(\bar{\tau}+\pi)\beta-(\epsilon+\bar{\epsilon})\gamma-(\gamma+\bar{\gamma})\epsilon+\tau\pi-\nu\kappa+\Psi_2-\Lambda+\Phi_{11}; $
\item $D\lambda-\bar{\delta}\pi=(\rho-3\epsilon+\bar{\epsilon})\lambda+\bar{\sigma}\mu+(\pi+\alpha-\bar{\beta})\pi-\nu\bar{\kappa}+\Phi_{20}; $
\item $D\mu-\delta\pi=(\bar{\rho}-\epsilon-\bar{\epsilon})\mu+\sigma\lambda+(\bar{\pi}-\bar{\alpha+\beta})\pi-\nu\kappa+\Psi_2+2\Lambda; $
\item $D\nu-\Delta\pi=(\pi+\bar{\tau})\mu+(\bar{\pi}+\tau)\lambda+(\gamma-\bar{\gamma})\pi-(3\epsilon+\bar{\epsilon})\nu+\Psi_3+\Phi_{21};$
\item $\Delta\lambda-\bar{\delta}\nu=-(\nu+\bar{\nu}+3\gamma-\bar{\gamma})\lambda+(3\alpha+\bar{\beta}+\pi-\bar{\tau})\nu-\Psi_4;$
\item $\delta\rho-\bar{\delta}\sigma=(\bar{\alpha}+\beta)\rho-(3\alpha-\bar{\beta})\sigma+(\rho-\bar{\rho})\tau+(\nu-\bar{\nu})\kappa-\Psi_1+\Phi_{01};$
\item $\delta\alpha-\bar{\delta}\beta=\nu\rho-\lambda\sigma-\alpha\bar{\alpha}+\beta\bar{\beta}-2\alpha\beta+(\rho-\bar{\rho})\gamma+(\nu-\bar{\nu})\epsilon-\Psi_2+\Lambda+\Phi_{11};$
\item $\delta\lambda-\bar{\delta}\mu=(\rho-\bar{\rho})\nu+(\mu-\bar{\mu})\pi+(\alpha+\bar{\beta})\mu+(\bar{\alpha}-3\beta)\lambda-\Psi_3+\Phi_{21};$
\item $\Delta\mu-\delta\nu=-(\mu+\gamma+\bar{\gamma})\mu-\lambda\bar{\lambda}+\bar{\mu}\pi+(\bar{\alpha}+3\beta-\tau)\nu-\Phi_{22};$
\item $\Delta\beta-\delta\gamma=(\bar{\alpha}+\beta-\tau)\gamma-\mu\tau+\sigma\nu+\epsilon\bar{\nu}+(\gamma-\bar{\gamma}-\mu)\beta-\alpha\bar{\lambda}-\Phi_{12};$
\item $\Delta\sigma-\delta\tau=-(\mu-3\gamma+\bar{\gamma})\sigma-\bar{\lambda}\rho-(\tau+\beta-\bar{\alpha})\tau+\kappa\bar{\nu}-\Phi_{02};$
\item $\Delta\rho-\bar{\delta}\tau=(\gamma+\bar{\gamma}-\bar{\nu})\rho-\sigma\lambda+(\bar{\beta}-\alpha-\bar{\tau})\tau+\nu\kappa-\Psi_2-2\Lambda;$
\item $\Delta\alpha-\bar{\delta}\gamma=(\rho+\epsilon)\nu-(\tau+\beta)\lambda+(\bar{\gamma}-\bar{\mu})\alpha+(\bar{\beta}-\bar{\tau})\gamma-\Psi_3.$
\end{enumerate}
\chapter{Christoffel Symbols}
\label{C}
The metric tensor in \eqref{eqn:136} is 
\begin{equation*}
g_{ab}=\left(\begin{matrix}\frac{V}{r}e^{2\beta}-r^2h_{AB}U^AU^B & e^{2\beta} & r^2h_{2B}U^B & r^2h_{3B}U^B \\
e^{2\beta} & 0 & 0 & 0 \\
r^2h_{2A}U^A & 0 & -r^2h_{22} &-r^2h_{23} \\
r^2h_{3A}U^A & 0 & -r^2h_{32} & -r^2h_{33}
\end{matrix}\right),
\end{equation*}
and its inverse is 
\begin{equation*}
g^{ab}=\left(\begin{matrix} 0 & e^{-2\beta} & 0 & 0 \\
e^{-2\beta} & -\frac{V}{r}e^{-2\beta} & U^2e^{-2\beta} & U^3e^{-2\beta} \\
0 & U^2e^{-2\beta}& -\frac{h_{22}}{r^2} &-\frac{h_{23}}{r^2} \\
0 & U^3e^{-2\beta} & -\frac{h_{32}}{r^2} & -\frac{h_{33}}{r^2}
\end{matrix}\right),
\end{equation*}
while the inverse of the matrix $h_{AB}$ in \eqref{eqn:141} is
\begin{equation*}
h^{AB}=\left(\begin{matrix}
e^{-2\gamma}\cosh 2\delta & -\displaystyle{\frac{\sinh 2\delta}{\sin\theta}}\\ \\
-\displaystyle{\frac{\sinh 2\delta}{\sin\theta}} & \displaystyle{\frac{e^{2\gamma}\cosh 2\delta}{\sin^2\theta}}
\end{matrix}\right).
\end{equation*}
It can be easily verified that 
\begin{equation*}
h^{AB}\partial_r h_{AB}=h^{AB}\partial_u h_{AB}=0.
\end{equation*}
The Christoffel symbols are
\begin{align*}
&\Gamma^{u}{}_{rr}=0,\\
&\Gamma^{r}{}_{rr}=2\partial_r\beta,\\
&\Gamma^{A}{}_{rr}=0,\\
&\Gamma^{u}{}_{rA}=0,\\
&\Gamma^{r}{}_{rA}=\frac{e^{-2\beta}r^2}{2}h_{AB}\left(\partial_r U^B\right)+\partial_{A} \beta,\\
&\Gamma^{B}{}_{rA}=\frac{\delta^B_A}{r}+\frac{\left(\partial_r h_{AC}\right)h^{BC}}{2},\\
&\Gamma^{u}{}_{AB}=e^{-2\beta}rh_{AB}+\frac{e^{-2\beta}r^2}{2}\left(\partial_rh_{AB}\right),\\
&\Gamma^{r}{}_{AB}=\frac{e^{-2\beta}r^2}{2}\left(\partial_AU_B+\partial_BU_A\right)+\frac{e^{-2\beta}r^2}{2}\left(\partial_uh_{AB}\right)-Ve^{-2\beta}h_{AB}\\
&-\frac{rVe^{-2\beta}}{2}\left(\partial_rh_{AB}\right)-U^C\frac{e^{-2\beta}r^2}{2}\left(\partial_Ah_{CB}+\partial_Bh_{AC}-\partial_Ch_{AB}\right),\\
&\Gamma^{C}{}_{AB}=rU^Ch_{AB}e^{-2\beta}+\frac{r^2e^{-2\beta}}{2}U^C\left(\partial_rh_{AB}\right)\\
&+\frac{h^{CD}}{2}\left(\partial_Ah_{DB}+\partial_Bh_{DA}-\partial_{D}h_{AB}\right),\\
&\Gamma^{u}{}_{Au}=\partial_A\beta-re^{-2\beta}U_A-\frac{r^2e^{-2\beta}}{2}\left(\partial_rU_A\right),\\
&\Gamma^{r}{}_{Au}=\frac{\partial_AV}{2r}-\frac{r^2e^{-2\beta}}{2}U^B\left(\partial_AU_B\right)+e^{-2\beta}VU_A+\frac{Vre^{-2\beta}}{2}\left(\partial_rU_A\right)\\&
-\frac{r^2e^{-2\beta}}{2}U^B\left(\partial_uh_{AB}\right)-\frac{r^2e^{-2\beta}}{2}U^B\left(\partial_BU_A\right),\\
&\Gamma^{B}{}_{Au}=U^B\left(\partial_A\beta\right)-re^{-2\beta}U^B\left(\partial_rU_A\right)-\frac{r^2e^{-2\beta}}{2}U^B\left(\partial_rU_A\right)-\frac{h^{BC}}{2}\left(\partial_AU_C\right)\\
&+\frac{h^{BC}}{2}\left(\partial_CU_A\right)+\frac{h^{BC}}{2}\left(\partial_uh_{AC}\right),\\
&\Gamma^{u}{}_{ru}=re^{-2\beta}U^AU_A+\frac{r^2e^{-2\beta}}{2}U^A\left(\partial_rU_A\right)-U^A\left(\partial_A\beta\right),\\
&\Gamma^{r}{}_{ru}=\frac{\partial_rV}{2r}-\frac{V}{2r^2}+\frac{V}{r}\left(\partial_r\beta\right)-\frac{r^2e^{-2\beta}}{2}U^A\left(\partial_rU_A\right)-U^A\left(\partial_A\beta\right),\\
&\Gamma^{A}{}_{ru}=-\frac{U^A}{r}-\frac{h^{AB}}{2}\left(\partial_rU_B\right)+\frac{h^{AB}}{r^2}\left(\partial_B\beta\right).
\end{align*}
\bibliographystyle{apa} 
\bibliography{biblio}
\nocite{Wald}\nocite{HawBH}\nocite{GerSing}\nocite{GerHor}\nocite{Pen60}\nocite{Pen63}\nocite{Pen64}\nocite{Pen67}\nocite{Esp99}\nocite{Car71}\nocite{Ger70b}\nocite{Fra2000}
\nocite{Ger71}\nocite{Sachs1}\nocite{BSF}\nocite{Pen82}\nocite{Held70}\nocite{Oblak16}\nocite{McC1}\nocite{McC2}\nocite{McC3}\nocite{McC4}\nocite{McC5}

\end{document}